\documentclass[acmsmall,screen]{acmart}

\startPage{1}

\usepackage{booktabs}   %
\usepackage{subcaption} %

\usepackage{algorithmicx}
\usepackage{algpseudocode}
\usepackage{algorithm}

\usepackage{listings}

\makeatletter
\let\old@lstKV@SwitchCases\lstKV@SwitchCases
\def\lstKV@SwitchCases#1#2#3{}
\makeatother
\usepackage{lstlinebgrd}
\makeatletter
\let\lstKV@SwitchCases\old@lstKV@SwitchCases

\lst@Key{numbers}{none}{%
    \def\lst@PlaceNumber{\lst@linebgrd}%
    \lstKV@SwitchCases{#1}%
    {none:\\%
     left:\def\lst@PlaceNumber{\llap{\normalfont
                \lst@numberstyle{\thelstnumber}\kern\lst@numbersep}\lst@linebgrd}\\%
     right:\def\lst@PlaceNumber{\rlap{\normalfont
                \kern\linewidth \kern\lst@numbersep
                \lst@numberstyle{\thelstnumber}}\lst@linebgrd}%
    }{\PackageError{Listings}{Numbers #1 unknown}\@ehc}}
\makeatother

\usepackage{array}
\newcolumntype{R}[1]{>{\raggedleft\let\newline\\\arraybackslash\hspace{0pt}}m{#1}}

\usepackage{mathpartir}

\usepackage{color}

\usepackage{wrapfig}

\usepackage{cleveref}

\usepackage{amsthm}

\usepackage{tikz}
\usepackage{marginalia}

\usepackage{pifont}

\usepackage{pgf}

\usepackage{multirow}

\usepackage[normalem]{ulem}

\newtheorem{assumption}{Assumption}[section]

\theoremstyle{definition}
\newtheorem{definition}{Definition}[section]

\algblock[TryCatchFinally]{try}{endtry}
\algcblockdefx[TryCatchFinally]{TryCatchFinally}{catch}{endtry}
  [1]{\textbf{catch} #1}
  {\textbf{end try}}

\newenvironment{DIFnomarkup}{}{}

\lstdefinelanguage{muf}{
  morekeywords={
    let, in, rec, where, end, if, then, else, do, done, run, 
    fun, const, 
    or, and, not, as, mod, 
    gaussian, beta, bernoulli, binomial, delta, observe, student_t,
    gamma, invgamma, fold, fold_resample,
    var, dist, const, hd, cons
  },
  emph={symbolic, sample},
  emphstyle={\color{@red}},
  sensitive=true,
  morestring=[b]",
  escapechar=\%,
  columns=fullflexible,
  keepspaces=true,
  basicstyle=\ttfamily,
  comment=[l][\color{gray}]{\#},
  mathescape=true
}

\def\zl{\lstinline[basicstyle=\small\ttfamily\color{blue!50!black}]}
\def\zlm{\lstinline[basicstyle=\footnotesize\ttfamily\color{blue!50!black}]}

\definecolor{@red}{RGB}{222, 62, 70}
\definecolor{@lightblue}{RGB}{129, 140, 166}
\definecolor{@green}{RGB}{112, 130, 0}
\definecolor{@lightlightblue}{RGB}{236, 243, 248}

\newcommand{\symbolic}[0]{\mfSymbolic}
\newcommand{\sample}[0]{\mfSample}

\newcommand{\siren}{\textsc{Siren}}

\newcommand{\abs}[1]{\hat{#1}}

\newcommand{\var}{\textit{var}}
\newcommand{\fail}{\textbf{fail}}
\newcommand{\afail}{\widehat{\textbf{fail}}}

\newcommand{\forgetr}{\textsc{forgetr}}

\newcommand{\symvalue}{\textsc{Value}}
\newcommand{\symassume}{\textsc{Assume}}
\newcommand{\symobserve}{\textsc{Observe}}

\newcommand{\swap}{\textsc{swap}}
\newcommand{\aswap}{\widehat{\swap}}

\newcommand{\const}{\textsc{const}}
\newcommand{\aconst}{\widehat{\const}}
\newcommand{\pure}{\textsc{pure}}

\newcommand{\hoist}{\textsc{hoist}}
\newcommand{\ahoist}{\widehat{\hoist}}
\newcommand{\graft}{\textsc{graft}}
\newcommand{\agraft}{\widehat{\graft}}

\newcommand{\prune}{\textsc{prune}}
\newcommand{\aprune}{\widehat{\prune}}

\newcommand{\marginalize}{\textsc{marginalize}}
\newcommand{\amarginalize}{\widehat{\marginalize}}

\newcommand{\ahoisthelper}{\widehat{\textsc{hoist\_helper}}}

\newcommand{\aeval}{\widehat{\textsc{eval}}}
\newcommand{\intervene}{\textsc{intervene}}
\newcommand{\aintervene}{\widehat{\intervene}}

\newcommand{\asymassume}{\widehat{\textsc{Assume}}}
\newcommand{\asymvalue}{\widehat{\textsc{Value}}}
\newcommand{\asymobserve}{\;\widehat{\textsc{Observe}}}

\newcommand{\arandomvar}{\hat{\randomvar}}
\newcommand{\aRandomvar}{\widehat{\Randomvar}}

\newcommand{\aStuff}{\hat{\Stuff}}

\newcommand{\anormal}[2]{{\widehat{\mathcal{N}}({#1}, {#2})}}
\newcommand{\anormalbig}[2]{{\widehat{\mathcal{N}}\left({#1}, {#2}\right)}}
\newcommand{\abern}[1]{{\widehat{\mathrm{B}}({#1})}}

\newcommand{\adeltad}[1]{{\hat{\delta}({#1})}}
\newcommand{\adeltasample}[1]{\hat{\delta_{s}}({#1})}

\newcommand{\agamma}[2]{\widehat{\Gamma}({#1}, {#2})}
\newcommand{\ainvgamma}[2]{\widehat{\Gamma}^{-1}({#1}, {#2})}

\newcommand{\adistr}[1]{\hat{D}^{#1}}
\newcommand{\aDistr}{\hat{\Distr}}

\newcommand{\distr}[1]{D^{#1}}
\newcommand{\Distr}{D}

\newcommand{\Expr}{E}
\newcommand{\expr}[1]{E^{#1}}

\newcommand{\progexpr}{e}

\newcommand{\aExpr}{\hat{\Expr}}
\newcommand{\aexpr}[1]{\hat{E}^{#1}}
\newcommand{\aplus}[2]{{{#1} \; \hat{\texttt{+}} \; {#2}}}

\newcommand{\aminus}[2]{{{#1} \; \hat{\texttt{-}} \; {#2}}}
\newcommand{\amult}[2]{{{#1} \; \hat{\texttt{*}} \; {#2}}}
\newcommand{\adivision}[2]{{#1}\; \hat{/}\; {#2}}
\newcommand{\aite}[3]{{\widehat{\texttt{ite}}({#1}, {#2}, {#3})}}

\newcommand{\constant}{c}
\newcommand{\aconstant}{\hat{\constant}}
\newcommand{\val}{v}
\newcommand{\aval}{\hat{\val}}
\newcommand{\Val}{\mathbb{V}}
\newcommand{\aVal}{\hat{\mathbb{\Val}}}

\newcommand{\listhd}{l_{\mit{hd}}}
\newcommand{\listtl}{l_{\mit{tl}}}
\newcommand{\alisthd}{\hat{l}_{\mit{hd}}}
\newcommand{\alisttl}{\hat{l}_{\mit{tl}}}
\newcommand{\distop}{\ensuremath{\mit{dist}}}

\newcommand{\etop}{\mtt{TopE}}
\newcommand{\dtop}{\mtt{TopD}}
\newcommand{\eunk}[2]{\mtt{UnkE}_{#1}(#2)}
\newcommand{\dunk}[2]{\mtt{UnkD}_{#1}(#2)}
\newcommand{\cunk}{\mtt{UnkC}}

\newcommand{\dom}{\mathrm{dom}}

\newcommand{\rename}{\textsc{rename}}
\newcommand{\narrowjoin}{\textsc{rename\_join}}

\newcommand{\rvnames}{\mit{\widehat{RV}_{canon}}}
\newcommand{\rvname}[1]{\mit{\widehat{RV}^{#1}}}

\newcommand{\initialize}{\textsc{initialize}}
\newcommand{\ainitialize}{\widehat{\initialize}}
\newcommand{\realize}{\textsc{realize}}
\newcommand{\arealize}{\widehat{\realize}}

\newcommand{\score}{\textsc{score}}
\newcommand{\conjdistr}{\textsc{conjugate\_dist}}
\newcommand{\aconjdistr}{\widehat{\conjdistr}}
\newcommand{\scoreval}{s}

\newcommand{\annotation}{\kappa}
\newcommand{\aannotation}{\hat{\annotation}}

\newcommand{\benchmark}[1]{\emph{#1}}
\newcommand{\bOutlier}{\benchmark{Outlier}}
\newcommand{\bEnvnoise}{\benchmark{EnvNoise}}
\newcommand{\bGtree}{\benchmark{Tree}}
\newcommand{\bOutlierheavy}{\benchmark{OutlierHeavy}}
\newcommand{\bSlds}{\benchmark{SLDS}}
\newcommand{\bRunner}{\benchmark{Runner}}
\newcommand{\bRadar}{\benchmark{Radar}}
\newcommand{\bAircraft}{\benchmark{Aircraft}}
\newcommand{\bNoise}{\benchmark{Noise}}
\newcommand{\bSlam}{\benchmark{SLAM}}
\newcommand{\bWheels}{\benchmark{Wheels}}

\newcommand{\ssi}{SSI}
\newcommand{\ds}{DS}
\newcommand{\bp}{SMC w/ BP}

\newcommand{\weight}{w}
\newcommand{\Weight}{W}

\newcommand{\categorical}{\textsc{Cat}}
\newcommand{\draw}{\textsc{Draw}}
\newcommand{\listl}{\ensuremath{l}}
\newcommand{\alistl}{\ensuremath{\hat{l}}}

\newcommand{\varset}{\ensuremath{S_{\mit{rv}}}}
\newcommand{\avarset}[1]{\ensuremath{\hat{S}_{\mit{rv}#1}}}
\newcommand{\particleset}{\ensuremath{S_{\mit{p}}}}
\newcommand{\traceset}{\ensuremath{S_{\mit{t}}}}

\algnewcommand\algorithmicswitch{\textbf{switch}}
\algnewcommand\algorithmiccase{\textbf{case}}
\algnewcommand\algorithmicassert{\texttt{assert}}
\algnewcommand\Assert[1]{\State \algorithmicassert(#1)}%
\algdef{SE}[SWITCH]{Switch}{EndSwitch}[1]{\algorithmicswitch\ #1\ \algorithmicdo}{\algorithmicend\ \algorithmicswitch}%
\algdef{SE}[CASE]{Case}{EndCase}[1]{\algorithmiccase\ #1}{\algorithmicend\ \algorithmiccase}%
\algdef{SE}[OTHERWISE]{Otherwise}{EndOtherwise}{\algorithmicelse}{\algorithmicend\ \algorithmicelse}%
\algtext*{EndSwitch}%
\algtext*{EndCase}%
\algtext*{EndOtherwise}%

\newcommand{\plan}[1]{\tiny{\textcolor{gray}{(Plan #1)}}}
\newcommand{\defaultplan}[0]{\tiny{\textcolor{gray}{(Default)}}}

\newcommand{\is}{\mathrel{\leftarrow}}
\newcommand{\mit}[1]{\ensuremath{\mathit{#1}}}
\newcommand{\mtt}[1]{\ensuremath{\texttt{\small{#1}}}}
\newcommand{\mttm}[1]{\ensuremath{\texttt{\footnotesize{#1}}}}
\newcommand{\mkw}[1]{\ensuremath{\textcolor{blue!50!black}{\mtt{#1}}}}
\newcommand{\mkwm}[1]{\ensuremath{\textcolor{blue!50!black}{\mttm{#1}}}}

\newcommand{\sem}[2]{\ensuremath{\llbracket #1 \rrbracket_{#2}}}
\newcommand{\psem}[2]{\{\mkern-4.8mu[ #1 ]\mkern-4.8mu\}_{#2}}

\newcommand{\prog}{\mit{prog}}
\newcommand{\distribution}{\textsc{Distribution}}
\newcommand{\adistribution}{\widehat{\textsc{Distribution}}}

\newcommand{\reachable}{\textsc{reachable}}

\newcommand{\mfInt}{\ensuremath{\mtt{int}}}
\newcommand{\mfReal}{\ensuremath{\mtt{real}}}
\newcommand{\mfUnit}{\ensuremath{\mtt{unit}}}
\newcommand{\amfUnit}{\ensuremath{\widehat{\mfUnit}}}

\newcommand{\mfError}{\ensuremath{\mtt{error}}}

\newcommand{\mfLetFun}[3]{\ensuremath{\mkw{let}\ #1\ \mtt{=} \ \mkw{fun} \ #2 \to \ #3}}

\newcommand{\mfPair}[2]{\ensuremath{\mtt{(} #1 \mtt{,}\ #2 \mtt{)}}}
\newcommand{\mfNil}{\ensuremath{\mkw{nil}}}

\newcommand{\mfCons}[2]{\ensuremath{\mkw{cons}\mtt{(}#1 \mtt{,}\ #2 \mtt{)}}}

\newcommand{\mfOp}[1]{\ensuremath{\mit{op}\mtt{(} #1 \mtt{)}}}

\newcommand{\mfApp}[2]{\ensuremath{#1\mtt{(} #2 \mtt{)}}}
\newcommand{\mfIf}[3]{\ensuremath{\mkw{if} \ #1 \ \mkw{then} \ #2 \ \mkw{else} \ #3}}

\newcommand{\mfFold}[3]{\ensuremath{\mkw{fold}\mtt{(}#1\mtt{,}\ #2 \mtt{,} \ #3 \mtt{)}}}

\newcommand{\mfLetIn}[3]{\ensuremath{\mkw{let}\ #1\ \mtt{=}\ #2\ \mkw{in}\ #3}}

\newcommand{\mfLetRvNoAnn}[3]{\ensuremath{\mkw{let}\ #1 \is #2\ \mkw{in}\ #3}}
\newcommand{\mfLetRv}[4]{\ensuremath{\mkw{let}\ #1 \ #2 \is #3\ \mkw{in}\ #4}}

\newcommand{\mfObserve}[2]{\ensuremath{\mkw{observe}\mtt{(}#1\mtt{,}\ #2 \mtt{)}}}

\newcommand{\mfResample}{\ensuremath{\mkw{resample}}}

\newcommand{\noannotation}{\varepsilon}
\newcommand{\mfSymbolic}{\ensuremath{\mkw{symbolic}}}
\newcommand{\amfSymbolic}{\ensuremath{\widehat{\mfSymbolic}}}
\newcommand{\mfSample}{\ensuremath{\mkw{sample}}}
\newcommand{\amfSample}{\ensuremath{\widehat{\mfSample}}}

\newcommand{\abstr}{\alpha}
\newcommand{\concret}{\gamma}

\newcommand{\Real}{{\mathbb{R}}}

\newcommand{\Randomvar}{{\mathcal{RV}}}
\newcommand{\randomvar}{{X}}

\newcommand{\programvar}{{x}}

\newcommand{\Stuff}{N}

\newcommand{\normal}[2]{{\mathcal{N}({#1}, {#2})}}
\newcommand{\bern}[1]{{\mathrm{B}({#1})}}

\newcommand{\deltad}[1]{{\delta({#1})}}
\newcommand{\SSIdeltasample}[1]{\delta_{s}(#1)}
\newcommand{\deltasample}[1]{\delta_{s}(#1)}
\newcommand{\invgamma}[2]{\Gamma^{-1}({#1}, {#2})}
\newcommand{\SSIgamma}[2]{\Gamma({#1}, {#2})}

\newcommand{\decl}{d}

\newcommand{\plus}[2]{{{#1} \; \texttt{+} \; {#2}}}

\newcommand{\minus}[2]{{{#1} \; \texttt{-} \; {#2}}}
\newcommand{\mult}[2]{{{#1} \; \texttt{*} \; {#2}}}
\newcommand{\division}[2]{{{#1}}\;/\;{{#2}}}

\newcommand{\ite}[3]{{\texttt{ite}({#1}, {#2}, {#3})}}

\newcommand{\bnfarrow}{\Coloneqq}

\newcommand{\Annotation}{A}
\newcommand{\aAnnotation}{\hat{\Annotation}}

\newcommand{\symbstate}{g}
\newcommand{\Symbstate}{G}
\newcommand{\symbstateD}[2]{\symbstate_{#1}({#2})_\mathrm{d}}
\newcommand{\symbstatePV}[2]{\symbstate_{#1}({#2})_\mathrm{a}}
\newcommand{\symbstateS}[2]{\symbstate_{#1}({#2})_\mathrm{n}}

\newcommand{\symbstateP}[2]{\symbstate_{#1}({#2})_\pi}
\newcommand{\symbstateK}[2]{\symbstate_{#1}({#2})_\kappa}

\newcommand{\av}[1]{\aval_{#1}}
\newcommand{\asymbstate}[1]{\hat{g}_{#1}}
\newcommand{\aSymbstate}{\hat{G}}
\newcommand{\asymbstateD}[2]{\asymbstate{#1}({#2})_\mathrm{d}}
\newcommand{\asymbstatePV}[2]{\asymbstate{#1}({#2})_\mathrm{a}}
\newcommand{\asymbstateS}[2]{\asymbstate{#1}({#2})_\mathrm{n}}

\newcommand{\letin}[1]{\ensuremath{\mathit{let} \; #1 \; \textit{in} \;}}

\newcommand{\pclstep}[6]{\ensuremath{{#1}, {#2} \downarrow^{#6} {#3},{#4},{#5}}}

\newcommand{\pclsstepmh}[3]{\ensuremath{{#1} \downdownarrows_{#3} {#2}}}

\newcommand{\pclsstep}[2]{\ensuremath{{#1} \downdownarrows {#2}}}
\newcommand{\pclpstep}[3]{\ensuremath{{#2} \Downarrow_{#1} {#3}}}
\newcommand{\doresample}{\ensuremath{\mathit{r}}}
\newcommand{\false}{\ensuremath{\mathit{false}}}
\newcommand{\true}{\ensuremath{\mathit{true}}}
\newcommand{\atrue}{\widehat{\true}}
\newcommand{\afalse}{\widehat{\false}}
\newcommand{\set}[1]{\ensuremath{\left\{ {#1}\right\}}}

\newcommand{\cpclstep}[2]{\ensuremath{{#1}\; \tilde{\downarrow}\; {#2}}}
\newcommand{\cpclsstep}[2]{\ensuremath{{#1}\; \tilde{\downdownarrows}\; {#2}}}
\newcommand{\cpclpstep}[2]{\ensuremath{{#1}\; \tilde{\Downarrow}\; {#2}}}
\newcommand{\cpclstepstar}[2]{\ensuremath{{#1}\; \tilde{\downarrow}^*\; {#2}}}

\newcommand{\apclstep}[4]{\ensuremath{{#1}, {#2}\; \hat{\downarrow}\; {#3},{#4}}}
\newcommand{\apclsstep}[2]{\ensuremath{{#1}\; \hat{\downdownarrows}\; {#2}}}
\newcommand{\apclpstep}[2]{\ensuremath{{#1}\; \hat{\Downarrow}\; {#2}}}

\newcommand{\weakeq}[3]{{#1} \cong_{#3} {#2}}

\newcommand{\evalset}{S_{c}}
\newcommand{\dset}{S_{D}}
\newcommand{\adset}{\hat{S}_{D}}
\newcommand{\pset}{S_{p}}
\newcommand{\apset}{\hat{S}_{p}}
\newcommand{\vset}{S_{v}}

\newcommand{\marginalized}[3]{\mtt{marginalized}(#1,#2,#3)}
\newcommand{\marginalizedroot}[1]{\mtt{marginalized}(#1)}
\newcommand{\initialized}[2]{\mtt{initialized}(#1,#2)}
\newcommand{\realized}{\mtt{realized}}

\newcommand{\amarginalized}[3]{\widehat{\mtt{marginalized}}(#1,#2,#3)}
\newcommand{\amarginalizedroot}[1]{\widehat{\mtt{marginalized}}(#1)}
\newcommand{\ainitialized}[2]{\widehat{\mtt{initialized}}(#1,#2)}
\newcommand{\arealized}{\widehat{\mtt{realized}}}

\newcommand{\topnode}[1]{\mtt{TopN}(#1)}

\newcommand{\setunk}{\textsc{set\_top}}

\setcopyright{rightsretained}
\acmDOI{10.1145/3704846}
\acmYear{2025}
\acmJournal{PACMPL}
\acmVolume{9}
\acmNumber{POPL}
\acmArticle{10}
\acmMonth{1}
\received{2024-07-07}
\received[accepted]{2024-11-07}

\begin{document}

\title{Inference Plans for Hybrid Particle Filtering}         %

\author{Ellie Y. Cheng}
\orcid{0009-0006-6128-0351}             %
\affiliation{
  \institution{MIT CSAIL}            %
  \country{USA}                    %
}
\email{ellieyhc@csail.mit.edu}          %

\author{Eric Atkinson}
\orcid{0000-0002-8396-1258}             %
\affiliation{
  \institution{Binghamton University}            %
  \country{USA}                    %
}
\email{eatkinson2@binghamton.edu}          %

\author{Guillaume Baudart}
\orcid{0000-0003-2230-1616}             %
\affiliation{
  \institution{Université Paris Cité -- CNRS -- Inria -- IRIF}           %
  \country{France}                   %
}
\email{guillaume.baudart@inria.fr}

\author{Louis Mandel}
\orcid{0000-0002-5291-6067}
\affiliation{
  \institution{IBM Research}            %
  \country{USA}                    %
}
\email{lmandel@us.ibm.com}

\author{Michael Carbin}
\orcid{0000-0002-6928-0456}
\affiliation{
  \institution{MIT CSAIL}            %
  \country{USA}                    %
}
\email{mcarbin@csail.mit.edu}

\begin{abstract}

Advanced probabilistic programming languages (PPLs) using \emph{hybrid particle filtering} combine symbolic exact inference and Monte Carlo methods to improve inference performance.
These systems use heuristics to partition random variables within the program into variables that are encoded symbolically and variables that are encoded with sampled values,
and the heuristics are not necessarily aligned with the developer's performance evaluation metrics.
In this work, we present \emph{inference plans}, a programming interface that enables developers to control the partitioning of random variables during hybrid particle filtering.
We further present \siren{}, a new PPL that enables developers to use annotations to specify inference plans the inference system must implement.
To assist developers with statically reasoning about whether an inference plan can be implemented, we present an abstract-interpretation-based static analysis for \siren{} for determining inference plan \emph{satisfiability}. We prove the analysis is sound with respect to \siren{}'s semantics.
Our evaluation applies inference plans to three different hybrid particle filtering algorithms on a suite of benchmarks.
It shows that the control provided by inference plans enables speed ups of 1.76x on average and up to 206x to reach a target accuracy, compared to the inference plans implemented by default heuristics; the results also show that inference plans improve accuracy by 1.83x on average and up to 595x with less or equal runtime, compared to the default inference plans.
We further show that our static analysis is precise in practice, identifying all satisfiable inference plans in 27 out of the 33 benchmark-algorithm evaluation settings.

\end{abstract}




\maketitle

\section{Introduction}
\label{sec:intro}

Probabilistic programming languages (PPLs) support primitives for modeling random variables and performing probabilistic inference~\citep{dippl,tolpin2016design,murray2018automated,narayanan2016probabilistic,holtzen2020scaling}. They provide high-level abstractions for probabilistic modeling that hide away the complex details of inference algorithms while leveraging common programming language constructs such as functions, loops, and control flow. 
PPLs serve as an expressive tool for solving such problems. Users can focus on modeling the problem, rather than the details of inference techniques.

\paragraph{Hybrid Inference Systems} 
\begin{wrapfigure}[22]{r}{0.56\textwidth}
  \centering
  \vspace{-3.1em}
  \includegraphics[width=0.53\textwidth]{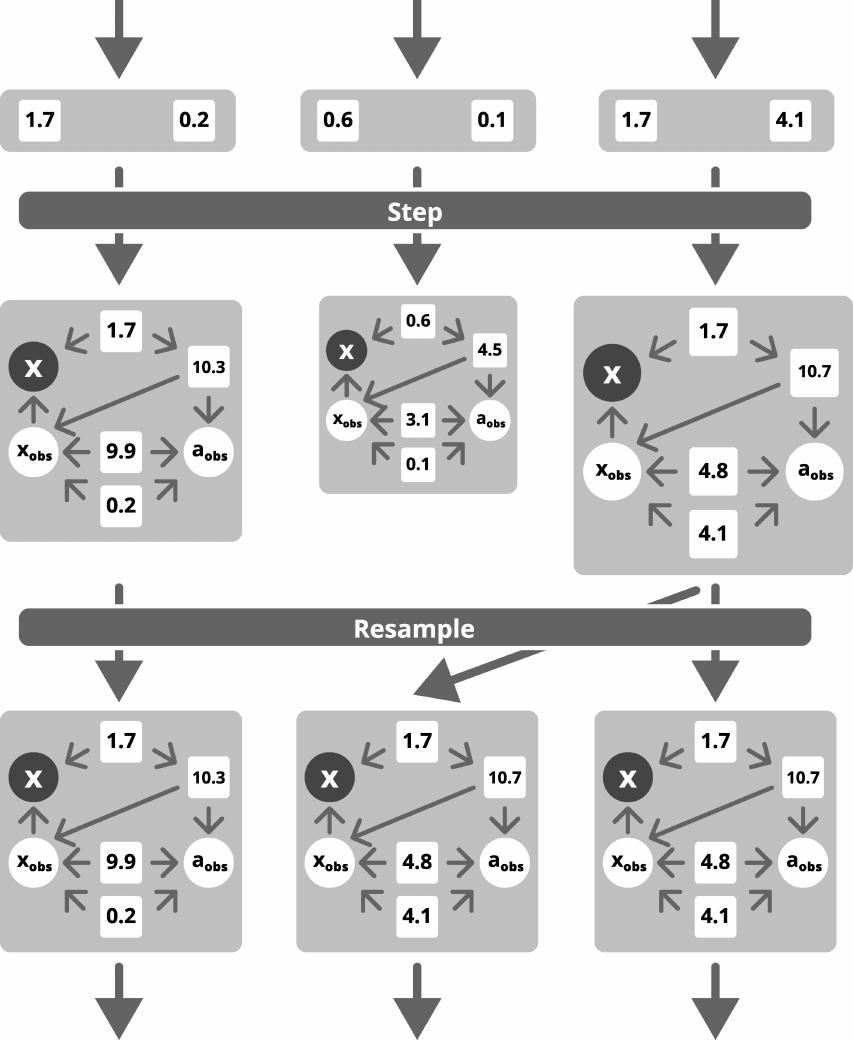}
  \vspace{-0.3em}
  \caption{Hybrid inference with particle filtering.}
  \label{fig:hybrid-inference}
\end{wrapfigure}
Hybrid inference systems -- such as delayed sampling~\citep{murray2018delayed,lunden2017delayed}, semi-symbolic inference~\citep{atkinson2022semi}, Sequential Monte Carlo with 
belief propagation~\citep{azizian2023automatic}, and automatically marginalized MCMC \citep{lai2023automatically} -- automatically incorporate exact inference with Monte Carlo methods
to improve performance. They utilize a \emph{symbolic encoding} of random variables to represent some or all parts of the model, enabling symbolic computation that 
lowers the variance of estimations. 
Hybrid inference algorithms that apply to particle filters~\citep{gordon1993novel} implement an automatic \emph{Rao-Blackwellization} of the particle filter~\citep{doucet2000rao}; hybrid inference algorithms that apply to MCMC algorithms automatically implement \emph{collapsed sampling}~\citep{liu1994collapsed}. 
In this work, we focus on hybrid inference systems that perform symbolic computations dynamically at runtime. %
These systems are able to take advantage of exact inference opportunities that only become available once the inference system replaces some variables with concrete Monte Carlo samples. 

\Cref{fig:hybrid-inference} depicts hybrid inference using particle filtering as the approximate inference method. The algorithm maintains a collection of parallel instances of executions, represented by the big boxes. Each instance contains a symbolic structure that encodes certain random variables symbolically, as shown by the circles in the diagram; other random variables are encoded as constant samples drawn from probability distributions, depicted as squares. 
The sizes of the boxes correspond to the associated \emph{weight}, which indicates how likely these instances are based on observed data (the white circles). The instances execute in parallel until they hit a checkpoint at which they are \emph{resampled}, meaning the distribution of instances is adjusted based on the weights. The arrows between each box indicate the transition or transformation of a particle from one stage to the next, including when a particle is duplicated during the resampling step.

\paragraph{Objective Oblivious Heuristics} 
When a hybrid inference system cannot solve the entire program with symbolic computation, it must partition the random variables into variables that it encodes symbolically and variables it encodes with concrete sampled values.
Different partitions make different subsets of random variables more accurate. Choosing which partition to use then depends on 1)~the probabilistic model and 2)~the  metrics used by the developer to evaluate the program. 
Hybrid inference systems automatically choose a partition to use based on the program structure using built-in heuristics. However, they are oblivious to the objectives of the developer and to how the developer measures performance. Their selected partition might not produce good inference performance as a result.
In these cases, developers need an interface for applying alternative heuristics that incorporate their evaluation objectives to achieve better performance. %

\paragraph{Inference Plans}
In this work, we present \emph{inference plans}, a programming interface that gives developers fine-grained control over how random variables are partitioned into sampled and symbolic variables during hybrid inference. 
An inference plan consists of a sequence of \emph{distribution encoding} annotations -- \mfSymbolic{} and \mfSample{} -- that specify whether the inference runtime should encode each random variable with a symbolic distribution or an approximate Monte Carlo sample. 
Our evaluation shows that the control provided
by inference plans enables speed ups of 1.76x on average and up to 206x to reach a target accuracy, compared to
the inference plans implemented by default heuristics; the results also show that inference plans improve
accuracy by 1.83x on average and up to 595x with less or equal runtime, compared to the default inference
plans.

In applications, such as control systems~\citep{burkhart1996adaptive, mehra1995adaptive,huang2017new}, robotics~\citep{doucet2000rao}, signal processing~\citep{dunik2017noise}, and data science~\citep{mansinghka2018probabilistic}, developers must tradeoff between accuracy and runtime during program development before deploying the program in production. 
By testing different inference plans during program development, developers can apply heuristics best-suited to their applications and optimize their programs on custom performance metrics.

\paragraph{Satisfiability Analysis} A key challenge in delivering the inference plans interface is that, depending on the program, the inference algorithm may not be able to maintain a random variable symbolically if the corresponding inference problem is too hard for the system to solve exactly.
Furthermore, because hybrid inference systems are dynamic and can change their behavior based on both the program inputs and random samples, the failure to maintain symbolic computation may only manifest in some executions.
This can lead to unpredictable or random accuracy degradations.
We call an inference plan \emph{unsatisfiable} in some execution if in that execution, the algorithm cannot evaluate the program while encoding the $\mfSymbolic$ annotated variables symbolically.
We present in this work an abstract-interpretation-based program analysis that can statically determine whether or not an inference plan is satisfiable in all executions for the given program. Against a suite of 11 benchmarks, using the 3 algorithms, our analysis identifies all satisfiable plans in 27 out of the 33 benchmark-algorithm evaluation settings.

\paragraph{Contributions}
In this paper, we present the following contributions:
\begin{itemize}
    \item We present \siren{}, a first-order functional PPL.
    \siren{} introduces distribution encoding annotations that programmers can use to assert an overall specification for how variables are represented by the inference algorithm; we term this specification an inference plan. \siren{} implements several existing hybrid particle filtering systems, unified via the hybrid inference interface, an extension of the symbolic interface from~\citep{atkinson2022semi}. We define the syntax and semantics of the language in Section~\ref{sec:background}. 
    \item We present an \emph{inference plan satisfiability analysis}, which determines statically if the annotated inference plan is satisfiable for all executions of a program. We formalize this analysis via abstract interpretation and present a proof of its soundness in Section~\ref{sec:analysis}.
    \item We implement the hybrid inference interface with semi-symbolic inference, delayed sampling, and Sequential Monte Carlo with belief propagation, and empirically show in \Cref{sec:eval} that inference plans speed up inference by 1.76x on average and up to 206x to reach target accuracy, compared to the default inference plans. We also show that inference plans improve accuracy by 1.83x on average and up to 595x with less or equal runtime, compared to the default plans.
    \item We empirically evaluate the precision of our analysis in \Cref{sec:eval}. Against a suite of 11 benchmarks, using the 3 algorithms, our analysis identifies all satisfiable plans in 27 out of the 33 benchmark-algorithm evaluation settings.
\end{itemize}

Inference plans enable developers to apply alternative heuristics to hybrid particle filtering algorithms within a probabilistic program. These customizations can improve inference performance, all the while maintaining the separation of probabilistic modeling and inference.

\section{Example}
\label{sec:example}

\begin{wrapfigure}[9]{r}{0.34\textwidth}
  \vspace{-1.2em}
  \centering
  \includegraphics[width=0.32\textwidth]{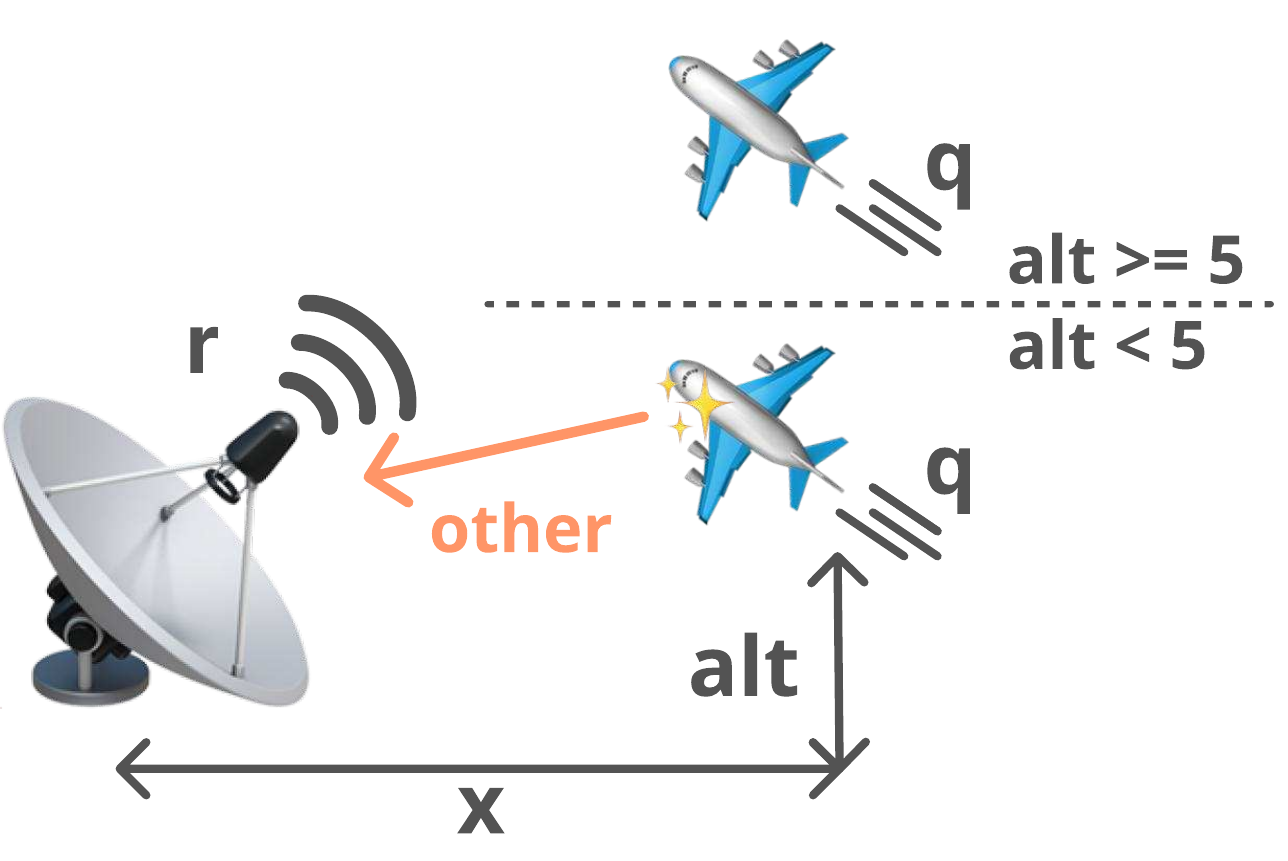}
\vspace{-0.8em}
  \caption{Diagram of a radar tracker.}
  \label{fig:example-diagram}
\end{wrapfigure}

To demonstrate how a developer can use inference plans to customize and improve inference performance, we present a simplified example adapted from \citet{bilik2010maneuvering}. 
\Cref{fig:example-diagram} shows a cartoon diagram of a radar tracker.
The goal of the radar tracker is to track the movement of an aircraft by estimating the position \zl{x} and altitude \zl{alt} of the aircraft over time.
The radar tracker can be specified as a probabilistic model. 
The model captures both the aircraft's \emph{movement noise} and also the radar's \emph{measurement noise}.
The movement noise \zl{q} captures the uncertainty in the model's belief of the aircraft movement relative to its previous position due to external factors such as turbulence.
The measurement noise captures inaccuracies in the radar measurements. For example, measurements naturally have white noise from radio interference. Additionally, when the aircraft is at lower altitudes, the measurements can be affected by spiking noise induced by electromagnetic waves reflecting off of the aircraft at an angle~\citep{bilik2010maneuvering}. In our example tracker, when the aircraft is at an altitude greater than 5, the measurement noise is modeled by only the white noise \zl{r}. Otherwise, it has additional noise \zl{other}.

\subsection{Specification in \siren{}}

\Cref{fig:example} presents two versions of the \siren{} program that implements the radar tracker. 
The programs share the same probabilistic model, but differ in the annotations used to specify random variable encodings.
Each program models the movement noise \zl{q} and white noise \zl{r} defined on \Cref{line:q-x,line:r-x} with Inverse-Gamma distributions.
The \zl{step} function defined on Lines~\ref{line:step1-x}-\ref{line:step2-x} updates the predicted position and altitude and conditions the model on radar measurement data.

On \Cref{line:x-x,line:a-x}, the \zl{x} positions and \zl{alt} altitudes are modeled as Gaussian random walks, i.e. as Gaussian distributions with a mean equal to the respective previous value. They have variance equal to \zl{q}. 
\Cref{line:other-x,line:if-x} model the measurement noise~\zl{v}. 
If the current estimated altitude \zl{alt} is less than 5, the program models the measurement noise as \zl{r+other}; the variable \zl{other} models the additional spiking noise as a separate, time-varying Inverse-Gamma distribution. Otherwise, the program models the measurement noise as only the white noise \zl{r}. 
The measured position and altitude are modeled as Gaussian distribution centered around the estimated position \zl{x} and estimated altitude \zl{alt}, respectively, and with the measurement noise as variance.

\Cref{line:obsx-x,line:obsalt-x} condition the model on the measured position being equal to the data value \zl{x_o} and on the measured altitude being equal to \zl{a_o}. The \zl{fold_resample} operation on Line~\ref{line:fold-x} iterates \zl{step} on the measurement data (provided by program inputs stored in variable \zl{data}) with initial position 0, altitude 10, and the noises \zl{q} and \zl{r}.
The program returns the final accumulator: the list of estimated positions, estimated altitudes, the movement noise, and the white noise.

\newcommand{\Hilight}{\makebox[0pt][l]{\color{lightgray}\rule[-4pt]{0.82\linewidth}{14pt}}}
\begin{figure}
\centering
\begin{subfigure}{0.46\textwidth}
\begin{lstlisting}[basicstyle=\footnotesize\ttfamily,aboveskip=0em,numbers=left,escapechar=|,linebackgroundcolor={\ifnum\value{lstnumber}=3\color{@lightlightblue}\else\ifnum\value{lstnumber}=11\color{@lightlightblue}\fi\fi}]
let step = fun ((x_o,a_o),(xs,alts,q,r)) ->|\label{line:step1-x}|
  let x0,alt0 = hd(xs),hd(alts) in
  let symbolic x <- gaussian(x0,q) in|\label{line:x-x}|
  let sample alt <- gaussian(alt0,q) in|\label{line:a-x}|
  let sample other <- invgamma(1.,10.) in|\label{line:other-x}|
  let v = if alt < 5 then r+other else r in|\label{line:if-x}|
  let () = observe(gaussian(x,v),x_o) in|\label{line:obsx-x}|
  let () = observe(gaussian(alt,v),a_o) in|\label{line:obsalt-x}|
  (cons(x,xs),cons(alt,alts),q,r)|\label{line:step2-x}|
let sample q <- invgamma(1.,1.) in|\label{line:q-x}|
let sample r <- invgamma(1.,1.) in|\label{line:r-x}|
fold_resample(step,data,([0.],[10.],q,r))|\label{line:fold-x}|
\end{lstlisting}
\caption{Annotated with \emph{Symbolic x Inference Plan}.}
\label{fig:example-annotated-x}
\end{subfigure}
\quad
\begin{subfigure}{0.46\textwidth}
\begin{lstlisting}[basicstyle=\footnotesize\ttfamily,aboveskip=0em,numbers=left,escapechar=|,linebackgroundcolor={\ifnum\value{lstnumber}=3\color{@lightlightblue}\else\ifnum\value{lstnumber}=11\color{@lightlightblue}\fi\fi}]
let step = fun ((x_o,a_o),(xs,alts,q,r)) ->|\label{line:step1-r}|
  let x0,alt0 = hd(xs),hd(alts) in
  let sample x <- gaussian(x0,q) in|\label{line:x-r}|
  let sample alt <- gaussian(alt0,q) in|\label{line:a-r}|
  let sample other <- invgamma(1.,10.) in|\label{line:other-r}|
  let v = if alt < 5 then r+other else r in|\label{line:if-r}|
  let () = observe(gaussian(x,v),x_o) in|\label{line:obsx-r}|
  let () = observe(gaussian(alt,v),a_o) in|\label{line:obsalt-r}|
  (cons(x,xs),cons(alt,alts),q,r)|\label{line:step2-r}|
let sample q <- invgamma(1.,1.) in|\label{line:q-r}|
let symbolic r <- invgamma(1.,1.) in|\label{line:r-r}|
fold_resample(step,data,([0.],[10.],q,r))|\label{line:fold-r}|
\end{lstlisting}
\caption{Annotated with \emph{Symbolic r Inference Plan}.}
\label{fig:example-annotated}
\end{subfigure}
\vspace{-0.2em}
\caption{A program written in \siren{} that implements a radar tracking. 
}
\vspace{-1em}
\label{fig:example}
\end{figure}

\subsection{Inference Plans}
Each random variable in \Cref{fig:example} has an optional \emph{distribution encoding} annotation that specifies how the inference system should encode the variable.
The inference system should encode a variable annotated with \mfSymbolic{} as a symbolic expression representing the corresponding distribution, and should encode one annotated with \mfSample{} with a concrete sample drawn from the distribution.
The two programs in \Cref{fig:example} differ in only the annotations on \Cref{line:x-x,line:r-x}. 

The annotations control the execution of \siren{}'s hybrid particle filtering implementation -- a combination of particle filtering with symbolic computation. 
Each particle contains a symbolic state, which collects the symbolic distributions of symbolically-encoded random variables. The \siren{} runtime treats variables encoded using samples as constant values during program execution.

The \Cref{fig:hybrid-inference} diagram shows an execution of the \Cref{fig:example-annotated-x} program. In one particle of the execution, the variable \zl{q} is bound to the constant value $1.7$, a random sample drawn from the Inverse-Gamma distribution on \Cref{line:q-x}. Then, on the first iteration, the symbolic random variable $\randomvar_x$ created on \Cref{line:x-x} has the symbolic distribution $\normal{0.}{1.7}$. This particle also has \zl{alt} bound to the constant sample $10.3$ and \zl{r} to $0.2$; the observed variable on \Cref{line:obsx-x} has the distribution $\normal{\randomvar_x}{0.2}$.
In another particle, \zl{q} is bound to $0.6$, \zl{r} to $0.1$, \zl{alt} to $4.5$, and \zl{other} to $3.1$. Then, $\randomvar_x$ has the symbolic distribution $\normal{0.}{0.6}$ and the \Cref{line:obsx-x} variable has $\normal{\randomvar_x}{3.2}$, as \zl{r+other} evaluates to 3.2.

The inference algorithm uses symbolic computation to evaluate the symbolic components of the particle. For example, in \Cref{fig:example-annotated-x}, the observed variable on \Cref{line:obsx-x} has a conditional distribution dependent on the symbolic variable $\randomvar_x$. The algorithm symbolically transforms the conditional distribution into a marginal distribution so that the algorithm can condition the model on the input value. 
During evaluation, the particle accumulates a weight from conditioning on input values. Then, at resampling checkpoints, a new collection of particles is resampled from the existing collection. 
In \Cref{fig:example} the resampling step occurs at the end of each \zl{fold_resample} iteration.

\subsubsection{Accuracy and Performance}

To deploy the radar tracker, the developer must ensure that it will achieve adequate performance. 
In this aircraft tracking application, the program must be within the acceptable margins of error and allowable latency or it can lead to catastrophic collisions~\citep{ali2015causal}.
This is challenging because the accuracy and runtime of hybrid particle filtering depend on both the number of particles and the inference plan used to perform inference.

\paragraph{Particle Count} While the runtime of a hybrid particle filter typically varies proportionally to the particle count, the relationship between accuracy and particle count is difficult to determine.
In general, developers need to search for a count that meets their accuracy and runtime constraints.

\paragraph{Inference Plan} A hybrid inference system's accuracy and runtime also depend on the inference plan.
This enables developers to use inference plans to control the inference system's performance and improve upon the system's default performance (i.e.\ its accuracy and runtime under the \emph{default inference plan} automatically selected by the system when no annotations are present in the program).
However, as inference plans operate by adjusting the partition between symbolic and sampled variables, the effect on accuracy and runtime is often hard to predict. 
Keeping variables symbolic \emph{can} reduce the number of particles the system requires to achieve adequate accuracy, but this behavior is not guaranteed.
Additionally, different inference plans can achieve higher accuracy for \emph{different} variables.
In an application such as radar tracking 
where highly accurate results are required within a time constraint, the developer can determine 
the program's behavior by executing it with different inference plans to build a \emph{performance profile}.
The developer can then choose the best inference plan and particle count to use in production.

\paragraph{Performance Profile}
In the tracking example, the most important variables are the target's position and altitude, so the developer would like to optimize the program for the highest accuracy for \zl{x} and \zl{alt}. 
However, the program runtime cannot exceed the acceptable latency, or the estimations will be too out of date. For this example, we assume 3 seconds as the maximum allowable runtime.

\Cref{fig:example-results} presents a performance profile for the programs in \Cref{fig:example}.
These graphs present scatter plots of accuracy and latency over a range of particle counts and for a variety of different inference plans.
Each graph uses an accuracy metric measuring the error of the inference algorithm's estimate of either \zl{x} or \zl{alt}. 
The red squares present the time and accuracy tradeoff for the default inference plan that makes no annotations. 
An alternative inference plan -- called the \emph{Symbolic x Inference Plan} because it annotates \zl{x} with $\mfSymbolic$ -- is the one from Figure~\ref{fig:example-annotated-x}, and is shown in green diamonds. 
An additional alternative plan -- called \emph{Symbolic r Inference Plan} because it annotates \zl{r} symbolic -- is the one from \Cref{fig:example-annotated}, and is shown in purple circles.

\begin{wrapfigure}[18]{r}{0.53\textwidth}
  \vspace{-1.5em}
  \centering
  \includegraphics[width=0.52\textwidth]{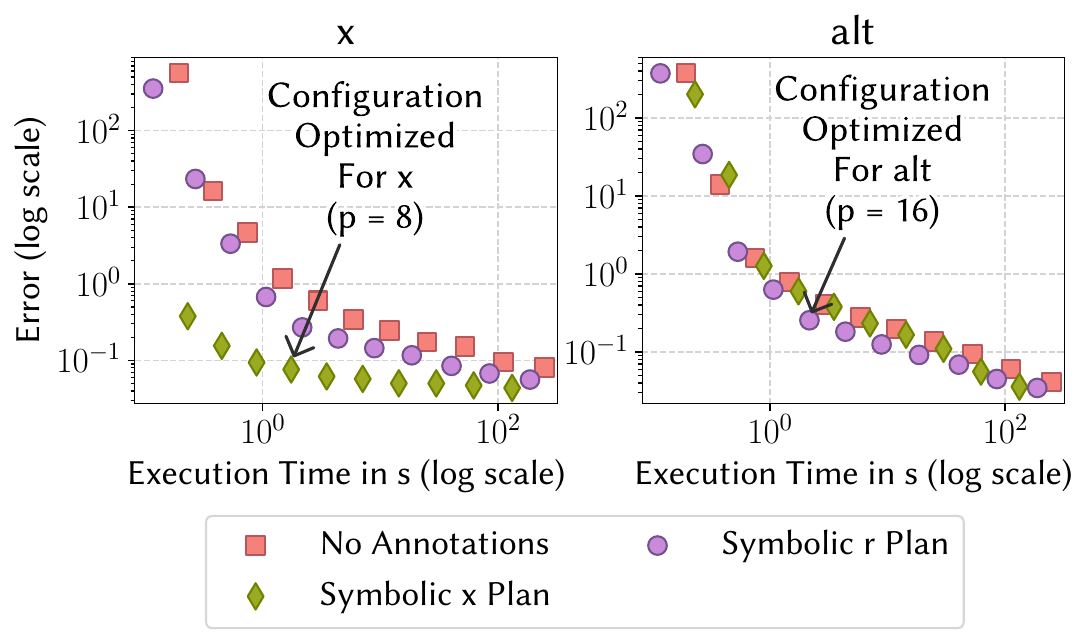}
  \caption{Accuracy and runtime performance of the Figure~\ref{fig:example} programs.
  Each scatter plot presents the program execution time and accuracy for particle counts $p$ ranging from 1 to 1024 and multiple inference plans.
  Each data point is an experiment that -- across 100 runs -- measures the median runtime and the 90th percentile of the Mean Squared Error for the relevant random variable -- \zl{x} or \zl{alt}.}
  \label{fig:example-results}
\end{wrapfigure}

\paragraph{Inference Plan Comparison} The developer can conclude from the performance profile in \Cref{fig:example-results} that, given a time constraint of 3 seconds, the best plan to optimize for \zl{x} accuracy is the \emph{Symbolic x Plan}. 
This demonstrates the value of using inference plans to control the inference system's behavior, 
as the default behavior of the system is oblivious to the developer's objectives and cannot adapt accordingly. 
Consequently, the default plan is inferior in accuracy.

The profile also shows that the \emph{Symbolic r Plan} achieves the best accuracy for \zl{alt} within the 3-second constraint. This plan also achieves better accuracy than the default plan on both \zl{x} and \zl{alt}. 
None of the plans achieves the best accuracy within 3 seconds for both variables, so the developer has to decide which is more critical to optimize for. 
The differing performance outcomes further illustrate the importance of inference plans: Developers can apply alternative inference plans to achieve the best performance on the variables they care about the most.

\subsubsection{Unsatisfiable Annotations}
According to the performance profiles, the developer may select either the \emph{Symbolic x Plan} of \Cref{fig:example-annotated-x} with $8$ particles or the \emph{Symbolic r Plan} of \Cref{fig:example-annotated} with $16$ particles as the configuration to deploy in production with acceptable accuracy and latency.
However, this profile also demonstrates the challenges in drawing conclusions from empirical data: The \emph{Symbolic r Plan} is \emph{unsatisfiable} in general.

To collect the performance profile, we generated the data for the performance profile assuming the aircraft stays at its cruising altitude of 10. If we instead generate data to model a descent where the aircraft eventually descends to an altitude below 5, the input altitude data passed to the variable \zl{a_o} on \Cref{line:step1-x} will be different. Then, the performance profile may no longer be valid.
Namely, with the original generated cruising data, the estimated altitude was never less than 5. However, with the generated descent data, the probability of the estimated altitude being less than 5 is significantly higher and, as a result, so is the probability of encountering an \emph{unsatisfiable annotation}. The $\mfSymbolic$ annotation on \zl{r} is unsatisfiable in those executions because
the inference runtime cannot evaluate the program while encoding \zl{r} symbolically.

In Figure~\ref{fig:example-annotated}, if the altitude of the aircraft is at or above 5, then the observed variables on \Cref{line:obsx-r,line:obsalt-r} have symbolic Gaussian distributions with variance $\randomvar_r$.
The Inverse-Gamma distribution of $\randomvar_r$ is conjugate with the Gaussian distribution, which means that the inference system can find a closed-form solution to the model and can encode \zl{r} symbolically.
However, if the altitude of the aircraft drops below 5, the program specifies that  
the observed variables have symbolic Gaussian distributions with variance equal to the sum of two Inverse-Gamma random variables: \zl{r} and \zl{other}.
This sum does not have an Inverse-Gamma distribution and is not conjugate with the Gaussian distributions on \Cref{line:obsx-r,line:obsalt-r}.
Without a conjugacy relationship to exploit, the inference system cannot solve the model analytically, even though the developer annotated variable \zl{r} as \mfSymbolic{}.

\begin{wrapfigure}[21]{r}{0.55\textwidth}
  \vspace{-1.5em}
  \centering
  \begin{subfigure}{0.54\textwidth}
    \centering
    \includegraphics[width=1\textwidth]{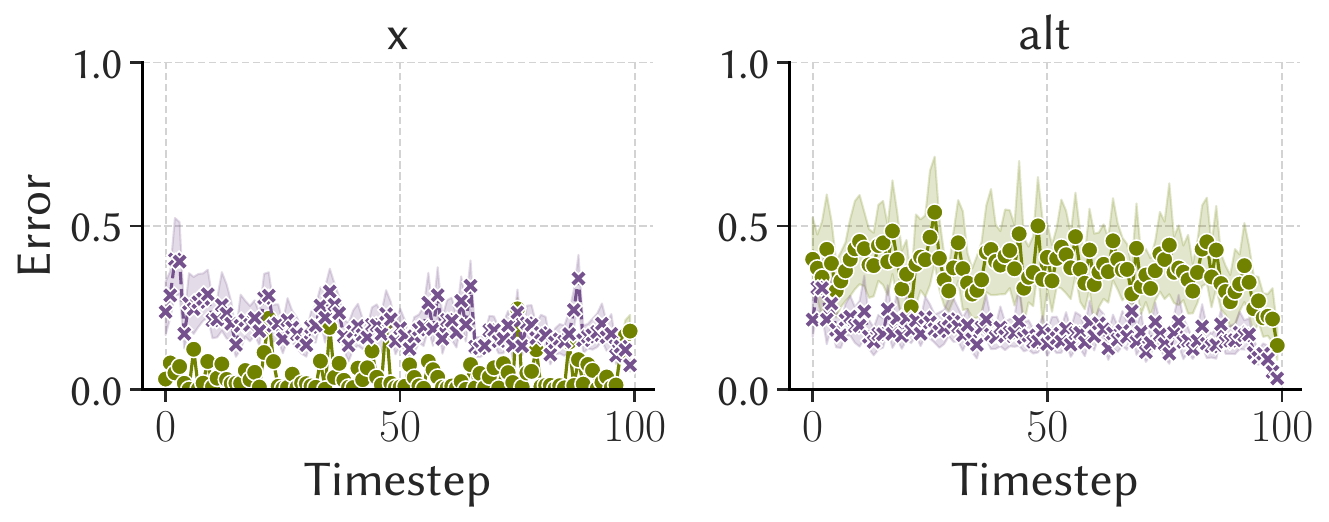}
    \vspace{-1.8em}
    \caption{Aircraft cruising at altitude 10.}
    \label{fig:example-results2-good}
  \end{subfigure}
  \begin{minipage}{0.54\textwidth}
    \hspace{0.01em}
      \begin{subfigure}{1\textwidth}
    \centering
    \includegraphics[width=1\textwidth]{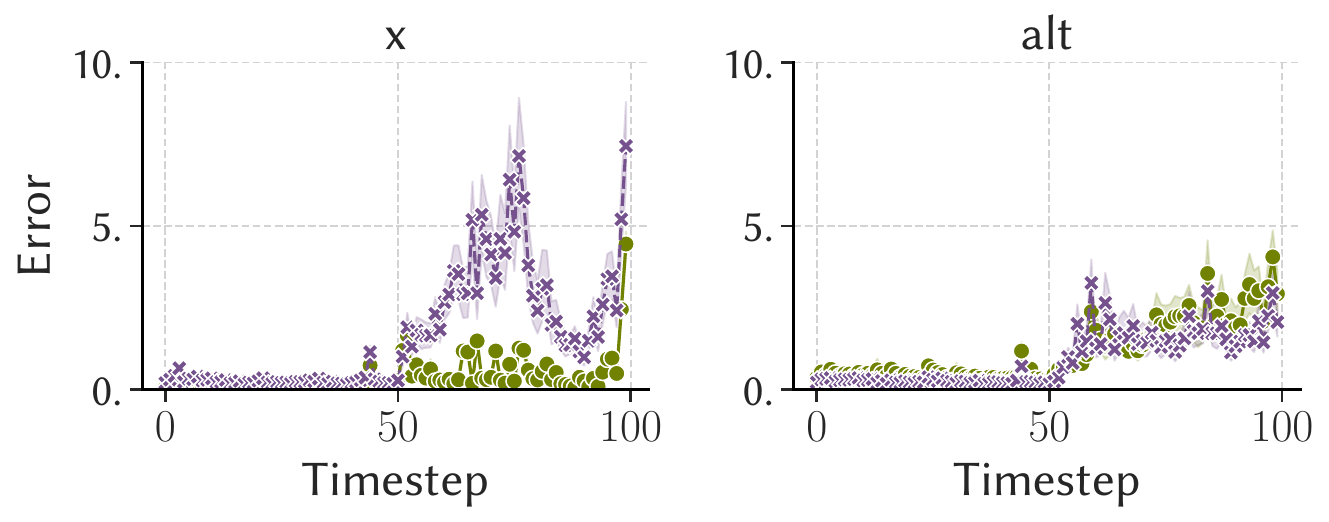}
    \label{fig:example-results2-bad}
  \end{subfigure}
  \begin{subfigure}{1\textwidth}
    \vspace{-1.5em}
    \hspace{1.8em}
    \includegraphics[width=0.92\textwidth]{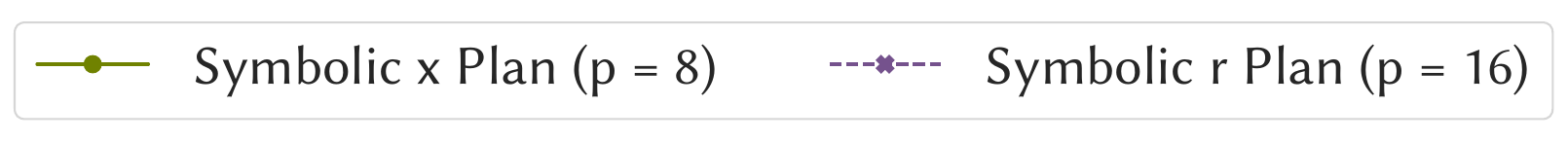}
    \vspace{-1.7em}
    \caption{Aircraft descending from altitude 10 to 0.}
    \label{fig:example-results2-bad}
  \end{subfigure}
\end{minipage}
\vspace{-0.6em}
  \caption{Accuracy of \zl{x} and \zl{alt} over 100 timesteps at altitudes, measured as the Squared Error of the estimated value to the true value at that timestep.}
  \label{fig:example-results-2}
\end{wrapfigure}

\paragraph{Dynamic Encoding Cast}
In such a scenario, \siren{}, like other hybrid inference systems \citep{atkinson2022semi,lunden2017delayed,murray2018delayed,azizian2023automatic}, will, conceptually, \emph{dynamically cast} the offending $\mfSymbolic$ annotation to a $\mfSample$ annotation, changing the underlying distribution encoding to a concrete sample. 
However, \Cref{fig:example-results-2} illustrates the impact of such a coercion. 
\Cref{fig:example-results-2} shows the accuracy of the position and altitude of a simulated aircraft
over 100 timesteps using the two programs in the two different flight modes.
In \Cref{fig:example-results2-good}, where the aircraft is cruising at altitude 10, the pattern observed in the performance profiles is replicated. The \emph{Symbolic x Plan} consistently achieves better accuracy for the \zl{x} position, and the \emph{Symbolic r Plan} for \zl{alt}. In either case, both plans have errors less than 1 for both variables at all timesteps.
However, when the aircraft is descending, this pattern is broken. In \Cref{fig:example-results2-bad}, the aircraft is now operating in a noisier environment. 
The accuracy and runtime of a probabilistic model inherently depend on the conditioned inputs, so both plans experience an error spike in \zl{alt}. The \emph{Symbolic x Plan} still maintains a relatively low error for \zl{x}, whereas the \emph{Symbolic r Plan} has a significant error spike at around 60 timesteps such that the estimation error of \zl{x} is a magnitude larger than before.
The error spike cannot be explained by the noisier environment alone, because the \emph{Symbolic x Plan} does not exhibit an \zl{x} error spike of the same magnitude.
The accuracy degradation in \zl{x} by the \emph{Symbolic r Plan} is due to a dynamic encoding cast on an unsatisfiable annotation.

\paragraph{Performance Degradation} A dynamic encoding cast has implications for inference performance, as it means that the runtime cannot implement the annotated inference plan and has to implement a different inference plan to continue execution. 
While flexible and enables the program execution to continue, it can cause imprecision that is unacceptable in certain applications.
In this example, casting the unsatisfiable annotation causes a significant decrease in accuracy in \zl{x}. The significant accuracy degradation negates the advantage the \emph{Symbolic r Plan} has over the \emph{Symbolic x Plan} for \zl{alt}. 
Given the possible impact on accuracy, the developer should use the \emph{Symbolic x Plan} instead.

\subsection{Inference Plan Satisfiability Analysis}
To enable developers to ensure their program is free from dynamic encoding casts, \siren{} performs the \emph{inference plan satisfiability analysis} to determine whether the annotated inference plan is satisfiable during all possible executions of the program using the hybrid inference algorithm. 

\paragraph{Abstract Interpretation}
The analysis uses an abstract interpretation of the program. It maintains abstract symbolic distributions of random variables and uses abstract expressions to over-approximate program executions. Consider the unsatisfiable plan from \Cref{fig:example-annotated}. On \Cref{line:if-r}, because the analysis does not know the exact value of \zl{alt}, it uses an abstract value to over-approximate the subexpressions in the branches.
If \zl{alt} is above 5, the abstract subexpression is $\arandomvar_r$, referring to the abstract random variable from \Cref{line:r-r}. If \zl{alt} is below 5, the subexpression is $\aplus{\arandomvar_r}{\cunk}$. The abstract value $\cunk$ represents some unspecified constant.
All random samples are constant values, so
\zl{other} evaluates to $\cunk$. The analysis over-approximates the subexpressions as a single abstract expression: $\eunk{}{\{\arandomvar_r\}}$ -- an unspecified abstract expression that references $\arandomvar_r$. 

\paragraph{Unsatisfiable Inference Plan}
The analysis also approximates how the system decides when to perform symbolic computations.
The \Cref{line:obsx-r} variable has the abstract distribution $\anormal{\cunk}{\eunk{}{\{\arandomvar_r\}}}$. The variance is the expression computed on \Cref{line:if-r}. 
The inference algorithm cannot perform symbolic computations when the variance is a complex expression. Because \zl{r} (corresponding to $\arandomvar_r$) is annotated with \mfSymbolic{}, the variable will be dynamically cast in those instances. The analysis rejects the program, correctly identifying that the inference plan is unsatisfiable in some executions.

\paragraph{Satisfiable Inference Plan}
The satisfiability analysis will always reject unsatisfiable inference plans, but it may be overly conservative and erroneously reject satisfiable plans.  
Nevertheless, the analysis is precise in practice. For example, it correctly determines the program in \Cref{fig:example-annotated-x} is satisfiable.  
At \Cref{line:obsx-r}, the observed variable has the symbolic distribution $\anormal{\arandomvar_x}{\cunk}$, where $\arandomvar_x$ refers to the symbolic random variable created on \Cref{line:x-r}. The variance is an unknown constant, the mean is a linear expression, and $\arandomvar_x$ also has a Gaussian symbolic distribution. 
This model only consists of linear-Gaussian distributions -- a class of probabilistic models that the inference algorithm can solve entirely symbolically. Thus, \zl{x} will always be encoded symbolically as the annotation requires, so the analysis accepts the inference plan.

\subsection{Summary}

Using annotations, developers can select alternative inference plans that are better aligned with their objectives. With the satisfiability analysis, the developer can further guarantee the inference plan will always be satisfiable in any execution. 
If \siren{} successfully compiles a program, then any variable annotated with \mfSymbolic{} will always be represented symbolically in any execution, and any variable annotated with \mfSample{} will always be sampled. 
This provides developers with the guarantee that their program performance will not be degraded by unsatisfiable annotations.

\section{Language Syntax and Semantics}
\label{sec:background}

We present the syntax and semantics of the first-order functional PPL, \siren{}, adapted from \citet{staton17}. We have extended the language to support distribution encoding annotations and adapted it to use hybrid inference, and specify its semantics via the \emph{hybrid inference interface}.

\subsection{Syntax}

\begin{wrapfigure}[11]{r}{0.48\textwidth}
\small
\vspace{-3em}
\centering
$
\begin{array}{@{}r@{\ }r@{\ }l@{}}
\\[0.5em]
\val &::=& c \mid \programvar \mid \mfPair{\val}{\val} \mid \mfOp{\val} \mid \mfNil \mid \mfCons{\val}{\val}
\\[0.5em]
\progexpr &::=& \val \mid \mfApp{f}{\val} \mid \mfIf{\val}{\progexpr}{\progexpr} \\
  &&\mid \mfLetIn{\programvar}{\progexpr}{\progexpr} 
  \mid \mfFold{f}{\val}{\val} \\
  &&\mid \mfLetRv{\annotation}{\programvar}{\mfOp{\val}}{\progexpr}
  \mid \mfObserve{\val}{\val}\\
  &&\mid \mfResample{}
\\[0.5em]
\annotation &::=& \varepsilon \mid \mfSymbolic \mid \mfSample
\\[0.5em]
\decl &::=& \mfLetFun{f}{\programvar}{\progexpr}
\\[0.5em]
\prog &::=& \decl^* \ \progexpr
\end{array}
$
\vspace{-0.7em}
    \caption{Syntax of the \siren{} language.}
    \label{fig:muf-syntax}
\end{wrapfigure}

\Cref{fig:muf-syntax} presents the syntax of \siren{}.
An expression $\progexpr$ is a value $\val$ (constant; variable; pair; or the application of an operator, e.g., arithmetic operation, distribution, or list), a function application, a conditional, or a local definition.
We add the classic \mkw{fold} operator as well.
\siren{} supports probabilistic operators. The expression 
\mfLetRvNoAnn{\programvar}{\mit{op}(\val)}{e} introduces a new local random variable $\programvar$ with distribution $\mit{op}(\val)$ to be used in $\progexpr$.
Optionally, a \mfSymbolic{} or \mfSample{} annotation adorns a random variable declaration. 
The \mfObserve{\val_1}{\val_2}{} expression conditions the model on a variable with distribution $\val_1$ having value $\val_2$.
The \mfResample{} operator instructs the program to perform resampling for particle filtering. The \zl{fold_resample} operation used in \Cref{sec:example} is syntactic sugar for applying \mfResample{} at the end of all \zl{fold} iterations. 
Finally, a program is a sequence of function declarations $\decl$ with a main expression.

\subsection{Operational Semantics}
\label{sec:operational-semantics}

While \siren{} has an ideal measure-based semantics (see Appendix~\ref{appendix:idealsem}),
the measure is, in general, intractable.
An alternative is to interpret the model as a \emph{weighted sampler} that returns a value and a \emph{score} measuring the likelihood of the result with respect to the model.
To approximate the posterior distribution, a weighted sampler launches a set of independent executions of the sampler, the \emph{particles}, and returns
a categorical distribution that associates each value with its score. 
\siren{} implements a \emph{particle filter} which occasionally resamples the set of particles according to their score during executions. 
Following~\cite{LundenBB21}, we add an explicit \mfResample{} operator to the language to enable programs to explicitly trigger resampling.\footnote{Automatic selection of resampling locations for optimal performance is an open problem~\citep{LundenBB21}.}
We present the operational semantics of \siren{} which is a big-step semantics extended with checkpoints for resampling. 

\subsubsection{Hybrid Inference Interface}
Hybrid particle filtering algorithms reduce variance in particle filters by computing closed-form distributions where possible and only drawing random samples if symbolic computation fails. 
We first present definitions for the hybrid inference interface, an extension of the symbolic interface~\citep{atkinson2022semi}, that underpins our operational semantics. 

\paragraph{Symbolic Expressions}

\Cref{fig:symb-grammar} presents the grammar of symbolic expressions used by algorithms implementing the hybrid inference interface.
It specifies a grammar of distributions $\Distr$ that includes, but is not limited to, Gaussian, Bernoulli, Inverse-Gamma, 
and Dirac Delta distributions. $\Distr$ can also be a sampled-Delta distribution (denoted as~$\delta_s$), which is a Dirac Delta distribution that is used only to represent a sample drawn from a probability distribution.
Figure~\ref{fig:symb-grammar} further specifies a grammar of expressions $\Expr$ that uses operators to combine constant values $\constant$ and random variables $\randomvar$.
The operators include standard arithmetic and comparison operators and a conditional operator \texttt{ite}.

\paragraph{Symbolic State}
A symbolic state $\symbstate \in \Symbstate = \Randomvar \rightarrow \Annotation \times \Distr \times \Stuff$ is a finite mapping whose domain is the set of random variable names. 
It maps each random variable to an entry $\symbstate(\randomvar)$ consisting of an optional annotation ($\Annotation = \set{\mfSample, \mfSymbolic, \varepsilon}$ where $\varepsilon$ represents no annotation) denoted $\symbstatePV{}{\randomvar}$, a symbolic representation of a distribution $\symbstateD{}{\randomvar}$, and an implementation-specific data field $\symbstateS{}{\randomvar}$.

\begin{figure}
    \small
    \centering
    \begin{minipage}[b][][b]{.6\textwidth}
      \centering
      \begin{align*} 
        \Distr \bnfarrow \; & \normal{\Expr}{\Expr}\;
        \mid \; \bern{\Expr}\;
        \mid \; \invgamma{\Expr}{\Expr} \;
        \mid \; \deltad{\Expr} \;
        \mid \; \deltasample{\Expr} \;
        \mid \; \dots \\
        \Expr \bnfarrow \; & \constant\;
        \mid \; \randomvar\;
        \mid \; \plus{\Expr}{\Expr}\;
        \mid \; \minus{\Expr}{\Expr}\;
        \mid \; \mult{\Expr}{\Expr}\\
        &\mid \; \division{\Expr}{\Expr}\;
        \mid \; \ite{\Expr}{\Expr}{\Expr}\;
        \mid \; \Expr \; \mathit{Cmp} \; \Expr\\
        \mathit{Cmp} \bnfarrow \; & \texttt{=} \;
        \mid \; \texttt{!=} \;
        \mid \; \texttt{<}  \;
        \mid \; \texttt{<=} \qquad
        \constant \in \Val, \randomvar \in \Randomvar
        \end{align*}
        \vspace{-2em}
        \captionof{figure}{Grammar of symbolic expressions.}
        \label{fig:symb-grammar}
    \end{minipage}%
    \begin{minipage}[b][][b]{.4\textwidth}
        \centering
        \begin{align*}
            \symassume &: \Annotation \times \Distr \times \Symbstate 
                \rightarrow \Randomvar \times \Symbstate\\
            \symobserve &: \Randomvar \times \Val \times \Symbstate 
                \rightarrow \Symbstate \times \Real\\
            \symvalue &: \Randomvar \times \Symbstate 
                \rightarrow \Val \times \Symbstate\ %
        \end{align*}
        \vspace{-2em}
        \captionof{figure}{Hybrid Inference Interface.}
        \label{fig:hybrid-interface}
    \end{minipage}
\vspace{-1.5em}
\end{figure}

\paragraph{Interface} The hybrid inference interface uses three operations to manipulate the symbolic state, shown in \Cref{fig:hybrid-interface}.
The $\symassume$ operation takes an annotation, a distribution, and a symbolic state, and returns a new random variable with the updated symbolic state. The $\symobserve$ operation conditions the symbolic state on the input variable having the given value and returns the updated state and a score for the particle filter to use as the weight. The $\symvalue$ operation replaces the input variable with a sample from its distribution, turning it into a sampled-Delta distribution. 
These operations decide whether the runtime samples a random variable or encodes a random variable symbolically in the symbolic state. By doing so, they determine the default inference plan in the absence of annotations.
We will discuss different implementations of the interface in \Cref{sec:instantiations}.

\subsubsection{Unsatisfiable Annotation}
When a distribution encoding annotation is unsatisfiable, the \siren{} runtime performs a dynamic encoding cast, enabling the execution to continue. In particular, the $\symvalue$ operation executes as if the annotation of the input random variable is $\varepsilon$ even if the annotation of the input random variable is $\mfSymbolic$.

\subsubsection{Big-Step Semantics with Checkpoints}

Next, we present the semantics of \siren{}. \Cref{fig:op:sem-particle,fig:op:sem-set,fig:op:sem-model} show a fragment, and the full semantics is in \Cref{appendix:op-sem}. A particle is represented by a pair (expression, symbolic state). A \siren{} program is described by three types of rules: 
\begin{itemize}
    \item \emph{Particle Evaluation.} The evaluation relation $\pclstep{\progexpr{}}{\symbstate}{\progexpr{}'}{\symbstate'}{\weight}{\doresample}$ evaluates a particle $(\progexpr{}, \symbstate)$ and returns an updated particle $(\progexpr{}', \symbstate')$, the associated score $\weight$, and a resample flag $\doresample$ indicating if the evaluation was interrupted.
    \item \emph{Particle Set Evaluation.} The evaluation relation $\pclsstep{\set{\progexpr{}_i, \symbstate_i}_{1 \le i \le N}}{\distr{}}$ gathers the results of a set of particles into a distribution $\distr{}$.
    \item \emph{Model Evaluation.} The evaluation relation $\pclpstep{N}{\progexpr{}}{\distr{}}$ evaluates a program expression $\progexpr{}$ into a distribution $\distr{}$ using $N$ particles.
\end{itemize}

\begin{figure}
\begin{small}
\begin{mathpar}
\inferrule%
{ }
{\pclstep{\val}{\symbstate}
         {\val}{\symbstate}{1}{\false}}

\inferrule%
{ }
{\pclstep{\mfResample}{g}
         {\mfUnit}{g}{1}{\true}}

\inferrule%
{\pure(\progexpr{}_1, \progexpr{}_2) \\
\neg\const(\val)\\\\
 \pclstep{\progexpr{}_1}{\symbstate}
         {\val{}_1}{\symbstate_1}{1}{\false} \\
 \pclstep{\progexpr{}_2}{\symbstate_1}
         {\val{}_2}{\symbstate'}{1}{\false}
}
{\pclstep{\mfIf{\val}{\progexpr{}_1}{\progexpr{}_2}}{\symbstate}
         {\ite{\val}{\val_1}{\val_2}}{\symbstate'}{1}{\false}}

\inferrule%
{\pclstep{\progexpr_1}{\symbstate}
         {\progexpr_1'}{\symbstate'}{\weight}{\true}}
{\pclstep{\mfLetIn{\programvar}{\progexpr_1}{\progexpr_2}}{\symbstate}
         {\mfLetIn{\programvar}{\progexpr_1'}{\progexpr_2}}{\symbstate'}{\weight}{\true}}

\inferrule%
{\pclstep{\progexpr_1}{\symbstate}
         {\val_1}{\symbstate_1}{\weight_1}{\false}\\
 \pclstep{\progexpr_2[\programvar \leftarrow \val_1]}{\symbstate_1}
         {\progexpr_2'}{\symbstate_2}{\weight_2}{\doresample}}
{\pclstep{\mfLetIn{\programvar}{\progexpr_1}{\progexpr_2}}{\symbstate}
         {\progexpr_2'}{\symbstate_2}{\weight_1 * \weight_2}{\doresample}}

\inferrule%
{
    {\pclstep{\mfLetIn{\programvar}{\mfApp{f}{\mfPair{\listhd}{v}}}{\mfFold{f}{\listtl}{\programvar}}}{\symbstate}
         {\progexpr}{\symbstate'}{\weight}{\doresample}}
}
{\pclstep{\mfFold{f}{\mfCons{\listhd}{\listtl}}{\val}}{\symbstate}
         {\progexpr}{\symbstate'}{\weight}{\doresample}}

\inferrule%
{\annotation \in \{\varepsilon, \mfSymbolic\}\\
\symassume(\annotation, \mfApp{\distop}{\val}, \symbstate) = \randomvar, \symbstate_\randomvar\\
 \pclstep{\progexpr[\programvar \leftarrow \randomvar]}{\symbstate_\randomvar}
         {\progexpr'}{\symbstate'}{\weight}{\doresample}}
{\pclstep{\mfLetRv{\annotation}{\programvar}{\mfApp{\distop}{\val}}{\progexpr}}{\symbstate}
         {\progexpr'}{\symbstate'}{\weight}{\doresample}}

\inferrule%
{\symassume(\mfSample, \mfApp{\distop}{\val}, \symbstate) = \randomvar, \symbstate_\randomvar\\
\symvalue(\randomvar, \symbstate_\randomvar) = \val_\programvar, \symbstate_\randomvar'\\
 \pclstep{\progexpr[\programvar \leftarrow \val_\programvar]}{\symbstate_\randomvar'}
         {\progexpr'}{\symbstate'}{\weight}{\doresample}}
{\pclstep{\mfLetRv{\mfSample}{\programvar}{\mfApp{\distop}{\val}}{\progexpr}}{\symbstate}
         {\progexpr'}{\symbstate'}{\weight}{\doresample}}

\inferrule%
{\symassume(\varepsilon, \mfApp{\distop}{\val_1}, \symbstate) = \randomvar, \symbstate_\randomvar\\
 \symvalue^*(\val_2, \symbstate_\randomvar) = \val, \symbstate_\val\\
 \symobserve(\randomvar, \val, \symbstate_\val) = \symbstate', \weight}
{\pclstep{\mfObserve{\mfApp{\distop}{\val_1}}{\val_2}}{\symbstate}
         {\mfUnit}{\symbstate'}{\weight}{\false}}
\end{mathpar}
\end{small}
\vspace{-0.8em}
\caption{Fragment of the particle evaluation rules. The full semantics is in \Cref{appendix:op-sem}.}
\label{fig:op:sem-particle}
\vspace{-0.8em}
\end{figure}

\paragraph{Particle Evaluation} \Cref{fig:op:sem-particle} shows a fragment of the particle evaluation rules.
A constant $\val$ is already fully reduced so the resample flag $\doresample$ is set to \false{}. Since it is a deterministic value, the associated score is~$1$. 
The $\mfResample$ operator interrupts reductions by setting the resample flag $\doresample$ to \true{}; it reduces to the unit value. 
The semantics of $\mfIf{\val}{\progexpr_1}{\progexpr_2}$ consists of two cases.
If~$\progexpr_1$ and~$\progexpr_2$ are pure (i.e. they do not perform any \mkw{observe} or \mkw{resample}) and the condition $\val$ is not a constant (i.e. it contains some symbolic random variables) then the rule reduces to an \mkw{ite} symbolic expression used for symbolic computation.
Since there is no \mkw{observe}, the score is~$1$.
Otherwise, the rule evaluates the condition to a constant value and uses it to decide which branch to execute. This rule is elided in \Cref{fig:op:sem-particle}, but can be found in \Cref{appendix:op-sem}. The \siren{} semantics uses the $\symvalue^*$ helper operation to evaluate an expression to a constant by calling $\symvalue$ on all random variables in the expression before evaluation.
The semantics of a local declaration $\mfLetIn{\programvar}{\progexpr_1}{\progexpr_2}$ depends on \mfResample{} operators. 
If there is a \mfResample{} in $\progexpr_1$ the first rule reduces $\progexpr_1$ up to the first \mfResample{} and stops the evaluation (i.e. the resample flag is set to \true{}). 
Otherwise, $\progexpr_1$ fully reduces without interruption, and the total score is the product of the score of $\progexpr_1$ and $\progexpr_2$.
The expression $\mfFold{f}{\listl}{\val}$ is standard; it evaluates to the accumulator $\val$ if the list is empty. Otherwise, the semantics evaluates the function call of $f$ on the first element of $\listl$ and the accumulator $\val$ and recurses on the rest of the list. We include the rules of function calls, impure $\mkw{if}$-expressions, and applying \mkw{fold} to an empty list in \Cref{appendix:op-sem}.

A random variable $\mfLetRv{\annotation}{\programvar}{\mfApp{\distop}{\val}}{\progexpr}$ uses $\symassume$ to create a new random variable in the symbolic state. 
The annotation $\mfSample$ denotes that the variable must be encoded using samples, $\mfSymbolic$ denotes it must be encoded symbolically, and $\varepsilon$ denotes that the runtime should decide. If a random variable is annotated $\mfSample$, the operation $\symvalue$ draws a random sample from the variable's distribution and updates the symbolic state.
The \mkw{observe} operator uses $\symassume$ to create a new random variable without annotations in the symbolic state and uses $\symobserve$ to condition the random variable and compute the score of the particle. 
The value to condition on must be a constant value, so the rule uses $\symvalue^*$ to turn $\val_2$ into a constant value.

\begin{figure}
\begin{small}
\begin{mathpar}
\inferrule{
  \set{ \pclstep{\progexpr_i}{\symbstate_i}{\val_i}{\symbstate_i'}{\weight_i}{\false} }_{1 \le i \le N}\\
  \set{\distribution(\val_i, \symbstate_i') = \distr{}_i}_{1 \le i \le N}\\
  \Weight = \textstyle{\sum_{1 \le i \le N}} \weight_i
}
{\pclsstep{\set{\progexpr_i, \symbstate_i}_{1 \le i \le N}}{\textstyle{\sum_{1 \le i \le N}} \dfrac{\weight_i}{\Weight} \times \distr{}_i }}

\inferrule{
  \set{ \pclstep{\progexpr_i}{\symbstate_i}
                {\progexpr_i'}{\symbstate_i'}{\weight_i}{\doresample_i} }_{1 \le i \le N}\\
  \textstyle{\bigvee_{1 \le i \le N}}{\doresample_i} \\
  \mu = \categorical\left(\set{ \weight_i, (\progexpr_i', \symbstate_i') }_{1 \le i \le N}\right) \\
  \pclsstep{\set{\draw(\mu)}_{1 \le i \le N}}{\distr{}}
}
{\pclsstep{\set{\progexpr_i, \symbstate_i}_{1 \le i \le N}}{\distr{}}}
\end{mathpar}
\end{small}
\vspace{-0.6em}
\caption{Particle set evaluation rules.}
\label{fig:op:sem-set}
\vspace{-0.9em}
\end{figure}

\paragraph{Particle Set Evaluation}

\Cref{fig:op:sem-set} shows the particle set evaluation rules. 
The rules handle resuming execution from checkpoints on a set of $N$ particles.
If all the particles finish execution (i.e. the resample flag is \false{}), 
the rule gather the computed distributions into a mixture distribution where $\distribution(\val, \symbstate)$ returns the distribution of $\val$ with respect to the symbolic state~$\symbstate$ and $\weight_i / \Weight$ are the normalized scores.
Otherwise, the rule builds a categorical distribution $\mu$ of particles using the weights and resamples a fresh set of particles $\set{\draw(\mu)}$ before resuming the execution. 

\paragraph{Model Evaluation}
\begin{wrapfigure}[5]{r}{0.3\textwidth}
\centering
\vspace{-1.5em}
\begin{minipage}[b][][b]{.3\textwidth}
\begin{small}
\begin{mathpar}
\inferrule{
    \pclsstep{\set{\progexpr, \emptyset}_{1 \le i \le N}}{\distr{}}
}
{\pclpstep{N}{\progexpr}{\distr{}}}
\end{mathpar}
\end{small}
\end{minipage}
\vspace{-0.5em}
\caption{Model evaluation.}
\label{fig:op:sem-model}
\end{wrapfigure}

To evaluate the expression $\progexpr$ with a set of $N$ particles, the model evaluation rule launches the particle set evaluation with $N$ independent particles $(\progexpr, \emptyset)$ where each particle starts with an initially empty symbolic state. The program evaluates to a distribution $\distr{}$. \Cref{fig:op:sem-model} shows the rule.

\subsection{Implementing the Hybrid Inference Interface}
\label{sec:instantiations}
The \siren{} semantics enables inference plans to be used with different hybrid particle filtering algorithms through the hybrid inference interface. 
The hybrid inference interface serves as a barrier between two halves of hybrid particle filtering algorithms. Above the interface is the implementation of the programming language and the particle filtering component of the algorithm (\Cref{fig:op:sem-particle,fig:op:sem-set,fig:op:sem-model}). Below the interface, i.e. the implementation of the interface operations from \Cref{fig:hybrid-interface}, is the specification of how to apply symbolic computation.
Developers can extend the \siren{} runtime with a new hybrid particle filtering algorithm by implementing only the interface operations, without needing to re-implement the programming language or the particle filter.

To illustrate how the interface can be implemented, 
we present the implementation for two algorithms -- semi-symbolic inference~\citep{atkinson2022semi} and delayed sampling~\citep{murray2018delayed,lunden2017delayed} -- to illustrate 1)~how to detect opportunities for symbolic computation during inference and 2)~how to extend the symbolic state with additional runtime information. We note that the interface is more general; for evaluation in \Cref{sec:eval}, we implement a third inference algorithm -- Sequential Monte Carlo with belief propagation~\citep{azizian2023automatic} -- that exhibits both features. 
The full implementations of these algorithms are quite complex, so we defer the full details to the respective works. 
In this section, we focus on presenting only details that pertain to either 1)~how an inference algorithm's internal control flow relates to its inference plans, and 2)~providing a semantic foundation for the analysis and its soundness (in particular, for Section~\ref{sec:abs-instantiation}). 

\subsubsection{Semi-symbolic Inference}
\label{sec:ssi}
Semi-symbolic inference (SSI) implements the hybrid inference interface using a series of helper functions.
For example, SSI defines the $\symvalue$ operation -- which replaces a random variable with a sample from its distribution -- using the $\hoist{}$ and $\intervene{}$ helper functions. The $\hoist{}$ function manipulates the given random variable to have no parent variables and $\intervene{}$ updates the variable with the sample. 
The $\symvalue$ operation 
is defined as:
\begin{align*}
\symvalue(\randomvar, \symbstate) = \;
        \mit{let} \; \symbstate' = \hoist{}(\randomvar, \symbstate) \; \mit{in}\; 
        \mit{let} \; \val = \draw(\symbstate'(\randomvar)_{d}) \mit{in} \;        
        (\val, \intervene(\randomvar, \deltasample{\val}, \symbstate'))
\end{align*}

We defer a full discussion of the implementation of $\symvalue$, $\hoist{}$, and other helper functions to \citep{atkinson2022semi}.
However, the SSI implementation of the hybrid inference interface depends on a key core operation called $\swap$.
A partial definition of $\swap$ is as follows:
\begin{align*}
    \swap(\randomvar_1, \randomvar_2, \symbstate) = \; & 
    \begin{array}[t]{@{}l@{}}
        \mit{match}\; \symbstateD{}{\randomvar_1}, \symbstateD{}{\randomvar_2}\; \mit{with}\\
        \begin{array}[t]{@{}l@{}}
            |\; \normal{\mu_0}{\var_0}, \normal{\mu}{\var}\;
                \mit{if}\; (\mu = a*\randomvar_1+b) \wedge \const(\var_0, \var):\\
            \quad \begin{array}[t]{@{}l@{}}
            \letin{(\mu_0', \var_0') = (\plus{(\mult{a}{\mu_0})}{b}, \mult{(\mult{a}{a})}{\var_0})}\\
            \letin{(\var_0'', \mu_0'') = (\division{1}{(\plus{\division{1}{\var_0}}{\division{1}{\var}})}, \mult{(\plus{\division{\mu_0'}{\var_0'}}{\division{\randomvar_2}{\var}})}{\var_0''})}\\
            (\symbstate[\randomvar_1 \mapsto \normal{\division{(\minus{\mu_0''}{b})}{a}}{\division{\var_0''}{(\mult{a}{a})}}][\randomvar_2 \mapsto \normal{\mu_0'}{\plus{\var_0'}{\var}}], \true)\\
            \end{array}\\
            \dots\\
            |\; \_: \; (\symbstate, \false)
        \end{array}
    \end{array}
\end{align*}
The $\swap$ operation enables the SSI runtime to symbolically transform different conjugate distributions in different cases; here we show the case for linear-Gaussians. When 1) both $\randomvar_1$ and $\randomvar_2$ are Gaussian-distributed, 2) the variance of each distribution is constant (i.e., does not depend on any random variables), and 3) the mean of $\randomvar_2$ is expressible as an affine function of $\randomvar_1$, the \swap{} operation performs linear-Gaussian swapping.
Note that all operations inside the $\swap$ construct symbolic expressions and perform no actual computation (e.g. $\mult{a}{a}$ construct a symbolic expression representing $a^2$).
The $\swap$ operation computes the new parameters of the swapped distributions according to the standard rules for conjugate priors~\cite{conjugate_priors} and updates the new distributions in the symbolic state.
It returns the updated state which represents the same distribution but where $\randomvar_2$ no longer depends on $\randomvar_1$ and  $\randomvar_1$ now depends on $\randomvar_2$. It also returns a $\true$ flag indicating a swap occurred.
If no conjugate distributions are available, the $\swap$ operation returns a $\false$ flag, indicating that exact inference is not possible and the algorithm must use approximate sampling. 

The success and failure of the $\swap$ transformation determine whether the SSI runtime encodes a random variable symbolically or samples it, influencing the inference plan it implements. 
For example, in \Cref{fig:example-annotated-x}, the \zl{x} variable in \Cref{line:x-x} and the observed Gaussian in \Cref{line:obsx-x} are linear-Gaussians. The $\swap$ function will apply the linear-Gaussian swapping, maintaining \zl{x} symbolically. Whereas, if \zl{v} is non-constant, the linear-Gaussian case will not apply. If no other conjugate prior case applies, the SSI runtime will be forced to sample \zl{x}.

\subsubsection{Delayed Sampling}
\label{sec:ds}
Delayed sampling (DS) is a hybrid inference algorithm that also exploits conjugacy relationships \citep{murray2018delayed,lunden2017delayed}.
It is an alternative implementation of the interface that represents the symbolic state using a forest of disjoint trees, where each node in each tree is a random variable. %

While we defer the full discussion of DS to prior work, we note here that compared to SSI, DS specifies additional information about each random variable, as each node is one of 3 types -- Initialized, Marginalized, or Realized -- and the inference plan satisfiability analysis needs to incorporate this additional information. %
In particular, Initialized nodes represent random variables that have a conditional distribution dependent on their parent; Marginalized nodes represent variables that have marginal distributions, and may need to track an optional prior distribution (and a reference to its original parent); and Realized nodes represent variables that have been replaced by a constant value through sampling or observing. 
The DS symbolic state uses the data field $\symbstateS{}{\randomvar}$ to track the node type for each random variable, where $\varset \subseteq \Randomvar$ are the children of the node:
\begin{align*}
    \Stuff ::= \marginalizedroot{\varset} \;|\; &\marginalized{\randomvar}{\Distr}{\varset} \; | \; \initialized{\randomvar}{\varset} \;|\; \realized 
\end{align*}

DS maintains invariants about the symbolic state, including that each tree contains at most one path of Marginalized nodes.
These invariants further influence the inference plans implemented by the DS runtime and require DS to implement the symbolic interface using a series of unique helper functions.
For example, DS implements the $\symvalue$ operation using the helpers $\graft$ and $\realize$:
\begin{align*}
\symvalue(\randomvar, \symbstate) = \;
        \mathit{let}\; \symbstate' = \graft(\randomvar, \symbstate) \; \mathit{in}\;
        \mathit{let}\; \val = \draw(\symbstate'(\randomvar)_{d}) \; \mathit{in}\;
        (\val, \realize(\randomvar, \deltasample{\val}, \symbstate'))
\end{align*}

While we defer the full details to prior work \citep{murray2018delayed, lunden2017delayed}, we note that $\graft$ and $\realize$ utilize the node types to manipulate the symbolic state. These operations determine whether the DS runtime samples random variables as well as the default inference plan when no annotations are provided.

\section{Inference Plan Satisfiability Analysis}
\label{sec:analysis}
Using the \mfSymbolic{} and \mfSample{} distribution encoding annotations, developers can express an inference plan specifying their requirements for how each random variable is encoded. However, the hybrid inference runtime performs a dynamic encoding cast on unsatisfiable annotations, enabling the program execution to continue at the risk of potential accuracy degradation.
In this section, we present the \emph{inference plan satisfiability analysis}, which statically identifies unsatisfiable annotations to assist developers reason about which inference plans to use.
If the analysis passes, the inference is guaranteed to encode all \mfSample{} variables with samples and all \mfSymbolic{} variables symbolically. 
We next formalize the analysis as an abstract interpretation and prove its soundness.

\subsection{Abstract Hybrid Inference}
\label{sec:abs-int}
Our analysis performs an abstract interpretation of the program by relying on an analogous version the hybrid inference interface that operates over the abstract domain. We refer to this version of the interface as the \emph{abstract hybrid inference interface}.
We construct abstract symbolic expressions and abstract symbolic states that the abstract interface operates over and manipulates.
The abstract interface operations mirror the concrete operations, except that $\asymobserve$ and $\asymvalue$ do not perform scoring or sampling.

\begin{figure}
\small
\begin{align*}
\aDistr \bnfarrow \; & 
\color{lightgray} 
\anormal{\aExpr}{\aExpr}\;
\mid \; \abern{\aExpr}\;
\mid \; \ainvgamma{\aExpr}{\aExpr} \;
\mid \; \adeltad{\aExpr}\;
\mid \; \adeltasample{\aExpr}\;
\mid \; \dots \;
\color{black}
\mid \; \dtop\;
\mid \; \dunk{}{\avarset{}}\;\\
\aExpr \bnfarrow \; & 
\color{lightgray}
\aconstant\;
\mid \; \arandomvar\;
\mid \; \aplus{\aExpr}{\aExpr}\;
\mid \; \dots \;
\color{black}
\mid \; \cunk\;
\mid \; \etop\;
\mid \; \eunk{}{\avarset{}}\\
\color{lightgray}
\abs{\mathit{Cmp}} \bnfarrow \; & 
\color{lightgray}
\hat{\texttt{=}} \;
\mid \; \hat{\texttt{!=}} \;
\mid \; \hat{\texttt{<}}  \;
\mid \; \hat{\texttt{<=}} \qquad
\aconstant \in \aVal, \arandomvar \in \aRandomvar, \color{black} \avarset{} \subseteq \aRandomvar
\end{align*}
\vspace{-1.9em}
\caption{Grammar of abstract symbolic expressions. Grayed-out expressions are identical to those in Figure~\ref{fig:symb-grammar}.}
\label{fig:abs-symb-grammar}
\vspace{-0.9em}
\end{figure}

\paragraph{Abstract Symbolic Expressions} 
Figure~\ref{fig:abs-symb-grammar} shows the grammar of abstract expressions.
For every symbolic expression, there is a corresponding abstract symbolic expression. 
Abstract expressions can also be $\cunk$, representing all constants, or $\etop$, representing all possible expressions. Additionally, they can also be the $\eunk{}{\avarset{}}$ expression, where $\avarset{}$ is a set of abstract random variables; the $\eunk{}{\avarset{}}$ expression represents expressions that reference any number of the random variables in $\avarset{}$.
Likewise, abstract distributions also can be $\dtop$ or $\dunk{}{\avarset{}}$.

Abstract symbolic expressions are equipped with a partial order, which we summarize as follows:
$$
\begin{array}{@{}rclcrcl@{}}
    \aconstant & \leq & \cunk & & \cunk & \leq & \arandomvar \\
    \arandomvar & \leq & \eunk{}{\{\arandomvar\}} & & & & \\
    \eunk{}{\avarset{}} & \leq & \etop & & \dunk{}{\avarset{}} & \leq & \dtop\\
    \eunk{}{\avarset{}} & \leq & \eunk{}{\avarset{}'} & \Leftarrow & \avarset{} &\subseteq& \avarset{}'\\
    \aplus{\eunk{}{\avarset{,1}}}{\eunk{}{\avarset{,2}}} & \leq & \eunk{}{\avarset{}'} & \Leftarrow & \avarset{}' &=& \avarset{,1} \cup \avarset{,2}\\
    \aplus{\aexpr{}_1}{\aexpr{}_2} & \leq & \aplus{\aexpr{\prime}_1}{\aexpr{\prime}_2} &\Leftarrow& \multicolumn{3}{c}{\aexpr{}_1 \leq \aexpr{\prime}_1,\; \aexpr{}_2 \leq \aexpr{\prime}_2}\\
    &&&\dots\\
\end{array}
$$

\noindent
The abstract expression $\cunk$ subsumes all constants and is itself considered a constant. Abstract random variables subsume constants. 
The $\eunk{}{\{\arandomvar\}}$ expression subsumes the variable $\arandomvar$. The $\eunk{}{\avarset{}}$ expression is a refinement of the top expression $\etop$ and the relative ordering between $\eunk{}{\avarset{}}$ expressions is defined by their variable sets.
Likewise, $\dunk{}{\avarset{}}$ is a refinement of $\dtop$.
A plus expression subsumes another plus expression if the subexpressions also subsumes the subexpressions of the other expression. Other complex expressions and distributions are analogous. 

Multiple abstract expressions can be over-approximated as one expression via a \emph{joining} expression, defined as taking the least upper bound according to the partial ordering. The partial ordering is designed to maintain more precise abstract expressions and distributions by recursing on subexpressions if the top-level expression type matches. 
For example, the join of $\aplus{\arandomvar_1}{\hat{1}}$ and $\aplus{\arandomvar_2}{\hat{2}}$ produces $\aplus{\eunk{}{\{\arandomvar_1, \arandomvar_2\}}}{\cunk}$. Random variables are not equivalent to each other, so their join can only be the $\eunk{}{\{\arandomvar_1, \arandomvar_2\}}$ expression. 

Retaining a precise representation of expressions is essential to the precision of the analysis because the hybrid inference algorithms depend on identifying expressions of certain classes (e.g. linear-Gaussians) to perform symbolic computation. However, symbolic computations can also cause the expression size to grow exponentially and infinitely, which is computationally expensive. To make the abstract domain finite, while trading off between precision and runtime, the analysis widens abstract expressions that are over the expression tree depth threshold $T$ to $\eunk{}{\avarset{}}$.  We use $T=5$ in our implementation.

\paragraph{Abstract Symbolic State}
An abstract symbolic state $\asymbstate{}$ is a finite mapping of abstract random variables (which reside in their own namespace) to tuples of annotations, abstract distributions, and the implementation-specific abstract data field: $\asymbstate{} \in \aSymbstate = \aRandomvar \rightarrow \aAnnotation \times \aDistr \times \aStuff$. 
Each entry is a constraint on entries in the concrete state. For example, $\asymbstateD{}{\arandomvar} = \anormal{\cunk}{\cunk}$ requires the concrete state to map the corresponding variable to Gaussian distributions with constant parameters.

$\aAnnotation = \set{\amfSample, \amfSymbolic, \noannotation}$ is the abstraction of annotations that is equipped with a partial ordering that defines its join operation: $\noannotation \leq \amfSample \leq \amfSymbolic$. The partial ordering is designed such that only a single abstract annotation needs to be tracked for each abstract random variable. An abstract random variable represents one or more concrete random variables, and so its abstract annotation must represent one or more concrete annotations.
A \mfSample{} annotation is always satisfiable because the \siren{} semantics collapses the random variable to a sampled value upon instantiation. This means the analysis does not need to remember $\amfSample{}$ annotations. However, a \mfSymbolic{} annotation could be an unsatisfiable annotation if the \siren{} runtime needs to sample it to continue execution. Then, the analysis must identify if any concrete random variable annotated with $\mfSymbolic$ will be sampled in the program, so $\amfSymbolic$ is the greatest in the partial ordering, to encompass the presence of one or more $\mfSymbolic$ annotations.

Abstract states are equipped with a partial order and a join operation:
\begin{align*}
    \asymbstate{1} \leq \asymbstate{2} \iff \textstyle{\forall_{\arandomvar \in \dom(\asymbstate{1})}}\; \asymbstate{1}(\arandomvar) \leq \asymbstate{2}(\arandomvar)
\end{align*}

\noindent
During the analysis, the join operation may introduce random variables that are unreachable from the abstract expression accompanying the abstract state. 
However, additional unreachable random variables in a symbolic state do not affect executions, so abstract symbolic states are equivalent for a given expression if all reachable variables from the expression have the same entries. This notion is useful when comparing two abstract symbolic states during the analysis.
We define a weak equivalence between two abstract symbolic states that only compares random variables that are reachable from the given expression to capture this notion.
\begin{align*}
    \weakeq{\asymbstate{1}}{\asymbstate{2}}{\progexpr} \iff \textstyle{\forall_{\arandomvar \in \reachable(\progexpr,\; \asymbstate{1},\; \asymbstate{2}) }}\; \asymbstate{1}(\arandomvar) = \asymbstate{2}(\arandomvar)
\end{align*}

\paragraph{Precision of Joining Expressions}
\label{sec:join-expr}
\begin{wrapfigure}[6]{r}{0.45\textwidth}
  \vspace{-1.5em}
  \centering
  \begin{lstlisting}
let symbolic x1 <- gaussian(1.,1.) in
let symbolic x2 <- gaussian(0.,1.) in
let x = if cond then x1+1. else x2+2. in
observe(gaussian(x,5.), obs)
  \end{lstlisting}
  \vspace{-1.8em}
  \caption{Example program.}
  \label{fig:join-example}
\end{wrapfigure}
When a program contains data-dependent or stochastic control flow, the static analysis does not know which branch would be evaluated. In such cases, the analysis must over-approximate the true states of the program by joining the abstract expressions and symbolic states from the branches. It then loses critical information about the structures of the 
subexpressions. Hybrid inference relies on matching symbolic expressions 
to detect exact 
inference opportunities, 
so the over-approximation can significantly impact the precision of the analysis. 

For example, consider the program in \Cref{fig:join-example}, where \mkw{cond} and \mkw{obs} are constant values. The variables \mkw{x1} and \mkw{x2} refer to the abstract random variables $\arandomvar_1$ and $\arandomvar_2$. 
No matter which branch the runtime executes, the observed Gaussian is a linear-Gaussian.
In SSI, the parent variable in both cases  would remain as symbolic expressions, so the inference plan is satisfiable for all possible executions.
However, the analysis does not know the value of \zl{cond}, so it must over-approximate the program state by joining $\aplus{\arandomvar_1}{\hat{1}}$ and $\aplus{\arandomvar_2}{\hat{2}}$ into $\aplus{\eunk{}{\{\arandomvar_1, \arandomvar_2\}}}{\cunk}$. It concludes the resulting abstract observed Gaussian is not necessarily linear-Gaussian, as the $\eunk{}{\{\arandomvar_1, \arandomvar_2\}}$ expression also represents expressions that are non-linear to $\arandomvar_1$ and $\arandomvar_2$. Consequently, it cannot be sure $\arandomvar_1$ and $\arandomvar_2$ will not be sampled and cannot determine the inference plan is satisfiable even though it is. 

We define a special operation for joining expressions that takes advantage of the fact that abstract random variables can be renamed without compromising soundness, which we prove in \Cref{appendix:proofs}. The key idea is that abstract random variables and concrete random variables exist in different namespaces. A single abstract variable can represent more than one concrete variable if its abstract symbolic state entry over-approximates the entries of those concrete variables.
By renaming two otherwise disparate variables to the same name, their entries are forced to be joined into one. 
We define the special join of two expressions $\aexpr{}_1$ and $\aexpr{}_2$ under symbolic states $\asymbstate{1}$ and $\asymbstate{2}$ as:
\begin{align*}
  \narrowjoin(\aexpr{}_1, \aexpr{}_2, \asymbstate{1}, \asymbstate{2})  = %
  \letin{(\aexpr{}_3, \asymbstate{3}) = \rename(\aexpr{}_1, \aexpr{}_2, \asymbstate{2})}
  (\aexpr{}_1 \sqcup \aexpr{}_3, \asymbstate{1} \sqcup \asymbstate{3})
\end{align*}

\noindent
where $\sqcup$ refers to the basic join operation implied by the partial orders for abstract expressions and abstract symbolic states.
The $\rename(\aexpr{}_1, \aexpr{}_2, \asymbstate{2})$ function returns renamed versions of $\aexpr{}_2$ and $\asymbstate{2}$ that maximize the similarities between $\aexpr{}_1$ and $\aexpr{}_2$ with capture-avoiding substitution.

For the program in \Cref{fig:join-example}, the analysis renames the variable $\arandomvar_2$ to $\arandomvar_1$. 
The joined expression is the more structurally precise expression $\aplus{\arandomvar_1}{\cunk}$, and 
the joined state assigns $\arandomvar_1$ to $\anormal{{\cunk}}{\hat{1}}$. Then, the observed variable has the distribution $\anormal{\aplus{\arandomvar_1}{\cunk}}{5}$.
$\narrowjoin$ retains the shared structure in the expressions, and the analysis recognizes the observed variable as a linear-Gaussian.

\paragraph{Fail}
A program execution may encounter an unsatisfiable annotation that is dynamically cast to be satisfiable. When the analysis detects the possibility of this event (i.e. when a $\amfSymbolic$ abstract random variable could be sampled), it returns $\afail$, the top of all abstract values: Any value joined with $\afail$ results in $\afail$. Any operation that receives $\afail$ as input also returns $\afail$ as the output.

\subsection{Abstract Interpretation Rules}
\begin{figure}
  \begin{small}
  \begin{mathpar}
  \inferrule%
  {
   \asymvalue^*(\aval, \asymbstate{}) = \cunk, \asymbstate{\aval} \\
   \apclstep{\progexpr_1}{\asymbstate{v}}
           {\aval_1}{\asymbstate{1}'} \\
   \apclstep{\progexpr_2}{\asymbstate{v}}
           {\aval_2}{\asymbstate{2}'}\\
   \narrowjoin(\aval_1, \aval_2, \asymbstate{1}', \asymbstate{2}') = \aval'', \asymbstate{}'
  }
  {\apclstep{\mfIf{\aval}{\progexpr_1}{\progexpr_2}}{\asymbstate{}}
           {\aval''}{\asymbstate{}'}}
  
  \inferrule%
  {
    \apclstep{\mfApp{f}{\mfPair{\alistl}{\aval}}}{\asymbstate{}}
             {\aval_f}{\asymbstate{f}}\\
    \narrowjoin(\aval, \aval_f, \asymbstate{}, \asymbstate{f}) = \aval_j, \asymbstate{j}\\
    \aval = \aval_j\\
    \weakeq{\asymbstate{}}{\asymbstate{j}}{(\alistl,\aval)}
  }
  {\apclstep{\mfFold{f}{\alistl}{\aval}}{\asymbstate{}}
            {\aval}{\asymbstate{}}}
  
  \inferrule%
  {
    \apclstep{\mfApp{f}{\mfPair{\alistl}{\aval}}}{\asymbstate{}}
             {\aval_f}{\asymbstate{f}}\\
    \narrowjoin(\aval, \aval_f, \asymbstate{}, \asymbstate{f}) = \aval_j, \asymbstate{j}\\
    \apclstep{\mfFold{f}{\alistl}{\aval_j}}{\asymbstate{j}}
             {\aval'}{\asymbstate{}'}
  }
  {\apclstep{\mfFold{f}{\alistl}{\aval}}{\asymbstate{}}
           {\aval'}{\asymbstate{}'}}
  
  \inferrule{
    \adset = \set{\adistribution(\aval,\asymbstate{}') \;|\; (\progexpr,\asymbstate{}) \in \apset, (\progexpr,\asymbstate{} \;\abs{\downarrow}\; \aval,\asymbstate{}') }
  }
  {\apclsstep{\apset}{\textstyle{\bigsqcup_{\adistr{}_i \in \adset}}\; \adistr{}_i}}

  \inferrule{
    \ensuremath{\set{{\progexpr}, {\emptyset}}\; \hat{\downdownarrows}\; \afail}
    }
    {\apclpstep{\progexpr}{\afail}}
  
  \inferrule{
      \apclsstep{\set{{\progexpr}, {\emptyset}}}
              {\adistr{}}
  }
  {\apclpstep{\progexpr}{\mit{satisfiable}}}
  
  \end{mathpar}
  \end{small}
  \vspace{-0.9em}
  \caption{Fragment of the abstract interpretation rules. The full set of rules is in \Cref{appendix:abs-interp-rules}. 
  }
  \label{fig:abstract-interp}
\vspace{-0.5em} 
\end{figure}

Figure~\ref{fig:abstract-interp} presents a fragment of the interpretation rules of a \siren{} program using the abstract hybrid inference operations. 
We present here only the rules that differ from the concrete semantics and include the full abstract semantics in \Cref{appendix:abs-interp-rules}.

\paragraph{Conditionals}
When $\asymvalue^*$ returns $\cunk$, the analysis cannot determine which branch of a conditional is taken, so it interprets both branches and joins the resulting abstract expressions with $\narrowjoin$ to approximate the program execution.

\paragraph{Fold.}
If the \zl{fold} operation receives a list argument $\alistl$ that is not a constant list, such as $\eunk{}{\emptyset}$, the analysis over-approximates the operation by computing the abstract fixpoint of the function $f$. 
The analysis first interprets $f$ on $(\alistl, \aval)$. Since $\alistl$ is not a constant list, it is either $\eunk{}{\avarset{}}$ or $\etop$, each of which also over-approximates any particular item in the list, respectively. The analysis computes $(\aval_j, \asymbstate{j})$ using $\narrowjoin$, which are the over-approximations of the current inputs and the inputs of the next iteration for $f$. If they are weakly equal to the current inputs, no further application of $f$ could be different; the fixpoint computation stops when $\aval$ and $\aval_j$ are equal and $\asymbstate{}$ and $\asymbstate{j}$ are weakly equal with respect to $(\alistl, \aval)$.

During fixpoint computations, the analysis could be joining $\eunk{}{\avarset{}}$ expressions. However, $\avarset{}$ can grow arbitrarily large. To bound the growth, the analysis widens the joined expression by converting $\eunk{}{\avarset{}}$ to $\etop$ if $|\avarset{}| \geq N$ for some parameter $N$. Our implementation uses $N=4$, but the parameter may be adjusted for greater precision at the cost of more fixpoint iterations.

\paragraph{Particle Set and Model Evaluation}
Unlike the concrete semantics, the abstract semantics spawns only a singleton set of particles to evaluate, as there are no weights to consider.
All possible particles from a program are accounted for in a single abstract particle evaluation. Thus, there is no abstract equivalent of the resampling step, and \mfResample{} is a no-op.
The abstract interpretation of a program is then simply whether the abstract particle evaluation rules encounter failures or not. 

\subsection{Implementing the Abstract Hybrid Inference Interface}
\label{sec:abs-instantiation}
The analysis, analogous to the concrete semantics, relies on the abstract hybrid inference interface. As such, the analysis is also unified across different hybrid particle filtering algorithms. 
Only the implementation of the interface is required to extend the analysis to a new algorithm. 
This section presents how the analysis implements an abstract version of the hybrid interface for SSI and DS. 
While we defer the full details to \Cref{appendix:abs-interface}, we note the similarity of the abstract operations to the concrete operations from Section~\ref{sec:instantiations}, except that the abstract operations have extensions to handle the analysis's imprecision.
We present here the abstract version of SSI's $\swap$ operation and how the analysis incorporates the additional information in DS's node types.

\subsubsection{Semi-Symbolic Inference}
Each SSI operation has an abstract version that mirrors the concrete operation and differs only in how it handles abstract values like $\dunk{}{\avarset{}}$. For instance, the $\asymvalue$ operation depends on $\ahoist$ and $\aintervene$, and it uses $\cunk$ instead of drawing values. However, for $\asymvalue$, there is the additional difference that it returns $\afail$ if the input variable has the $\amfSymbolic$ annotation i.e. the concrete random variables represented by the abstract random variable may have an unsatisfiable annotation.
\begin{align*}
\asymvalue(\arandomvar, \asymbstate{}) = \;
        \begin{array}[t]{@{}l@{}}
        \mit{if}\; \asymbstatePV{}{\arandomvar} = \amfSymbolic\; \mit{then}\;
        \afail\;\\
        \mit{else}\;
        \mathit{let} \; \asymbstate{}' = \ahoist(\arandomvar, \asymbstate{}) \; \mathit{in}\; 
        (\cunk, \aintervene(\arandomvar, \adeltasample{\cunk}, \asymbstate{}'))
        \end{array}
\end{align*}

To show how implementations handle abstract values, we next describe the abstract $\swap$ operation used in SSI. The abstract operation simulates the concrete $\swap$ operation as defined in \Cref{sec:ssi} by detecting conjugacy and performing the equivalent computation where it can:
\\
\vspace{-1em}
\begin{align*}
  \aswap(\arandomvar_1, \arandomvar_2, \asymbstate{}) = \; & 
    \begin{array}[t]{@{}l@{}}
        \mit{match}\; \asymbstateD{}{\arandomvar_1}, \asymbstateD{}{\arandomvar_2}\; \mit{with}\\
        \begin{array}[t]{@{}l@{}}
            |\; \anormal{\mu_0}{\var_0}, \anormal{\mu}{\var}\;
                \mit{if}\; (\mu = \aplus{\amult{a}{\arandomvar_1}}{b}) \wedge \aconst(\var_0, \var):\\
            \quad \begin{array}[t]{@{}l@{}}
            \letin{(\mu_0', \var_0') = (\aplus{(\amult{a}{\mu_0})}{b}, \amult{(\amult{a}{a})}{\var_0})}\\
            \letin{(\var_0'', \mu_0'') = (\adivision{\hat{1}}{(\aplus{\adivision{\hat{1}}{\var_0}}{\adivision{\hat{1}}{\var}})}, \amult{(\aplus{\adivision{\mu_0'}{\var_0'}}{\adivision{\arandomvar_2}{\var}})}{\var_0''})}\\
            (\asymbstate{}[\arandomvar_1 \mapsto \anormal{\adivision{(\aminus{\mu_0''}{b})}{a}}{\adivision{\var_0''}{(\amult{a}{a})}}][\arandomvar_2 \mapsto \anormalbig{\mu_0'}{\aplus{\var_0'}{\var}}], \true)\\
            \end{array}\\
            \dots\\
            |\; \dunk{}{\_}, \_ : \; (\setunk(\arandomvar_1, \asymbstate{}), \afalse)\\
            |\; \_: \; (\asymbstate{}, \afalse)
        \end{array}
    \end{array}
\end{align*}

\noindent
Because the analysis is performed at compile-time, it can only perform a best-effort detection of conjugates, given that parameters might be represented by the opaque $\eunk{}{\avarset{}}$ and $\dunk{}{\avarset{}}$ expressions.
If the abstract distribution of the parent variable $\arandomvar_1$ is $\dunk{}{\avarset{}}$ or $\dtop$, the analysis recursively sets  $\arandomvar_1$ and its ancestors to $\dtop$ using $\setunk(\arandomvar_1, \asymbstate{})$. During this process, if the random variable or any of its ancestors (i.e. the parents of the variable and their parents and so on) are annotated $\amfSymbolic$, the analysis cannot be sure if the variable is a conjugate prior nor that it is not a conjugate prior (meaning that it will be sampled), so the analysis conservatively $\afail$s. 

For example, consider an abstract symbolic state where that the parent variable $\arandomvar_1$ has the abstract distribution $\dunk{}{\{\arandomvar_3\}}$, and variable $\arandomvar_3$ has $\anormal{1.}{1.}$. Also, $\arandomvar_3$ has the $\amfSymbolic{}$ annotation. Because the parent variable $\arandomvar_1$ can be any distribution referencing $\arandomvar_3$, the analysis cannot determine if $\arandomvar_1$ has a conjugate prior distribution to the child variable $\arandomvar_2$. Subsequently, it cannot determine if $\arandomvar_1$ will be sampled or not and will invoke $\setunk(\arandomvar_1, \asymbstate{})$, resulting in $\arandomvar_1$ and $\arandomvar_3$ both having the distribution $\dtop$, because it also cannot determine the representations of upstream distributions. $\arandomvar_3$ has the $\amfSymbolic$ annotation, so the analysis will return $\afail$.

\subsubsection{Delayed Sampling}
The abstract node types of DS can be Initialized, Marginalized, or Realized. Initialized and Marginalized nodes still track their parent variables, their prior distributions (using abstract expressions), and their children. The fourth abstract node type, \mtt{TopN}, is the top of all node states. It indicates that the analysis does not know the node's type. No prior and no parent to the random variable are tracked for \mtt{TopN}, only its children.
\begin{align*}
    \aStuff ::=\; & \amarginalizedroot{\avarset{}} \;|\; \amarginalized{\arandomvar}{\adistr{}}{\avarset{}} \; | \; \ainitialized{\arandomvar}{\avarset{}} \;|\; \arealized \; | \; \topnode{\avarset{}}
\end{align*}

\noindent
Here, $\avarset{} \subseteq \aRandomvar$ represents the set of possible child random variables. The abstract node states are equipped with a partial ordering:
$$
\begin{array}{@{}rclcl@{}}
    \arealized & \leq & \topnode{\emptyset}  & & \\
    \amarginalizedroot{\avarset{}} & \leq & \amarginalizedroot{\avarset{}'} & \Leftarrow & \avarset{} \subseteq \avarset{}' \\
    \amarginalizedroot{\avarset{}} & \leq & \topnode{\avarset{}'} & \Leftarrow & \avarset{} \subseteq \avarset{}' \\
    \amarginalized{\arandomvar}{\adistr{}}{\avarset{}} & \leq & \amarginalized{\arandomvar'}{\adistr{\prime}}{\avarset{}'} & \Leftarrow & (\arandomvar, \adistr{}) \leq (\arandomvar', \adistr{\prime}), \avarset{} \subseteq \avarset{}' \\
    \amarginalized{\arandomvar}{\adistr{}}{\avarset{}} & \leq & \topnode{\avarset{}'} & \Leftarrow & \avarset{} \subseteq \avarset{}' \\
    \ainitialized{\arandomvar}{\avarset{}} & \leq & \ainitialized{\arandomvar'}{\avarset{}'} & \Leftarrow & \arandomvar \leq \arandomvar', \avarset{} \subseteq \avarset{}' \\
    \ainitialized{\arandomvar}{\avarset{}} & \leq & \topnode{\avarset{}'} & \Leftarrow & \avarset{} \subseteq \avarset{}' \\
    \topnode{\avarset{}} & \leq & \topnode{\avarset{}'} & \Leftarrow & \avarset{} \subseteq \avarset{}'
\end{array}
$$

The join operation defined by the partial ordering is used when the analysis has to compute the join of two abstract symbolic states.
Even though Initialized and Marginalized node states may be tracking parent and prior distributions, \mtt{TopN} does not. This is because as soon as an abstract node type becomes \mtt{TopN}, the analysis recursively over-approximates the parent and Marginal children to the node by setting them to \mtt{TopN} node states as well. During this process, if any of these random variables has a $\amfSymbolic$ annotation, the analysis returns $\afail$.

For example, consider two abstract symbolic states $\asymbstate{1}$ and  $\asymbstate{2}$. In both states, the variable $\arandomvar_1$ has the node type 
$\amarginalizedroot{\{\arandomvar_2\}}$. $\arandomvar_2$ has the node type $\ainitialized{\arandomvar_1}{\emptyset}$ in $\asymbstate{1}$ but $\amarginalized{\arandomvar_1}{\anormal{\arandomvar_1}{1.}}{\arandomvar_3}$ in $\asymbstate{2}$. The symbolic state $\asymbstate{2}$ also has $\arandomvar_3$ having the node type $\amarginalized{\arandomvar_2}{\anormal{\arandomvar_2}{1.}}{\emptyset}$ and the $\amfSymbolic$ annotation. The join operation between these two states results in $\arandomvar_2$ having \mtt{TopN} node type. Because delayed sampling uses the node types to determine whether to sample variables, the analysis would not be able to determine whether its parents and children would be sampled. As such, it would then set $\arandomvar_1$ and $\arandomvar_3$ to both be \mtt{TopN} node types as well. The resulting state would have $\arandomvar_1$ having the node type $\topnode{\{\arandomvar_2\}}$, $\arandomvar_2$ having $\topnode{\{\arandomvar_3\}}$, and $\arandomvar_3$ having $\topnode{\emptyset}$. $\arandomvar_3$ has the $\amfSymbolic$  annotation, so the analysis will return $\afail$.

\subsection{Properties}
\label{sec:properties}
In this section, we show that the inference plan satisfiability analysis is sound. 
The approach is mostly standard~\citep{cousot1977abstract,cousot1992abstract}, except for how it handles random variable names and the variable sets in abstract expressions.
We will highlight these nonstandard elements throughout the formal development. 
First, we define the collecting semantics for sets of program states that serves as the basis of our soundness proof. The collecting semantics accumulates from program executions the information relevant to the program properties under study. 
The abstract states computed by the analysis must over-approximate the collected concrete states to ensure soundness.
Next, we define the abstraction and concretization functions that relate abstract values to concrete values.
Finally, we present key lemmas and theorems that prove the analysis is sound.

\subsubsection{Collecting Semantics}
\begin{wrapfigure}[5]{r}{0.4\textwidth}
  \centering
  \vspace{-1.1em}
  \begin{minipage}[b][][b]{.4\textwidth}
    \begin{small}
    \begin{mathpar}
    \inferrule%
    {\evalset = \set{ (\progexpr',\symbstate',\doresample) \;|\; (\progexpr, \symbstate \downarrow^{\doresample} \progexpr', \symbstate', w)}}
    {\cpclstep{\progexpr, \symbstate}{\evalset}}
    \end{mathpar}
\end{small}
\end{minipage}
  \vspace{-0.5em}
  \caption{Collecting particle evaluation rule.}
  \label{fig:collecting-particle}
\end{wrapfigure}
The collecting semantics is a forward collecting semantics~\citep{cousot1992abstract} based on our operational semantics. 
The program states collected differ between our three types of evaluation rules. Even though the operational semantics uses weight values and performs resampling, 
the collecting semantics ignores these aspects. 
The analysis only depends on the possible particles produced during program execution. 
The resampling step does not introduce any new particles to the execution, only duplicating or removing existing particles.
Additionally, weight values do not affect the representation of random variables. 
As such, weights are not collected and the resampling step in the particle set evaluation rules is a no-op.

\paragraph{Particle Evaluation.}
The collecting semantics of particle evaluation collects any particle that can result from applying concrete particle evaluation rules to the particle, shown in \Cref{fig:collecting-particle}. The rules may produce more than one evaluated particle due to data-dependent or randomized control flow. The collecting semantics returns all such possible evaluated particles with the resample flags, dropping weight values. We call these tuples \emph{configurations}.

When a \mfSymbolic{} distribution encoding annotation is unsatisfiable, the \siren{} runtime performs a dynamic encoding cast by sampling the annotated random variable anyway, enabling the execution to continue. The cast is the event that the inference plan satisfiability analysis must detect.
In the collecting semantics, the $\symvalue$ function must return the $\fail$ value if the annotation of the input random variable is \mfSymbolic{}. The $\fail$ value propagates in the standard way.
\begin{align*}
\symvalue(\randomvar, \symbstate) = \;
        \begin{array}[t]{@{}l@{}}
        \mit{if}\; \symbstatePV{}{\randomvar} = \mfSymbolic\; \mit{then}\;
        \fail\;\\
        \mit{else}\;
        \mit{let} \; \symbstate' = \hoist{}(\randomvar, \symbstate) \; \mit{in}\; 
        \mit{let} \; \val = \draw(\symbstate'(\randomvar)) \mit{in} \;        
        (\val, \intervene(\randomvar, \deltasample{\val}, \symbstate'))
        \end{array}
\end{align*}

\paragraph{Particle Set and Model Evaluation}
While the bulk of the soundness proof refers to the collecting particle evaluation $\tilde{\downarrow}$, the top-level theorems also refer to collecting analogs of the particle set and model evaluation semantics of Figure~\ref{fig:op:sem-set} and~\ref{fig:op:sem-model}.
The relation $(\cpclsstep{\pset}{\dset})$ means that the set of particles $\pset$ evaluates to the set of distributions $\dset$, and the definition of $\tilde{\downdownarrows}$ refers to the definition of $\tilde{\downarrow}$.
Similarly, the relation $(\cpclpstep{\progexpr}{\dset})$ means that the program $\progexpr$ evaluates to the set of distributions $\dset$, and the definition $\tilde{\Downarrow}$ depends on $\tilde{\downdownarrows}$.
The definitions for both $\tilde{\downdownarrows}$ and $\tilde{\Downarrow}$ are in Appendix~\ref{appendix:proofs}.

\subsubsection{Abstraction}
The abstraction function $\abstr$ maps sets of concrete values to an abstract value. We define the function first for singleton sets of concrete values.
The abstraction of sets of multiple values is then the join of the corresponding abstracted values.

Concrete random variables and abstract random variables have different namespaces. 
To account for this, the abstraction function assumes the existence of a default mapping $\rvnames : \Randomvar \rightarrow \aRandomvar$ that maps concrete variable names to abstract variable names.
The abstractions of both random variables and symbolic states use this to produce the appropriate name in abstract values.

\begin{definition}[Abstraction Function]
We define the abstraction function $\abstr$ as follows. 
The default mapping function $\rvnames{}$ maps every concrete variable to a unique, canonical abstract variable.
\[
\begin{array}{@{}rclcrcl@{}}
\abstr(\set{\constant}) &=& \aconstant & &  \abstr(\set{\randomvar}) &=& \rvnames(\randomvar)\\
\abstr(\{\plus{\expr{}_1}{\expr{}_2}\}) &=& \aplus{\abstr(\{\expr{}_1\})}{\abstr(\{\expr{}_2\})} & & \abstr(\set{(\val_1, \val_2)}) &=& (\abstr(\set{\val_1}), \abstr(\set{\val_2}))\\
\abstr(\set{\mfSymbolic}) &=& \amfSymbolic & &  \\
&&&\cdots\\
\abstr(\set{\symbstate}) &=& \multicolumn{5}{l}{
\left\{ \rvnames(\randomvar) \mapsto \abstr(\set{\symbstate(\randomvar)})\; \middle|\; \randomvar \in \dom(\symbstate)\right\}}\\
\abstr(\set{\fail}) &=& \afail & &\abstr(\vset) &=& \bigsqcup_{\val \in \vset} \abstr(\set{v})
\end{array}
\]
\end{definition}
For example, consider a particle $(\expr{}, \symbstate)$ with symbolic expression $\expr{} = \plus{\randomvar_1}{1}$ and symbolic state
$
\symbstate = \{ \randomvar_1 \mapsto (\mfSymbolic, \SSIgamma{1}{1}, \realized) \}
$.
Assuming that $\rvnames$ maps $\randomvar_1$ to the abstract variable $\arandomvar_a$, we have that $\abstr(\set{\expr{}, \symbstate}) = (\aplus{\arandomvar_a}{\abs{1}}, \abstr(\set{\symbstate}))$ where $\abstr(\set{\symbstate}) = \{\arandomvar_a \mapsto ( \amfSymbolic, \agamma{\abs{1}}{\abs{1}}, \arealized) \}$.

\subsubsection{Concretization} 
The concretization function $\concret$ plays the opposite role to the abstraction function: it maps every abstract value to a set of concrete values.
While we formalize abstraction with a default mapping function, 
the concretization needs to account for all possible mappings.
We first define a version of the concretization function that is parameterized by a surjective function $\rvname{} : \Randomvar \rightarrow \aRandomvar$ that maps concrete random variables to abstract random variables.
The function must be surjective since every abstract random variable must have a corresponding concrete random variable. The function does not have to be injective, because an abstract variable can represent multiple concrete variables that share properties in the symbolic state.
The concretizations for $\eunk{}{\avarset{}}$ and $\dunk{}{\avarset{}}$ incorporate the variable set $\avarset{}$ by ensuring that the concretization includes only those expressions whose \emph{free variables} are a subset of $\avarset{}$.
We formalize this using the operation $\mathit{FV}(\expr{}, \rvname{})$ that returns the set of free variables in $\expr{}$, mapped to abstract names using $\rvname{}$.
Finally, we define the concretization function by taking the union over all possible name-mapping functions;
the set resulting from the concretization function is closed under name re-mappings.

\begin{definition}[Concretization Function]
We define the concretization function $\concret$ as follows.
First, we define $\concret$ as a function that takes in an abstract state and the name-mapping function $\rvname{}$. The function $\rvname{}$ is surjective and it maps concrete random variables into abstract random variables.\\
\vspace{-1em}
\[
\begin{array}{@{}rclcrcl@{}}
\concret(\aconstant, \rvname{}) &=& \{ \constant \} & & \concret(\cunk, \rvname{}) &=& \Val \\
\concret(\arandomvar, \rvname{}) &=& \{ \randomvar \;|\; \rvname{}(\randomvar) = \arandomvar \} \bigcup \Val & & \concret(\eunk{}{\avarset{}}, \rvname{}) &=& {\{\expr{} \; | \; \mathit{FV}(\expr{}, \rvname{}) \subseteq \avarset{} \} } \\
\concret(\etop, \rvname{}) &=& \Expr & & \concret(\amfSymbolic, \rvname{}) &=& \set{\mfSymbolic, \mfSample, \varepsilon}\\
\concret(\aplus{\aexpr{}_1}{\aexpr{}_2}, \rvname{}) &=& \multicolumn{5}{l}{
\{ \plus{\expr{}_1}{\expr{}_2}\ |\ \expr{}_1 \in \concret(\aexpr{}_1, \rvname{}), \expr{}_2 \in \concret(\aexpr{}_2, \rvname{})\}}\\
\concret((\aval_1, \aval_2), \rvname{}) &=& \multicolumn{5}{l}{\{ (\val_1, \val_2) \ |\ \val_1 \in \concret(\aval_1, \rvname{}), \val_2 \in \concret(\aval_2,\rvname{}) \}}\\
&&&\cdots\\
\concret(\asymbstate{}, \rvname{}) &=& \multicolumn{5}{l}{ 
\{ \symbstate{} \;|\; 
\textstyle{\forall_{\randomvar}}\;
\mit{if}\; \rvname{}(\randomvar) \in \dom(\asymbstate{})\; \mit{then}\;
\symbstate(\randomvar) \in \asymbstate{}(\rvname{}(\randomvar))
\;\mit{else}\;
\randomvar \notin \dom(\symbstate)
\} }\\
\end{array}
\]

\noindent
Now, we define $\concret$ by taking the union over all possible surjective naming functions: 
$$ \concret(\aval) = \left\{\;
\val \;\middle|\; 
\val \in \concret(\aval, \rvname{}), 
\rvname{} \in \randomvar \rightarrow \arandomvar, 
\rvname{}\; \mit{is}\; \mit{surjective}
\;\right\}$$
\end{definition}
For example, the concrete symbolic state $\set{ \randomvar_1 \mapsto (\mfSample, \deltad{1}), \randomvar_2 \mapsto (\mfSample, \deltad{2}) }$  is included in the concretization $\concret(\{ \arandomvar_a \mapsto (\amfSample, \dunk{}{\emptyset}) \})$. Conversely, the concrete symbolic state $\set{\randomvar_1 \mapsto (\mfSample, \normal{\randomvar_2}{1.}), \randomvar_2 \mapsto (\mfSample, \deltasample{1})}$ is not.
Additionally, the concretization $\concret(\adistr{})$ where $\adistr{} = \anormal{\arandomvar_a}{\aplus{\arandomvar_b}{\hat{1.}}}$ includes both $\normal{\randomvar_0}{\plus{\randomvar_1}{1.}}$ and $\normal{\randomvar_1}{\plus{\randomvar_0}{1.}}$.

\subsubsection{Soundness of Analysis}
We now present the key ideas and properties necessary to prove the soundness of the analysis and defer the full formalization and proofs to \Cref{appendix:proofs}.
Our treatment of soundness is limited in that we assume the analysis has a sound implementation of the symbolic interface, and show that under this assumption, the overall analysis is sound. We formalize this assumption as follows:
\begin{assumption}[Abstract Hybrid Inference Interface Soundness]
For every $i \in \{\textsc{Assume},\allowbreak \textsc{Value},\allowbreak \textsc{Observe}\}$ and input values $\val_i$, we have that $i(\val_i) \in \concret(\abs{i}(\abstr(\{\val_i\})))$.
\label{assumption:syminterface}
\end{assumption}

Because of the join operation on abstract symbolic states, an abstract operation might compute an abstract symbolic state that has variables that are not reachable from the computed expression. The concretization of abstract symbolic states retains those unreachable variables in the concrete symbolic states. 
Symbolic states with different domains are not strictly equal. 
However, unreachable variables do not alter the evaluated expression nor the reachable entries in the resulting symbolic state. 
To account for this property, we define a weak equivalence relation for concrete symbolic states, analogous to the weak equivalence relation for abstract states.

The formalization uses an auxiliary operation $\tilde{\downarrow}^*$ for repeatedly evaluating a particle until it has terminated, which we define precisely in \Cref{appendix:proofs}. We write configuration sets that have weakly equivalent symbolic states as $\weakeq{\evalset}{\evalset'}{}  \iff \evalset' = \set{(\progexpr, \symbstate', \doresample) \;|\; (\progexpr,\symbstate,\doresample) \in \evalset, \weakeq{\symbstate}{\symbstate'}{\progexpr}}$.
We use an auxiliary operation to drop the resample flag in configurations:
    $\forgetr(\evalset) = \set{(\progexpr, \symbstate)\;|\; (\progexpr,\symbstate,\doresample) \in \evalset}$.

We first show the analysis is sound when the particle evaluation terminates, and resuming particle evaluation preserves the soundness of the analysis. 

\begin{lemma}[Terminating Particle Evaluation Soundness]
    \label{lem:terminating-particle-short}
For every particle $(\progexpr, \symbstate)$, such that $(\cpclstep{\progexpr, \symbstate}{\evalset})$ and 
$\textstyle{\forall_{(\val,\; \symbstate',\; \doresample) \in \evalset}}\; (\val = \fail) \vee \neg \doresample$, 
we have $(\progexpr, \abstr(\{\symbstate\}) \;\abs{\downarrow}\; \aval', \asymbstate{}')$ and 
there exists a configuration set $\evalset'$ such that $\weakeq{\evalset}{\evalset'}{}$ and 
$\forgetr(\evalset') \subseteq \concret((\aval', \asymbstate{}'))$.
\end{lemma}

\begin{lemma}[Preservation]
\label{lem:preservation-short}
If $(\pclstep{\progexpr}{\symbstate}{\progexpr'}{\symbstate'}{\weight}{\doresample})$, then there exists abstract values $\aval, \aval'$ and abstract symbolic states  $\asymbstate{}, \asymbstate{}', \asymbstate{}''$ such that 1)\ $(\apclstep{\progexpr}{\abstr(\{\symbstate\})}{\aval}{\asymbstate{}}) \iff (\apclstep{\progexpr'}{\abstr(\{\symbstate'\})}{\aval'}{\asymbstate{}'})$, 2)\ $\weakeq{\asymbstate{}'}{\asymbstate{}''}{\aval'}$, and 3)\ $\concret((\aval', \asymbstate{}'')) \subseteq \concret((\aval, \asymbstate{}))$.
\end{lemma}

It follows that the analysis is sound for evaluating any particle until termination. 

\begin{lemma}[Particle Evaluation Soundness]
\label{thm:particle-short}
For every particle $(\progexpr, \symbstate)$, such that $(\cpclstepstar{\progexpr, \symbstate}{\evalset})$, we have
$(\progexpr, \abstr(\{\symbstate\}) \;\abs{\downarrow}\; \aval, \asymbstate{})$ and a configuration set $\evalset'$ such that $\weakeq{\evalset}{\evalset'}{}$ and 
$\forgetr(\evalset')  \subseteq \concret((\aval, \asymbstate{}))$.
\end{lemma}

Additionally, every distribution resulting from a particle set evaluation can be traced back to a particle in the particle set and be equivalently derived by evaluating the particle until termination. 

\begin{lemma}[Particle Trace]
    \label{lem:particle-trace-short}
If $(\cpclsstep{\pset}{\dset})$, we have for all $\;\distr{} \in \dset$, there exists $(\progexpr, \symbstate) \in \pset$ such that $(\cpclstepstar{\progexpr, \symbstate}{\evalset})$ and $\;\distr{} \in \set{\distribution(\val,\symbstate) \;|\; (\val,\symbstate,\doresample) \in \evalset }$.
\end{lemma}

From the particle trace property with the fact the analysis is sound when evaluating a particle until termination, we have that the analysis is sound with respect to evaluating sets of particles. The soundness of the model evaluation follows. 

\begin{theorem}[Particle Set Evaluation Soundness]
    \label{thm:particle-set-short}
    For every particle set $\pset$, and distribution set $\dset$ such that $(\cpclsstep{\pset}{\dset})$, we have that \apclsstep{\set{\progexpr, \abstr(\{\symbstate\}) \;|\; (\progexpr, \symbstate) \in \pset}}{\adistr{}} and $\dset \subseteq \concret(\adistr{})$.
\end{theorem}

\begin{corollary}[Model Evaluation Soundness]
If $\cpclpstep{\progexpr}{\{\fail\}}$, then $\apclpstep{\progexpr}{\afail}$.
\end{corollary}
Overall, the soundness results show that if the analysis does not produce $\afail$, the collecting semantics does not produce $\fail$ and therefore every execution of the program is satisfiable with respect to the inference plan.

\section{Evaluation}
\label{sec:eval}
In this section, we empirically evaluate the efficacy of \siren{} on a set of probabilistic programs. 
We also empirically evaluate how good the inference plan satisfiability analysis is at identifying whether an inference plan is satisfiable. We seek to answer these research questions: 

\textbf{\emph{RQ1.}} Can inference plans improve hybrid particle filtering performance? In other words, does there exist an inference plan that improves program performance compared to the default plan?

\textbf{\emph{RQ2.}} How precise is the inference plan satisfiability analysis? Section~\ref{sec:properties} proves the analysis is sound, so it will never state an unsatisfiable inference plan is satisfiable. The task remains to empirically determine whether the analysis can detect satisfiable inference plans in practice.

\textbf{\emph{RQ3.}} How long does the inference plan satisfiability analysis take?

\subsection{Benchmarks}
We evaluate the performance of different hybrid particle filtering algorithms on a set of benchmark programs. 
We describe here the 11 benchmarks and identify the variables evaluated for accuracy.
The following benchmarks are benchmarks with multiple inference plans from prior work on SSI and DS by \citet{atkinson2022semi} and \citet{baudart2020reactive}: \bOutlier{}, \bGtree{}, \bSlam{}, and \bWheels{}. We describe the programs and the evaluated variables in \Cref{appendix:benchmarks}. We also added the following additional benchmarks, each of which cannot be solved purely with exact inference. They demonstrate the advantages of using inference plans for improving performance against the default behavior.

\bAircraft{} is the program presented in \Cref{sec:example}. 

\bNoise{} is a one-dimensional particle filter with a hidden state modeled by Gaussian distributions (\zl{x}) with variance modeled by an Inverse-Gamma distribution (\zl{q}). Observations are made on Gaussian distributions centered around the previous state with variance also modeled by an inverse-Gamma distribution (\zl{r}). This program is adapted from~\citet{dunik2017noise}. 

\bRadar{} is a radar tracker with glint noise modeled as a random, rarely occurring spike~\citep{wu1993target}. The observation noise is modeled by the sum of two independent Inverse-Gamma distributions (\zl{r} and \zl{other}) if the \zl{env} random variable -- modeled by a Bernoulli -- indicates a spike. This differs from the \bAircraft{} program in that it only models the Gaussian-distributed \zl{x} position, and the Bernoulli random variable determines whether to observe a spike in the measurement noise.

\bEnvnoise{} is similar to \bRadar{}, but the noise variable \zl{other} is modeled by the more flexible Beta distribution~\citep{ma2011bayesian,arazo2019unsupervised}.

\bOutlierheavy{} models a one-dimensional particle filter where there might be sensor errors producing outlier observations. The hidden state is modeled by Gaussian distributions (\zl{xt}) and the sensor error rate as a Beta prior (\zl{outlier_prob}) to a Bernoulli. The regular observations are made on Gaussian distributions and the outlier observations are modeled by a long-tailed location-scale $t$ distribution as used in \citet{chang2014robust}. This program is an extension of the \bOutlier{} benchmark implemented by~\citet{atkinson2022semi} and adapted from~\citet{minka2013expectation}.

\bSlds{} is a switching linear dynamical system adapted from \citet{obermeyer2019functional}. The model switches between two nonlinear Kalman filters that each have unknown measurement noises. The switching label follows Markovian dynamics and is modeled by two Beta priors (\zl{trans_prob0} and \zl{trans_prob1}). The filters use Gaussian distributed hidden states (\zl{x0} and \zl{x1}). The measurement noises are modeled by Inverse-Gamma distributions (\zl{obs_noise0}, \zl{obs_noise1}).

\bRunner{}, adapted from \citet{azizian2023automatic}, models the 2-D position (\zl{x}, \zl{y}) and speed (\zl{sx}, \zl{xy}) of a runner based on speedometer readings and the altitude, modeled by Gaussian distributions.

We include the source code and the annotations of each evaluated inference plan for each benchmark in \Cref{appendix:benchmarks}.

\subsection{Methodology}

We implemented \siren{} in Python. In addition to semi-symbolic inference and delayed sampling, we also implemented a third hybrid inference algorithm: SMC with belief propagation (SMC with BP)~\citep{azizian2023automatic}. The algorithm swaps parent-child dependencies using conjugate distributions similar to SSI and maintains node types in the data field like DS. The implementation is available at \url{https://github.com/psg-mit/siren/tree/main}.

\subsubsection{RQ1 Methodology}
To determine \siren{}'s performance for \emph{RQ1},
we execute each benchmark 100 times for 100 timesteps with an exponentially increasing particle count from 1 to 1024.
We execute each benchmark using all satisfiable inference plans, except for \bSlds{} with SSI and DS, where due to the large number of inference plans, we sort the plans by the number of \mfSymbolic{} variables in descending order and only compare the first 4 plans (and any plans tied with those) against the plan with all $\mfSample$d variables and the default plan. We set the timeout to 300 seconds.
We measure the accuracy by the Mean Squared Error of the posterior expected value of the variable compared to its ground truth value, which is available as all data were generated by sampling from the prior. For each benchmark, we compute the speedup each inference plan achieves compared to the default plan to reach the \emph{target accuracy}, defined as the 90th percentile of error by the default plan using the greatest particle count evaluated that did not timeout. Following \citet{atkinson2022semi} and \citet{baudart2020reactive}, reaching target accuracy is defined as $\log(P_{90\%}(\mit{loss})) - \log(\mit{loss}_{\mit{target}}) < 0.5$. We also compute the summary statistics of the increase in accuracy each plan achieves compared to the default plan with less than or equal runtime. We conduct experiments on a 60-vCPU Intel Xeon Cascade Lake (up to 3.9 GHz) node with 240 GB RAM. 

\subsubsection{RQ2 Methodology}
To determine \siren{}'s analysis precision for \emph{RQ2}, we enumerated all satisfiable and unsatisfiable inference plans, and measured if the inference plan satisfiability analysis correctly determines the satisfiability of each plan. 

\subsubsection{RQ3 Methodology}
For each program, algorithm, and inference plan (satisfiable or not), we measured the runtime of the analysis to answer RQ3. 

\subsection{Results}

\subsubsection{RQ1 Results}

Across all benchmarks, variables, and inference algorithms, using the best inference plans produces an average speedup of 1.76x to reach the same target accuracy compared to the default plans, with a maximum speedup of 206x. The best inference plans achieve 1.83x better accuracy with equal or less runtime compared to the default plans, with a maximum increase of 595x. See \Cref{appendix:evaluation} for the breakdowns for each benchmark.

Figure~\ref{fig:performance-results-1} plots the experiment results of \bOutlier{} and \bNoise{} using the SSI algorithm. 
For each particle count, we plot the median runtime to the 90th percentile of error for each evaluated variable. 
For the plots of the remaining algorithms and benchmarks, see Appendix~\ref{appendix:evaluation}.

\begin{figure}[t]
  \centering
  \begin{subfigure}[c]{1\textwidth}
    \centering
    \includegraphics[width=0.5\textwidth]{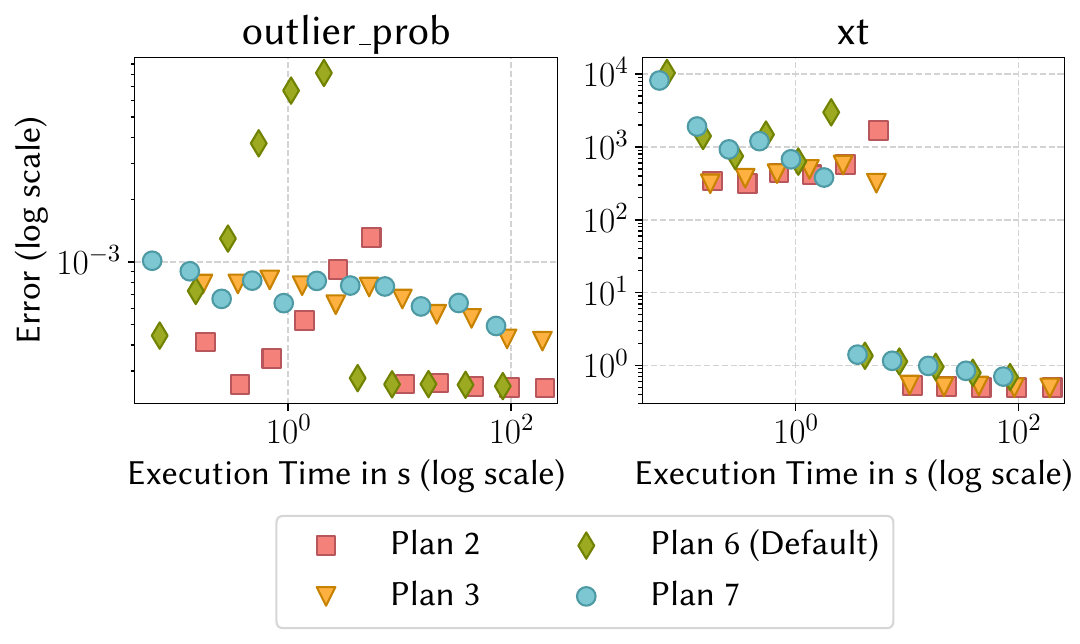}
    \caption{\bOutlier{}.}
  \end{subfigure}%
  \\
  \begin{subfigure}[c]{1\textwidth}
    \centering
    \includegraphics[width=0.75\textwidth]{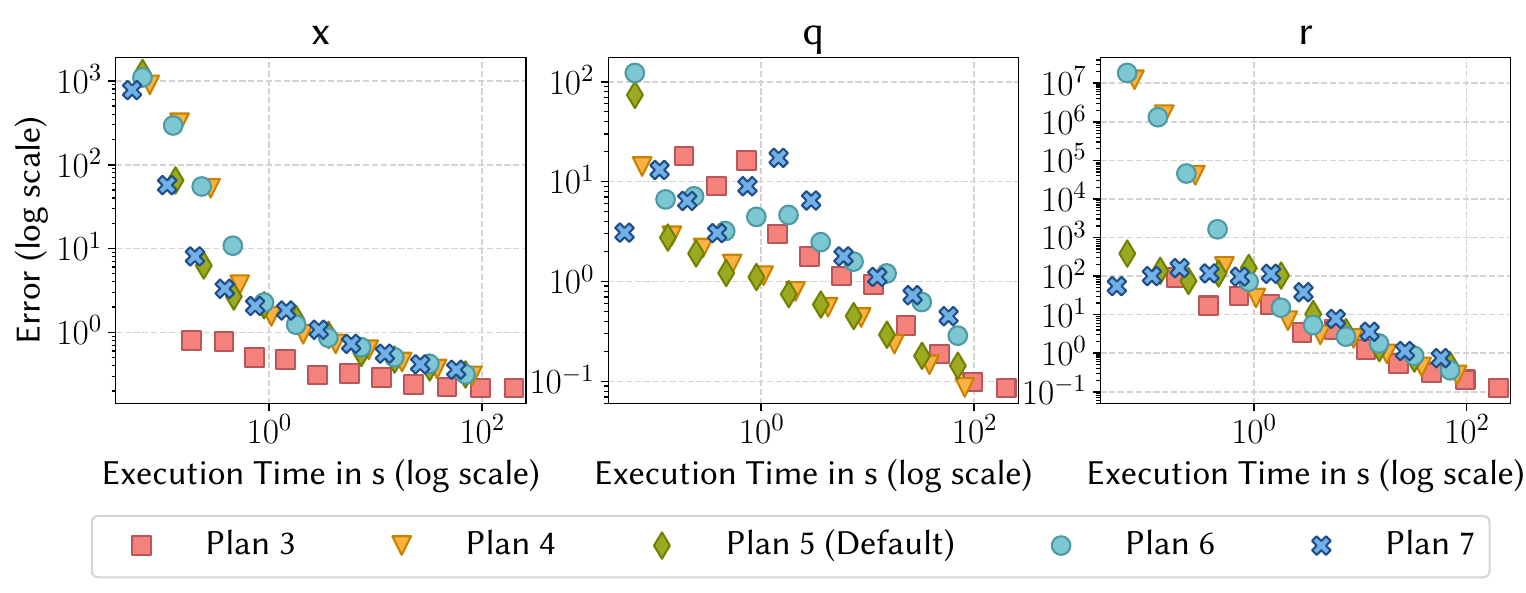}
    \caption{\bNoise{}.}
  \end{subfigure}%
  \caption{For each particle count, the plots show the median execution time to the 90th percentile of error for each variable using different satisfiable inference plan.}
  \label{fig:performance-results-1}
  \vspace{-0.8em}
\end{figure}

The \bOutlier{} results in \Cref{fig:performance-results-1} show that, when the execution time is greater than 10 seconds, the annotated Plan 2 achieves better accuracy for \zl{xt} than the default Plan 6 and would be preferred in a context where the developer cares the most about \zl{xt}. Overall, aggregated across particle counts, the best accuracy achieved by any plan with less or equal runtime to the default plan is 2.75x better than the default for \zl{outlier_prob} and 1.89x for \zl{xt}.
Additionally, while the default Plan 6 has the fastest execution time to reach target accuracy for \zl{outlier_prob}, the alternative Plan 7 achieves the fastest execution time to reach target accuracy for \zl{xt} offering a speedup of 1.16x over the default plan.
Thus, the best inference plan to use depends on the context in which the system is deployed.

For \bNoise{}, the best accuracy achieved by any plan with less or equal runtime to the default plan is 2.36x better than the default plan for \zl{x}, 1.33x for \zl{q}, and 2.83x for \zl{r} on average across particle counts.
The results in \Cref{fig:performance-results-1} show that the annotated Plan 3 achieves the lowest error for \zl{x}, but the alternative Plan 4 and the default Plan 5 achieve the lowest error for \zl{q}. In terms of efficiency, the default Plan 5 has the fastest execution time to reach target accuracy for \zl{q}. But the alternative Plan 3 has the fastest execution time for \zl{x} and \zl{r}, with a speedup of 19x and 1.33x over the default plan, respectively. 
Thus, the inference plan that should be used depends on the context -- in particular which variable the developer considers most important.

Overall, these results demonstrate that inference performance depends strongly on the inference plan, that the plan that should be used to execute the program is context-dependent, and that annotated inference plans enable substantial speedups and increase in accuracy over the default.

\subsubsection{RQ2 Results}
To evaluate the analysis precision for \emph{RQ2}, we run the analysis on all 11 benchmarks and 3 inference algorithms, and summarize the results in Table~\ref{tab:analysis-results}. 
We manually count the total number of satisfiable plans 
in each case as well as how many of these plans the analysis identifies.
The analysis identifies all satisfiable plans in 27 out of the 33 evaluated settings.
This shows that the analysis is precise in practice. We describe below the edge cases where the analysis does not identify satisfiable plans.

\begin{table}
  \small
  \centering
  \caption{Number of satisfiable inference plans identified by the inference plan satisfiability analysis out of the total number of satisfiable inference plans for each benchmark and algorithm. }
  \begin{tabular}[h]{lccccccccccc}
    \toprule
    & \multicolumn{11}{c}{Identified Satisfiable Plans / Satisfiable Plans}\\
    \cmidrule(lr){2-12}
    Algorithm & \bNoise{} & \bRadar{} & \emph{EnvN} & \emph{Out} & \emph{OutH} & \bGtree{} & \bSlds{} & \bRunner{} & \bWheels{} & \bSlam{} & \emph{Air} \\
\midrule
\ssi & 5/5 & 3/3 & 3/3 & 4/4 & 2/2 & 3/4 & 28/36 & 4/4 & 3/4 & 3/4 & 3/3 \\
\ds & 4/4 & 2/2 & 2/2 & 2/2 & 2/2 & 3/3 & 16/16 & 1/1 & 1/3 & 2/2 & 2/2 \\
\bp & 2/2 & 2/2 & 2/2 & 2/2 & 1/1 & 3/4 & 4/4 & 4/4 & 3/3 & 1/1 & 2/2 \\
\midrule
Total Plans & 8 & 32 & 32 & 8 & 8 & 4 & 128 & 16 & 4 & 4 & 32 \\
  \bottomrule
  \end{tabular}
  \label{tab:analysis-results}
\end{table}

\paragraph{Aliasing}
The loss of precision in \bSlds{} executed with SSI is due to \emph{aliasing}.
When joining expressions in conditionals and \zl{fold} fixpoint computations, the analysis loses information. 
This introduces an aliasing problem because inconsistent branch conditions are not detected. This is illustrated in the following program,
\begin{lstlisting}
  let symbolic x1 <- gaussian(0.,1.) in
  let symbolic var1 <- invgamma(1.,1.) in
  observe(gaussian(if cond then x1 else 1., if !cond then var1 else 1.), obs)
\end{lstlisting}
where \mkw{cond} and \mkw{obs} are constants. 
Because \mkw{cond} and \mkw{!cond} are inconsistent branch conditions and the \mkw{else}-branches are both constants, in any execution SSI only needs a conjugacy relation for \emph{either} \mkw{x1} or \mkw{var1}.
Such a single-variable conjugacy relation always exists, so annotating both \mkw{x1} and \mkw{var1} $\mfSymbolic$ will not throw an error in any execution.
However, the analysis approximates the observed Gaussian distribution as $\anormal{\arandomvar_{x1}}{\arandomvar_{\mit{var}1}}$, meaning that the distribution is potentially depending on \emph{both} \mkw{x1} and \mkw{var1}, which would require a conjugacy relationship for both variables simultaneously. 
SSI does not support this, so the analysis concludes that SSI may throw an error when in fact no such error-throwing execution exists.
This problem manifests in \bSlds{} with SSI but does not affect any other benchmarks. 

\paragraph{Widening Expressions}
In \bSlam{} with SSI, the analysis widens abstract expressions to 
$\eunk{}{\avarset{}}$ when the expression tree depth is over the preset
threshold because large symbolic expressions can be computationally expensive. It also widens $\eunk{}{\avarset{}}$ to \etop{} when the number of variables in $\avarset{}$ is over the preset threshold to hasten the convergence of the fixpoint computation during a \mkw{fold} because $\avarset{}$ can grow indefinitely large.
However, 
the \etop{} expression is not precise enough for the analysis to detect conjugacies that SSI needs to perform symbolic computation in \bSlam{} when all the variables are annotated \mfSymbolic.

\paragraph{Syntactic Partial Ordering Comparisons}
The analysis fails to identify satisfiable plans in \bGtree{} with SSI and SMC with BP and \bWheels{} with SSI and DS because the partial ordering of abstract expressions performs comparisons syntactically. For example, $\amult{\cunk}{\arandomvar}$ and $\aplus{\amult{\cunk}{\arandomvar}}{\cunk}$ do not share syntactic expression structure, so joining the expressions produces $\eunk{}{\{\arandomvar\}}$. The over-approximated expression $\eunk{}{\{\arandomvar\}}$ is not considered an affine expression with respect to $\arandomvar$, causing the analysis to fail to identify linear-Gaussian conjugacies. 

\subsubsection{RQ3 Results}
\label{app:analysis-runtime}
We summarize the time taken by the analysis averaged over the all inference plans for each evaluated setting in Table~\ref{tab:analysis-time}.
Overall, the analysis takes less than 1 second for all benchmarks, algorithms, and inference plans, except for SLAM with SSI using the default Plan 0. Plan 0 annotates all variables with symbolic, which SSI can implement on SLAM. However, the symbolic computation here results in performing exact inference on a program with only Bernoullis, which is intractable as the conjugacy transformation exponentially increases the expression size. By lowering the widening thresholds, this time could be sped up, at the expense of precision. Nevertheless, the analysis is in general fast to perform.

\begin{table}
  \small
  \centering
  \caption{Time taken by the analysis averaged over all inference plans for each benchmark and algorithm in seconds. The analysis time for \bSlam{} with SSI is exceptionally long due to the intractable symbolic computation of full exact inference.}
  \begin{tabular}[h]{lccccccccccc}
    \toprule
    Algorithm & \bNoise{} & \bRadar{} & \emph{EnvN} & \emph{Out} & \emph{OutH} & \bGtree{} & \bSlds{} & \bRunner{} & \bWheels{} & \bSlam{} & \emph{Air} \\
\midrule
\ssi & 0.40 & 0.41 & 0.42 & 0.41 & 0.41 & 0.40 & 0.50 & 0.54 & 0.45 & 40.93 & 0.46 \\
\ds & 0.39 & 0.40 & 0.41 & 0.41 & 0.40 & 0.38 & 0.50 & 0.54 & 0.45 & 0.53 & 0.46 \\
\bp & 0.39 & 0.40 & 0.42 & 0.41 & 0.40 & 0.38 & 0.49 & 0.54 & 0.45 & 0.51 & 0.46 \\
  \bottomrule
  \end{tabular}
  \label{tab:analysis-time}
\end{table}

\section{Related Work}
\label{sec:related-work}

\paragraph{Hybrid Inference}
\siren{} supports hybrid inference algorithms based on particle filtering. Other Monte Carlo inference methods can also be combined with exact inference; Hakaru~\citep{narayanan2016probabilistic}, Autoconj~\citep{hoffman2018autoconj}, and automatic marginalization~\citep{lai2023automatically} perform static transformations to solve the model analytically, and allow the rest to be solved with Monte Carlo methods such as Metropolis-Hastings and HMC. 
These Monte Carlo methods have the same key feature as particle filtering, which is that there is a subset of random variables that, when reduced to constant values, allows the algorithm to analytically solve the rest of the inference problem. This paradigm naturally leads to the concept of partitioning random variables in a probabilistic model into \mfSample{} and \mfSymbolic{} random variables.
In systems that perform static transformations, the partitioning is inherently known at compile time. However, the concept of inference plans can still provide an explicit interface for reasoning about these partitions.
In the dynamic setting where the partitioning is only entirely determined at run time, a static analysis such as the one in \Cref{sec:analysis} is necessary to determine the satisfiability of inference plans.

To the best of our knowledge, no prior works have combined dynamic symbolic computations on Monte Carlo methods other than particle filtering. Developing such an algorithm is out of the scope of this work. Nevertheless, as a proof of concept of how the key ideas presented in this work extend to other Monte Carlo-based methods, we present an alternative semantics for \siren{} using a basic Metropolis-Hastings implementation combined with symbolic computations in \Cref{appendix:mh-siren}. We show that using inference plans can improve performance and that the analysis is still precise.

\paragraph{Programmable Inference}
Inference plans is an instance of programmable inference, where the programming system hands over control of the inference procedure to the user.
Other works in the programmable inference space~\citep{cusumano2019gen,mansinghka2014venture,mansinghka2018probabilistic,tehrani2020bean,lew2019trace} hand over control to the user at varying stages of inference for different inference paradigms. Our interface applies specifically to enable users to use alternative heuristics for hybrid inference algorithms. 

\paragraph{Probabilistic Program Analyses}
Several efforts have been made on using program analyses to detect structure to optimize in probabilistic programs~\citep{ritchie2016deep, huang2017compiling, gorinova2020automatic, zhou2020divide, cheng2021flip}. 
Other works also use program analyses to statically infer properties about the outputs or resource usage of probabilistic programs~\citep{gorinova2021conditional,lee2023smoothness,atkinson21statically,ngo2018bounded,monniaux2001abstract,monniaux2000abstract,cousot2012probabilistic,di2000concurrent,smith2008probabilistic,trilla2020probabilistic,wang2023newtonian}.
Our analysis infers properties about the runtime behavior of the probabilistic inference algorithm.

\section{Conclusion}
\label{sec:conclusion}

In this work, we present \siren{}, a new probabilistic programming language for hybrid inference.
\siren{} enables developers to use \emph{inference plans} to control the partitioning of random variables into sampled and symbolic variables.
To assist programmers in reasoning about inference plans in hybrid inference systems, \siren{} employs a static analysis that determines if an inference plan is satisfiable in all possible executions of the program.
Our design of the hybrid inference interface enables \siren{} to work with multiple hybrid inference algorithms, including semi-symbolic inference, delayed sampling, and SMC with belief propagation.

The promise of PPLs is to separate the task of probabilistic modeling from the complex low-level details of building an inference algorithm.
However, to achieve good performance in practice, developers often need control over the behavior of the inference system.
\siren{} brings custom hybrid inference to the paradigm of probabilistic programming: developers can adjust the behavior of the inference algorithm to achieve better performance while maintaining the separation of modeling and inference.

\section*{Data-Availability Statement}
The artifact of this work is available on Zenodo~\citep{cheng_2024_13924216}.

\begin{DIFnomarkup}
\begin{acks}                            %
This material is based upon work supported in-part by the National Science Foundation Graduate Research Fellowship under Grant No. 2141064, Sloan Foundation, SRC JUMP 2.0 (CoCoSys), and Amazon MIT Science Hub.
We thank Alex Renda, Charles Yuan, Tian Jin, Jesse Michel, and Logan Weber for helpful feedback
on this work.
\end{acks}
\end{DIFnomarkup}

\bibliography{main}


\par\bigskip\noindent{\small\normalfont{Received 2024-07-07; accepted 2024-11-07}\par}
\label{lastpage}

\newpage
\appendix
\section{Ideal Semantics}
\label{appendix:idealsem}

The ideal semantics of \siren{} is a measure-based semantics adapted from~\cite{staton17} and presented in \Cref{fig:sem-muf}. The type $t$ of an expression is \mfUnit, \mfReal, \mfInt, a product type, or a list. In the paper, we assume all programs are well-typed. 
In the ideal semantics, types are interpreted as measurable spaces.
Given an environment $\gamma$, the semantics of an expression $e$ of type $t$, $\sem{e}{\gamma} : \Sigma_{\sem{t}{}} \to [0, \infty]$ is a \emph{measure} which maps a measurable set of values of type $t$ to a positive score.\footnote{We write $\Sigma_A$ for a $\sigma$-algebra on $A$.}

The semantics of a deterministic value $v$ is the Dirac Delta measure on the corresponding value~($\delta_{\sem{v}{\gamma}}(U) = 1$ if $v \in U$ and $0$ otherwise).
The semantics of a local definition \mfLetIn{x}{e_1}{e_2}{} corresponds to integration, i.e., integrate the semantics of $e_2$ over all possible values for $e_1$.
The semantics of a local random variable is very similar, but instead of an arbitrary expression, a distribution is a value that is already a measure.
In this ideal semantics, annotations \mfSample{} and \mfSymbolic{} are simply ignored.
Finally, \mfObserve{v_1}{v_2} conditions the model using the likelihood of the value $v_2$ in the distribution $v_1$, i.e., the value of the \emph{probability density function} (pdf) of $d$ in $v$.\footnote{For simplicity, we also write \mit{pdf}{} for the \emph{probability mass function} of discrete distributions}

The semantics of the \mkw{fold} operator is defined with a local definition and a function call.
On an empty list, \mfFold{f}{\mfNil}{a} returns the value of the accumulator $a$.
Otherwise \mfFold{f}{\mfCons{h}{t}}{a} first applies the function $f$ on the head of the list $h$ and the accumulator $a$ to compute the next accumulator value, and then recursively applies \mkw{fold} on the new value and the tail of the list $t$.

Finally, the semantics of a program incrementally builds the environment $\gamma$ required to compute the measure corresponding to the main expression.
Because of the \mkw{observe} construct, this measure is, in general, \emph{unnormalized}.
We thus renormalize it to return the corresponding probability distribution.
If normalization is not possible, we return an error.

\[
\begin{array}{@{}lrl@{}}
\sem{\mit{prog} = d_1\ \dots \ d_n \ e}{} =
\begin{array}[t]{@{}l@{}}
\mit{let} \ \gamma = \sem{d_1 \ \dots \ d_n}{[]} \ \mit{in}\\
\mit{let}\ \mit{norm} = \psem{e}{\gamma}(\sem{\mit{typeOf}(e)}{}) \ \mit{in}\\
\begin{cases}
\left. \psem{e}{\gamma} \middle/ \mit{norm} \right. &\textit{if \; $0 < \mit{norm} < \infty$}\\
\left. \mfError \right. & \textit{otherwise}
\end{cases}
\end{array}
\end{array}
\]

\begin{figure}[H]
\small
\[
\begin{array}{@{}lcl@{}}
\sem{\mfUnit}{} &=& \textit{discrete measurable space with one element $()$}\\
\sem{\mfReal}{} &=& \textit{reals with Borel sets}\\
\sem{\mfInt}{} &=& \textit{integers}\\
\sem{t_1 \times t_2}{} &=& \textit{product space $\sem{t_1}{} \times \sem{t_2}{}$}\\
\sem{t\ \mtt{list}}{} &=& \textit{product space $\left\{\bigtimes_n \ \sem{t}{} \mid n \in \mathbb{N} \right\}$}\\
\\
\sem{c}{\gamma} &=& c\\
\sem{x}{\gamma} &=& \gamma(x)\\
\sem{\mfPair{v_1}{v_2}}{\gamma} &=& (\sem{v _1}{\gamma}, \sem{v_2}{\gamma})\\
\sem{\mfNil}{\gamma} &=& []\\
\sem{\mfCons{v_1}{v_2}}{\gamma} &=& \sem{v_1}{\gamma} :: \sem{v_2}{\gamma}\\
\\
\psem{v}{\gamma} &=& \delta_{\sem{v}{\gamma}}\\
\psem{\mfApp{f}{v}}{\gamma} &=& \gamma(f)(\sem{v}{\gamma})\\
\psem{\mfIf{v}{e_1}{e_2}}{\gamma} &=& \mit{if}\ \sem{v}{\gamma} \ \mit{then} \ \psem{e_1}{\gamma} \ \mit{else} \ \psem{e_2}{\gamma}\\
\psem{\mfLetIn{x}{e_1}{e_2}}{\gamma} &=& \int \psem{e_1}{\gamma}(dr)\ \psem{e_2}{\gamma + [x \is r]}\\
\psem{\mfLetRv{\annotation}{x}{\mfOp{v}}{e}}{\gamma} &=& \int \sem{\mfOp{v}}{\gamma}(dr)\ \psem{e}{\gamma + [x \is r]}\\
\psem{\mfObserve{v_1}{v_2}}{\gamma} &=& \mit{pdf}(\sem{v_1}{\gamma})(\sem{v_2}{\gamma}) * \delta_{()}\\
\psem{\mfFold{f}{\mfNil}{v}}{\gamma} &=& \psem{v}{\gamma}\\
\psem{\mfFold{f}{\mfCons{h}{t}}{v}}{\gamma} &=& 
\psem{\mfLetIn{a}{f(h, v)}{\mfFold{f}{t}{a}}}{\gamma}
\\
\\
\sem{\mfLetFun{f}{x}{e}}{\gamma} &=& \gamma + \left[f \is \left(\lambda v. \ \psem{e}{\gamma + [x \is v]}\right)\right]\\
\sem{d_1 \ d_2}{\gamma} &=& \sem{d_2}{\sem{d_1}{\gamma}}
\end{array}
\]
\caption{Ideal measure-based semantics of \siren.}
\label{fig:sem-muf}
\end{figure}

\newpage

\section{Operational Semantics}
\label{appendix:op-sem}

\subsection{Big Step Semantics with Checkpoints}
We present the full operational semantics with resampling checkpoints in \Cref{fig:op:sem-full}.
\begin{figure}[H]
\begin{small}
\begin{mathpar}
\inferrule%
{ }
{\pclstep{\val}{\symbstate}
         {\val}{\symbstate}{1}{\false}}

\inferrule%
{ }
{\pclstep{\mfResample}{g}
         {\mfUnit}{g}{1}{\true}}

\inferrule%
{\mfLetFun{f}{\programvar}{\progexpr} \\
 \pclstep{\progexpr[\programvar \leftarrow \val]}{\symbstate}
         {\progexpr'}{\symbstate'}{\weight}{\doresample}}
{\pclstep{\mfApp{f}{\val}}{\symbstate}
         {\progexpr'}{\symbstate'}{\weight}{\doresample}}

\inferrule%
{\pure(\progexpr{}_1, \progexpr{}_2)\\
\neg\const(\val)\\
 \pclstep{\progexpr{}_1}{\symbstate}
         {\val{}_1}{\symbstate_1}{1}{\false}\\
 \pclstep{\progexpr{}_2}{\symbstate_1}
         {\val{}_2}{\symbstate'}{1}{\false} 
}
{\pclstep{\mfIf{\val}{\progexpr{}_1}{\progexpr{}_2}}{\symbstate}
         {\ite{\val}{\val_1}{\val_2}}{\symbstate'}{1}{\false}}
         
\inferrule%
{\symvalue^*(\val, \symbstate) = \true, \symbstate_\val \\
 \pclstep{\progexpr_1}{\symbstate_\val}
         {\progexpr_1'}{\symbstate'}{\weight}{\doresample}}
{\pclstep{\mfIf{\val}{\progexpr_1}{\progexpr_2}}{\symbstate}
         {\progexpr_1'}{\symbstate'}{\weight}{\doresample}}

\inferrule%
{\symvalue^*(\val, \symbstate) = \false, \symbstate_\val \\
 \pclstep{\progexpr_2}{\symbstate_\val}
         {\progexpr_2'}{\symbstate'}{\weight}{\doresample}}
{\pclstep{\mfIf{\val}{\progexpr_1}{\progexpr_2}}{\symbstate}
         {\progexpr_2'}{\symbstate'}{\weight}{\doresample}}

\inferrule%
{\pclstep{\progexpr_1}{\symbstate}
         {\progexpr_1'}{\symbstate'}{\weight}{\true}}
{\pclstep{\mfLetIn{\programvar}{\progexpr_1}{\progexpr_2}}{\symbstate}
         {\mfLetIn{\programvar}{\progexpr_1'}{\progexpr_2}}{\symbstate'}{\weight}{\true}}

\inferrule%
{\pclstep{\progexpr_1}{\symbstate}
         {\val_1}{\symbstate_1}{\weight_1}{\false}\\
 \pclstep{\progexpr_2[\programvar \leftarrow \val_1]}{\symbstate_1}
         {\progexpr_2'}{\symbstate_2}{\weight_2}{\doresample}}
{\pclstep{\mfLetIn{\programvar}{\progexpr_1}{\progexpr_2}}{\symbstate}
         {\progexpr_2'}{\symbstate_2}{\weight_1 * \weight_2}{\doresample}}

\inferrule%
{ }
{\pclstep{\mfFold{f}{\mfNil}{\val}}{\symbstate}
         {\val}{\symbstate}{1}{\false}}

\inferrule%
{
    {\pclstep{\mfLetIn{\programvar}{\mfApp{f}{\mfPair{\listhd}{v}}}{\mfFold{f}{\listtl}{\programvar}}}{\symbstate}
         {\progexpr}{\symbstate'}{\weight}{\doresample}}
}
{\pclstep{\mfFold{f}{\mfCons{\listhd}{\listtl}}{\val}}{\symbstate}
         {\progexpr}{\symbstate'}{\weight}{\doresample}}

\inferrule%
{\annotation \in \{\varepsilon, \mfSymbolic\}\\
\symassume(\annotation, \mfApp{\distop}{\val}, \symbstate) = \randomvar, \symbstate_\randomvar\\
 \pclstep{\progexpr[\programvar \leftarrow \randomvar]}{\symbstate_\randomvar}
         {\progexpr'}{\symbstate'}{\weight}{\doresample}}
{\pclstep{\mfLetRv{\annotation}{\programvar}{\mfApp{\distop}{\val}}{\progexpr}}{\symbstate}
         {\progexpr'}{\symbstate'}{\weight}{\doresample}}

\inferrule%
{\symassume(\mfSample, \mfApp{\distop}{\val}, \symbstate) = \randomvar, \symbstate_\randomvar\\
\symvalue(\randomvar, \symbstate_\randomvar) = \val_\programvar, \symbstate_\randomvar'\\
 \pclstep{\progexpr[\programvar \leftarrow \val_\programvar]}{\symbstate_\randomvar'}
         {\progexpr'}{\symbstate'}{\weight}{\doresample}}
{\pclstep{\mfLetRv{\mfSample}{\programvar}{\mfApp{\distop}{\val}}{\progexpr}}{\symbstate}
         {\progexpr'}{\symbstate'}{\weight}{\doresample}}

\inferrule%
{\symassume(\varepsilon, \mfApp{\distop}{\val_1}, \symbstate) = \randomvar, \symbstate_\randomvar\\
 \symvalue^*(\val_2, \symbstate_\randomvar) = \val, \symbstate_\val\\
 \symobserve(\randomvar, \val, \symbstate_\val) = \symbstate', \weight}
{\pclstep{\mfObserve{\mfApp{\distop}{\val_1}}{\val_2}}{\symbstate}
         {\mfUnit}{\symbstate'}{\weight}{\false}}

\inferrule{
  \set{ \pclstep{\progexpr_i}{\symbstate_i}{\val_i}{\symbstate_i'}{\weight_i}{\false} }_{1 \le i \le N}\\
  \set{\distribution(\val_i, \symbstate_i') = \distr{}_i}_{1 \le i \le N}\\
  \Weight = \textstyle{\sum_{1 \le i \le N}} \weight_i
}
{\pclsstep{\set{\progexpr_i, \symbstate_i}_{1 \le i \le N}}{\textstyle{\sum_{1 \le i \le N}} \dfrac{\weight_i}{\Weight} \times \distr{}_i }}

\inferrule{
  \set{ \pclstep{\progexpr_i}{\symbstate_i}
                {\progexpr_i'}{\symbstate_i'}{\weight_i}{\doresample_i} }_{1 \le i \le N}\\
  \textstyle{\bigvee_{1 \le i \le N}}{\doresample_i} \\
  \mu = \categorical\left(\set{ \weight_i, (\progexpr_i', \symbstate_i') }_{1 \le i \le N}\right) \\
  \pclsstep{\set{\draw(\mu)}_{1 \le i \le N}}{\distr{}}
}
{\pclsstep{\set{\progexpr_i, \symbstate_i}_{1 \le i \le N}}{\distr{}}}

\inferrule{
    \pclsstep{\set{\progexpr, \emptyset}_{1 \le i \le N}}{\distr{}}
}
{\pclpstep{N}{\progexpr}{\distr{}}}

\end{mathpar}
\end{small}
\caption{Big-step semantics with checkpoints for the resample step. The semantics are described by three rules: particle evaluation rules, denoted by $\downarrow^r$; particle set evaluation rules, denoted by $\downdownarrows$; and model evaluation rules, denoted by $\Downarrow$. 
}
\label{fig:op:sem-full}
\end{figure}

\subsection{Implementing Hybrid Inference Interface}
We present here the definitions for the interface operations. We defer the auxiliary operations to the respective works.

\subsubsection{Semi-symbolic Inference}
The interface operations are implemented using \textsc{hoist} and \textsc{intervene}. We discuss \textsc{hoist} here, but leave other auxiliary operations to \citet{atkinson2022semi}.
\begin{align*}
\symassume(\annotation, \distr{}, \symbstate) &= \; 
    \letin{\symbstate' = \symbstate[\randomvar_\mit{new} \mapsto (\annotation,\distr{})]}
    (\randomvar_\mit{new}, \symbstate')\\
\symvalue(\randomvar, \symbstate) &= \;
    \mit{let} \; \symbstate' = \textsc{hoist}(\randomvar, \symbstate) \; \mit{in}\; 
    \mit{let} \; v = \textsc{draw}(\symbstate'(\randomvar)_d) \mit{in} \;        
    (v, \textsc{intervene}(\randomvar, \SSIdeltasample{v}, \symbstate'))\\
\symobserve(\randomvar, \val, \symbstate) &= \;
    \begin{array}[t]{@{}l@{}}
    \mit{let} \; \symbstate' = \textsc{hoist}(\randomvar, \symbstate) \; \mit{in}\; 
    \mit{let} \; \scoreval = \score(\symbstate'(\randomvar)_d) \mit{in} \;        
    (\textsc{intervene}(\randomvar, \deltad{\val}, \symbstate'), \scoreval)
    \end{array}
\end{align*}

\paragraph{\textsc{hoist}}
We present the definition of \textsc{hoist} in \Cref{alg:hoist}.
The $\textsc{hoist\_helper}$ operation recursively calls itself on the parents of $\randomvar_\mit{cur}$ in topological order. The topological ordering is a requirement to prevent creating cycles in the dependencies of the random variables as we perform swaps. After each recursive call, the hoisted parent is added to the next calls' $\texttt{roots}$ set. Members of the $\texttt{roots}$ set are not hoisted or swapped. The $\texttt{roots}$ set enforces that subsequent parents will only be descendants of the earlier parents so no cycles will be created.
Then, $\textsc{hoist\_helper}$ iterates through all parents in reverse topological order and swaps $\randomvar_\mit{cur}$ with each.
If $\randomvar_\mit{cur}$ cannot be swapped with one of its parents, the operation throws an exception that is caught at the outermost level.

The $\textsc{hoist}$ operation calls $\textsc{hoist\_helper}$ with an empty $\texttt{roots}$ set, thus turning $\randomvar_\mit{in}$ into a root.
In the event of an exception, it forces the parent of the offending variable to be sampled using the \textsc{Value} operation.
It then uses the $\textsc{eval}^*$ operation to eliminate the resulting Delta distribution from the symbolic state.
After thus eliminating this variable, it again attempts to \textsc{hoist} the variable $\randomvar_\mit{in}$.

\begin{algorithm}[H]
  \small
  \begin{algorithmic}
  \Function{hoist\_helper}{$\randomvar_\mit{cur}$, \texttt{roots}, $\symbstate$}
  \State \texttt{parents} $\gets$ $\textsc{topo\_sort}(\textsc{get\_parents}(\randomvar_\mit{cur}, \symbstate))$
  \State \texttt{roots}' $\gets$ \texttt{roots};\; $\symbstate' \gets \symbstate$
  \For{$\randomvar_\mit{par} \in $ \texttt{parents}}
  \If{$\randomvar_\mit{par} \not \in $ \texttt{roots}}
  \State $\symbstate' \gets$ \Call{hoist\_helper}{$\randomvar_\mit{par}$, \texttt{roots}', $\symbstate'$}$;\;$ \texttt{roots}' $\gets \randomvar_\mit{par} ::$ \texttt{roots}'
  \EndIf
  \EndFor
  \State $\symbstate'' \gets \symbstate'$
  \For{$\randomvar_\mit{par} \in $ \textsc{reverse}(\texttt{parents})}
  \If{$\randomvar_\mit{par} \not \in $ \texttt{roots}'}
  \State $(\symbstate'', \texttt{conjugate}) \gets$ \Call{swap}{$\randomvar_\mit{par}$, $\randomvar_\mit{cur}$, $\symbstate''$}
  \If{not \texttt{conjugate}}
  \State \textbf{throw} $(\randomvar_\mit{par}, \randomvar_\mit{cur})$
  \EndIf
  \EndIf
  \EndFor
  \State \Return $\symbstate''$
  \EndFunction
  \Function{hoist}{$\randomvar_\mit{in}$, $\symbstate$}
  \try
  \State \Return \Call{hoist\_helper}{$\randomvar_\mit{in}$, $\{\}$, $\symbstate$}
  \catch{$(\randomvar_\mit{par}, \randomvar_\mit{child})$}
  \State $(\_, \symbstate') \gets$ \Call{\textsc{value}}{$\randomvar_{\mit{par}}, \symbstate$}$;\;$ $\symbstate'' \gets \textsc{eval}^*(\randomvar_\mit{child}, \symbstate');\; $ \Return \Call{hoist}{$\randomvar_\mit{in}$, $\symbstate''$}
  \endtry
  \EndFunction
  \end{algorithmic}
  \caption{\textsc{hoist}: Hoisting a random variable to be a root depending on no other random variables.}
  \label{alg:hoist}
\end{algorithm}

\subsubsection{Delayed Sampling}
The interface operations for DS are implemented using $\graft$ and other helper operations. 
The recursive structure of $\graft$, when combined with the additional helper functions $\prune$ -- which realizes a Marginalized child -- and $\marginalize$ -- which symbolically marginalizes a conditional distribution -- maintains the invariant that each tree has a single Marginalized path.
The operation $\conjdistr$ takes a symbolic distribution $\distr{}$ and ensures that $\distr{}$ has at most one parent and that the parent must be a conjugate prior. If $\distr{}$ has more than one parent or the parent is not a conjugate prior, those parents are sampled with $\symvalue$. 
The operation $\initialize$ is used for inserting nodes into the tree. It takes a new random variable $\randomvar$ and a symbolic distribution $\distr{}$ and checks if $\distr{}$ is a conditional distribution with parents. If $\randomvar$ has a single parent that is a conjugate prior to $\distr{}$, then $\randomvar$ is assigned the Initialized node state, with $\distr{}$ as the prior. If $\randomvar$ has no parents, then $\randomvar$ is assigned the Marginalized state with no parent. The operation also inserts the annotation into the entry of the random variable. We discuss $\graft$ here, but we defer the full definitions of these operations to \citet{lunden2017delayed,murray2018delayed}.
\begin{align*}
  \symassume(\annotation, \distr{}, \symbstate) &= \; 
      \begin{array}[t]{@{}l@{}}
        \letin{\distr{\prime}, \symbstate' = \conjdistr(\distr{}, \symbstate)}\\
      (\randomvar_\mit{new}, \initialize(\randomvar_\mit{new}, \annotation, \distr{\prime}, \symbstate'))\\
      \end{array}\\
  \symvalue(\randomvar, \symbstate) &= \;
      \mit{let} \; \symbstate' = \textsc{graft}(\randomvar, \symbstate) \; \mit{in}\; 
      \mit{let} \; v = \draw(\symbstate'(\randomvar)_d) \mit{in} \;        
      (v, \textsc{realize}(\randomvar, \SSIdeltasample{v}, \symbstate'))\\
  \symobserve(\randomvar, \val, \symbstate) &= \;
      \begin{array}[t]{@{}l@{}}
      \mit{let} \; \symbstate' = \textsc{graft}(\randomvar, \symbstate) \; \mit{in}\; 
      \mit{let} \; \scoreval = \score(\symbstate'(\randomvar)_d) \mit{in} \;        
      (\textsc{realize}(\randomvar, \deltad{\val}, \symbstate'), \scoreval)
      \end{array}
  \end{align*}

\paragraph{Graft}
DS observes or samples a random variable that is a \emph{terminal} node (i.e. a Marginalized node with no Marginalized children). DS turns variables into terminal nodes with a \emph{grafting} algorithm that recursively converts Initialized nodes into Marginalized nodes by symbolically marginalizing the parent of Initialized nodes and removing edges to its Marginalized child by sampling. We present the definition of $\graft$ below, which uses helper operations \textsc{marginalized\_child}, $\prune$, and $\marginalize$. \textsc{marginalized\_child} returned a Marginalized child node, and $\prune$ recursively samples the downstream Marginalized children using $\symvalue$. $\marginalize$ symbolically transforms an Initialized node with a conditional distribution into a Marginalized node with a marginal distribution. The $\graft$ operation uses the node type to determine which operations convert the node into a terminal node: If the node is a Marginalized node, $\graft$ retrieves its Marginalized child and removes the edge to that subtree using $\prune$; if the node is an Initialized node, $\graft$ ensures the parent node is also terminal by recursively calling $\graft$ and then symbolically marginalizing with $\marginalize$. At the end of $\graft$, the variable $\randomvar$ is a terminal node.
\begin{align*}
\graft(\randomvar, \symbstate) = \;& 
\begin{array}[t]{@{}l@{}}
    \mit{match}\; \symbstateS{}{\randomvar}\; \mit{with}\\
    \begin{array}[t]{@{}l@{}}
        |\; \marginalizedroot{\varset} \;|\; \marginalized{\randomvar'}{\distr{}}{\varset}:\\
        \quad \begin{array}[t]{@{}l@{}}
            \letin{\randomvar_\mit{child} = \textsc{marginalized\_child}(\varset)}
            \prune(\randomvar_\mit{child}, \symbstate)
        \end{array}\\
        |\; \initialized{\randomvar'}{\varset}:
        \begin{array}[t]{@{}l@{}}
            \letin{\symbstate' = \graft(\randomvar', \symbstate)}
            \marginalize(\randomvar, \symbstate')
        \end{array}\\
    \end{array}
\end{array}
\end{align*}

\section{Inference Plan Satisfiability Analysis}
\label{appendix:abs-interp-rules}

\subsection{Interpretation Rules}
We present the full set of rules in \Cref{fig:abstract-interp-full}.

\begin{figure}[H]
  \centering
  \begin{small}
  \begin{mathpar}
  \inferrule%
  { }
  {\apclstep{\aval}{\asymbstate{}}
            {\aval}{\asymbstate{}}}

  \inferrule%
  { }
  {\apclstep{\mfResample}{\asymbstate{}}
           {\amfUnit}{\asymbstate{}}}
  
  \inferrule%
  {\mfLetFun{f}{\programvar}{\progexpr} \\
   \apclstep{\progexpr[\programvar \leftarrow \aval]}{\asymbstate{}}
           {\aval'}{\asymbstate{}'}}
  {\apclstep{\mfApp{f}{\aval}}{\asymbstate{}}
           {\aval'}{\asymbstate{}'}}

  \inferrule%
  {\pure(\progexpr_1, \progexpr_2) \\
  \neg\aconst(\aval)\\
   \apclstep{\progexpr_1}{\asymbstate{}}
           {\aval_1}{\asymbstate{1}} \\
   \apclstep{\progexpr_2}{\asymbstate{1}}
           {\aval_2}{\asymbstate{}'}
  }
  {\apclstep{\mfIf{\aval}{\progexpr_1}{\progexpr_2}}{\asymbstate{}}
           {\aite{\aval}{\aval_1}{\aval_2}}{\asymbstate{}'}}
  
  \inferrule%
  {\asymvalue^*(\aval, \asymbstate{}) = \atrue, \asymbstate{\aval}\\
   \apclstep{\progexpr_1}{\asymbstate{\aval}}
           {\aval_1}{\asymbstate{}'}}
  {\apclstep{\mfIf{\aval}{\progexpr_1}{\progexpr_2}}{\asymbstate{}}
           {\aval_1}{\asymbstate{}'}}
  
  \inferrule%
  {\asymvalue^*(\aval, \asymbstate{}) = \afalse, \asymbstate{\aval}\\
   \apclstep{\progexpr_2}{\asymbstate{\aval}}
           {\aval_2}{\asymbstate{}'}}
  {\apclstep{\mfIf{\aval}{\progexpr_1}{\progexpr_2}}{\asymbstate{}}
           {\aval_2}{\asymbstate{}'}}

  \inferrule%
  {\asymvalue^*(\aval, \asymbstate{}) = \cunk, \asymbstate{\aval} \\
   \apclstep{\progexpr_1}{\asymbstate{v}}
           {\aval_1}{\asymbstate{1}'} \\
   \apclstep{\progexpr_2}{\asymbstate{v}}
           {\aval_2}{\asymbstate{2}'}\\
   \narrowjoin(\aval_1, \aval_2, \asymbstate{1}', \asymbstate{2}') = \aval'', \asymbstate{}'
  }
  {\apclstep{\mfIf{\aval}{\progexpr_1}{\progexpr_2}}{\asymbstate{}}
           {\aval''}{\asymbstate{}'}}
  
  \inferrule%
  {\apclstep{\progexpr_1}{\asymbstate{}}
           {\aval_1}{\asymbstate{1}}\\
   \apclstep{\progexpr_2[\programvar \leftarrow \aval_1]}{\asymbstate{1}}
           {\aval_2}{\asymbstate{2}}}
  {\apclstep{\mfLetIn{x}{\progexpr_1}{\progexpr_2}}{\asymbstate{}}
           {\aval_2}{\asymbstate{2}}}

  \inferrule%
  { }
  {\apclstep{\mfFold{f}{\mfNil}{\aval}}{\asymbstate{}}
           {\aval}{\asymbstate{}}}

  \inferrule%
  {
    \apclstep{\mfLetIn{\programvar}{\mfApp{f}{\mfPair{\alisthd}{\aval}}}{\mfFold{f}{\alisttl}{\programvar}}}{\asymbstate{}}
         {\aval'}{\asymbstate{}'}
  }
  {\apclstep{\mfFold{f}{\mfCons{\alisthd}{\alisttl}}{\aval}}{\asymbstate{}}
           {\aval'}{\asymbstate{}'}}
  
  \inferrule%
  {
    \apclstep{\mfApp{f}{\mfPair{\alistl}{\aval}}}{\asymbstate{}}
             {\aval_f}{\asymbstate{f}}\\
    \narrowjoin(\aval, \aval_f, \asymbstate{}, \asymbstate{f}) = \aval_j, \asymbstate{j}\\
    \aval = \aval_j\\
    \weakeq{\asymbstate{}}{\asymbstate{j}}{(\alistl,\aval)}
  }
  {\apclstep{\mfFold{f}{\alistl}{\aval}}{\asymbstate{}}
            {\aval}{\asymbstate{}}}
  
  \inferrule%
  {
    \apclstep{\mfApp{f}{\mfPair{\alistl}{\aval}}}{\asymbstate{}}
             {\aval_f}{\asymbstate{f}}\\
    \narrowjoin(\aval, \aval_f, \asymbstate{}, \asymbstate{f}) = \aval_j, \asymbstate{j}\\
    \apclstep{\mfFold{f}{\alistl}{\aval_j}}{\asymbstate{j}}
             {\aval'}{\asymbstate{}'}
  }
  {\apclstep{\mfFold{f}{\alistl}{\aval}}{\asymbstate{}}
           {\aval'}{\asymbstate{}'}}
  
  \inferrule%
  {\annotation \in \{\varepsilon, \mfSymbolic\}\\
  \asymassume(\annotation, \distop(\aval), \asymbstate{}) = \arandomvar, \asymbstate{\arandomvar}\\
   \apclstep{\progexpr[\programvar \leftarrow \arandomvar]}{\asymbstate{\arandomvar}}
           {\aval'}{\asymbstate{}'}}
  {\apclstep{\mfLetRv{\annotation}{\programvar}{\distop(\aval)}{\progexpr}}{\asymbstate{}}
           {\aval'}{\asymbstate{}'}}

  \inferrule%
  {\asymassume(\amfSample, \distop(\aval), \asymbstate{}) = \arandomvar, \asymbstate{\arandomvar}\\
   \asymvalue(\arandomvar, \asymbstate{\arandomvar}) = \aval_\programvar, \asymbstate{\arandomvar}'\\
   \apclstep{\progexpr[\programvar \leftarrow \aval_\programvar]}{\asymbstate{\arandomvar}'}
           {\aval'}{\asymbstate{}'}}
  {\apclstep{\mfLetRv{\mfSample}{\programvar}{\distop(\aval)}{\progexpr}}{\asymbstate{}}
           {\aval'}{\asymbstate{}'}}

  \inferrule%
  { \asymassume(\varepsilon, \distop(\aval_1), \asymbstate{}) = \arandomvar, \asymbstate{\arandomvar}\\
   \asymvalue^*(\aval_2, \asymbstate{\arandomvar}) = \aval, \asymbstate{\aval}\\
   \asymobserve({\arandomvar}, \aval, \asymbstate{\aval}) = \asymbstate{}'}
  {\apclstep{\mfObserve{\distop(\aval_1)}{\aval_2}}{\asymbstate{}}
           {\amfUnit}{\asymbstate{}'}}

  \inferrule{
    \adset = \set{\adistribution(\aval,\asymbstate{}') \;|\; (\progexpr,\asymbstate{}) \in \apset, (\progexpr,\asymbstate{} \;\abs{\downarrow}\; \aval,\asymbstate{}') }
  }
  {\apclsstep{\apset}{\textstyle{\bigsqcup_{\adistr{}_i \in \adset}}\; \adistr{}_i}}

\inferrule{
  \ensuremath{\set{{\progexpr}, {\emptyset}}\; \hat{\downdownarrows}\; \afail}
  }
  {\apclpstep{\progexpr}{\afail}}

\inferrule{
    \apclsstep{\set{{\progexpr}, {\emptyset}}}
            {\adistr{}}
}
{\apclpstep{\progexpr}{\mit{satisfiable}}}  
  \end{mathpar}
  \end{small}
  
  \caption{Abstract interpretation rules. Failure rules are elided.}
  \label{fig:abstract-interp-full}
\vspace{-1em}
\end{figure}

\subsection{Implementing Abstract Inference Interface}
\label{appendix:abs-interface}
The abstract interface operations mirror their concrete counterparts, except that they do not perform any scoring or sampling. 

\subsubsection{Semi-symbolic Inference}

\begin{align*}
\asymassume(\aannotation, \adistr{}, \asymbstate{}) &= \; 
    \letin{\asymbstate{}' = \asymbstate{}[\arandomvar_\mit{new} \mapsto (\aannotation,\adistr{})]}
    (\arandomvar_\mit{new}, \asymbstate{}')\\
\asymvalue(\arandomvar, \asymbstate{}) &= \;
    \begin{array}[t]{@{}l@{}}
    \mit{if}\; \asymbstatePV{}{\arandomvar} = \amfSymbolic\; \mit{then}\;
    \afail\;\\
    \mit{else}\;
    \mit{let} \; \asymbstate{}' = \ahoist(\arandomvar, \asymbstate{}) \; \mit{in}\;      
    (\cunk, \aintervene(\arandomvar, \adeltasample{\cunk}, \asymbstate{}'))
    \end{array}\\
\asymobserve(\arandomvar, \aval, \asymbstate{}) &= \;
    \begin{array}[t]{@{}l@{}}
    \mit{let} \; \asymbstate{}' = \ahoist(\arandomvar, \asymbstate{}) \; \mit{in}\; 
    \aintervene(\arandomvar, \adeltad{\aval}, \asymbstate{}')
    \end{array}
\end{align*}

The abstract $\ahoist$ operation is the same algorithm as the concrete $\hoist{}$, where the helper functions are replaced with their abstract counterparts. The helper function $\widehat{\textsc{get\_parents}}$, like its concrete counterpart, returns parent random variables of $\arandomvar_\mit{cur}$.

\begin{algorithm}[H]
  \small
  \begin{algorithmic}
  \Function{$\ahoisthelper$}{$\arandomvar_\mit{cur}$, \texttt{roots}, $\asymbstate{}$}
  \State \texttt{parents} $\gets$ $\textsc{topo\_sort}(\widehat{\textsc{get\_parents}}(\arandomvar_\mit{cur}, \asymbstate{}))$
  \State \texttt{roots}' $\gets \texttt{roots}; \; \asymbstate{}' \gets \asymbstate{}$
  \For{$\arandomvar_\mit{par} \in $ \texttt{parents}}
  \If{$\arandomvar_\mit{par} \not \in $ \texttt{roots}}
  \State $\asymbstate{}' \gets$ \Call{$\ahoisthelper$}{$\arandomvar_\mit{par}$, \texttt{roots}', $\asymbstate{}'$}$;\;$ \texttt{roots}' $\gets \arandomvar_\mit{par} ::$ \texttt{roots}'
  \EndIf
  \EndFor
  \State $\asymbstate{}'' \gets \asymbstate{}'$
  \For{$\arandomvar_\mit{par} \in $ \textsc{reverse}(\texttt{parents})}
  \If{$\arandomvar_\mit{par} \not \in $ \texttt{roots}'}
  \State $(\asymbstate{}'', \texttt{conjugate}) \gets$ \Call{$\aswap$}{$\arandomvar_\mit{par}$, $\arandomvar_\mit{cur}, \asymbstate{}''$}
  \If{not \texttt{conjugate}}
  \State \textbf{throw} $(\arandomvar_\mit{par}, \arandomvar_\mit{cur})$
  \EndIf
  \EndIf
  \EndFor
  \State \Return $\asymbstate{}''$
  \EndFunction
  \Function{$\ahoist$}{$\arandomvar_\mit{in}$, $\asymbstate{}$}
  \try
  \State \Return \Call{$\ahoisthelper$}{$\arandomvar_\mit{in}$, $\{\}$, $\asymbstate{}$}
  \catch{$(\arandomvar_\mit{par}, \arandomvar_\mit{child})$}
  \State $(\_, \asymbstate{}') \gets$ \Call{$\asymvalue$}{$\arandomvar_{\mit{par}}, \asymbstate{}$}$;\;$ $\asymbstate{}'' \gets \aeval^*(\arandomvar_\mit{child}, \asymbstate{}')$ 
  \State \Return \Call{$\ahoist$}{$\arandomvar_\mit{in}$, $\asymbstate{}''$}
  \endtry
  \EndFunction
  \end{algorithmic}
  \caption{Abstract \hoist{}.}
  \label{alg:abs-hoist}
\end{algorithm}

\subsubsection{Delayed Sampling}
\begin{align*}
  \asymassume(\aannotation, \adistr{}, \asymbstate{}) &= \; 
      \begin{array}[t]{@{}l@{}}
        \letin{\adistr{\prime}, \asymbstate{}' = \aconjdistr(\adistr{}, \asymbstate{})}\\
      (\arandomvar_\mit{new}, \ainitialize(\arandomvar_\mit{new}, \aannotation, \adistr{\prime}, \asymbstate{}'))\\
      \end{array}\\
  \asymvalue(\arandomvar, \asymbstate{}) &= \;
      \begin{array}[t]{@{}l@{}}
      \mit{if}\; \asymbstatePV{}{\arandomvar} = \amfSymbolic\; \mit{then}\;
      \afail\;\\
      \mit{else}\;
      \mit{let} \; \asymbstate{}' = \agraft(\arandomvar, \asymbstate{}) \; \mit{in}\; 
      (\cunk, \arealize(\arandomvar, \adeltasample{\cunk}, \asymbstate{}'))
      \end{array}\\
  \symobserve(\arandomvar, \aval, \asymbstate{}) &= \;
      \begin{array}[t]{@{}l@{}}
      \mit{let} \; \asymbstate{}' = \agraft(\arandomvar, \asymbstate{}) \; \mit{in}\; 
      \arealize(\arandomvar, \adeltad{\cunk}, \asymbstate{}')
      \end{array}
  \end{align*}

The implementation of the interface using DS uses the abstract version of the grafting algorithm. If an abstract node has more than one marginal child it is unclear which marginal child is pruned. Thus, $\agraft$ invokes $\setunk$ on each child. $\aprune$ is only called on Marginalized nodes.
\begin{align*}
\agraft(\arandomvar, \asymbstate{}) = \;& 
\begin{array}[t]{@{}l@{}}
    \mit{match}\; \asymbstateS{}{\arandomvar}\; \mit{with}\\
    \begin{array}[t]{@{}l@{}}
        |\; \amarginalizedroot{\avarset{}} \;|\; \amarginalized{\arandomvar'}{\adistr{}}{\avarset{}}:\\
        \quad \begin{array}[t]{@{}l@{}}
          \letin{\avarset{,\mit{children}} = \widehat{\textsc{marginalized\_child}}(\avarset{})}\\
          \mit{if}\; |\avarset{,\mit{children}}| > 1 \; \mit{then}\; \\
          \quad \mit{for}\; \arandomvar_\mit{child} \in \avarset{,\mit{children}}: \setunk(\arandomvar_\mit{child}, \asymbstate{})\\
          \mit{else}\; 
          \letin{\{\arandomvar_\mit{child}\} = \avarset{,\mit{children}}}\;
          \aprune(\arandomvar_\mit{child}, \asymbstate{})
        \end{array}\\
        |\; \ainitialized{\arandomvar'}{\avarset{}}:\\
        \quad \begin{array}[t]{@{}l@{}}
            \letin{\asymbstate{}' = \agraft(\arandomvar', \asymbstate{})}
            \amarginalize(\arandomvar, \asymbstate{}')
        \end{array}\\
        |\; \topnode{\avarset{}}: \setunk(\arandomvar, \asymbstate{})
    \end{array}
\end{array}\\
\end{align*}

\section{Additional Proof Details}
\label{appendix:proofs}

We present here additional details necessary for showing that the inference plan satisfiability analysis is sound.

\subsection{Collecting Semantics}

\paragraph{Particle Set and Model Evaluation}
The collecting semantics of the particle set evaluation operates on a single set of particles. The analysis is interested in the possible particles produced during program execution; how the particles group together into particle sets is not relevant, which is why the semantics does not collect sets of sets. The given set of particles is first evaluated using the collecting semantics of the particle evaluation (which operates on a set of particles) to produce a set of configurations. If any of the particles return $\fail$, the semantics return a singleton set containing \fail.
If all of the particles have a \false{} resample flag (i.e. they all terminated), the semantics return a set of distributions using the \textsc{Distribution} operation.
If at least one of the particles has a \true{} resample flag, the collecting semantics iterates on the particle set evaluation rule. 
\begin{small}
    \begin{mathpar}
    \inferrule%
    { 
        \evalset = \set{(\progexpr',\symbstate',\doresample') \;|\;  (\progexpr,\symbstate) \in \pset, (\cpclstep{\progexpr, \symbstate}{\evalset'}), (\progexpr',\symbstate',\doresample') \in \evalset'}\\
        \textstyle{\bigwedge_{(\progexpr, \symbstate, \doresample) \in \evalset} }\neg \doresample\\
        \textstyle{\bigwedge_{(\progexpr, \symbstate, \doresample) \in \evalset}}\ (\progexpr \neq \fail) \\
        \dset = \set{\distribution(\val,\symbstate) \;|\; (\val, \symbstate, \doresample) \in \evalset}
    }
    {\cpclsstep{\pset}{\dset}}

    \inferrule%
    { 
      \evalset = \set{(\progexpr',\symbstate',\doresample') \;|\;  (\progexpr,\symbstate) \in \pset, (\cpclstep{\progexpr, \symbstate}{\evalset'}), (\progexpr',\symbstate',\doresample') \in \evalset'}\\
        \textstyle{\bigvee_{(\progexpr, \symbstate, \doresample) \in \evalset} } \doresample\\
        \textstyle{\bigwedge_{(\progexpr, \symbstate, \doresample) \in \evalset}}\ (\progexpr \neq \fail) \\
        \pset' = \forgetr(\evalset)\\
        \cpclsstep{\pset'}{\dset}
    }
    {\cpclsstep{\pset}{\dset}}

    \inferrule%
    { 
      \evalset = \set{(\progexpr',\symbstate',\doresample') \;|\;  (\progexpr,\symbstate) \in \pset, (\cpclstep{\progexpr, \symbstate}{\evalset'}), (\progexpr',\symbstate',\doresample') \in \evalset'}\\
        \textstyle{\bigvee_{(\progexpr, \symbstate, \doresample) \in \evalset}} (\progexpr = \fail)
    }
    {\cpclsstep{\pset}{\set{\fail}}}
    \end{mathpar}
\end{small}

The collecting semantics of the model evaluation rules constructs a particle from the given program with an empty symbolic state and evaluates with the particle set evaluation rules.
\begin{small}
\begin{mathpar}
  \inferrule%
  {
  \cpclsstep{\set{\progexpr, \emptyset{}}}{\dset}\\
  }
  {\cpclpstep{\progexpr}{\dset}}
  \end{mathpar}
\end{small}

\subsection{Abstraction and Concretization of Delayed Sampling Node Type}
The abstraction and concretization of node type are recursively defined. We show here the abstraction and concretization of the Marginalized nodes. The cases for other nodes are symmetric.
\begin{align*}
\abstr(\{\marginalized{\randomvar}{\distr{}}{\varset}\}) &= \amarginalized{\rvnames(\randomvar)}{\abstr(\{\distr{}\})}{\{\rvnames(\randomvar_s) \;|\; \randomvar_s \in \varset\}}\\
\concret(\amarginalized{\arandomvar}{\adistr{}}{\avarset{}}, \rvname{}) &= 
\set{ \marginalized{\randomvar}{\distr{}}{\varset} \;|\; 
\substack{
\rvname{}(\randomvar) = \arandomvar, \distr{} \in \concret(\adistr{}, \rvname{}), \\
S \subseteq \{\randomvar_s \;|\;  \rvname{}(\randomvar_s) \in \avarset{} \} }
}
\end{align*}

\subsection{Proofs}

First, we define a weak equivalence comparison between two symbolic states that only compares random variables that are reachable from the given expression. 
\begin{align*}
  \weakeq{\symbstate_{1}}{\symbstate_{2}}{\progexpr} \iff \textstyle{\forall_{\randomvar \in \reachable(\progexpr,\; \symbstate_{1},\; \symbstate_{2}) }}\; \symbstate_{1}(\randomvar) = \symbstate_{2}(\randomvar)
\end{align*}

With this weak equivalence relation, we can prove the weakening properties of concrete and abstract symbolic states. Informally, this property states that adding random variables that are not already in a symbolic state does not alter the evaluated expression nor the original entries in the resulting symbolic state.

\begin{lemma}[Weakening]
  \label{lem:weakening}
  If $(\pclstep{\progexpr}{\symbstate{}}{\progexpr'}{\symbstate{}'}{\weight}{\doresample})$ and $\weakeq{\symbstate{}}{\symbstate{}_1}{\progexpr}$, then $(\pclstep{\progexpr}{\symbstate{}_1}{\progexpr'}{\symbstate{}_1'}{\weight}{\doresample})$ and $\weakeq{\symbstate{}'}{\symbstate{}_1'}{\progexpr'}$.
\end{lemma}
\begin{proof}
  We first prove a more general property: If $(\pclstep{\progexpr}{\symbstate{}}{\progexpr'}{\symbstate{}'}{\weight}{\doresample})$ and $\weakeq{\symbstate{}}{\symbstate{}_1}{\progexpr}$, then $(\pclstep{\progexpr}{\symbstate{}_1}{\progexpr'}{\symbstate{}_1'}{\weight}{\doresample})$ and $\weakeq{\symbstate{}'}{\symbstate{}_1'}{\progexpr'}$, and for all $\progexpr''$ such that $\weakeq{\symbstate{}}{\symbstate{}_1}{\progexpr''}$, we have $\weakeq{\symbstate{}'}{\symbstate{}_1'}{\progexpr''}$.

  By structural induction on derivations of $\downarrow$:
  \begin{itemize}
      \item \pclstep{\val}{\symbstate}{\val}{\symbstate}{1}{\false}: By definition of weak equivalence. 
      For every $\randomvar \in \symbstate$, if $\randomvar \in \reachable(\val, \symbstate, \symbstate_1)$, then $\val = \randomvar$ and $\symbstate(\randomvar) = \symbstate_1(\randomvar)$ so $(\pclstep{\val}{\symbstate_1}{\val}{\symbstate_1}{1}{\false})$. Otherwise, $\randomvar$ is unreachable from $\val$, so it remains unreachable in $\symbstate_1$ and $(\pclstep{\val}{\symbstate_1}{\val}{\symbstate_1}{1}{\false})$. 
      \item \pclstep{\mfResample}{\symbstate}{\mfUnit}{\symbstate}{1}{\true}: By definition of weak equivalence.
      \item \pclstep{\mfApp{f}{\val}}{\symbstate}{\val'}{\symbstate'}{\weight}{\doresample}: 
      Because $\weakeq{\symbstate}{\symbstate_1}{\mfApp{f}{\val}}$, so $\weakeq{\symbstate}{\symbstate_1}{\progexpr[\programvar \leftarrow \val]}$.
      By inductive hypothesis, $(\pclstep{\progexpr[\programvar \leftarrow \val]}{\symbstate_1}{\progexpr'}{\symbstate_1'}{\weight}{\doresample})$ and $\weakeq{\symbstate'}{\symbstate_1'}{\progexpr'}$ and for all $\progexpr''$ such that $\weakeq{\symbstate{}}{\symbstate{}_1}{\progexpr''}$, we have $\weakeq{\symbstate'}{\symbstate{}_1'}{\progexpr''}$.
      Then, by definition $(\pclstep{\mfApp{f}{\val}}{\symbstate_1}{\val'}{\symbstate_1'}{\weight}{\doresample})$ and for all $\progexpr''$ such that $\weakeq{\symbstate{}}{\symbstate{}_1}{\progexpr''}$, we have $\weakeq{\symbstate'}{\symbstate{}_1'}{\progexpr''}$.
      \item \pclstep{\mfIf{\val}{\progexpr_1}{\progexpr_2}}{\symbstate}{\ite{\val}{\val_1}{\val_2}}{\symbstate'}{1}{\false}: 
        Because $\weakeq{\symbstate}{\symbstate_1}{\progexpr}$, so $\weakeq{\symbstate}{\symbstate_1}{\val}$ and $\weakeq{\symbstate}{\symbstate_1}{\progexpr_1}$ and $\weakeq{\symbstate}{\symbstate_1}{\progexpr_2}$.
        By definition, we have $(\pclstep{\progexpr_1}{\symbstate}{\val_1}{\symbstate_2}{1}{\false})$ and $(\pclstep{\progexpr_2}{\symbstate_2}{\val_2}{\symbstate'}{1}{\false})$.
        By inductive hypothesis, $(\pclstep{\progexpr_1}{\symbstate_1}{\val_1'}{\symbstate_2'}{1}{\false})$ and $\weakeq{\symbstate_2}{\symbstate_2'}{\val_1}$ and for all $\progexpr''$ such that $\weakeq{\symbstate{}}{\symbstate{}_1}{\progexpr''}$, we have $\weakeq{\symbstate{}_2}{\symbstate{}_2'}{\progexpr''}$. Then, $\weakeq{\symbstate_2}{\symbstate_2'}{\progexpr_2}$ and $\weakeq{\symbstate_2}{\symbstate_2}{\val}$.
        By inductive hypothesis, we then have $(\pclstep{\progexpr_2}{\symbstate_2'}{\val_2}{\symbstate_1'}{1}{\false})$ and $\weakeq{\symbstate'}{\symbstate_1'}{\val_2}$ and for all $\progexpr''$ such that $\weakeq{\symbstate_2}{\symbstate_2'}{\progexpr''}$, we have $\weakeq{\symbstate'}{\symbstate{}_1'}{\progexpr''}$.
        Then, $\weakeq{\symbstate'}{\symbstate_1'}{\val_1}$ and $\weakeq{\symbstate'}{\symbstate_1'}{\val}$.
        Then, $(\pclstep{\mfIf{\val}{\progexpr_1}{\progexpr_2}}{\symbstate_1}{\ite{\val}{\val_1}{\val_2}}{\symbstate_1'}{1}{\false})$ and $\weakeq{\symbstate'}{\symbstate_1'}{\progexpr}$. Also, for all $\progexpr''$ such that $\weakeq{\symbstate{}}{\symbstate{}_1}{\progexpr''}$, then $\weakeq{\symbstate{}_2}{\symbstate{}_2'}{\progexpr''}$ and then $\weakeq{\symbstate{}'}{\symbstate{}_1'}{\progexpr''}$.
      \item \pclstep{\mfIf{\val}{\progexpr_1}{\progexpr_2}}{\symbstate}{\val_1}{\symbstate'}{\weight}{\doresample}: 
      Because $\weakeq{\symbstate}{\symbstate_1}{\progexpr}$, so $\weakeq{\symbstate}{\symbstate_1}{\val}$ and $\weakeq{\symbstate}{\symbstate_1}{\progexpr_1}$. Since $\symvalue$ cannot affect any unreachable variables from $\val$, there exists $\symbstate_v'$ such that for all $\progexpr''$, if $\weakeq{\symbstate}{\symbstate_1}{\progexpr''}$, then $\weakeq{\symbstate_v}{\symbstate_v'}{\progexpr''}$.
      By inductive hypothesis, we have $(\pclstep{\progexpr_1}{\symbstate_v'}{\progexpr_1'}{\symbstate_1'}{\weight}{\doresample})$ and $\weakeq{\symbstate'}{\symbstate_1'}{\progexpr_1'}$ and for all $\progexpr''$ such that $\weakeq{\symbstate_v}{\symbstate{}_v'}{\progexpr''}$, we have $\weakeq{\symbstate'}{\symbstate_1'}{\progexpr''}$. It follows that for all $\progexpr''$, if $\weakeq{\symbstate}{\symbstate_1}{\progexpr''}$, then $\weakeq{\symbstate'}{\symbstate_1'}{\progexpr''}$.
      \item \pclstep{\mfIf{\val}{\progexpr_1}{\progexpr_2}}{\symbstate}{\val_2}{\symbstate'}{\weight_2}{\doresample}: Because $\weakeq{\symbstate}{\symbstate_1}{\progexpr}$, so $\weakeq{\symbstate}{\symbstate_1}{\val}$ and $\weakeq{\symbstate}{\symbstate_1}{\progexpr_2}$. Since $\symvalue$ cannot affect any unreachable variables from $\val$, there exists $\symbstate_v'$ such that for all $\progexpr''$, if $\weakeq{\symbstate}{\symbstate_1}{\progexpr''}$, then $\weakeq{\symbstate_v}{\symbstate_v'}{\progexpr''}$.
      By inductive hypothesis, we have $(\pclstep{\progexpr_2}{\symbstate_v'}{\progexpr_2'}{\symbstate_1'}{\weight}{\doresample})$ and $\weakeq{\symbstate'}{\symbstate_1'}{\progexpr_2'}$ and for all $\progexpr''$ such that $\weakeq{\symbstate_v}{\symbstate{}_v'}{\progexpr''}$, we have $\weakeq{\symbstate'}{\symbstate_1'}{\progexpr''}$. It follows that for all $\progexpr''$, if $\weakeq{\symbstate}{\symbstate_1}{\progexpr''}$, then $\weakeq{\symbstate'}{\symbstate_1'}{\progexpr''}$.
      \item \pclstep{\mfLetIn{\programvar}{\progexpr_1}{\progexpr_2}}{\symbstate}{\mfLetIn{\programvar}{\progexpr_1'}{\progexpr_2}}{\symbstate'}{\weight}{\true}:
      Because $\weakeq{\symbstate}{\symbstate_1}{\progexpr}$, so $\weakeq{\symbstate}{\symbstate_1}{\progexpr_1}$. 
      By inductive hypothesis, we have $(\pclstep{\progexpr_1}{\symbstate_1}{\progexpr_1'}{\symbstate_1'}{\weight}{\true})$ and $\weakeq{\symbstate'}{\symbstate_1'}{\progexpr_1'}$ and for all $\progexpr''$ such that $\weakeq{\symbstate}{\symbstate{}_1}{\progexpr''}$, we have $\weakeq{\symbstate'}{\symbstate_1'}{\progexpr''}$. 
      \item \pclstep{\mfLetIn{\programvar}{\progexpr_1}{\progexpr_2}}{\symbstate}{\progexpr_2'}{\symbstate_2}{\weight_1*\weight_2}{\doresample}: 
      Because $\weakeq{\symbstate}{\symbstate_1}{\progexpr}$, so $\weakeq{\symbstate}{\symbstate_1}{\progexpr'}$ and $\weakeq{\symbstate}{\symbstate_1}{\progexpr_2}$. 
      By definition, we have $(\pclstep{\progexpr_1}{\symbstate}{\val_1}{\symbstate_2}{\weight_1}{\false})$ and $(\pclstep{\progexpr_2[\programvar \leftarrow v_1]}{\symbstate_2}{\progexpr_2'}{\symbstate'}{\weight_2}{\doresample})$.
      By inductive hypothesis, we have $(\pclstep{\progexpr_1}{\symbstate_1}{\val_1}{\symbstate_2'}{\weight_1}{\false})$ and $\weakeq{\symbstate_2}{\symbstate_2'}{\val_1}$ and for all $\progexpr''$ such that $\weakeq{\symbstate}{\symbstate_1}{\progexpr''}$, we have $\weakeq{\symbstate_2}{\symbstate_2'}{\progexpr''}$.
      By inductive hypothesis, we have $(\pclstep{\progexpr_2[\programvar \leftarrow v_1]}{\symbstate_2'}{\progexpr_2'}{\symbstate_1'}{\weight_2}{\doresample})$ and $\weakeq{\symbstate'}{\symbstate_1'}{\progexpr_2'}$ and for all $\progexpr''$ such that $\weakeq{\symbstate_2}{\symbstate_2'}{\progexpr''}$, we have $\weakeq{\symbstate'}{\symbstate_1'}{\progexpr''}$. It follows that for all $\progexpr''$, if $\weakeq{\symbstate}{\symbstate_1}{\progexpr''}$, then $\weakeq{\symbstate'}{\symbstate_1'}{\progexpr''}$.
      \item \pclstep{\mfFold{f}{\mfNil}{\val}}{\symbstate}{\val}{\symbstate}{1}{\false}: By definition of weak equivalence.
      \item \pclstep{\mfFold{f}{\mfCons{l_{\mit{hd}}}{l_{\mit{tl}}}}{\val}}{\symbstate}{\val'}{\symbstate'}{\weight}{\doresample}: 
      By inductive hypothesis, we have the fact that $(\pclstep{\mfLetIn{\programvar}{\mfApp{f}{\mfPair{l_{\mit{hd}}}{\val}}}{\mfFold{f}{l_{tl}}{\programvar}}}{\symbstate_1}{\progexpr'}{\symbstate_1'}{\weight}{\doresample})$ and $\weakeq{\symbstate'}{\symbstate_1'}{\progexpr'}$ and for all $\progexpr''$ such that $\weakeq{\symbstate}{\symbstate{}_1}{\progexpr''}$, we have $\weakeq{\symbstate'}{\symbstate_1'}{\progexpr''}$. 
      \item \pclstep{\mfLetRv{\annotation}{\programvar}{\distop(\val)}{\progexpr}}{\symbstate}{\val}{\symbstate'}{\weight}{\doresample}: 
      Because $\weakeq{\symbstate}{\symbstate_1}{\progexpr}$, so $\weakeq{\symbstate}{\symbstate_1}{\progexpr[\programvar \leftarrow \randomvar]}$. 
      Since $\symassume$ cannot affect any unreachable variables from $\distop(\val)$, there exists $\symbstate_\randomvar'$ such that for all $\progexpr''$, if $\weakeq{\symbstate}{\symbstate_1}{\progexpr''}$, then $\weakeq{\symbstate_\randomvar}{\symbstate_\randomvar'}{\progexpr''}$.
      By inductive hypothesis, we have $(\pclstep{\progexpr[\programvar \leftarrow \randomvar]}{\symbstate_\randomvar}{\progexpr'}{\symbstate'}{\weight}{\doresample})$ and $\weakeq{\symbstate'}{\symbstate_1'}{\progexpr'}$ and for all $\progexpr''$ such that $\weakeq{\symbstate_\randomvar}{\symbstate_\randomvar'}{\progexpr''}$, we have $\weakeq{\symbstate'}{\symbstate_1'}{\progexpr''}$. It follows that for all $\progexpr''$, if $\weakeq{\symbstate}{\symbstate_1}{\progexpr''}$, then $\weakeq{\symbstate'}{\symbstate_1'}{\progexpr''}$.
      \item \pclstep{\mfLetRv{\mfSample}{\programvar}{\distop(\val)}{\progexpr}}{\symbstate}{\val}{\symbstate'}{\weight}{\doresample}:
      Because $\weakeq{\symbstate}{\symbstate_1}{\progexpr}$, so $\weakeq{\symbstate}{\symbstate_1}{\progexpr[\programvar \leftarrow \randomvar]}$. 
      Since $\symassume$ and $\symvalue$ cannot affect any unreachable variables from $\distop(\val)$, there exists $\symbstate_\randomvar''$ such that for all $\progexpr''$, if $\weakeq{\symbstate}{\symbstate_1}{\progexpr''}$, then $\weakeq{\symbstate_\randomvar'}{\symbstate_\randomvar''}{\progexpr''}$.
      By inductive hypothesis, we have $(\pclstep{\progexpr[\programvar \leftarrow \randomvar]}{\symbstate_\randomvar}{\progexpr'}{\symbstate'}{\weight}{\doresample})$ and $\weakeq{\symbstate'}{\symbstate_1'}{\progexpr'}$ and for all $\progexpr''$ such that $\weakeq{\symbstate_\randomvar}{\symbstate_\randomvar'}{\progexpr''}$, we have $\weakeq{\symbstate'}{\symbstate_1'}{\progexpr''}$. It follows that for all $\progexpr''$, if $\weakeq{\symbstate}{\symbstate_1}{\progexpr''}$, then $\weakeq{\symbstate'}{\symbstate_1'}{\progexpr''}$.
      \item \pclstep{\mfObserve{\distop(\val_1)}{\val_2}}{\symbstate}{\mfUnit}{\symbstate'}{\weight}{\false}: By definition of weak equivalence and that $\symassume$, $\symvalue^*$, and $\symobserve$ cannot affect any unreachable random variables from $\progexpr$.
  \end{itemize}
\end{proof}

\begin{lemma}[Abstract Weakening]
  \label{lem:abs-weakening}
  If $(\apclstep{\progexpr}{\asymbstate{}}{\aval'}{\asymbstate{}'})$ 
  and $\weakeq{\asymbstate{}}{\asymbstate{1}}{\progexpr}$,
  then $(\apclstep{\progexpr}{\asymbstate{1}}{\aval'}{\asymbstate{1}'})$ and $\weakeq{\asymbstate{}'}{\asymbstate{1}'}{\aval'}$.
\end{lemma}
\begin{proof}
  We first prove a more general property: If $(\apclstep{\progexpr}{\asymbstate{}}{\progexpr'}{\asymbstate{}'})$ and $\weakeq{\asymbstate{}}{\asymbstate{1}}{\progexpr}$, then $(\apclstep{\progexpr}{\asymbstate{1}}{\progexpr'}{\asymbstate{1}'})$ and $\weakeq{\asymbstate{}'}{\asymbstate{1}'}{\progexpr'}$, and for all $\progexpr''$ such that $\weakeq{\asymbstate{}}{\asymbstate{1}}{\progexpr''}$, we have $\weakeq{\asymbstate{}'}{\asymbstate{1}'}{\progexpr''}$.
  
  By structural induction on derivations of $\;\hat{\downarrow}$:

  \begin{itemize}
    \item $\apclstep{\aval}{\asymbstate{}}{\aval}{\asymbstate{}}$: By definition of weak equivalence. 
    For every $\arandomvar \in \asymbstate{}$, if $\arandomvar \in \reachable(\aval, \asymbstate{}, \asymbstate{1})$, then $\aval = \arandomvar$ and $\asymbstate{}(\arandomvar) = \asymbstate{1}(\arandomvar)$ so $(\apclstep{\aval}{\asymbstate{1}}{\aval}{\asymbstate{1}})$. Otherwise, $\arandomvar$ is unreachable from $\aval$, so it remains unreachable in $\asymbstate{1}$ and $(\apclstep{\aval}{\asymbstate{1}}{\aval}{\asymbstate{1}})$. 
    \item $\apclstep{\mfResample}{\asymbstate{}}{\amfUnit}{\asymbstate{}}$: By definition of weak equivalence.
    \item $\apclstep{\mfApp{f}{\aval}}{\asymbstate{}}{\aval'}{\asymbstate{}'}$: By inductive hypothesis, 
    $(\apclstep{\progexpr[\programvar \leftarrow \aval]}{\asymbstate{1}}{\progexpr'}{\asymbstate{1}'})$ and $\weakeq{\asymbstate{}'}{\asymbstate{1}'}{\progexpr'}$ and for all $\progexpr''$ such that $\weakeq{\asymbstate{}}{\asymbstate{1}}{\progexpr''}$, we have $\weakeq{\asymbstate{}'}{\asymbstate{1}'}{\progexpr''}$.
    \item $\apclstep{\mfIf{\aval}{\progexpr_1}{\progexpr_2}}{\asymbstate{}}{\aite{\aval}{\aval_1}{\aval_2}}{\asymbstate{}'}$: 
      Because $\weakeq{\asymbstate{}}{\asymbstate{1}}{\progexpr}$, so $\weakeq{\asymbstate{}}{\asymbstate{1}}{\aval}$ and $\weakeq{\asymbstate{}}{\asymbstate{1}}{\progexpr_1}$ and $\weakeq{\asymbstate{}}{\asymbstate{1}}{\progexpr_2}$.
      By definition, we have $(\apclstep{\progexpr_1}{\asymbstate{}}{\aval_1}{\asymbstate{2}})$ and $(\apclstep{\progexpr_2}{\asymbstate{2}}{\aval_2}{\asymbstate{}'})$.
      By inductive hypothesis, $(\apclstep{\progexpr_1}{\asymbstate{1}}{\aval_1'}{\asymbstate{2}'})$ and $\weakeq{\asymbstate{2}}{\asymbstate{2}'}{\aval_1}$ and for all $\progexpr''$ such that $\weakeq{\asymbstate{}}{\asymbstate{}}{\progexpr''}$, we have $\weakeq{\asymbstate{2}}{\asymbstate{2}'}{\progexpr''}$. Then, $\weakeq{\asymbstate{2}}{\asymbstate{2}'}{\progexpr_2}$ and $\weakeq{\asymbstate{2}}{\asymbstate{2}}{\aval}$.
      By inductive hypothesis, we then have $(\apclstep{\progexpr_2}{\asymbstate{2}'}{\aval_2}{\asymbstate{1}'})$ and $\weakeq{\asymbstate{}'}{\asymbstate{1}'}{\aval_2}$ and for all $\progexpr''$ such that $\weakeq{\asymbstate{2}}{\asymbstate{2}'}{\progexpr''}$, we have $\weakeq{\asymbstate{}'}{\asymbstate{1}'}{\progexpr''}$.
      Then, $\weakeq{\asymbstate{}'}{\asymbstate{1}'}{\aval_1}$ and $\weakeq{\asymbstate{}'}{\asymbstate{1}'}{\aval}$.
      Then, $(\apclstep{\mfIf{\val}{\progexpr_1}{\progexpr_2}}{\asymbstate{1}}{\aite{\aval}{\aval_1}{\aval_2}}{\asymbstate{1}'})$ and $\weakeq{\asymbstate{}'}{\asymbstate{1}'}{\progexpr}$. Also, for all $\progexpr''$ such that $\weakeq{\asymbstate{}}{\asymbstate{1}}{\progexpr''}$, then $\weakeq{\asymbstate{2}}{\asymbstate{2}'}{\progexpr''}$ and then $\weakeq{\asymbstate{}'}{\asymbstate{1}'}{\progexpr''}$.
    \item $\apclstep{\mfIf{\aval}{\progexpr_1}{\progexpr_2}}{\asymbstate{}}{\aval_1}{\asymbstate{}'}$: Since $\asymvalue$ cannot affect any unreachable variables from $\aval$, there exists $\asymbstate{\aval}'$ such that for all $\progexpr''$, if $\weakeq{\asymbstate{}}{\asymbstate{1}}{\progexpr''}$, then $\weakeq{\asymbstate{\aval}}{\asymbstate{\aval}'}{\progexpr''}$. Then, by inductive hypothesis.
    \item $\apclstep{\mfIf{\aval}{\progexpr_1}{\progexpr_2}}{\asymbstate{}}{\aval_2}{\asymbstate{}'}$: Since $\asymvalue$ cannot affect any unreachable variables from $\aval$, there exists $\asymbstate{\aval}'$ such that for all $\progexpr''$, if $\weakeq{\asymbstate{}}{\asymbstate{1}}{\progexpr''}$, then $\weakeq{\asymbstate{\aval}}{\asymbstate{\aval}'}{\progexpr''}$. Then, by inductive hypothesis.
    \item $\apclstep{\mfIf{\aval}{\progexpr_1}{\progexpr_2}}{\asymbstate{}}{\aval''}{\asymbstate{}'}$ : Since $\asymvalue$ cannot affect any unreachable variables from $\aval$, there exists $\asymbstate{\aval}'$ such that for all $\progexpr''$, if $\weakeq{\asymbstate{}}{\asymbstate{1}}{\progexpr''}$, then $\weakeq{\asymbstate{\aval}}{\asymbstate{\aval}'}{\progexpr''}$. Then, apply inductive hypothesis and by definition of weak equivalence and $\narrowjoin$.
    \item $\apclstep{\mfLetIn{\programvar}{\progexpr_1}{\progexpr_2}}{\asymbstate{}}{\aval_2}{\asymbstate{2}}$: By inductive hypothesis. 
    \item $\apclstep{\mfFold{f}{\mfNil}{\aval}}{\asymbstate{}}{\aval}{\asymbstate{}}$: By definition of weak equivalence.
    \item $\apclstep{\mfFold{f}{\mfCons{\alisthd}{\alisthd}}{\aval}}{\asymbstate{}}{\aval'}{\asymbstate{}'}$: By inductive hypothesis.
    \item $\apclstep{\mfFold{f}{\alistl}{\aval}}{\asymbstate{}}{\aval}{\asymbstate{}}$: 
    By inductive hypothesis, $\apclstep{\mfApp{f}{\mfPair{\alistl}{\aval}}}{\asymbstate{1}}{\aval_f}{\asymbstate{f}'}$ and $\weakeq{\asymbstate{f}'}{\asymbstate{f}'}{\aval_f}$. By definition of weak equivalence and $\narrowjoin$, for all $\progexpr''$, if $\weakeq{\asymbstate{}}{\asymbstate{1}}{\progexpr''}$ and $\weakeq{\asymbstate{f}}{\asymbstate{f}'}{\progexpr''}$, then $\weakeq{\asymbstate{j}}{\asymbstate{j}'}{\progexpr''}$. Then by definition of weak equivalence.
    \item $\apclstep{\mfFold{f}{\alistl}{\aval}}{\asymbstate{}}{\aval'}{\asymbstate{}'}$: 
    By inductive hypothesis, $\apclstep{\mfApp{f}{\mfPair{\alistl}{\aval}}}{\asymbstate{1}}{\aval_f}{\asymbstate{f}'}$ and $\weakeq{\asymbstate{f}'}{\asymbstate{f}'}{\aval_f}$. By definition of weak equivalence and $\narrowjoin$, for all $\progexpr''$, if $\weakeq{\asymbstate{}}{\asymbstate{1}}{\progexpr''}$ and $\weakeq{\asymbstate{f}}{\asymbstate{f}'}{\progexpr''}$, then $\weakeq{\asymbstate{j}}{\asymbstate{j}'}{\progexpr''}$. Then by inductive hypothesis.

    \item $\apclstep{\mfLetRv{\annotation}{\programvar}{\distop(\aval)}{\progexpr'}}{\asymbstate{}}{\aval'}{\asymbstate{}'}$: Since $\asymassume$ cannot affect any unreachable variables from $\distop(\aval)$, there exists $\asymbstate{\arandomvar'}$ such that for all $\progexpr''$, if $\weakeq{\asymbstate{}}{\asymbstate{1}}{\progexpr''}$, then $\weakeq{\asymbstate{\arandomvar}}{\asymbstate{\arandomvar'}}{\progexpr''}$.
    \item $\apclstep{\mfObserve{\distop(\aval_{1})}{\aval_2}}{\asymbstate{}}{\amfUnit}{\asymbstate{}'}$: By definition of weak equivalence and that $\symassume$, $\symvalue^*$, and $\symobserve$ cannot affect any unreachable random variables from $\progexpr$.
  \end{itemize}
\end{proof}

The next lemma states the soundness condition for the $\rename$ operation.
Namely, this specifies that $\textsc{rename}(\aexpr{}_1, \aexpr{}_2, \asymbstate{})$ produces a result that preserves the concretization of the pair $(\aexpr{}_2, \asymbstate{})$.

\begin{lemma}[Renaming Soundness]
For any abstract expressions $\aexpr{}_1$ and $\aexpr{}_2$, and any abstract symbolic state $\asymbstate{}$, if $(\aexpr{\prime}_2, \asymbstate{}') = \rename(\aexpr{}_1, \aexpr{}_2, \asymbstate{})$ then $\concret((\aexpr{\prime}_2, \asymbstate{}')) = \concret((\aexpr{}_2, \asymbstate{}))$.
\label{lem:rename}
\end{lemma}
\begin{proof}
The key idea is that, for any $\rvname{}$ used in the computation of $\concret((\aexpr{}_2, \asymbstate{}))$, we can construct a $\rvname{\prime}$ that incorporates the renaming of $\arandomvar$ to $\arandomvar'$. 
Suppose that in $\rvname{}$, the concrete random variable $\randomvar$ is mapped to $\arandomvar$. Let $\rvname{\prime}$ be a copy of $\rvname{}$, except that $\randomvar$ maps to $\arandomvar'$ in $\rvname{\prime}$, and a fresh concrete random variable $\randomvar_\mit{new}$ maps to $\arandomvar$. Then, the computation of the union in $\concret((\aexpr{\prime}_2, \asymbstate{}'))$ is simply a reordering of the one from $\concret((\aexpr{}_2, \asymbstate{}))$.\\
\end{proof}

We now prove a lemma that establishes the soundness of the $\narrowjoin$ operation.
This lemma states that the $\narrowjoin$ operation soundly over-approximates its input, given that additional unused random variables may be present in the over-approximation.

\begin{lemma}[Rename-Join Expression Soundness]
For any two abstract expressions $\aexpr{}_1$ and $\aexpr{}_2$, abstract symbolic states $\asymbstate{1}$ and $\asymbstate{2}$, there exists some abstract symbolic states $\asymbstate{1}', \asymbstate{2}'$ such that $\weakeq{\asymbstate{1}}{\asymbstate{1}'}{\aexpr{}_1}$, $\weakeq{\asymbstate{2}}{\asymbstate{2}'}{\aexpr{}_2}$, and
$(\concret((\aexpr{}_1, \asymbstate{1}')) \cup \concret((\aexpr{}_2, \asymbstate{2}'))) \subseteq \concret(\narrowjoin(\aexpr{}_1, \aexpr{}_2, \asymbstate{1}, \asymbstate{2}))$.
\label{lem:join}
\end{lemma}
\begin{proof}
Let $\asymbstate{1}'$ and $\asymbstate{2}'$ be the following abstract symbolic states:
\begin{align*}
\asymbstate{1}' &= \left\{ \arandomvar \mapsto \substack{\begin{cases}
    \asymbstate{1}(\arandomvar) & \arandomvar \in \dom(\asymbstate{1})\\
    \asymbstate{2}(\arandomvar) & \arandomvar \in \dom(\asymbstate{2}) \setminus \dom(\asymbstate{1})
\end{cases}}\right\}\\
\asymbstate{2}' &= \left\{ \arandomvar \mapsto \substack{\begin{cases}
    \asymbstate{2}(\arandomvar) & \arandomvar \in \dom(\asymbstate{2}) \\
    \asymbstate{1}(\arandomvar) & \arandomvar \in \dom(\asymbstate{1}) \setminus \dom(\asymbstate{2})
\end{cases}}\right\}
\end{align*}
By construction, $\weakeq{\asymbstate{1}}{\asymbstate{1}'}{\aexpr{}_1}$ and $\weakeq{\asymbstate{2}}{\asymbstate{2}'}{\aexpr{}_2}$.

During rename-join, we have an intermediate, renamed abstract expression and abstract symbolic state, $\aexpr{}_3, \asymbstate{3} = \rename(\aexpr{}_1, \aexpr{}_2, \asymbstate{2})$. 
The operation may rename a random variable $\randomvar_2 \in \dom(\asymbstate{2})$ to a random variable $\arandomvar_1 \in \dom(\asymbstate{1})$ and $\arandomvar_1$ may or may not be in $\dom(\asymbstate{2})$. 
However, the operation never renames a random variable $\arandomvar_1 \notin \dom(\asymbstate{2})$ to a random variable $\arandomvar_2 \in \dom(\asymbstate{2})$. So the random variables in $\dom(\asymbstate{1}) \setminus \dom(\asymbstate{2})$ either becomes a part of $\dom(\asymbstate{3})$ or remains unreachable from $\aexpr{}_3$. 
Thus, we have $\aexpr{}_3, \asymbstate{3}' = \rename(\aexpr{}_1, \aexpr{}_2, \asymbstate{2}')$ where
\begin{align*}
    \asymbstate{3}' &= \left\{ \arandomvar \mapsto \substack{\begin{cases}
        \asymbstate{3}(\arandomvar) & \arandomvar \in \dom(\asymbstate{3})\\
        \asymbstate{1}(\arandomvar) & \arandomvar \in \dom(\asymbstate{1}) \setminus \dom(\asymbstate{3})
    \end{cases}}\right\}
\end{align*}

It follows from the definition of basic expression join, abstract symbolic state join, and concretization that $(\concret((\aexpr{}_1, \asymbstate{1}')) \cup \concret((\aexpr{}_3, \asymbstate{3}'))) \subseteq \concret((\aexpr{}_1 \sqcup \aexpr{}_3, \asymbstate{1} \sqcup \asymbstate{3}))$.
Following from Lemma~\ref{lem:rename}, the concrete set of the intermediate, renamed abstract expression and abstract symbolic state preserves the concrete set of the original: 
$\concret((\aexpr{}_3, \asymbstate{3}')) = \concret((\aexpr{}_2, \asymbstate{2}'))$.
Therefore, we have $(\concret((\aexpr{}_1, \asymbstate{1}'))\; \cup\; \concret((\aexpr{}_2, \asymbstate{2}'))) \subseteq \concret(\narrowjoin(\aexpr{}_1, \aexpr{}_2, \asymbstate{1}, \asymbstate{2}))$.

\end{proof}

The next lemma states that substituting abstract values that over-approximates the original abstract value, soundly over-approximates the original results.

\begin{lemma}[Substitution Soundness]
    \label{lem:substitution}
    If $(\apclstep{\progexpr}{\asymbstate{1}}{\av{1}'}{\asymbstate{1}'})$ and $(\apclstep{\progexpr[\av{1} \leftarrow \av{2}]}{\asymbstate{2}}{\av{2}'}{\asymbstate{2}'})$ and $\concret((\av{1}, \asymbstate{1})) \subseteq \concret((\av{2}, \asymbstate{2}))$, then $\concret((\av{1}', \asymbstate{1}')) \subseteq \concret((\av{2}', \asymbstate{2}'))$.
\end{lemma}
\begin{proof}
    By structural induction on derivations of $\;\hat{\downarrow}$, with appeal to \Cref{lem:join,lem:abs-weakening}. The cases are as follows:

    \begin{itemize}
      \item $\apclstep{\aval}{\asymbstate{1}}{\aval}{\asymbstate{1}}$:
      By induction on the value syntax. If $\aval$ is $c$, $\programvar$, or $\mfNil$, then by definition $\concret((\aval, \asymbstate{})) \subseteq \concret((\aval_{2}, \asymbstate{2}))$. Otherwise, apply inductive hypothesis. For example, if $\aval$ is $\mfPair{\aval_1}{\aval_4}$, then $\aval_2' = \mfPair{\aval_2}{\aval_4}$. It follows that $\concret((\mfPair{\aval_1}{\aval_4}, \asymbstate{})) \subseteq \concret((\mfPair{\aval_2}{\aval_4}, \asymbstate{2}))$. The other cases are analogous.
      \item $\apclstep{\mfResample}{\asymbstate{1}}{\amfUnit}{\asymbstate{1}}$: There can be no substitution. 
      \item $\apclstep{\mfApp{f}{\aval}}{\asymbstate{1}}{\aval_1'}{\asymbstate{1}'}$: 
        By inductive hypothesis.
      \item $\apclstep{\mfIf{\aval}{\progexpr_1}{\progexpr_2}}{\asymbstate{1}}{\aite{\aval}{\aval_1^0}{\aval_2^0}}{\asymbstate{1}'}$: 
        By inductive hypothesis.
      \item $\apclstep{\mfIf{\aval}{\progexpr_1}{\progexpr_2}}{\asymbstate{1}}{\aval_1'}{\asymbstate{1}'}$ where $\asymvalue^*(\aval, \asymbstate{1}) = \atrue, \asymbstate{\aval}$: 
        By inductive hypothesis.
      \item $\apclstep{\mfIf{\aval}{\progexpr_1}{\progexpr_2}}{\asymbstate{1}}{\aval_1'}{\asymbstate{1}'}$ where $\asymvalue^*(\aval, \asymbstate{1}) = \afalse, \asymbstate{\aval}$: 
        By inductive hypothesis.
      \item $\apclstep{\mfIf{\aval}{\progexpr_1}{\progexpr_2}}{\asymbstate{1}}{\aval_1'}{\asymbstate{1}'}$ where $\asymvalue^*(\aval, \asymbstate{1}) = \aval', \asymbstate{\aval}$:
        By definition, 
        $(\apclstep{\progexpr_1}{\asymbstate{\aval}}{\aval_1^0}{\asymbstate{1}^0})$ and 
        $(\apclstep{\progexpr_2}{\asymbstate{\aval}}{\aval_2^0}{\asymbstate{2}^0})$ and 
        $(\apclstep{\progexpr_1[\aval_1 \leftarrow \aval_2]}{\asymbstate{\aval}'}{\aval_1^1}{\asymbstate{1}^1})$ and
        $(\apclstep{\progexpr_2[\aval_1 \leftarrow \aval_2]}{\asymbstate{\aval}'}{\aval_2^1}{\asymbstate{2}^1})$.
        By definition, $\concret(\asymbstate{\aval}) \subseteq \concret(\asymbstate{\aval}')$.
        By inductive hypothesis, we have $\concret((\aval_1^0, \asymbstate{1}^0)) \subseteq \concret((\aval_1^1, \asymbstate{1}^1))$
        and
        $\concret((\aval_2^0, \asymbstate{2}^0)) \subseteq \concret((\aval_2^1, \asymbstate{2}^1))$. Then, apply \Cref{lem:join} and \Cref{lem:abs-weakening}.
      \item $\apclstep{\mfLetIn{\programvar}{\progexpr_1}{\progexpr_2}}{\asymbstate{1}}{\aval_1'}{\asymbstate{1}'}$:
        Then, by inductive hypothesis.
      \item $\apclstep{\mfFold{f}{\mfNil}{\aval}}{\asymbstate{}}{\aval}{\asymbstate{}}$: By definition.
      \item $\apclstep{\mfFold{f}{\mfCons{\alisthd}{\alisthd}}{\aval}}{\asymbstate{1}}{\aval_1'}{\asymbstate{1}'}$:
        By inductive hypothesis.
        
      \item $\apclstep{\mfFold{f}{\alistl}{\aval}}{\asymbstate{1}}{\aval_1'}{\asymbstate{1}'}$ where $\aval = \aval_j$ and $\weakeq{\asymbstate{}}{\asymbstate{j}}{(\alistl, \aval)}$:
        By definition, (\apclstep{\mfApp{f}{\mfPair{\alistl}{\aval}}}{\asymbstate{1}}{\av{f}}{\asymbstate{f}}) and (\apclstep{\mfFold{f}{\alistl}{\av{j}}}{\asymbstate{j}}{\av{1}'}{\asymbstate{1}'}) and (\apclstep{\mfApp{f}{\mfPair{\alistl}{\aval}}[\aval_1 \leftarrow \aval_2]}{\asymbstate{2}}{\av{f}'}{\asymbstate{f}'}) and \\(\apclstep{\mfFold{f}{\alistl}{\av{j}'}[\aval_1 \leftarrow \aval_2]}{\asymbstate{j}'}{\av{2}'}{\asymbstate{2}'}).
        By inductive hypothesis, $\concret((\aval_f, \asymbstate{f})) \subseteq \concret((\aval_f', \asymbstate{f}'))$.
        By \Cref{lem:join}, there exists $\asymbstate{1}^*, \asymbstate{2}^*, \asymbstate{f}^*, \asymbstate{f}^{\prime*}$ where $\weakeq{\asymbstate{1}}{\asymbstate{1}^*}{\aval}$  and $\weakeq{\asymbstate{2}}{\asymbstate{2}^*}{\aval[\aval_1 \leftarrow \aval_2]}$ and $\weakeq{\asymbstate{f}}{\asymbstate{f}^*}{\aval_f}$ and $\weakeq{\asymbstate{f}'}{\asymbstate{f}^{\prime*}}{\aval_f'}$, such that $\concret((\aval, \asymbstate{1}^*)) \cup \concret((\aval_f, \asymbstate{f}^*)) \subseteq \concret((\av{j}, \asymbstate{j}))$ and
        $\concret((\aval[\aval_1 \leftarrow \aval_2], \asymbstate{2}^*)) \cup \concret((\av{f}', \asymbstate{f}^{\prime*})) \subseteq \concret((\av{j}', \asymbstate{j}'))$.
        By \Cref{lem:abs-weakening}, $\concret((\aval_1, \asymbstate{1}^*)) \subseteq \concret((\aval_2, \asymbstate{2}^{*}))$ and $\concret((\av{f}, \asymbstate{f}^*)) \subseteq \concret((\av{f}', \asymbstate{f}^{\prime*}))$. By join definition, $\concret((\av{j}, \asymbstate{j})) \subseteq \concret((\av{j}', \asymbstate{j}'))$. By definition of weak equivalence.
      \item $\apclstep{\mfFold{f}{\alistl}{\aval}}{\asymbstate{1}}{\aval_1'}{\asymbstate{1}'}$: 
        By definition, (\apclstep{\mfApp{f}{\mfPair{\alistl}{\aval}}}{\asymbstate{1}}{\av{f}}{\asymbstate{f}}) and (\apclstep{\mfFold{f}{\alistl}{\av{j}}}{\asymbstate{j}}{\av{1}'}{\asymbstate{1}'}) and (\apclstep{\mfApp{f}{\mfPair{\alistl}{\aval}}[\aval_1 \leftarrow \aval_2]}{\asymbstate{2}}{\av{f}'}{\asymbstate{f}'}) and (\apclstep{\mfFold{f}{\alistl}{\av{j}'}[\aval_1 \leftarrow \aval_2]}{\asymbstate{j}'}{\av{2}'}{\asymbstate{2}'}).
        By inductive hypothesis, $\concret((\aval_f, \asymbstate{f})) \subseteq \concret((\aval_f', \asymbstate{f}'))$.
        By \Cref{lem:join}, there exists $\asymbstate{1}^*, \asymbstate{2}^*, \asymbstate{f}^*, \asymbstate{f}^{\prime*}$ where $\weakeq{\asymbstate{1}}{\asymbstate{1}^*}{\aval}$  and $\weakeq{\asymbstate{2}}{\asymbstate{2}^*}{\aval[\aval_1 \leftarrow \aval_2]}$ and $\weakeq{\asymbstate{f}}{\asymbstate{f}^*}{\aval_f}$ and $\weakeq{\asymbstate{f}'}{\asymbstate{f}^{\prime*}}{\aval_f'}$, such that $\concret((\aval, \asymbstate{1}^*)) \cup \concret((\aval_f, \asymbstate{f}^*)) \subseteq \concret((\av{j}, \asymbstate{j}))$ and
        $\concret((\aval[\aval_1 \leftarrow \aval_2], \asymbstate{2}^*)) \cup \concret((\av{f}', \asymbstate{f}^{\prime*})) \subseteq \concret((\av{j}', \asymbstate{j}'))$.
        By \Cref{lem:abs-weakening}, $\concret((\aval_1, \asymbstate{1}^*)) \subseteq \concret((\aval_2, \asymbstate{2}^{*}))$ and $\concret((\av{f}, \asymbstate{f}^*)) \subseteq \concret((\av{f}', \asymbstate{f}^{\prime*}))$. By join definition, $\concret((\av{j}, \asymbstate{j})) \subseteq \concret((\av{j}', \asymbstate{j}'))$. By definition of weak equivalence. %
        By inductive hypothesis.
      \item $\apclstep{\mfLetRv{\annotation}{\programvar}{\distop(\aval)}{\progexpr'}}{\asymbstate{1}}{\aval_1'}{\asymbstate{1}'}$: 
        By inductive hypothesis.
      \item $\apclstep{\mfLetRv{\sample}{\programvar}{\distop(\aval)}{\progexpr'}}{\asymbstate{1}}{\aval_1'}{\asymbstate{1}'}$: By inductive hypothesis and by definition.
      \item $\apclstep{\mfObserve{\distop(\aval_{1})}{\aval_2}}{\asymbstate{1}}{\amfUnit}{\asymbstate{1}'}$: By definition.
    \end{itemize}

\end{proof}

We define an auxiliary operation that evaluates a particle until it has terminated. 

\begin{definition}[Particle Evaluation until Termination]
  \label{def:particle-termination}
We define the operation $\tilde{\downarrow}^*$ for evaluating an expression on a set of symbolic states until completion.
\begin{small}
  \begin{mathpar}
  \inferrule%
  { 
      \cpclstep{\progexpr, \symbstate}{\evalset}\\
      \textstyle{\bigwedge_{(\progexpr',\symbstate',\doresample') \in \evalset}} \neg \doresample'\\
      \textstyle{\bigwedge_{(\progexpr', \symbstate', \doresample') \in \evalset}}\ (\progexpr' \neq \fail) 
  }
  {\cpclstepstar{\progexpr, \symbstate}{\evalset}}

  \inferrule%
  { 
      \cpclstep{\progexpr{}, \symbstate}{\evalset}\\
      \textstyle{\bigvee_{(e',g',r') \in \evalset}} e' = \fail
  }
  {\cpclstepstar{\progexpr{}, \symbstate}{\set{(\fail, \fail, \fail)}}}

  \inferrule%
  { 
      \cpclstep{\progexpr, \symbstate}{\evalset}\\
      \textstyle{\bigvee_{(\progexpr',\symbstate',\doresample') \in \evalset} } \doresample'\\
      \textstyle{\bigwedge_{(\progexpr',\symbstate',\doresample') \in \evalset}}\ (\progexpr' \neq \fail) \\
      \evalset'' = \set{(\progexpr'',\symbstate'',\doresample'') \;|\; (\progexpr',\symbstate',\doresample') \in \evalset, (\cpclstepstar{\progexpr', \symbstate'}{\evalset'}), (\progexpr'',\symbstate'',\doresample'') \in \evalset' }
  }
  {\cpclstepstar{\progexpr, \symbstate}{\evalset''}}
  \end{mathpar}
\end{small}
\end{definition}

This lemma states that every value in the result of the particle set evaluation of the collecting semantics can be traced back to a particle in the input particle set that produces the value when evaluated until termination.

\begin{lemma}[Particle Trace]
\label{lem:particle-trace}
If $(\cpclsstep{\pset}{\dset})$, we have for all $\;\distr{} \in \dset$, there exists $(\progexpr, \symbstate) \in \pset$ such that $(\cpclstepstar{\progexpr, \symbstate}{\evalset})$ and $\;\distr{} \in \set{\distribution(\val,\symbstate) \;|\; (\val,\symbstate,\doresample) \in \evalset }$.
\end{lemma}
\begin{proof}
  By structural induction derivations of $\tilde{\downdownarrows}$.
  In the case of the first $\tilde{\downdownarrows}$ rule, for each $\distr{} \in \dset$, there exists a $\evalset'$ where there is a $(\val, \symbstate, \false) \in \evalset'$ such that $\distr{} = \distribution(\val, \symbstate)$. $\evalset'$ is constructed by the union of each of the $\evalset$ configuration set produced by the collecting single particle evaluation on every particle $(\progexpr, \symbstate)$ in $\particleset$. Then, there exists $(\progexpr, \symbstate) \in \particleset$ and $\evalset$ such that $(\cpclstep{\progexpr, \symbstate}{\evalset})$ and $(\val, \symbstate, \false) \in \evalset$. By \Cref{def:particle-termination}, we have $(\cpclstepstar{\progexpr, \symbstate}{\evalset})$. It follows also that $\distr{} \in \set{\distribution(\val,\symbstate) \;|\; (\val,\symbstate,\false) \in \evalset }$. 

  In the case of the second $\tilde{\downdownarrows}$ rule, by IH, we have that for all $\distr{} \in \dset$, there exists $(\progexpr', \symbstate') \in \particleset'$ such that $(\cpclstepstar{\progexpr', \symbstate'}{\evalset'})$ and $\distr{} \in \set{\distribution(\val,\symbstate) \;|\; (\val,\symbstate,\doresample) \in \evalset' }$. $\particleset'$ is constructed from the union of each of the $\evalset''$ configuration set produced by the collecting single particle evaluation on every $(\progexpr, \symbstate)$ particle in $\particleset$ by forgetting the resample flags from the configurations. Then, for each $(\progexpr', \symbstate') \in \particleset'$, there exists $(\progexpr, \symbstate) \in \particleset$ such that there exists an $\evalset''$ where $(\cpclstep{\progexpr, \symbstate}{\evalset''})$ and $(\progexpr', \symbstate') \in \forgetr(\evalset'')$ and  
  $ \textstyle{\bigvee_{(\progexpr, \symbstate, \doresample) \in \evalset''} } \doresample$.
  Then, by \Cref{def:particle-termination}, we have that $(\cpclstepstar{\progexpr, \symbstate}{\evalset})$ and $\evalset' \subseteq \evalset$. It follows also that $\distr{} \in \set{\distribution(\val,\symbstate) \;|\; (\val,\symbstate,\doresample) \in \evalset }$. 

  For the third $\tilde{\downdownarrows}$ rule, apply \Cref{def:particle-termination}.

\end{proof}

We now prove that the analysis is sound for particle evaluations that terminate.

\begin{lemma}[Terminating Particle Evaluation Soundness]
  \label{lem:terminating-particle}
  For every particle $(\progexpr, \symbstate)$, such that $(\cpclstep{\progexpr, \symbstate}{\evalset})$ and 
  $\textstyle{\forall_{(\val,\; \symbstate',\; \doresample) \in \evalset}}\; (\val = \fail) \vee \neg \doresample$, 
  we have $(\progexpr, \abstr(\{\symbstate\}) \;\abs{\downarrow}\; \aval', \asymbstate{}')$ and 
  there exists a configuration set $\evalset'$ such that $\weakeq{\evalset}{\evalset'}{}$ and 
  $\forgetr(\evalset') \subseteq \concret((\aval', \asymbstate{}'))$.
\end{lemma}
\begin{proof}
  By definition, if $\forall_{(\progexpr, \symbstate, \doresample) \in \evalset}$, there exists $\symbstate'$ such that $\weakeq{\symbstate}{\symbstate'}{\progexpr}$ and $\{(\progexpr, \symbstate') \} \subseteq \concret((\aval, \asymbstate{}))$, then there exists $\evalset'$ such that $\weakeq{\evalset}{\evalset'}{}$ and 
  $\forgetr(\evalset') \subseteq \concret((\aval, \asymbstate{}))$. Thus, it suffices to show that for each particle, $(\progexpr, \abstr(\{\symbstate\}) \;\abs{\downarrow}\; \aval', \asymbstate{}')$ and there exists $\symbstate'$ such that $\weakeq{\symbstate}{\symbstate'}{\aval'}$ and $\{(\aval', \symbstate') \} \subseteq \concret((\aval', \asymbstate{}'))$.
  
  By structural induction on derivations of $\downarrow$, with each case using the definitions of abstraction, concretization, and interpretation rules, with appeal to \Cref{assumption:syminterface} and \Cref{lem:join,lem:weakening,lem:abs-weakening,lem:substitution}.
  The cases are as follows:
  \begin{itemize}
    \item \pclstep{v}{g}{v}{g}{1}{\false}: By definition.
    \item \pclstep{\mfApp{f}{v}}{g}{v'}{g'}{w}{\false}: By inductive hypothesis, using \Cref{lem:substitution} to ensure the premise of the inductive hypothesis holds.
    \item \pclstep{\mfIf{v}{e_1}{e_2}}{g}{\ite{v}{v_1}{v_2}}{g'}{1}{\false}: By inductive hypothesis.
    \item \pclstep{\mfIf{v}{e_1}{e_2}}{g}{v_1}{g'}{w_1}{\false}: We have $\asymvalue^*(\av{},\asymbstate{}) = \av{}', \asymbstate{v}$. If $\av{}' = \true$, by inductive hypothesis and \Cref{assumption:syminterface}. Otherwise, by \Cref{lem:join}, \Cref{lem:weakening}, and inductive hypothesis. 
    \item \pclstep{\mfIf{v}{e_1}{e_2}}{g}{v_2}{g'}{w_2}{\false}: We have $\asymvalue^*(\av{},\asymbstate{}) = \av{}', \asymbstate{v}$. If $\av{}' = \false$, by inductive hypothesis and \Cref{assumption:syminterface}. Otherwise, by \Cref{lem:join}, \Cref{lem:weakening}, and inductive hypothesis. 
    \item \pclstep{\mfLetIn{x}{e_1}{e_2}}{g}{v_2}{g_2}{w_1*w_2}{\false}: By inductive hypothesis and \Cref{lem:substitution}.
    \item \pclstep{\mfFold{f}{\mfNil}{v}}{g}{v}{g}{1}{\false}: By definition.
    \item \pclstep{\mfFold{f}{\mfCons{l_{\mit{hd}}}{l_{\mit{tl}}}}{v}}{g}{v'}{g'}{w}{\false}: By inductive hypothesis.
    \item \pclstep{\mfLetRv{\annotation}{x}{\distop(v)}{e}}{g}{v}{g'}{w}{\false}: By inductive hypothesis and \Cref{assumption:syminterface,lem:substitution}.
    \item \pclstep{\mfLetRv{\mfSample}{x}{\distop(v)}{e}}{g}{v}{g'}{w}{\false}: By inductive hypothesis and \Cref{assumption:syminterface,lem:substitution}.
    \item \pclstep{\mfObserve{\distop(v_1)}{v_2}}{g}{\mfUnit}{g'}{w}{\false}: By inductive hypothesis and \Cref{assumption:syminterface}.
  \end{itemize}
\end{proof}

The next lemma proves that resuming particle evaluation from a configuration preserves soundness.

\begin{lemma}[Preservation]
\label{lem:preservation}
If $(\pclstep{\progexpr}{\symbstate}{\progexpr'}{\symbstate'}{\weight}{\doresample})$, then there exists abstract values $\aval, \aval'$ and abstract symbolic states  $\asymbstate{}, \asymbstate{}', \asymbstate{}''$ such that 1)\ $(\apclstep{\progexpr}{\abstr(\{\symbstate\})}{\aval}{\asymbstate{}}) \iff (\apclstep{\progexpr'}{\abstr(\{\symbstate'\})}{\aval'}{\asymbstate{}'})$, 2)\ $\weakeq{\asymbstate{}'}{\asymbstate{}''}{\aval'}$, and 3)\ $\concret((\aval', \asymbstate{}'')) \subseteq \concret((\aval, \asymbstate{}))$.
\end{lemma}
\begin{proof}
By structural induction on derivations of $\downarrow$, by definition of abstraction, concretization, and interpretation rules, with appeal to \Cref{lem:terminating-particle,lem:join,lem:substitution,lem:abs-weakening}

The cases are as follows:
\begin{itemize}
  \item For all $\doresample = \false$ cases, we apply \Cref{lem:terminating-particle}.
  \item \pclstep{\mfResample}{g}{\mfUnit}{g}{w}{\true}: By definition.
  \item \pclstep{v}{g}{v}{g}{1}{\false}: By definition.
  \item \pclstep{\mfApp{f}{v}}{g}{e'}{g'}{w}{r}: 
  \begin{itemize}
    \item Destructing definitions yields $(\pclstep{e[x \leftarrow v]}{g}{e'}{g'}{w}{r})$.
    \item $(\Rightarrow)$: Destructing definitions yields $(\apclstep{e[x \leftarrow v]}{\abstr(\{g\})}{\aval}{\asymbstate{}})$. By inductive hypothesis. %
    \item $(\Leftarrow)$: By inductive hypothesis, we have $(\apclstep{e[x \leftarrow v]}{\abstr(\{g\})}{\aval}{\asymbstate{}})$. The remainder follows from definitions.
  \end{itemize}
  \item \pclstep{\mfIf{v}{e_1}{e_2}}{g}{\ite{v}{v_1}{v_2}}{g'}{1}{\false}: 
  \begin{itemize}
    \item Destructing definitions yields $(\pclstep{e_1}{g}{v_1}{g_1}{1}{\false})$ and $(\pclstep{e_2}{g_1}{v_2}{g'}{1}{\false})$. 
    \item $(\Rightarrow)$: Follows from definitions.
    \item $(\Leftarrow)$: By inductive hypothesis, $(\apclstep{e_1}{\abstr(g)}{\av{1}}{\asymbstate{1}})$ and $(\apclstep{e_2}{\abstr(g_1)}{\av{2}}{\asymbstate{}'})$. The remainder follows from definitions. 
  \end{itemize}
  \item \pclstep{\mfIf{v}{e_1}{e_2}}{g}{e_1'}{g'}{w_1}{\true}:
  \begin{itemize}
    \item Destructing definitions yields $\symvalue^*(v, g) = \true, \symbstate_v$ and $(\pclstep{e_1}{g_v}{e_1'}{g'}{w_1}{\true})$.
    \item $(\Rightarrow)$: Destructing definitions yields 2 cases: 
    \begin{itemize}
      \item Case $\asymvalue^*(v, g) = \true, \asymbstate{v}$ and $(\apclstep{e_1}{\asymbstate{v}}{\aval_1}{\asymbstate{}'})$: Follows from the inductive hypothesis.
      \item Case $\asymvalue^*(v, g) = \aval', \asymbstate{v}$ and $(\apclstep{e_1}{\asymbstate{v}}{\aval_1}{\asymbstate{1}'})$ and $(\apclstep{e_2}{\asymbstate{v}}{\aval_2}{\asymbstate{2}'})$ and \\$\narrowjoin(\aval_1, \aval_2, \asymbstate{1}', \asymbstate{2}') = \aval, \asymbstate{}$: By inductive hypothesis, we have $(\apclstep{e_1'}{\abstr(\{\symbstate'\})}{\aval_1}{\asymbstate{1}'})$ and $(\apclstep{e_2'}{\abstr(\{\symbstate'\})}{\aval_2}{\asymbstate{2}'})$. Then, apply \Cref{lem:abs-weakening,lem:join}.
    \end{itemize}
    \item $(\Leftarrow)$: By inductive hypothesis, we have that $(\apclstep{e_1}{\abstr(\{g_v\})}{\aval'}{\asymbstate{}'})$. The remainder follows from the definition of true branch evaluation.
  \end{itemize}
  \item \pclstep{\mfIf{v}{e_1}{e_2}}{g}{e_2'}{g'}{w_2}{\true}: Symmetric to the previous case.
  \item \pclstep{\mfLetIn{x}{e_1}{e_2}}{g}{\mfLetIn{x}{e_1'}{e_2}}{g'}{w}{\true}:
  \begin{itemize}
    \item Destructing definitions yields $(\pclstep{e_1}{g}{e_1'}{g'}{w}{\true})$. 
    \item $(\Rightarrow)$: Destructing definitions yields $(\apclstep{e_1}{\abstr(\{g\})}{\av{1}}{\asymbstate{1}})$ and $(\apclstep{e_2[x \leftarrow \av{1}]}{\asymbstate{1}}{\aval}{\asymbstate{}})$. By inductive hypothesis, we have that $(\apclstep{e_1'}{\asymbstate{}'}{\av{1}}{\asymbstate{1}})$. The remainder follows from definitions. 
    \item $(\Leftarrow)$: Destructing definitions yields $(\apclstep{e_1'}{\abstr(\{g'\})}{\av{1}}{\asymbstate{1}})$ and $(\apclstep{e_2[x \leftarrow \av{1}]}{\asymbstate{1}}{\aval}{\asymbstate{}})$. By inductive hypothesis, $(\apclstep{e_1}{\abstr(\{g\})}{\av{1}}{\asymbstate{1}})$. The remainder follows from definitions.  
  \end{itemize}
  \item \pclstep{\mfLetIn{x}{e_1}{e_2}}{g}{e_2'}{g_2}{w_1*w_2}{\true}:
  \begin{itemize}
    \item Destructing definitions yields $(\pclstep{e_1}{g}{v_1}{g_1}{w_1}{\false})$ and $(\pclstep{e_2[x\leftarrow v_1]}{g_1}{e_2'}{g_2}{w_2}{\true})$.
    \item $(\Rightarrow)$: Destructing definitions yields $(\apclstep{e_1}{\abstr(\{g\})}{\av{1}}{\asymbstate{1}})$ and $(\apclstep{e_2[x \leftarrow \av{1}]}{\asymbstate{1}}{\aval}{\asymbstate{}})$. The remainer follows from applying inductive hypothesis. %
    \item $(\Leftarrow)$: By inductive hypothesis, we have that $(\apclstep{e_2[x \leftarrow v_1]}{\abstr(\{g_1\})}{\aval}{\asymbstate{}})$. Then, by \Cref{lem:terminating-particle}, we have that $(\apclstep{e_1}{\abstr(\{g\})}{\av{1}}{\asymbstate{1}})$. The remainder follows from applying \Cref{lem:substitution}.
  \end{itemize}
  \item \pclstep{\mfFold{f}{\mfCons{\listhd}{\listtl}}{\val}}{\symbstate}{\progexpr'}{\symbstate'}{w}{\true}:
  \begin{itemize}
    \item Destructing definitions yields $\pclstep{\mfLetIn{x}{\mfApp{f}{\mfPair{l_{\mit{hd}}}{v}}}{\mfFold{f}{l_{\mit{tl}}}{x}}}{g}{e'}{g'}{w}{\true}$.
    \item $(\Rightarrow)$: Destructing definitions yields 2 cases
    \begin{itemize}
      \item Case $\pclstep{\mfApp{f}{\mfPair{l_{\mit{hd}}}{v}}}{g}{e'}{g'}{w}{\true}$: Destructing definitions yields $\apclstep{\mfApp{f}{\mfPair{l_{\mit{hd}}}{v}}}{\abstr(\{g\})}{\av{f}}{\asymbstate{f}}$ and $\apclstep{\mfFold{f}{\listtl}{\aval_f}}{\asymbstate{f}}{\aval}{\asymbstate{}}$. By inductive hypothesis, we have that $\apclstep{e'}{\abstr(\{g'\})}{\av{f}'}{\asymbstate{f}'}$. The remainder follows from \Cref{lem:substitution,lem:abs-weakening}.
      \item Case $\pclstep{\mfFold{f}{l_{tl}}{v_f}}{g_f}{e'}{g'}{w}{\true}$: By inductive hypothesis.
    \end{itemize}
    \item $(\Leftarrow)$: By the inductive hypothesis, $\apclstep{\mfLetIn{x}{\mfApp{f}{\mfPair{l_{\mit{hd}}}{v}}}{\mfFold{f}{l_{\mit{tl}}}{x}}}{g}{\aval}{\asymbstate{}}$. Destructing definitions yields $\apclstep{\mfApp{f}{\mfPair{l_{\mit{hd}}}{v}}}{g}{\av{f}}{\asymbstate{f}}$ and $\apclstep{\mfFold{f}{l_{\mit{tl}}}{\av{f}}}{\asymbstate{f}}{\aval}{\asymbstate{}}$. 
  \end{itemize}
  \item \pclstep{\mfLetRv{\annotation}{x}{\distop(v)}{e}}{g}{e'}{g'}{w}{\true}: symmetric to the case for let binding.
  \item \pclstep{\mfLetRv{\mfSample}{x}{\distop(v)}{e}}{g}{e'}{g'}{w}{\true}: symmetric to the case for let binding.
\end{itemize}
\end{proof}

Next, we show that the analysis is sound for evaluating any particle until termination. 

\begin{lemma}[Particle Evaluation Soundness]
\label{lem:particle}
For every particle $(\progexpr, \symbstate)$, such that $(\cpclstepstar{\progexpr, \symbstate}{\evalset})$, we have
$(\progexpr, \abstr(\{\symbstate\}) \;\abs{\downarrow}\; \aval, \asymbstate{})$ and a configuration set $\evalset'$ such that $\weakeq{\evalset}{\evalset'}{}$ and 
$\forgetr(\evalset')  \subseteq \concret((\aval, \asymbstate{}))$.
\end{lemma}
\begin{proof}

  By structural induction on $\tilde{\downarrow}^*$. For the two base case rules, apply \Cref{lem:terminating-particle}. 

  In the inductive case with the third rule, 
  we have that $(\cpclstep{\progexpr, \symbstate}{\evalset^1})$. 
  Also, for each $(\progexpr', \symbstate') \in \forgetr(\evalset^1)$, there exists $\evalset^2$ such that $(\cpclstepstar{\progexpr', \symbstate'}{\evalset^2})$. 
  By IH, we have that $(\apclstep{\progexpr'}{\abstr(\{\symbstate'\})}{\aval'}{\asymbstate{}'})$ and there is a configuration set $\evalset^{2\prime}$ such that $\weakeq{\evalset^2}{\evalset^{2\prime}}{}$ and $\forgetr(\evalset^{2\prime}) \subseteq \concret((\aval', \asymbstate{}'))$. 

  By \Cref{lem:preservation}, for each possible particle evaluation included in the collecting particle evaluation, we have $(\apclstep{\progexpr}{\abstr(\{\symbstate\})}{\aval}{\asymbstate{}})$ and there exists $\asymbstate{}^{1\prime}$ such that $\weakeq{\asymbstate{}'}{\asymbstate{}^{1\prime}}{\aval'}$ and $\concret((\aval', \asymbstate{}^{1\prime})) \subseteq \concret((\aval, \asymbstate{}))$. Then, there is a configuration set $\evalset^{2\prime\prime}$ that accounts the additional abstract variables in $\asymbstate{}^{1\prime}$ that are not in $\asymbstate{}'$ such that $\weakeq{\evalset^{2\prime}}{\evalset^{2\prime\prime}}{}$ and $\forgetr(\evalset^{2\prime\prime}) \subseteq \concret((\aval', \asymbstate{}^{1\prime})) \subseteq \concret((\aval, \asymbstate{}))$. 

  Since $\evalset$ is the union of the $\evalset^2$ resulting from each $(\progexpr', \symbstate') \in \forgetr(\evalset^1)$, we have $\evalset' = \bigcup \evalset^{2\prime\prime}$ such that $\weakeq{\evalset}{\evalset'}{}$ and $\forgetr(\evalset') \subseteq \concret((\aval, \asymbstate{}))$.

\end{proof}

Finally, we show that the analysis is sound with respect to evaluating sets of particles.

\begin{theorem}[Particle Set Evaluation Soundness]
    \label{thm:particle-set}
    For every particle set $\pset$, and distribution set $\dset$ such that $(\cpclsstep{\pset}{\dset})$, we have that \apclsstep{\set{\progexpr, \abstr(\{\symbstate\}) \;|\; (\progexpr, \symbstate) \in \pset}}{\adistr{}} and $\dset \subseteq \concret(\adistr{})$.
\end{theorem}
\begin{proof}

    By \Cref{lem:particle-trace}, for each $\distr{} \in \dset$, there exists $(\progexpr, \symbstate) \in \pset$ such that (\cpclstepstar{\progexpr, \symbstate}{\evalset}) and $\distr{} \in \set{\distribution(\val,\symbstate) \;|\; (\val,\symbstate,\doresample) \in \evalset}$. By \Cref{lem:particle}, we have $(\progexpr, \abstr(\{\symbstate\}) \;\abs{\downarrow}\; \aval, \asymbstate{})$ and a configuration set $\evalset'$ such that $\weakeq{\evalset}{\evalset'}{}$ and $\forgetr(\evalset')  \subseteq \concret((\aval, \asymbstate{}))$. By definition, $\apclsstep{\set{\progexpr, \abstr(\{\symbstate\}) \;|\; (\progexpr, \symbstate) \in \pset}}{\adistr{}}$ and $\set{\distr{}} \subseteq \concret(\adistr{})$. Then, $\dset \subseteq \concret(\adistr{})$.
\end{proof}

\section{Benchmarks}
\label{appendix:benchmarks}

\subsection{Additional Benchmark Descriptions}
We evaluated on the following benchmarks with multiple inference plans from \citet{atkinson2022semi} and \citet{baudart2020reactive}:

\bOutlier{} models a one-dimensional particle filter where there might be sensor errors producing outlier observations. The hidden state is modeled by Gaussian distributions (\zl{xt}), the sensor error rate as a Beta prior (\zl{outlier_prob}) to a Bernoulli, and observations are made on Gaussian distributions. This program was implemented by~\citet{atkinson2022semi} and adapted from~\citet{minka2013expectation}.

\bGtree{} is a particle filter that makes an observation centered around a single random variable modeled by a Gaussian distribution (\zl{a}) and another centered around random variables (\zl{b}) that are sampled from a Gaussian distribution every timestep. This program was previously implemented by \citet{atkinson2022semi} and adapted from~\citet{lunden2017delayed}. 

\bSlam{} is a Simultaneous Localization and Mapping model of a robot traveling on a black-and-white map (\zl{map}) while estimating its position (\zl{x}). This program was implemented by \citet{atkinson2022semi} and adapted from \citet{doucet2000rao}.

\bWheels{} from ~\citet{atkinson2022semi} models the velocity (\zl{velocity}) and angular velocity (\zl{omega}) of a robot with speed sensors on its two wheels. These variables are modeled by linear-Gaussians. 

\subsection{Source Code and Inference Plans}
We include here the source code and the evaluated inference plans. \verb|PLACEHOLDER| indicates where distribution encodings are added.

\begin{figure}[H]
\begin{lstlisting}[basicstyle=\footnotesize\ttfamily,aboveskip=0em,numbers=left]
let step = fun (zobs, (xs, q, r)) ->
  let prev_x = List.hd(xs) in
  let h = 2. in
  let f = 1.001 in
  let PLACEHOLDER x <- gaussian(f * prev_x, q) in
  let () = observe(gaussian(h*x, r), zobs) in
  let () = resample() in
  (cons(x, xs), q, r)
let PLACEHOLDER q <- invgamma (1., 1.) in
let PLACEHOLDER r <- invgamma (1., 1.) in
let x0 = 0. in
let (xs, q, r) = fold(step, data, ([x0], q, r)) in
let xs = List.tl(List.rev(xs)) in
(xs, q, r)
\end{lstlisting}
  \caption{Source code for \bNoise{}.}
  \label{fig:noise-code}
\end{figure}

\begin{table}[H]
  \small
  \centering
  \caption{Evaluated inference plans for \bNoise{}.}
  \begin{tabular}[h]{llll}
    \toprule
     & \multicolumn{3}{c}{Variables}\\
     \cmidrule(lr){2-4}
    {}                    & \zlm|x|   & \zlm|q|    & \zlm|r|   \\
    \midrule
    Plan 3              &  \symbolic & \sample    & \sample \\
    Plan 4              &  \sample   & \symbolic  & \symbolic \\
    Plan 5              &  \sample   & \symbolic  & \sample \\
    Plan 6              &  \sample   & \sample    & \symbolic \\
    Plan 7              &  \sample   & \sample    & \sample \\
    \bottomrule
  \end{tabular}
  \label{tab:noise-plans}
\end{table}

\begin{figure}[H]
  \begin{lstlisting}[basicstyle=\footnotesize\ttfamily,aboveskip=0em,numbers=left]
let step = fun (obs, (xs,q,r)) ->
  let PLACEHOLDER x <- gaussian(List.hd(xs)+5,q) in
  let PLACEHOLDER env <- bernoulli(0.0001) in
  let PLACEHOLDER other <- invgamma(1,1) in
  let v = if env then r + other else r in
  let () = observe(gaussian(x,v),obs) in
  let () = resample() in
  (cons(x, xs),q,r)
let PLACEHOLDER q <- invgamma(1.,1.) in
let PLACEHOLDER r <- invgamma(1.,1.) in
let (xs,q,r) = 
  fold(step,data,([0.],q,r)) in
(List.tl(List.rev(xs)),q,r)
  \end{lstlisting}
    \caption{Source code for \bRadar{}.}
    \label{fig:radar-code}
  \end{figure}
  
  \begin{table}[H]
    \small
    \centering
    \caption{Evaluated inference plans for \bRadar{}.}
    \begin{tabular}[h]{llllll}
      \toprule
        & \multicolumn{5}{c}{Variables}\\
        \cmidrule(lr){2-6}
      {}                    & \zlm|x|   & \zlm|env|    & \zlm|other| & \zlm|q|    & \zlm|r|   \\
      \midrule
      Plan 15              &  \symbolic & \sample    & \sample  & \sample    & \sample\\
      Plan 29              &  \sample   & \sample    & \sample  & \symbolic  & \sample \\
      Plan 31              &  \sample   & \sample    & \sample  & \sample    & \sample \\
      \bottomrule
    \end{tabular}
    \label{tab:radar-plans}
  \end{table}

\begin{figure}[H]
  \begin{lstlisting}[basicstyle=\footnotesize\ttfamily,aboveskip=0em,numbers=left]
let step = fun (obs, (xs,q,r)) ->
  let x_i = List.hd(xs) in
  let h = 2. in
  let f = 1.001 in
  let PLACEHOLDER x <- gaussian(f * x_i,q) in
  let PLACEHOLDER env <- bernoulli(0.0001) in
  let PLACEHOLDER other <- beta(1,1) in
  let () = if env then
      observe(gaussian(h * x, (r) + 1000*other),obs) 
    else
      observe(gaussian(h * x,r),obs)
  in 
  let () = resample() in
  (cons(x, xs),q,r)
let PLACEHOLDER q <- invgamma(1.,1.) in
let PLACEHOLDER r <- invgamma(1.,1.) in
let (xs,q,r) = 
  fold(step,data,([0.],q,r)) in
(List.tl(List.rev(xs)),q,r)
\end{lstlisting}
  \caption{Source code for \bEnvnoise{}.}
  \label{fig:envnoise-code}
\end{figure}

\begin{table}[H]
  \small
  \centering
  \caption{Evaluated inference plans for \bEnvnoise{}.}
  \begin{tabular}[h]{llllll}
    \toprule
      & \multicolumn{5}{c}{Variables}\\
      \cmidrule(lr){2-6}
    {}                    & \zlm|x|   & \zlm|env|    & \zlm|other| & \zlm|q|    & \zlm|r|   \\
    \midrule
    Plan 15              &  \symbolic & \sample    & \sample  & \sample    & \sample\\
    Plan 29              &  \sample   & \sample    & \sample  & \symbolic  & \sample \\
    Plan 31              &  \sample   & \sample    & \sample  & \sample    & \sample \\
    \bottomrule
  \end{tabular}
  \label{tab:envnoise-plans}
\end{table}

\begin{figure}[H]
  \begin{lstlisting}[basicstyle=\footnotesize\ttfamily,aboveskip=0em,numbers=left]
let step = fun (yobs, (first, outlier_prob, xs)) ->
  let prev_xt = List.hd(xs) in
  let xt_mu = if first then 0. else prev_xt in
  let xt_var = if first then 2500. else 1. in
  let PLACEHOLDER xt <- gaussian(xt_mu, xt_var) in
  let PLACEHOLDER is_outlier <- bernoulli(outlier_prob) in
  let mu = if is_outlier then 0. else xt in
  let var = if is_outlier then 10000. else 1. in
  let () = observe(gaussian(mu, var), yobs) in
  let () = resample() in
  (false, outlier_prob, cons(xt, xs))
let PLACEHOLDER outlier_prob <- beta(100., 1000.) in
let (_, outlier_prob, xs) = fold(step, data, (true, outlier_prob, [0.])) in
let xs = List.tl(List.rev(xs)) in
(outlier_prob, xs)
\end{lstlisting}
  \caption{Source code for \bOutlier{}.}
  \label{fig:outlier-code}
\end{figure}

\begin{table}[H]
  \small
  \centering
  \caption{Evaluated inference plans for \bOutlier{}.}
  \begin{tabular}[h]{llllll}
    \toprule
      & \multicolumn{3}{c}{Variables}\\
      \cmidrule(lr){2-4}
    {}                    & \zlm|xt|   & \zlm|is_outlier| & \zlm|outlier_prob|  \\
    \midrule
    Plan 2                & \symbolic & \sample & \symbolic \\
    Plan 3                & \symbolic & \sample & \sample \\
    Plan 6                & \sample   & \sample & \symbolic \\
    Plan 7                & \sample   & \sample & \sample \\
    \bottomrule
  \end{tabular}
  \label{tab:outlier-plans}
\end{table}

\begin{figure}[H]
  \begin{lstlisting}[basicstyle=\footnotesize\ttfamily,aboveskip=0em,numbers=left]
let step = fun (yobs, (first, outlier_prob, xs)) ->
  let prev_xt = List.hd(xs) in
  let xt_mu = if first then 0. else prev_xt in
  let xt_var = if first then 2500. else 1. in
  let PLACEHOLDER xt <- gaussian(xt_mu, xt_var) in
  let PLACEHOLDER is_outlier <- bernoulli(outlier_prob) in
  let () = if is_outlier then
    observe(student_t(xt, 1, 1.1), yobs)
  else
    observe(gaussian(xt, 1), yobs)
  in
  let () = resample() in
  (false, outlier_prob, cons(xt, xs))
let PLACEHOLDER outlier_prob <- beta(100., 1000.) in
let (_, outlier_prob, xs) = fold(step, data, (true, outlier_prob, [0.])) in
let xs = List.tl(List.rev(xs)) in
(outlier_prob, xs)
\end{lstlisting}
  \caption{Source code for \bOutlierheavy{}.}
  \label{fig:outlierheavy-code}
\end{figure}

\begin{table}[H]
  \small
  \centering
  \caption{Evaluated inference plans for \bOutlierheavy{}.}
  \begin{tabular}[h]{llllll}
    \toprule
      & \multicolumn{3}{c}{Variables}\\
      \cmidrule(lr){2-4}
      {}                    & \zlm|xt|   & \zlm|is_outlier| & \zlm|outlier_prob|  \\
      \midrule
      Plan 6                & \sample    & \sample          & \symbolic \\
      Plan 7                & \sample    & \sample          & \sample \\
    \bottomrule
  \end{tabular}
  \label{tab:outlierheavy-plans}
\end{table}

\begin{figure}[H]
  \begin{lstlisting}[basicstyle=\footnotesize\ttfamily,aboveskip=0em,numbers=left]
let step = fun ((obs_b, obs_a), (a, bs)) ->
  let PLACEHOLDER b <- gaussian(a, 10) in
  let () = observe(gaussian(b, 1000), obs_b) in
  let () = observe(gaussian(a, 1000), obs_a) in
  let () = resample() in
  (a, cons(b, bs))
let PLACEHOLDER a <- gaussian(0, 100) in
let (a, bs) = fold(step, data, (a, [])) in
(a, List.rev(bs))
\end{lstlisting}
  \caption{Source code for \bGtree{}.}
  \label{fig:gtree-code}
\end{figure}

\begin{table}[H]
  \small
  \centering
  \caption{Evaluated inference plans for \bGtree{}.}
  \begin{tabular}[h]{llllll}
    \toprule
      & \multicolumn{2}{c}{Variables}\\
      \cmidrule(lr){2-3}
    {}                    & \zlm|a|   & \zlm|b|   \\
    \midrule
    Plan 0                & \symbolic & \symbolic \\
    Plan 1                & \sample & \symbolic \\
    Plan 2                & \symbolic & \sample \\
    Plan 3                & \sample & \sample \\
    \bottomrule
  \end{tabular}
  \label{tab:gtree-plans}
\end{table}

\begin{figure}[H]
  \begin{lstlisting}[basicstyle=\footnotesize\ttfamily,aboveskip=0em,numbers=left]
let s_transition = fun (p, prev_st) ->
  let PLACEHOLDER s <- bernoulli(p) in
  s
let step = fun (yobs, ((trans_prob0, trans_prob1, obs_noise0, obs_noise1), ss, xs)) ->
  let prev_st = List.hd(ss) in
  let prev_xt = List.hd(xs) in
  (* p(s_t | s_{t-1}) *)
  let trans_prob = if prev_st then trans_prob1 else trans_prob0 in
  let st = s_transition(trans_prob) in
  (* p(x_t | x_{t-1}, s_t) *)
  let PLACEHOLDER x1 <- gaussian(prev_xt, 1.2429) in
  let PLACEHOLDER x0 <- gaussian(prev_xt, 0.02248262) in
  let xt = if st then x1 else x0 in
  (* p(z_t | x_t, s_t) *)
  let var = if st then obs_noise1 else obs_noise0 in
  let () = observe(gaussian(xt, var), yobs) in
  let () = resample() in
  ((trans_prob0, trans_prob1, obs_noise0, obs_noise1), cons(st, ss), cons(xt, xs))
let PLACEHOLDER trans_prob0 <- beta(1., 1.) in
let PLACEHOLDER trans_prob1 <- beta(1., 1.) in
let PLACEHOLDER obs_noise0 <- invgamma(1., 1.) in
let PLACEHOLDER obs_noise1 <- invgamma(1., 1.) in
let ((trans_prob0, trans_prob1, obs_noise0, obs_noise1), ss, xs) =
  fold(step, data, ((trans_prob0, trans_prob1, obs_noise0, obs_noise1), [false], [0])) in
let xs = List.tl(List.rev(xs)) in
let ss = List.tl(List.rev(ss)) in
(trans_prob0, trans_prob1, obs_noise0, obs_noise1, ss, xs)
\end{lstlisting}
  \caption{Source code for \bSlds{}.}
  \label{fig:slds-code}
\end{figure}

\begin{table}[H]
  \small
  \centering
  \caption{Evaluated inference plans for \bSlds{}.}
  \begin{tabular}[h]{llllllll}
    \toprule
      & \multicolumn{7}{c}{Variables}\\
      \cmidrule(lr){2-8}
    {}                    & \zlm|s|   & \zlm|x1|    & \zlm|x0|  & \zlm|trans_prob0|    & \zlm|trans_prob1| & \zlm|obs_noise0|    & \zlm|obs_noise1|   \\
    \midrule
    Plan 67               & \sample   & \symbolic   & \symbolic & \symbolic            & \sample           & \symbolic           & \symbolic \\
    Plan 81               & \sample   & \symbolic   & \sample   & \symbolic            & \symbolic         & \symbolic           & \sample   \\
    Plan 98               & \sample   & \sample     & \symbolic & \symbolic            & \symbolic         & \sample             & \symbolic   \\
    Plan 112              & \sample   & \sample     & \sample   & \symbolic            & \symbolic         & \symbolic           & \symbolic   \\
    Plan 113              & \sample   & \sample     & \sample   & \symbolic            & \symbolic         & \symbolic           & \sample   \\
    Plan 114              & \sample   & \sample     & \sample   & \symbolic            & \symbolic         & \sample             & \symbolic   \\
    Plan 115              & \sample   & \sample     & \sample   & \symbolic            & \symbolic         & \sample             & \sample   \\
    Plan 116              & \sample   & \sample     & \sample   & \symbolic            & \sample           & \symbolic           & \symbolic   \\
    Plan 120              & \sample   & \sample     & \sample   & \sample              & \symbolic         & \symbolic           & \symbolic   \\
    Plan 127              & \sample   & \sample     & \sample   & \sample              & \sample           & \sample             & \sample   \\
    \bottomrule
  \end{tabular}
  \label{tab:slds-plans}
\end{table}

\begin{figure}[H]
  \begin{lstlisting}[basicstyle=\footnotesize\ttfamily,aboveskip=0em,numbers=left]
let alt = fun (x, y) ->
  if 50 < y then 
    10 - 0.01 * x * x + 0.0001 * x * x * x
  else 
    x + 0.1
let step = fun ((sx_obs, sy_obs, a_obs), (sx_l, sy_l, x_l, y_l)) ->
  let sx_i = List.hd(sx_l) in
  let sy_i = List.hd(sy_l) in
  let x_i = List.hd(x_l) in
  let y_i = List.hd(y_l) in
  let PLACEHOLDER sx <- gaussian(sx_i, 0.1) in
  let PLACEHOLDER sy <- gaussian(sy_i, 0.1) in
  let PLACEHOLDER x <- gaussian(x_i + sx_i, 1) in
  let PLACEHOLDER y <- gaussian(y_i + sy_i, 1) in
  let a = alt(x, y) in
  let () = observe(gaussian(sx, 1), sx_obs) in
  let () = observe(gaussian(sy, 1), sy_obs) in
  let () = observe(gaussian(a, 1), a_obs) in
  let () = resample() in
  (cons(sx, sx_l), cons(sy, sy_l), cons(x, x_l), cons(y, y_l))
let (sx_l, sy_l, x_l, y_l) = fold(step, data, ([0], [0], [0.1], [0.1])) in
let sx_l = List.tl(List.rev(sx_l)) in
let sy_l = List.tl(List.rev(sy_l)) in
let x_l = List.tl(List.rev(x_l)) in
let y_l = List.tl(List.rev(y_l)) in
(sx_l, sy_l, x_l, y_l)
\end{lstlisting}
  \caption{Source code for \bRunner{}.}
  \label{fig:runner-code}
\end{figure}

\begin{table}[H]
  \small
  \centering
  \caption{Evaluated inference plans for \bRunner{}.}
  \begin{tabular}[h]{llllll}
    \toprule
      & \multicolumn{4}{c}{Variables}\\
      \cmidrule(lr){2-5}
    {}                    & \zlm|sx|   & \zlm|sy|    & \zlm|x| & \zlm|y| \\
    \midrule
    Plan 3                & \symbolic  & \symbolic   & \sample & \sample \\
    Plan 7                & \symbolic  & \sample     & \sample & \sample \\
    Plan 11               & \sample    & \symbolic   & \sample & \sample \\
    Plan 15               & \sample    & \sample     & \sample & \sample \\
    \bottomrule
  \end{tabular}
  \label{tab:runner-plans}
\end{table}

\begin{figure}[H]
  \begin{lstlisting}[basicstyle=\footnotesize\ttfamily,aboveskip=0em,numbers=left]
let bernoulli_priors = fun i ->
  let PLACEHOLDER cell <- bernoulli(0.5) in
  cell
let move = fun (max_pos, prev_x, cmd) ->
  let wheel_noise = 0.1 in
  let PLACEHOLDER wheel_slip <- bernoulli(wheel_noise) in
  let cmd2 = 
    if (prev_x + cmd) < 0 then 
      0
    else if max_pos < (prev_x + cmd) then
      0
    else
      cmd 
  in
  let x = if wheel_slip then prev_x else (prev_x + cmd2) in
  x
let find = fun (curr, (n, i, found, res)) -> 
  let found2 = if n = i then true else found in
  let res2 = if n = i then curr else res in
  (n, i+1, found2, res2)
let get = fun (map, n) ->
  (* Assumes n is within range *)
  let (_, _, _, res) = fold(find, map, (n, 0, false, -1)) in
  res
let step = fun ((obs, cmd), (max_pos, map, xs)) ->
  let prev_x = if List.len(xs) = 0 then 0 else List.hd(xs) in
  let sensor_noise = 0.1 in
  let x = move(max_pos, prev_x, cmd) in
  let o = get(map, x) in
  let () = observe(bernoulli(if o then (1 - sensor_noise) else sensor_noise), obs) in
  let () = resample() in
  (max_pos, map, cons(x, xs))
let max_pos = 10 in
let map = List.map(bernoulli_priors, List.range(0, max_pos + 1)) in
let (_, map, xs) = fold(step, data, (max_pos, map, [])) in
let x = List.hd(xs) in
(map, x)
\end{lstlisting}
  \caption{Source code for \bSlam{}.}
  \label{fig:slam-code}
\end{figure}

\begin{table}[H]
  \small
  \centering
  \caption{Evaluated inference plans for \bSlam{}.}
  \begin{tabular}[h]{llllll}
    \toprule
      & \multicolumn{2}{c}{Variables}\\
      \cmidrule(lr){2-3}
    {}                    & \zlm|cell|   & \zlm|wheel| \\
    \midrule
    Plan 0                & \symbolic  & \symbolic   \\
    Plan 1                & \symbolic  & \sample     \\
    Plan 2                & \sample    & \symbolic   \\
    Plan 3                & \sample    & \sample     \\
    \bottomrule
  \end{tabular}
  \label{tab:slam-plans}
\end{table}

\begin{figure}[H]
  \begin{lstlisting}[basicstyle=\footnotesize\ttfamily,aboveskip=0em,numbers=left]
let step = fun ((left_wheel_rate, right_wheel_rate), (start, vs, os)) ->
  let prev_o = List.hd(os) in
  let prev_v = List.hd(vs) in
  let wb = 2.0 in 
  let sensor_err_l = 1.0 in
  let sensor_err_r = 0.95 in
  let omega_noise = if start then 2500. else 1. in
  let velocity_noise = if start then 2500. else 1. in
  let PLACEHOLDER omega <- gaussian(prev_o, omega_noise) in
  let PLACEHOLDER velocity <- gaussian(prev_v, velocity_noise) in
  let () = observe(gaussian(velocity - (wb * omega), sensor_err_l), left_wheel_rate) in
  let () = observe(gaussian(velocity + (wb * omega), sensor_err_r), right_wheel_rate) in
  let () = resample() in
  (false, cons(velocity, vs), cons(omega, os))
let (_, vs, os) = fold(step, data, (true, [0.], [0.])) in
let velocity = List.hd(vs) in
let omega = List.hd(os) in
(velocity, omega)
\end{lstlisting}
  \caption{Source code for \bWheels{}.}
  \label{fig:wheels-code}
\end{figure}

\begin{table}[H]
  \small
  \centering
  \caption{Evaluated inference plans for \bWheels{}.}
  \begin{tabular}[h]{llllll}
    \toprule
      & \multicolumn{2}{c}{Variables}\\
      \cmidrule(lr){2-3}
    {}                    & \zlm|omega|   & \zlm|velocity| \\
    \midrule
    Plan 0                & \symbolic  & \symbolic   \\
    Plan 1                & \symbolic  & \sample     \\
    Plan 2                & \sample    & \symbolic   \\
    Plan 3                & \sample    & \sample     \\
    \bottomrule
  \end{tabular}
  \label{tab:wheels-plans}
\end{table}

\begin{figure}[H]
  \begin{lstlisting}[basicstyle=\footnotesize\ttfamily,aboveskip=0em,numbers=left]
let step = fun ((x_obs, alt_obs), (xs,alts,q,r)) ->
  let PLACEHOLDER x <- gaussian(List.hd(xs),q) in
  let PLACEHOLDER alt <- gaussian(List.hd(alts),q) in
  let PLACEHOLDER other <- invgamma(1.,10.) in
  let v = if (alt < 5) then (r + other) else (r) in
  let () = observe(gaussian(x,v),x_obs) in
  let () = observe(gaussian(alt,v),alt_obs) in
  let () = resample() in
  (cons(x, xs), cons(alt, alts), q, r)
let PLACEHOLDER q <- invgamma(1.,1.) in
let PLACEHOLDER r <- invgamma(1.,1.) in
let (xs,alts,q,r) = fold(step,data,([0.],[0.],q,r)) in
let xs = List.tl(List.rev(xs)) in
let alts = List.tl(List.rev(alts)) in
(xs, alts, q, r)
\end{lstlisting}
  \caption{Source code for \bAircraft{}.}
  \label{fig:aircraft-code}
\end{figure}

\begin{table}[H]
  \small
  \centering
  \caption{Evaluated inference plans for \bAircraft{}.}
  \begin{tabular}[h]{llllll}
    \toprule
      & \multicolumn{5}{c}{Variables}\\
      \cmidrule(lr){2-6}
    {}                    & \zlm|x|   & \zlm|alt|    & \zlm|other| & \zlm|q|    & \zlm|r|   \\
    \midrule
    Plan 15              &  \symbolic & \sample    & \sample  & \sample    & \sample\\
    Plan 29              &  \sample   & \sample    & \sample  & \symbolic  & \sample \\
    Plan 31              &  \sample   & \sample    & \sample  & \sample    & \sample \\
    \bottomrule
  \end{tabular}
  \label{tab:aircraft-plans}
\end{table}

\section{Performance Evaluation}
\label{appendix:evaluation}
For each particle count, we plot the median execution time to the 90th percentile of error for each variable of interest using each satisfiable inference plan. Additionally, we plot the median execution time to the median error with 10th and 90th percentile of error to see the estimation variance. We evaluated all possible satisfiable inference plans for each benchmark, except for \bSlds{}, which has 36 satisfiable plans with SSI and 16 satisfiable plans with DS. For this benchmark, we sort the plans by the number of \mfSymbolic{} variables in descending order and only compare the first 4 plans (and any plans tied with those) against the plan with all $\mfSample$d variables and the default plans. Finally, we present the best execution time by any satisfiable inference plan and by the default inference plan for 90\% of the executions to reach target accuracy. The target accuracy is defined as the 90th percentile of error by the default inference plan using the greatest number of particles evaluated that did not produce timeouts.

\subsection{Performance Profiles with 90th Perecentile Error}

\begin{figure}[H]
  \centering
  \begin{subfigure}[c]{0.75\textwidth}
    \centering
    \includegraphics[width=1\textwidth]{figures/noise/smc_ssi_particles.pdf}
    \caption{SSI.}
  \end{subfigure}%
  \\
  \begin{subfigure}[c]{0.75\textwidth}
    \centering
    \includegraphics[width=1\textwidth]{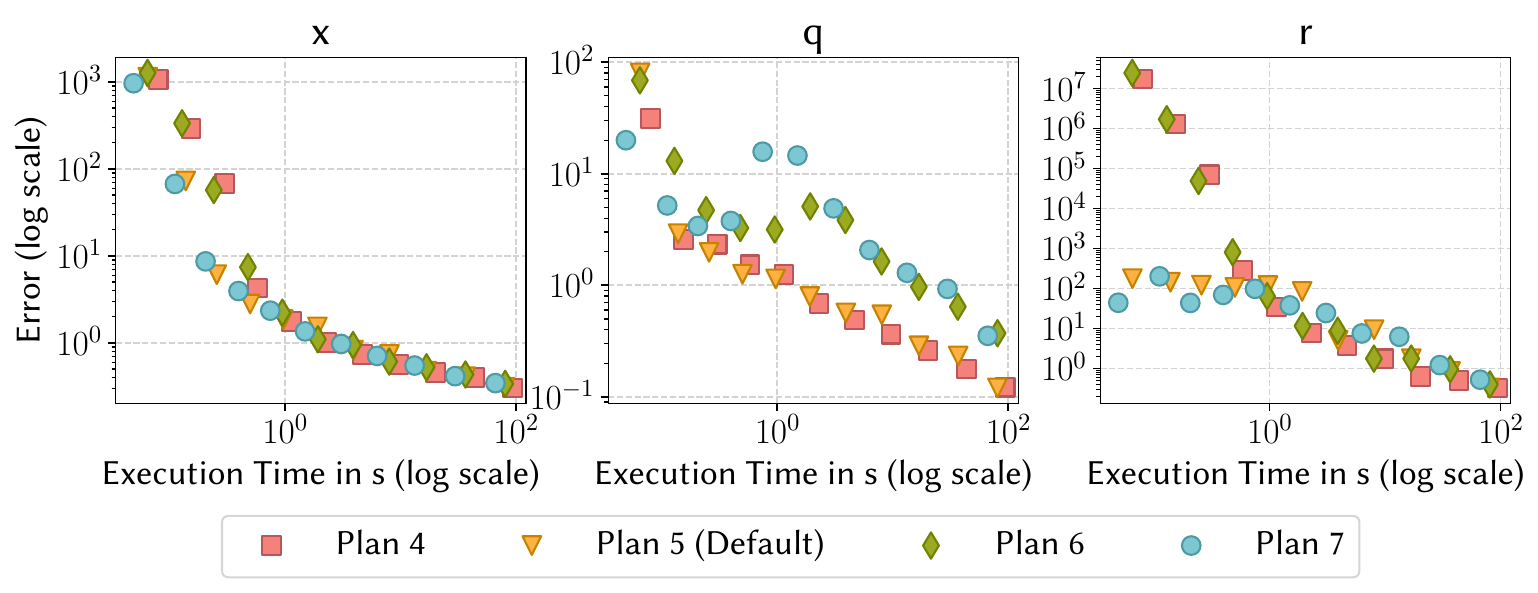}
    \caption{DS.}
  \end{subfigure}%
  \\
  \begin{subfigure}[c]{0.75\textwidth}
    \centering
    \includegraphics[width=1\textwidth]{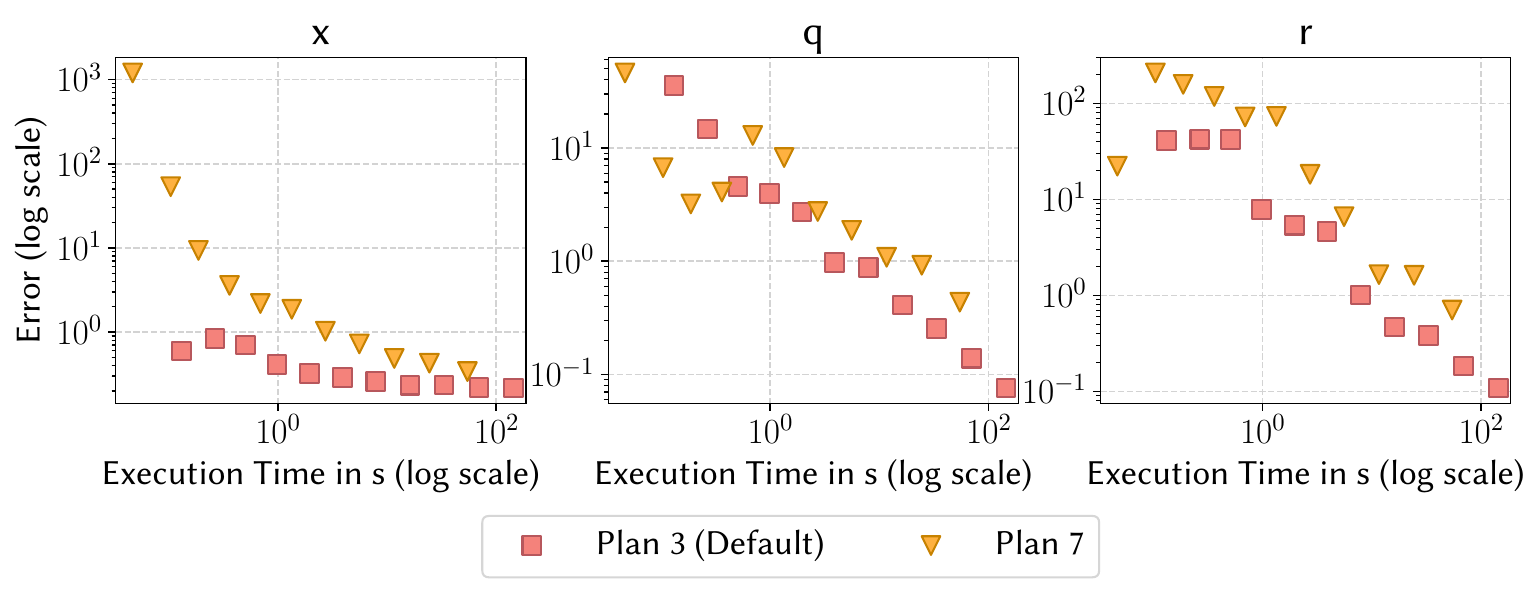}
    \caption{SMC w/ BP.}
  \end{subfigure}%
  \caption{\bNoise{}}
  \label{fig:performance-results-noise}
\end{figure}

\begin{figure}[H]
  \centering
  \begin{subfigure}[c]{0.75\textwidth}
    \centering
    \includegraphics[width=1\textwidth]{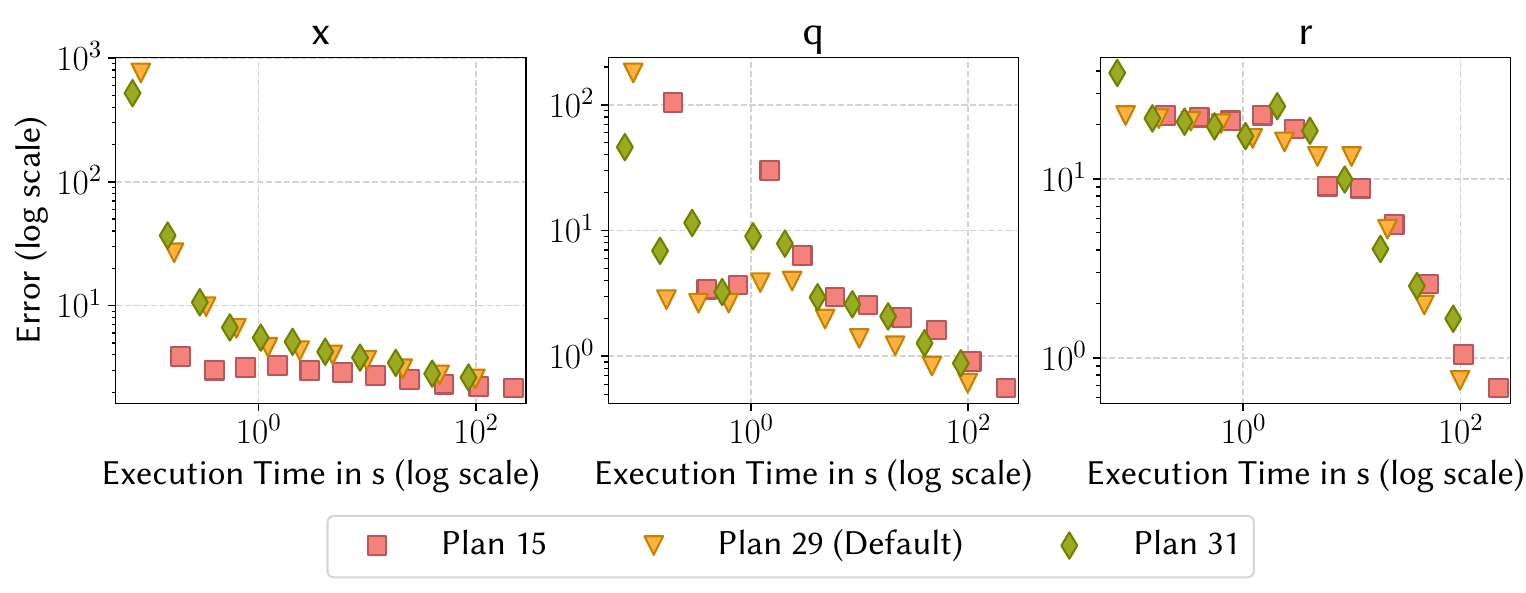}
    \caption{SSI.}
  \end{subfigure}%
  \\
  \begin{subfigure}[c]{0.75\textwidth}
    \centering
    \includegraphics[width=1\textwidth]{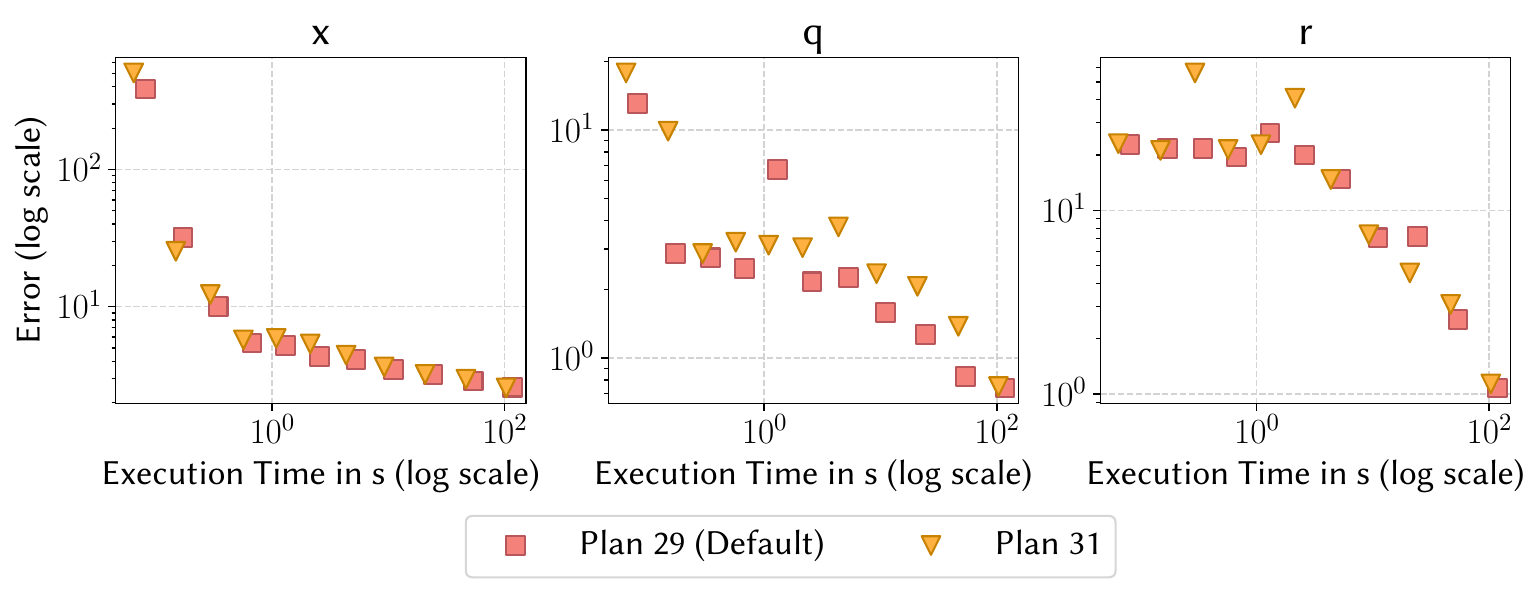}
    \caption{DS.}
  \end{subfigure}%
  \\
  \begin{subfigure}[c]{0.75\textwidth}
    \centering
    \includegraphics[width=1\textwidth]{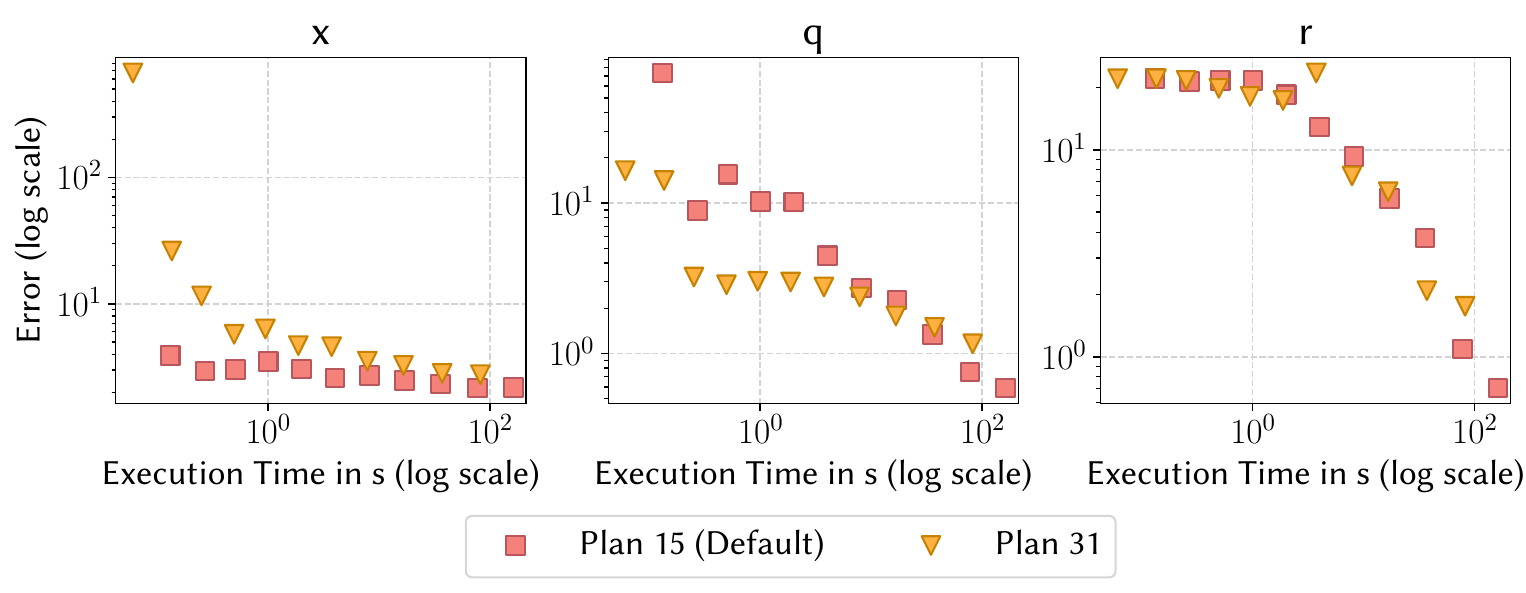}
    \caption{SMC w/ BP.}
  \end{subfigure}%
  \caption{\bRadar{}}
  \label{fig:performance-results-radar}
\end{figure}

\begin{figure}[H]
  \centering
  \begin{subfigure}[c]{0.75\textwidth}
    \centering
    \includegraphics[width=1\textwidth]{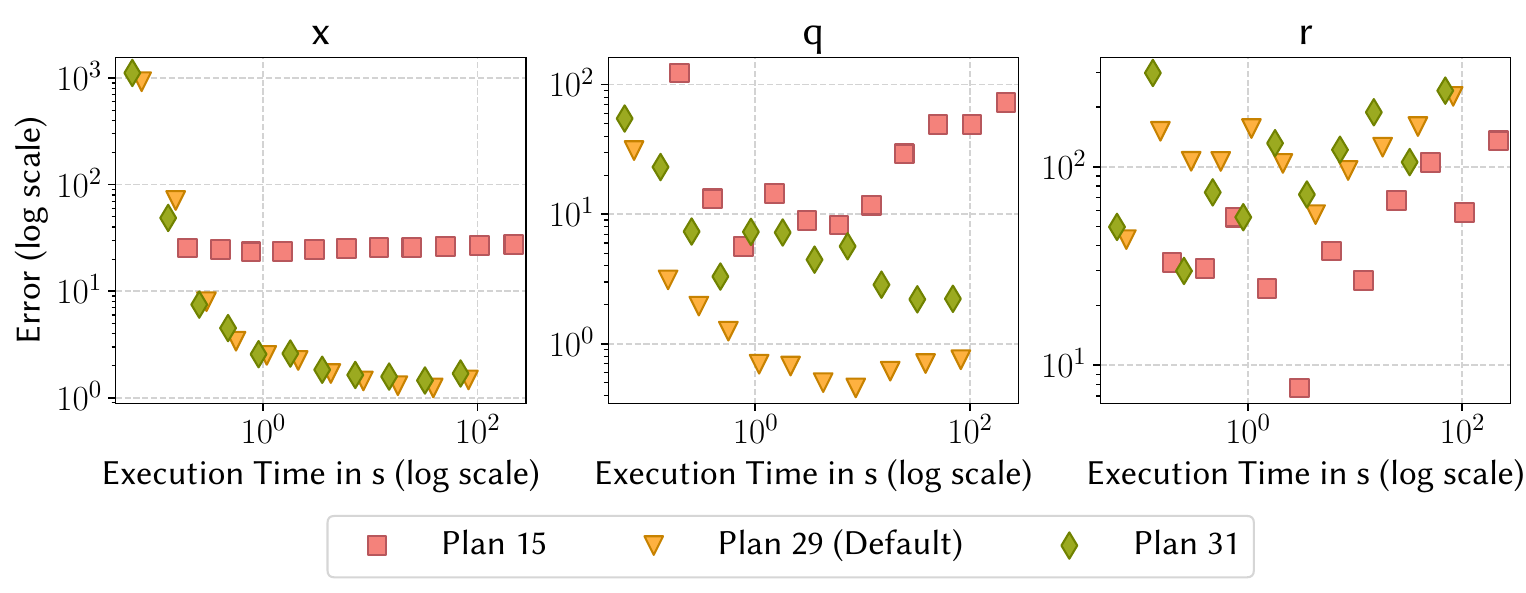}
    \caption{SSI.}
  \end{subfigure}%
  \\
  \begin{subfigure}[c]{0.75\textwidth}
    \centering
    \includegraphics[width=1\textwidth]{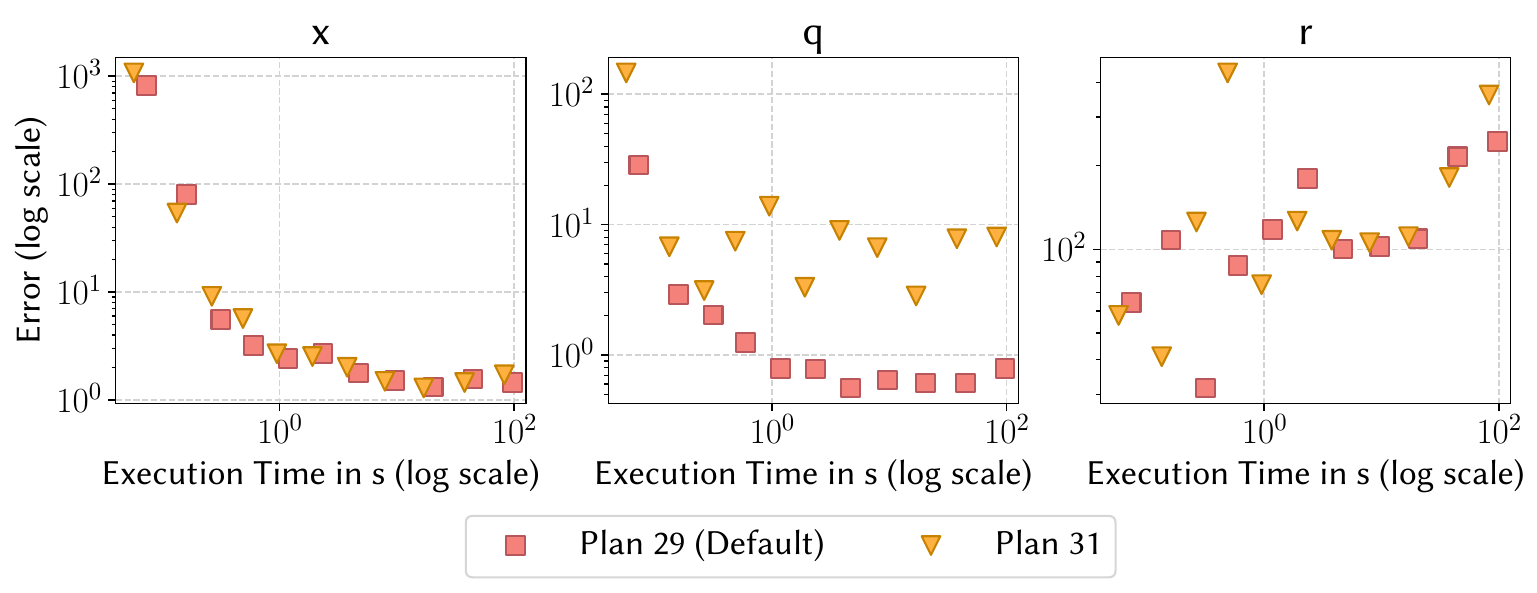}
    \caption{DS.}
  \end{subfigure}%
  \\
  \begin{subfigure}[c]{0.75\textwidth}
    \centering
    \includegraphics[width=1\textwidth]{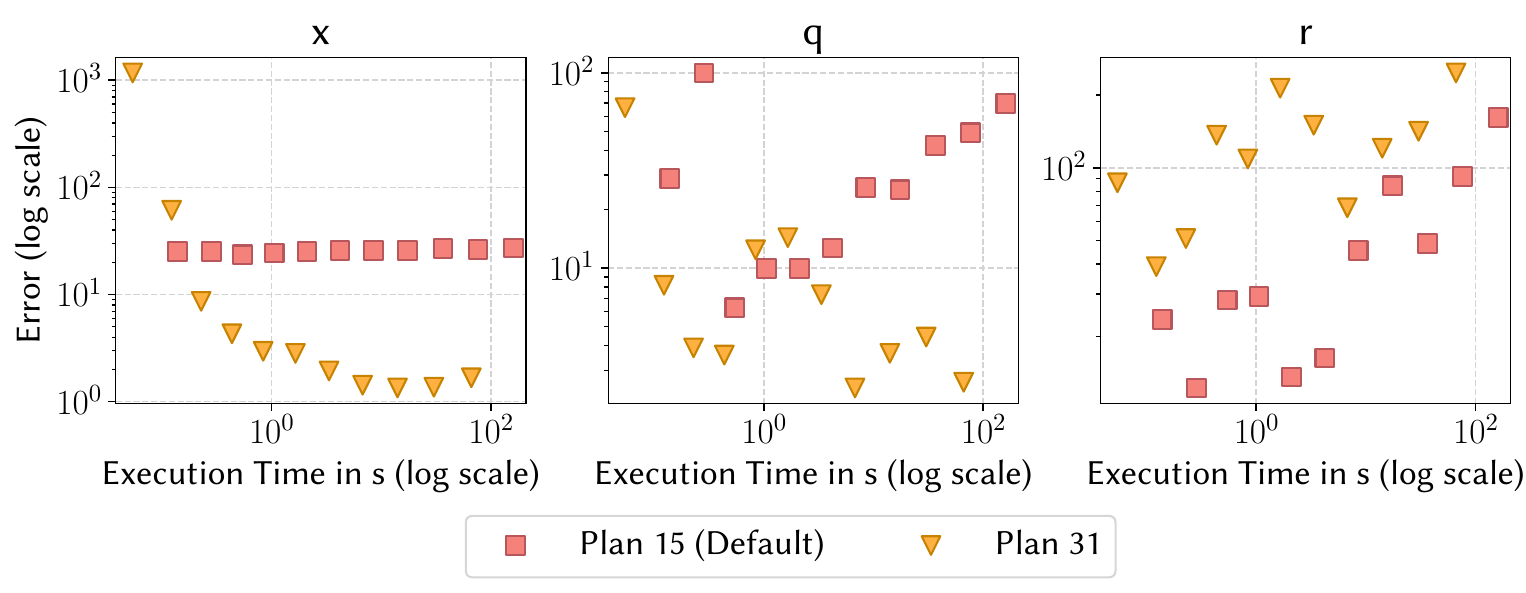}
    \caption{SMC w/ BP.}
  \end{subfigure}%
  \caption{\bEnvnoise{}}
  \label{fig:performance-results-envnoise}
\end{figure}

\begin{figure}[H]
  \centering
  \begin{subfigure}[c]{0.5\textwidth}
    \centering
    \includegraphics[width=1\textwidth]{figures/outlier/smc_ssi_particles.pdf}
    \caption{SSI.}
  \end{subfigure}%
  \\
  \begin{subfigure}[c]{0.5\textwidth}
    \centering
    \includegraphics[width=1\textwidth]{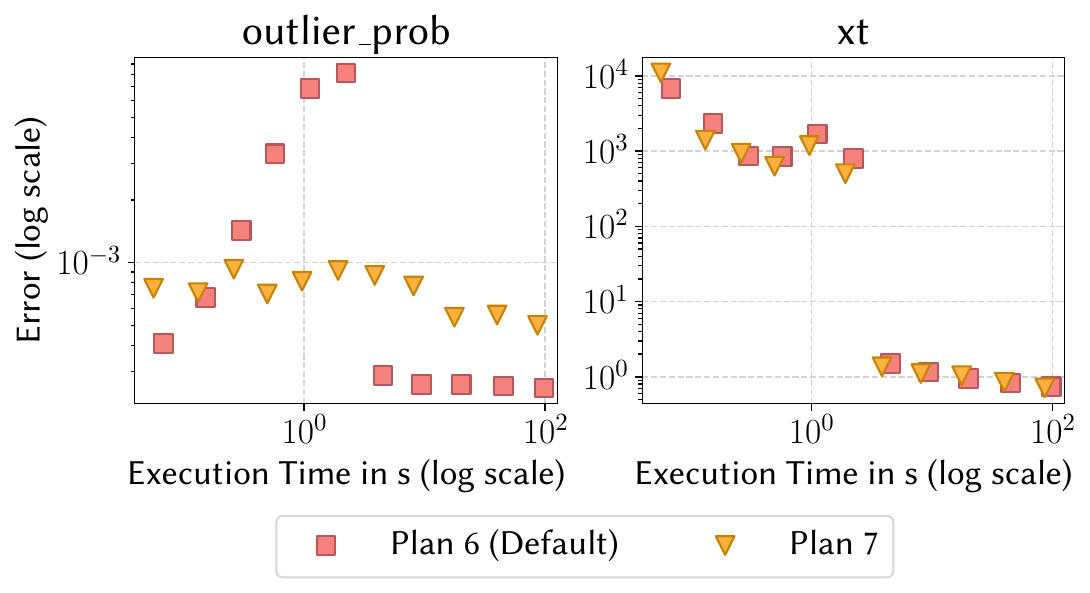}
    \caption{DS.}
  \end{subfigure}%
  \\
  \begin{subfigure}[c]{0.5\textwidth}
    \centering
    \includegraphics[width=0.99\textwidth]{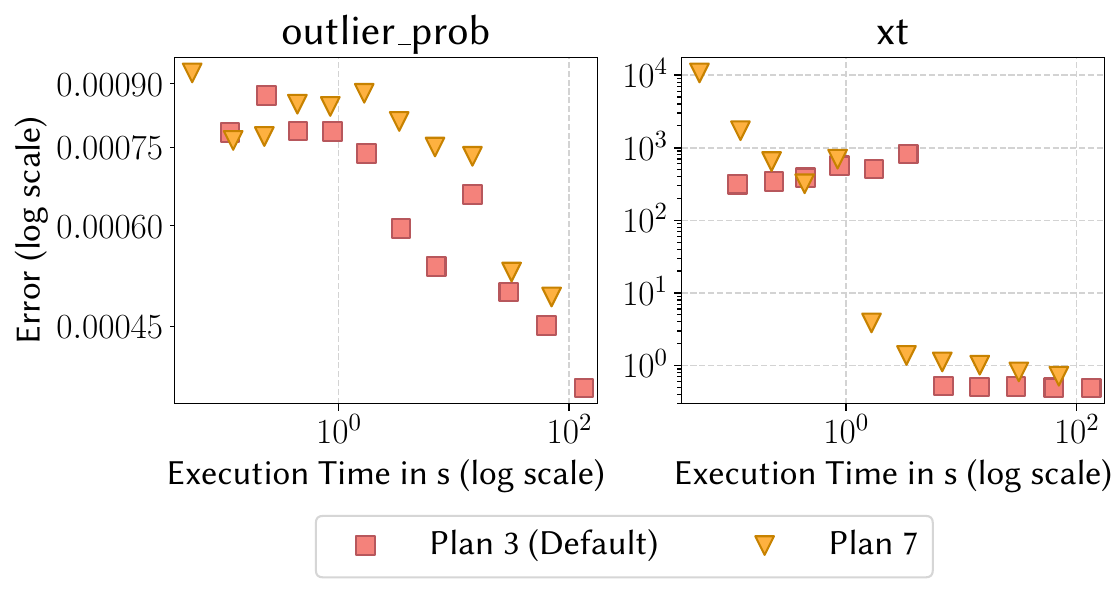}
    \caption{SMC w/ BP.}
  \end{subfigure}%
  \caption{\bOutlier{}}
  \label{fig:performance-results-outlier}
\end{figure}

\begin{figure}[H]
  \centering
  \begin{subfigure}[c]{0.5\textwidth}
    \centering
    \includegraphics[width=1\textwidth]{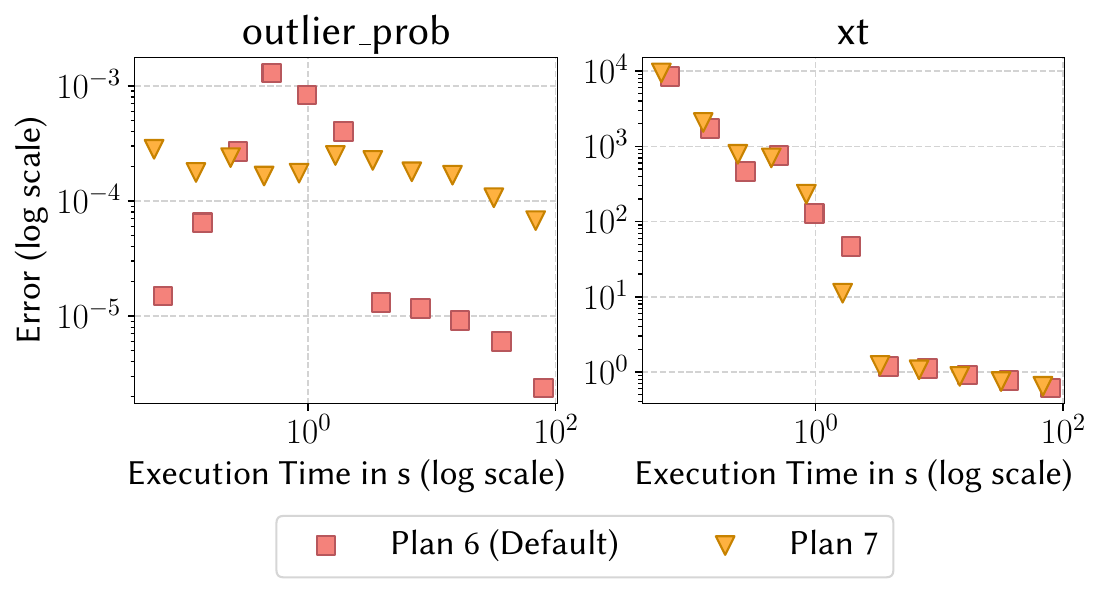}
    \caption{SSI.}
  \end{subfigure}%
  \\
  \begin{subfigure}[c]{0.5\textwidth}
    \centering
    \includegraphics[width=1\textwidth]{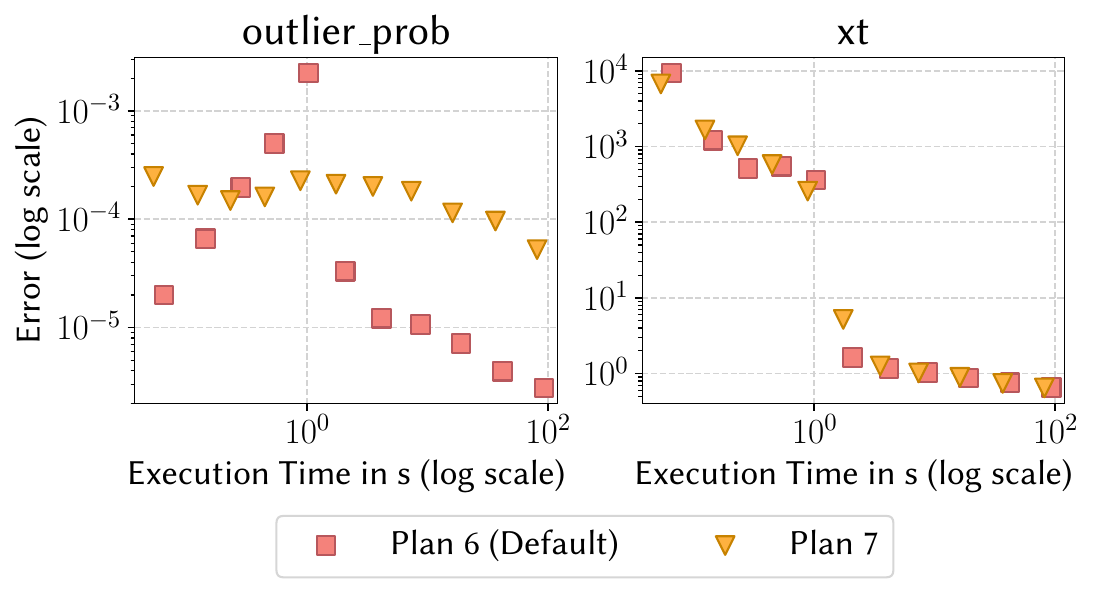}
    \caption{DS.}
  \end{subfigure}%
  \\
  \begin{subfigure}[c]{0.5\textwidth}
    \centering
    \includegraphics[width=1\textwidth]{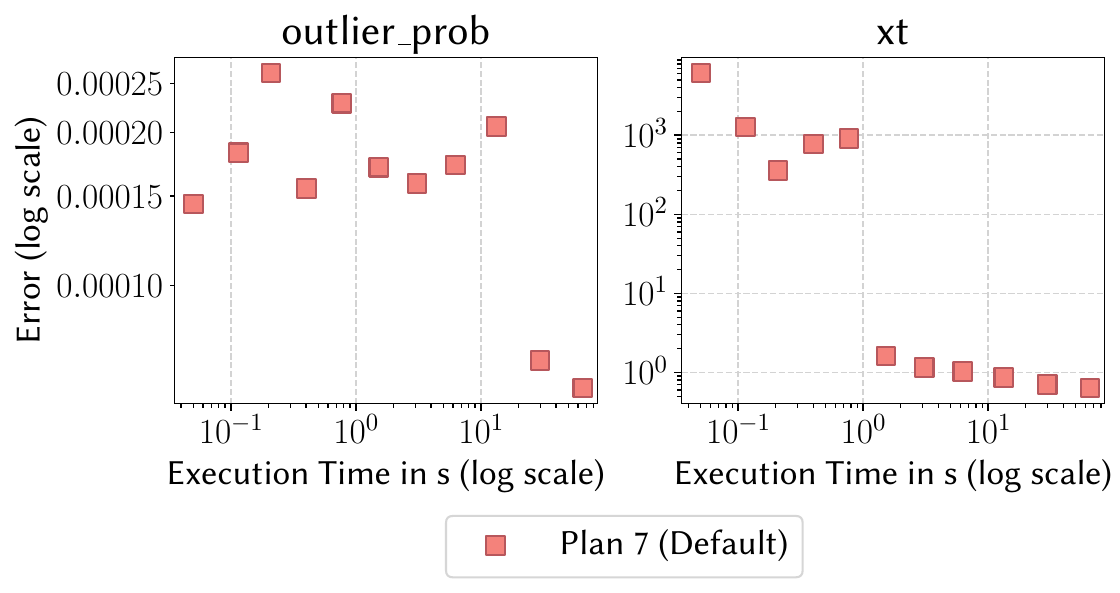}
    \caption{SMC w/ BP.}
  \end{subfigure}%
  \caption{\bOutlierheavy{}.}
  \label{fig:performance-results-outlierheavy}
\end{figure}

\begin{figure}[H]
  \centering
  \begin{subfigure}[c]{0.5\textwidth}
    \centering
    \includegraphics[width=1\textwidth]{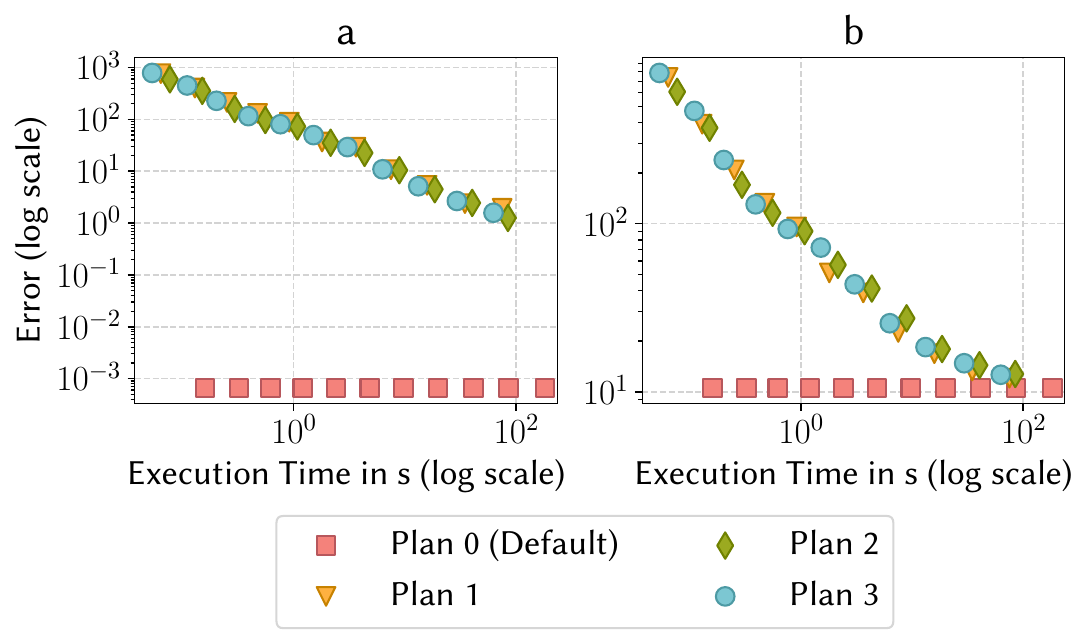}
    \caption{SSI.}
  \end{subfigure}%
  \\
  \begin{subfigure}[c]{0.5\textwidth}
    \centering
    \includegraphics[width=1\textwidth]{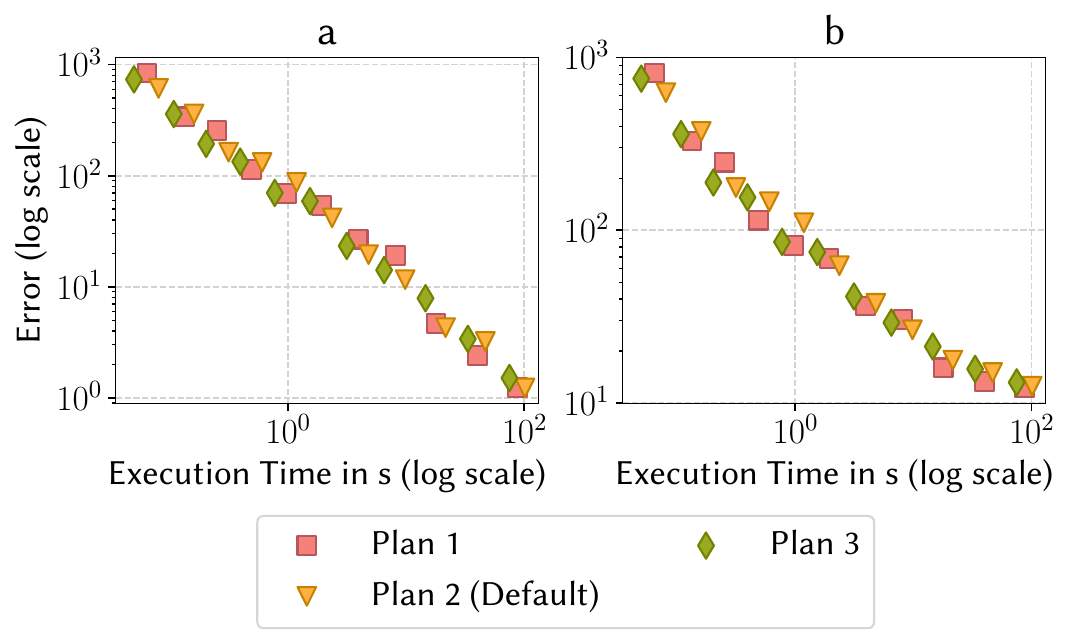}
    \caption{DS.}
  \end{subfigure}%
  \\
  \begin{subfigure}[c]{0.5\textwidth}
    \centering
    \includegraphics[width=1\textwidth]{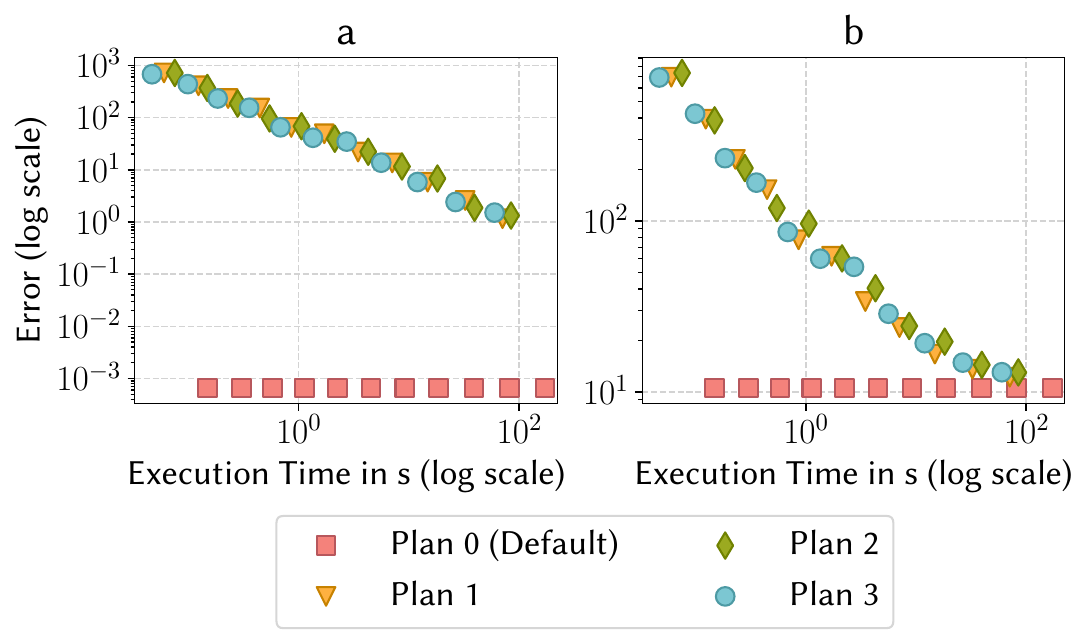}
    \caption{SMC w/ BP.}
  \end{subfigure}%
  \caption{\bGtree{}}
  \label{fig:performance-results-tree}
\end{figure}

\begin{figure}[H]
  \centering
  \begin{subfigure}[c]{0.75\textwidth}
    \centering
    \includegraphics[width=1\textwidth]{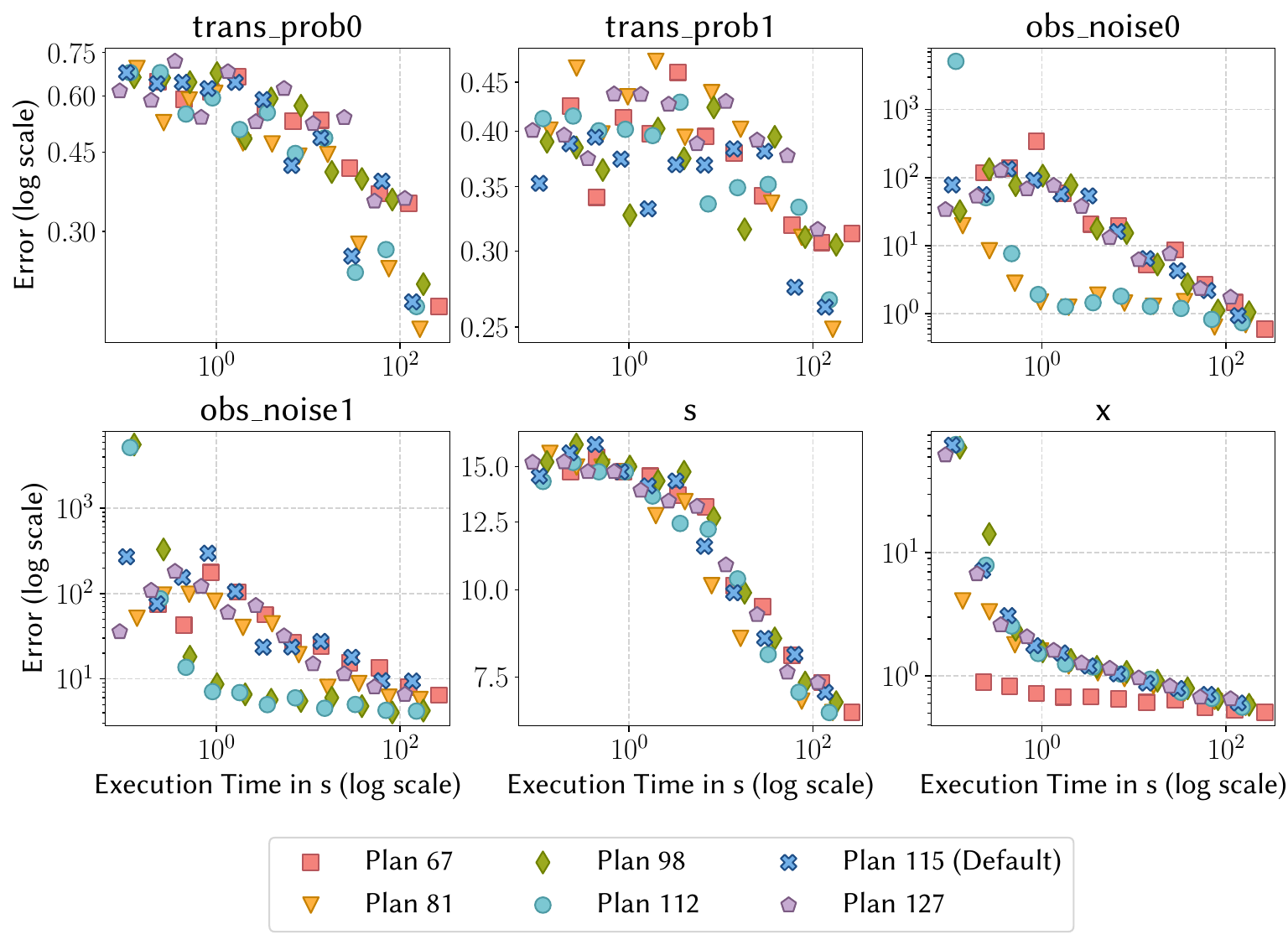}
    \caption{SSI.}
  \end{subfigure}%
  \\
  \begin{subfigure}[c]{0.75\textwidth}
    \centering
    \includegraphics[width=1\textwidth]{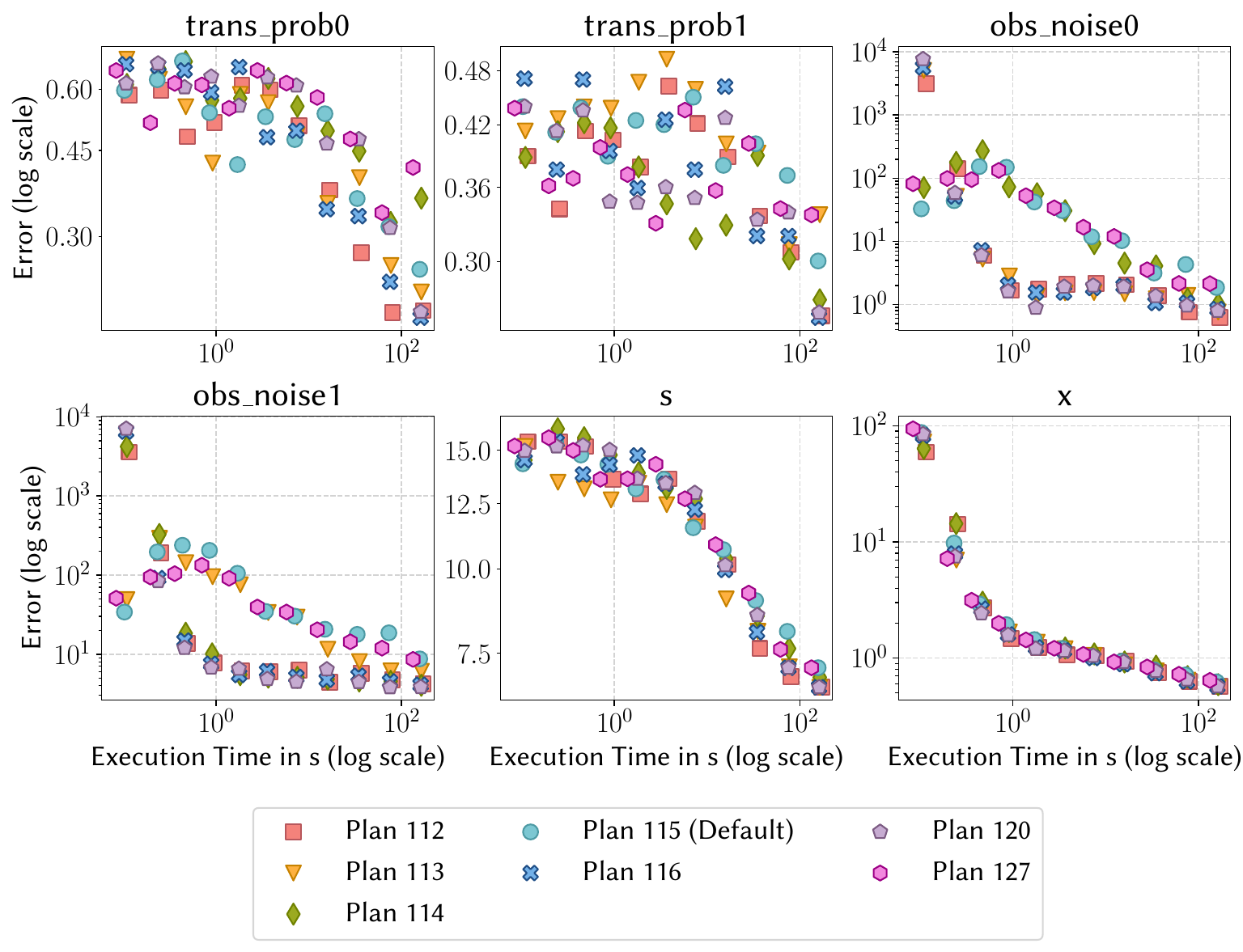}
    \caption{DS.}
  \end{subfigure}%
  \caption{\bSlds{}}
  \label{fig:performance-results-slds-1}
\end{figure}

\begin{figure}[H]
  \centering
  \begin{subfigure}[c]{0.75\textwidth}
    \centering
    \includegraphics[width=1\textwidth]{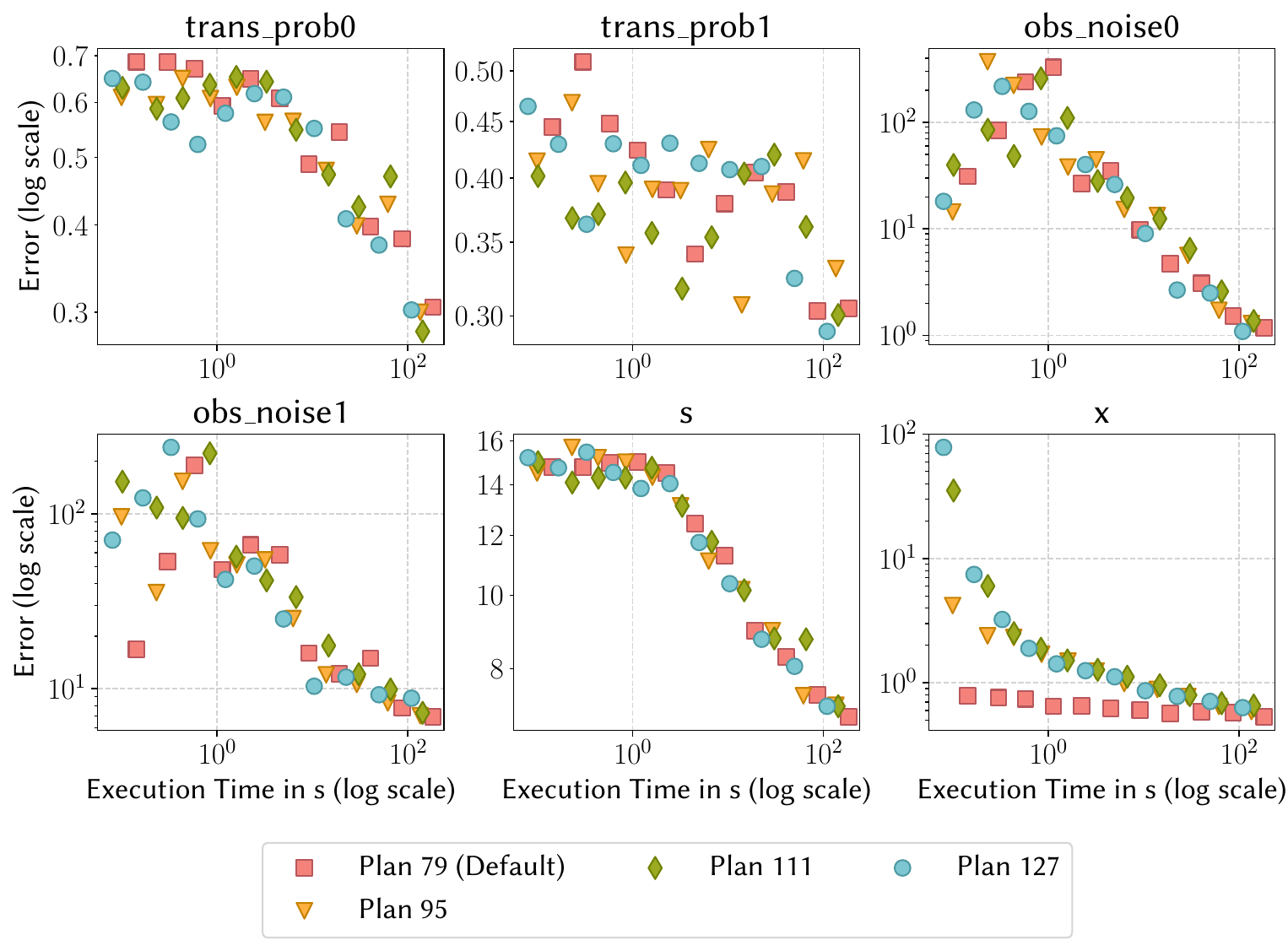}
    \caption{SMC w/ BP.}
  \end{subfigure}%
  \caption{\bSlds{} (continued)}
  \label{fig:performance-results-slds-2}
\end{figure}

\begin{figure}[H]
  \centering
  \begin{subfigure}[c]{1\textwidth}
    \centering
    \includegraphics[width=1\textwidth]{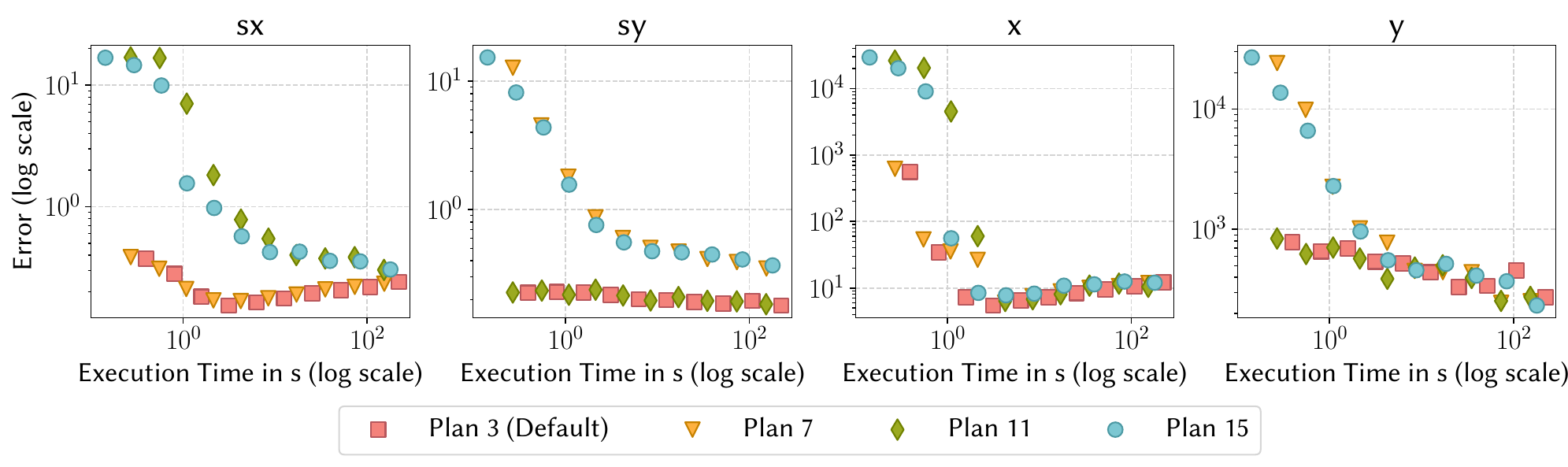}
    \caption{SSI.}
  \end{subfigure}%
  \\
  \begin{subfigure}[c]{1\textwidth}
    \centering
    \includegraphics[width=1\textwidth]{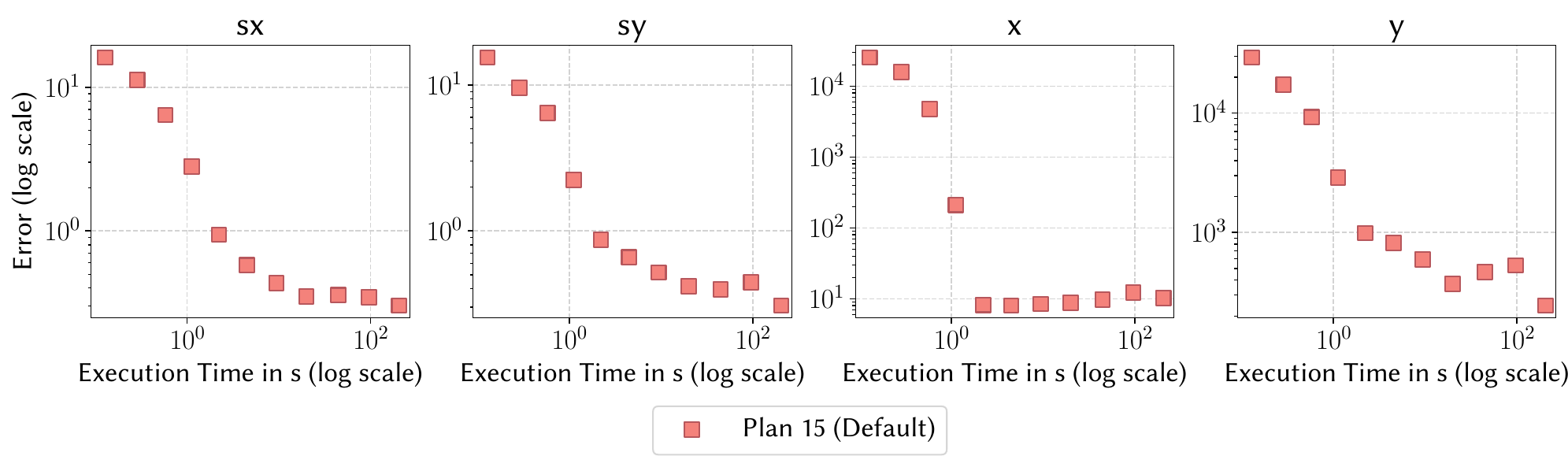}
    \caption{SSI.}
  \end{subfigure}%
  \\
  \begin{subfigure}[c]{1\textwidth}
    \centering
    \includegraphics[width=1\textwidth]{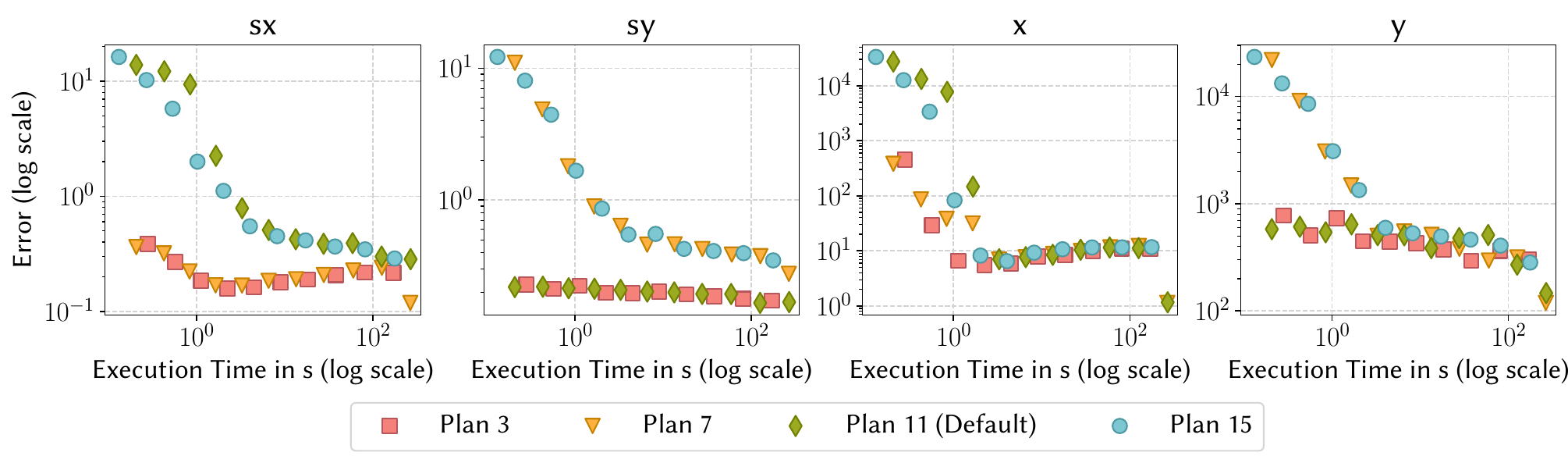}
    \caption{SMC w/ BP.}
  \end{subfigure}%
  \caption{\bRunner{}.}
  \label{fig:performance-results-runner}
\end{figure}

\begin{figure}[H]
  \centering
  \begin{subfigure}[c]{1\textwidth}
    \centering
    \includegraphics[width=0.5\textwidth]{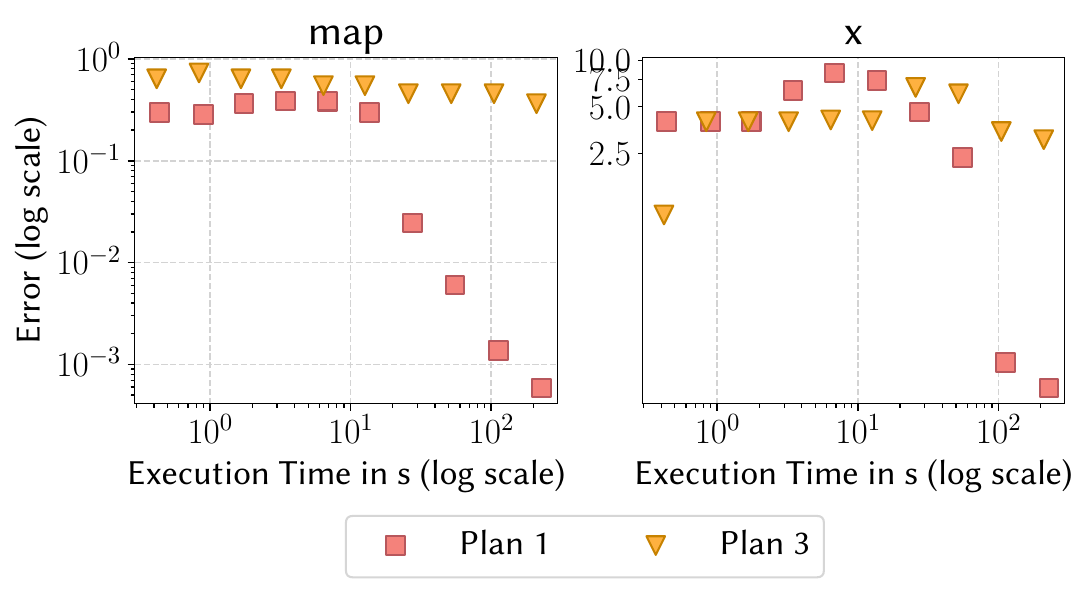}
    \caption{SSI. Plan 0 (Default) and Plan 2 time out for all particles.}
  \end{subfigure}%
  \\
  \begin{subfigure}[c]{1\textwidth}
    \centering
    \includegraphics[width=0.5\textwidth]{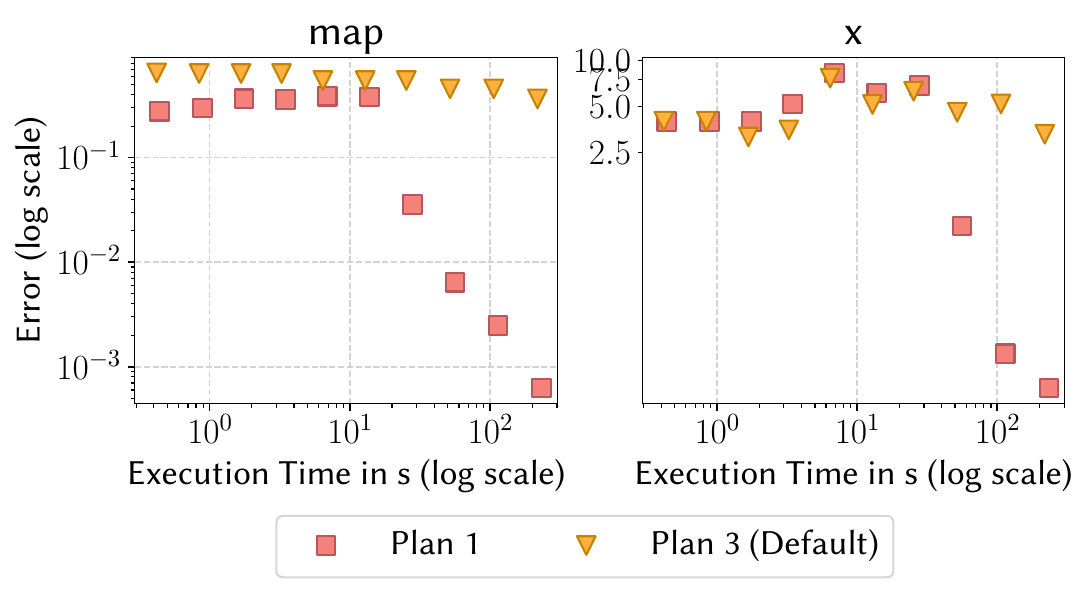}
    \caption{DS.}
  \end{subfigure}%
  \\
  \begin{subfigure}[c]{1\textwidth}
    \centering
    \includegraphics[width=0.5\textwidth]{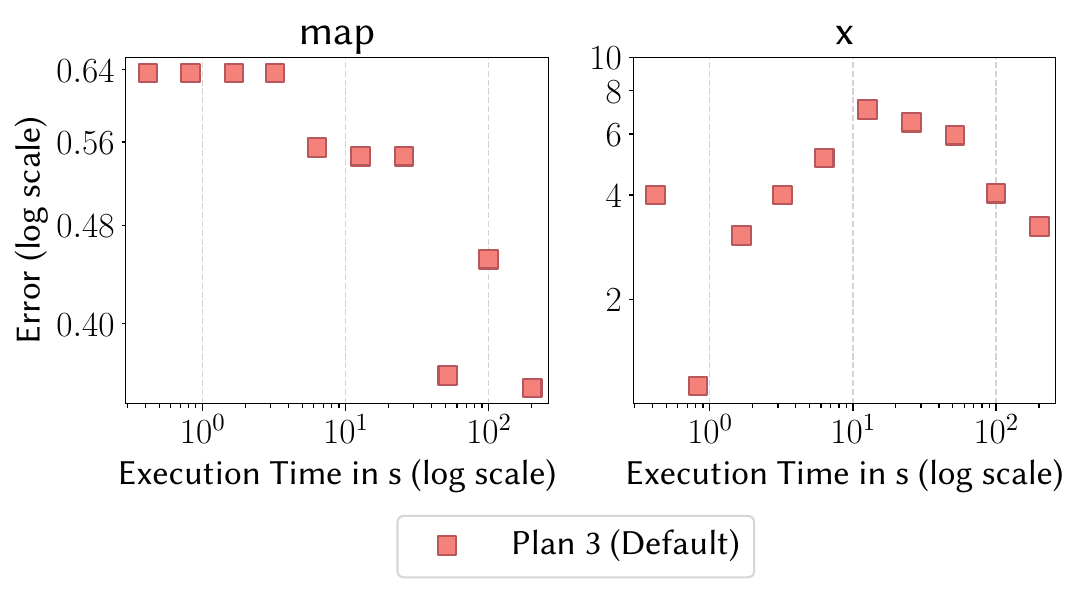}
    \caption{SMC w/ BP.}
  \end{subfigure}%
  \caption{\bSlam{}.}
  \label{fig:performance-results-slam}
\end{figure}

\begin{figure}[H]
  \centering
  \begin{subfigure}[c]{1\textwidth}
    \centering
    \includegraphics[width=0.5\textwidth]{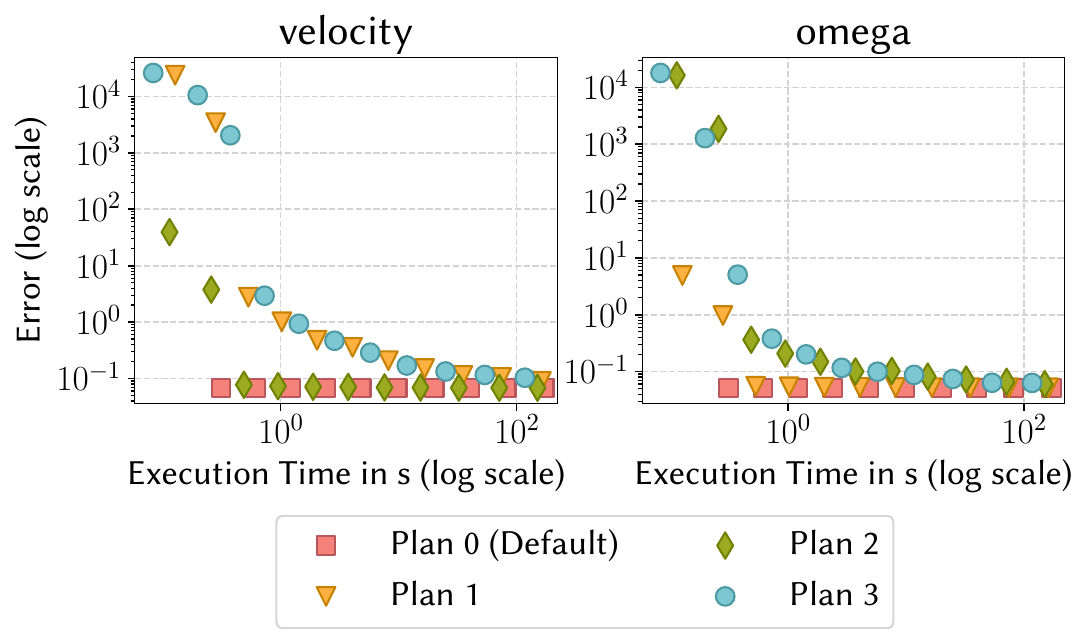}
    \caption{SSI. }
  \end{subfigure}%
  \\
  \begin{subfigure}[c]{1\textwidth}
    \centering
    \includegraphics[width=0.5\textwidth]{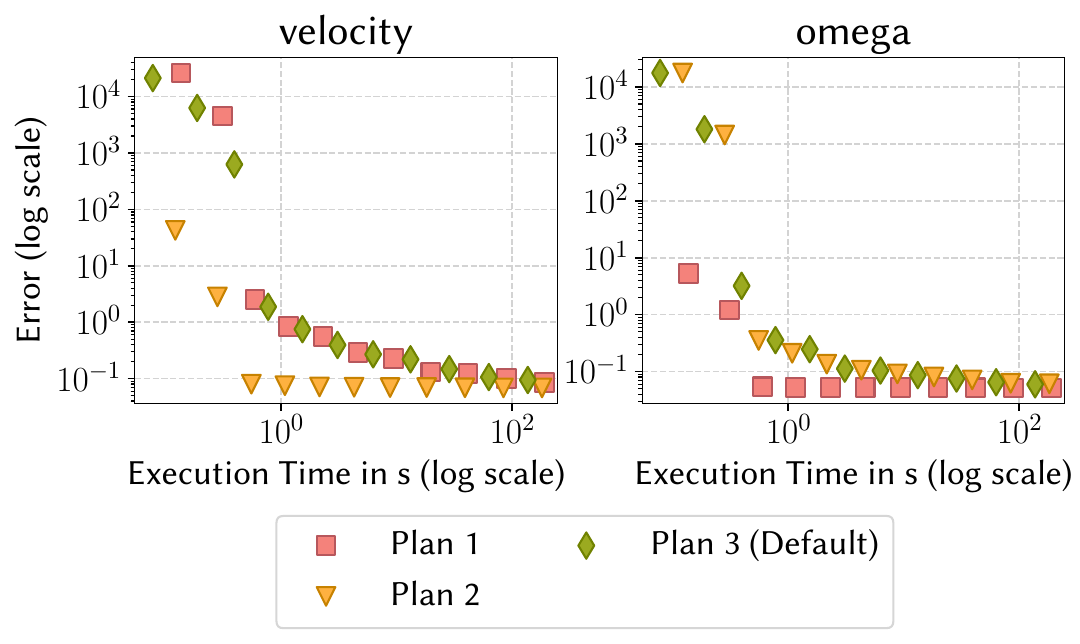}
    \caption{DS.}
  \end{subfigure}%
  \\
  \begin{subfigure}[c]{1\textwidth}
    \centering
    \includegraphics[width=0.5\textwidth]{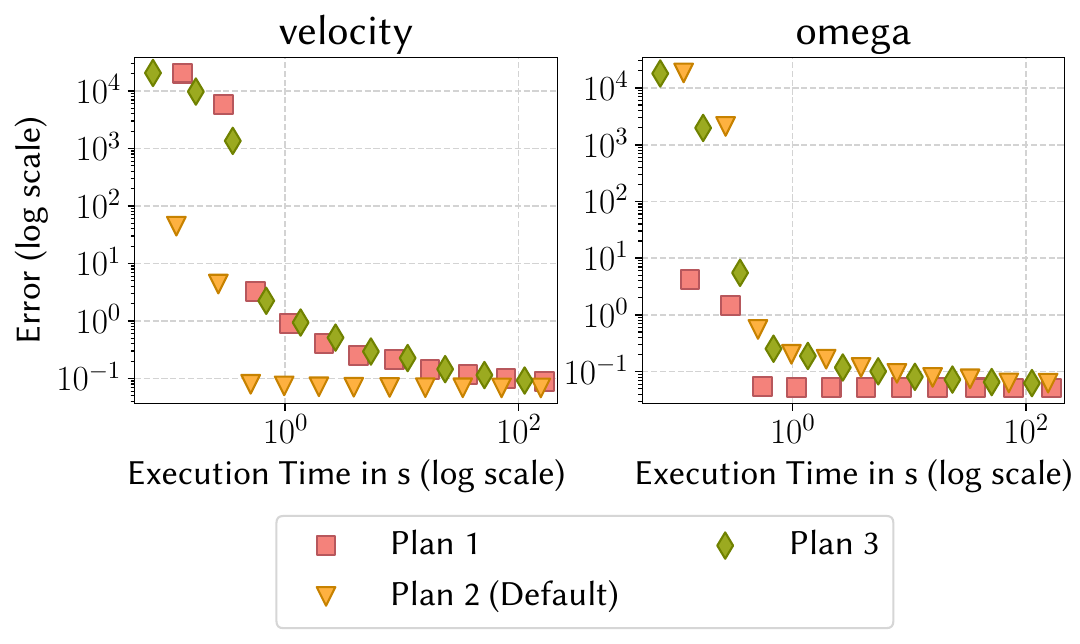}
    \caption{SMC w/ BP.}
  \end{subfigure}%
  \caption{\bWheels{}.}
  \label{fig:performance-results-wheels}
\end{figure}

\begin{figure}[H]
  \centering
  \begin{subfigure}[c]{1\textwidth}
    \centering
    \includegraphics[width=1\textwidth]{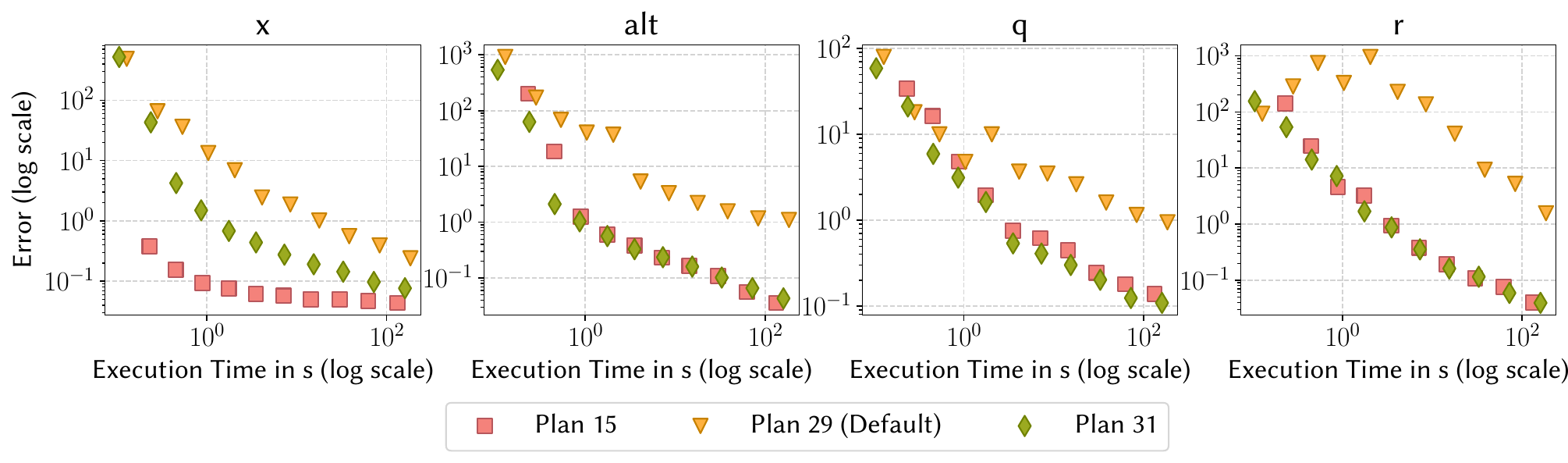}
    \caption{SSI.}
  \end{subfigure}%
  \\
  \begin{subfigure}[c]{1\textwidth}
    \centering
    \includegraphics[width=1\textwidth]{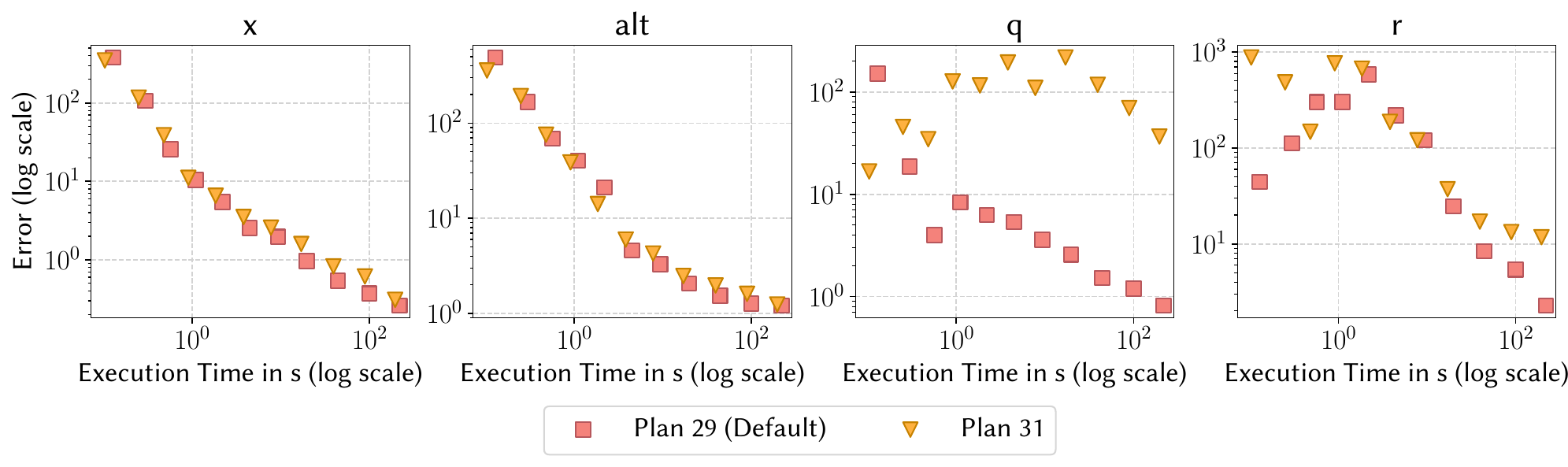}
    \caption{DS.}
  \end{subfigure}%
  \\
  \begin{subfigure}[c]{1\textwidth}
    \centering
    \includegraphics[width=1\textwidth]{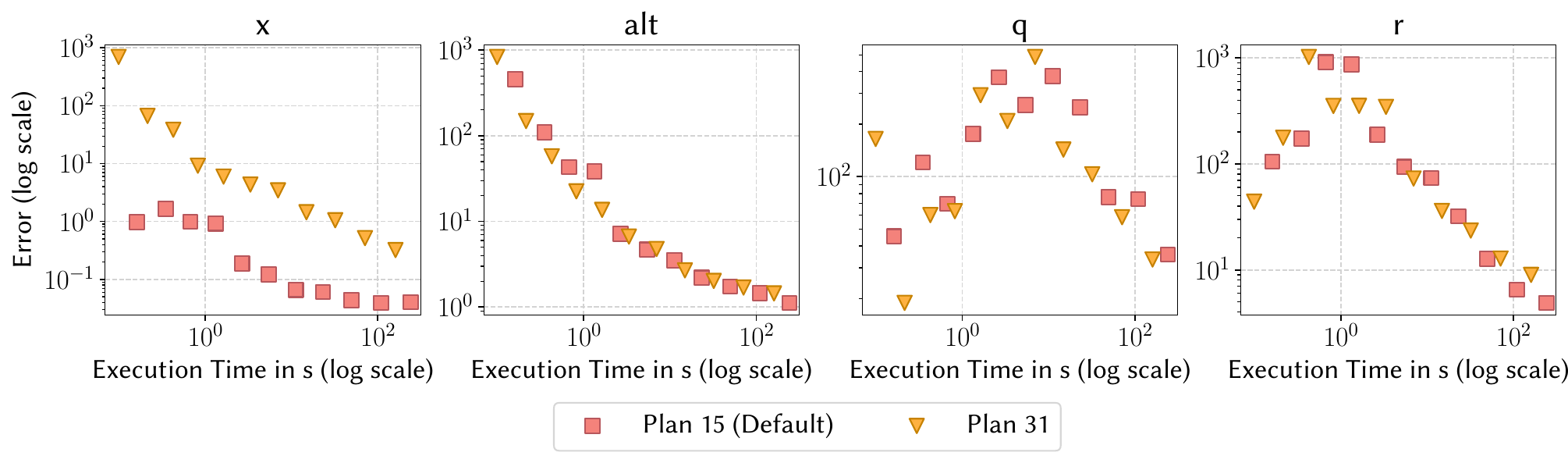}
    \caption{SMC w/ BP.}
  \end{subfigure}%
  \caption{\bAircraft{}. }
  \label{fig:performance-results-aircraft}
\end{figure}

\subsection{Performance Profiles with 10th, 50th, and 90th Perecentile Error}
\begin{figure}[H]
  \centering
  \begin{subfigure}[c]{0.75\textwidth}
    \centering
    \includegraphics[width=1\textwidth]{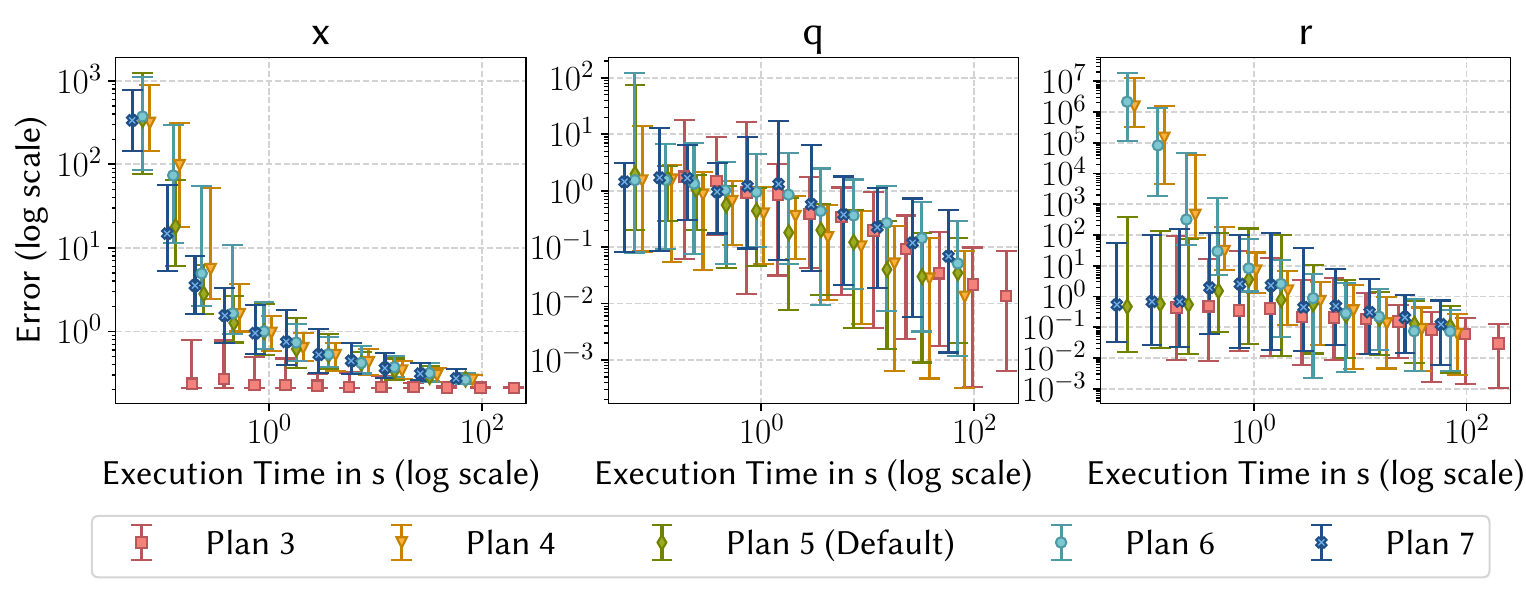}
    \caption{SSI.}
  \end{subfigure}%
  \\
  \begin{subfigure}[c]{0.75\textwidth}
    \centering
    \includegraphics[width=1\textwidth]{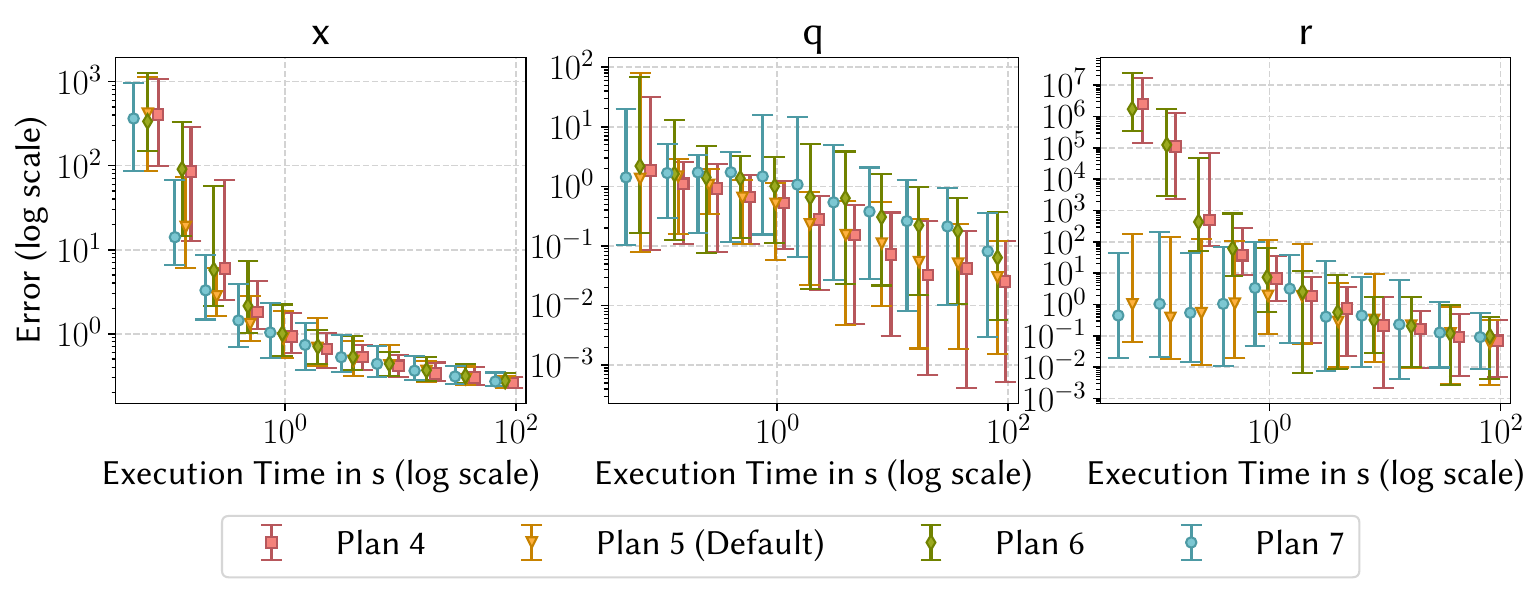}
    \caption{DS.}
  \end{subfigure}%
  \\
  \begin{subfigure}[c]{0.75\textwidth}
    \centering
    \includegraphics[width=1\textwidth]{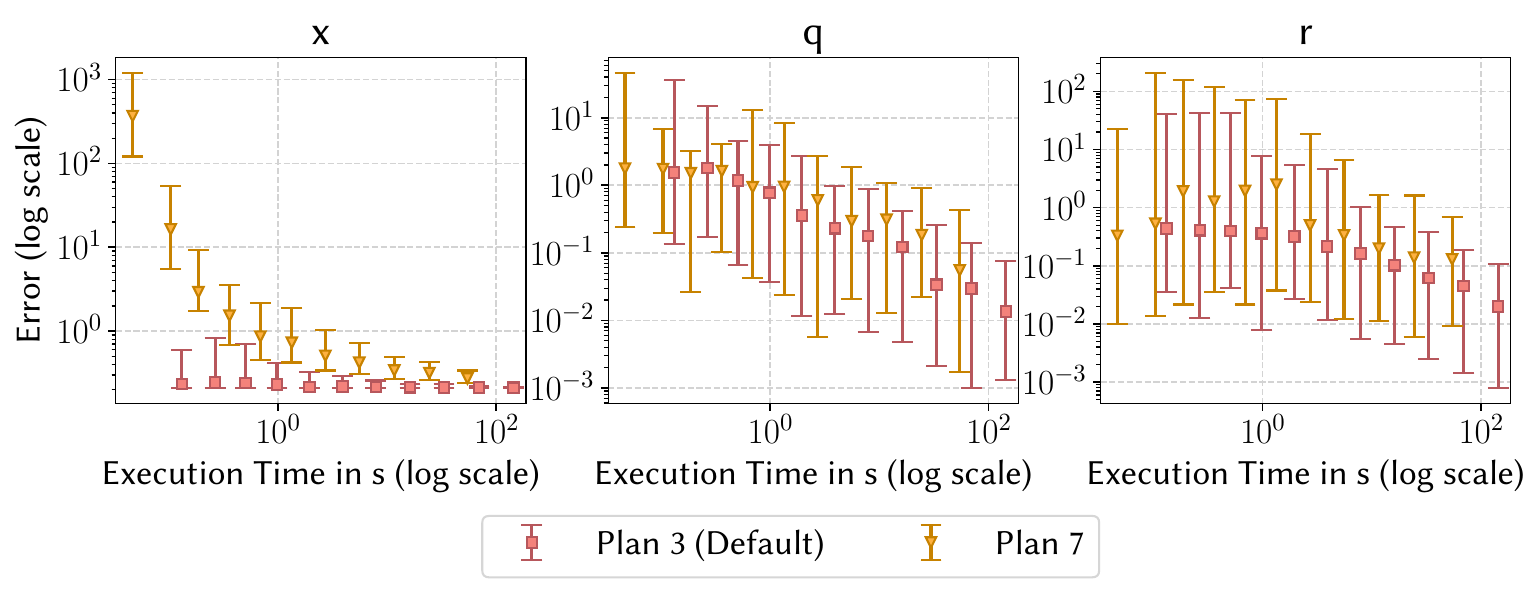}
    \caption{SMC w/ BP.}
  \end{subfigure}%
  \caption{\bNoise{}}
  \label{fig:performance-results-errorbar-noise}
\end{figure}

\begin{figure}[H]
  \centering
  \begin{subfigure}[c]{0.75\textwidth}
    \centering
    \includegraphics[width=1\textwidth]{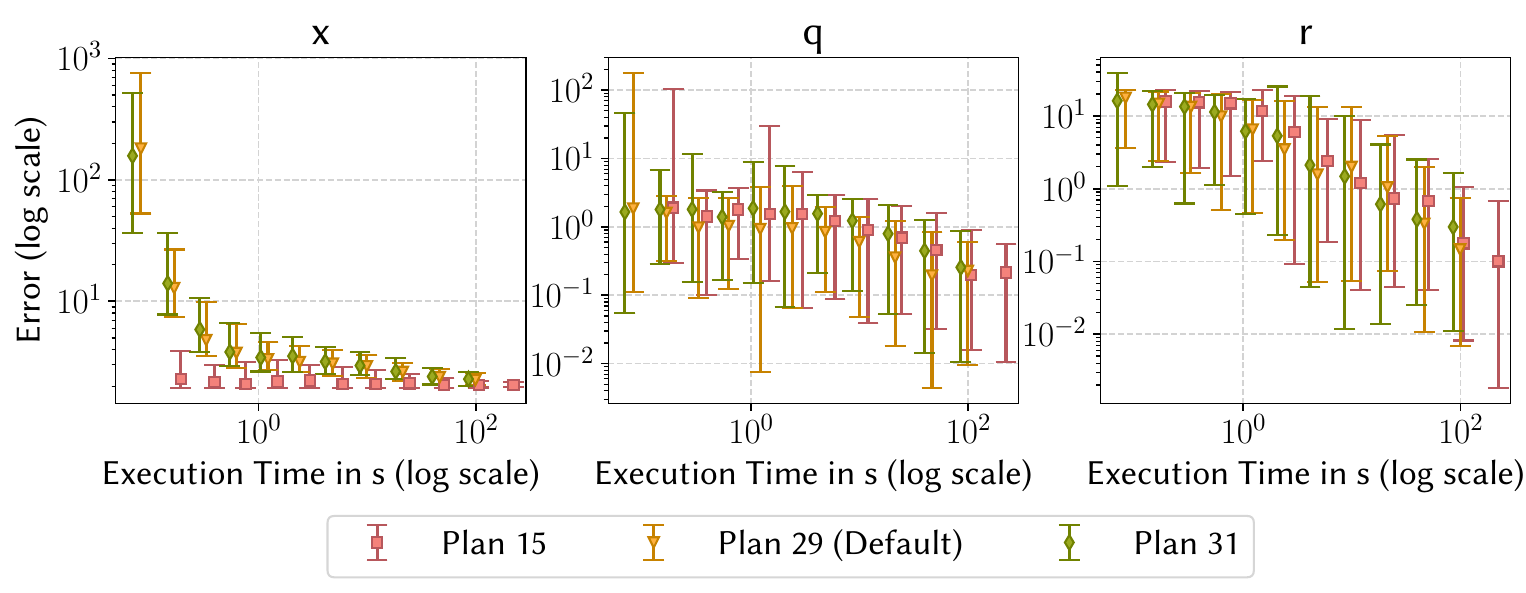}
    \caption{SSI.}
  \end{subfigure}%
  \\
  \begin{subfigure}[c]{0.75\textwidth}
    \centering
    \includegraphics[width=1\textwidth]{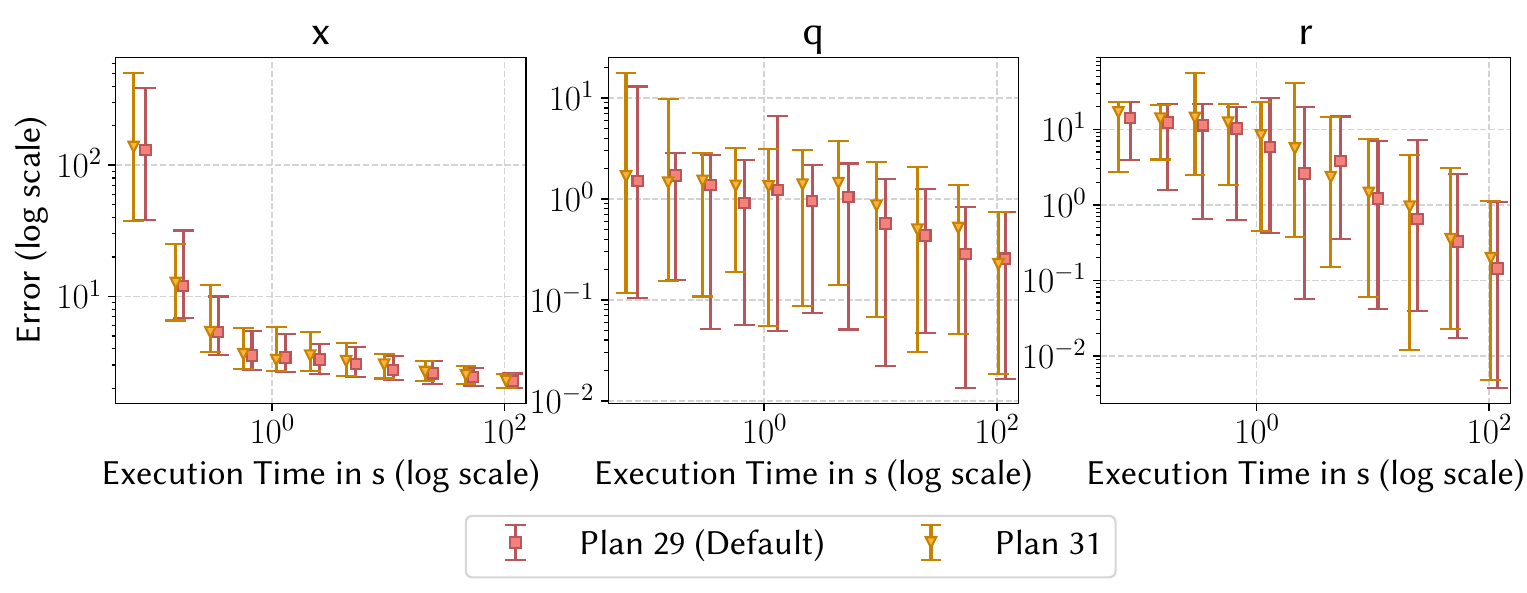}
    \caption{DS.}
  \end{subfigure}%
  \\
  \begin{subfigure}[c]{0.75\textwidth}
    \centering
    \includegraphics[width=1\textwidth]{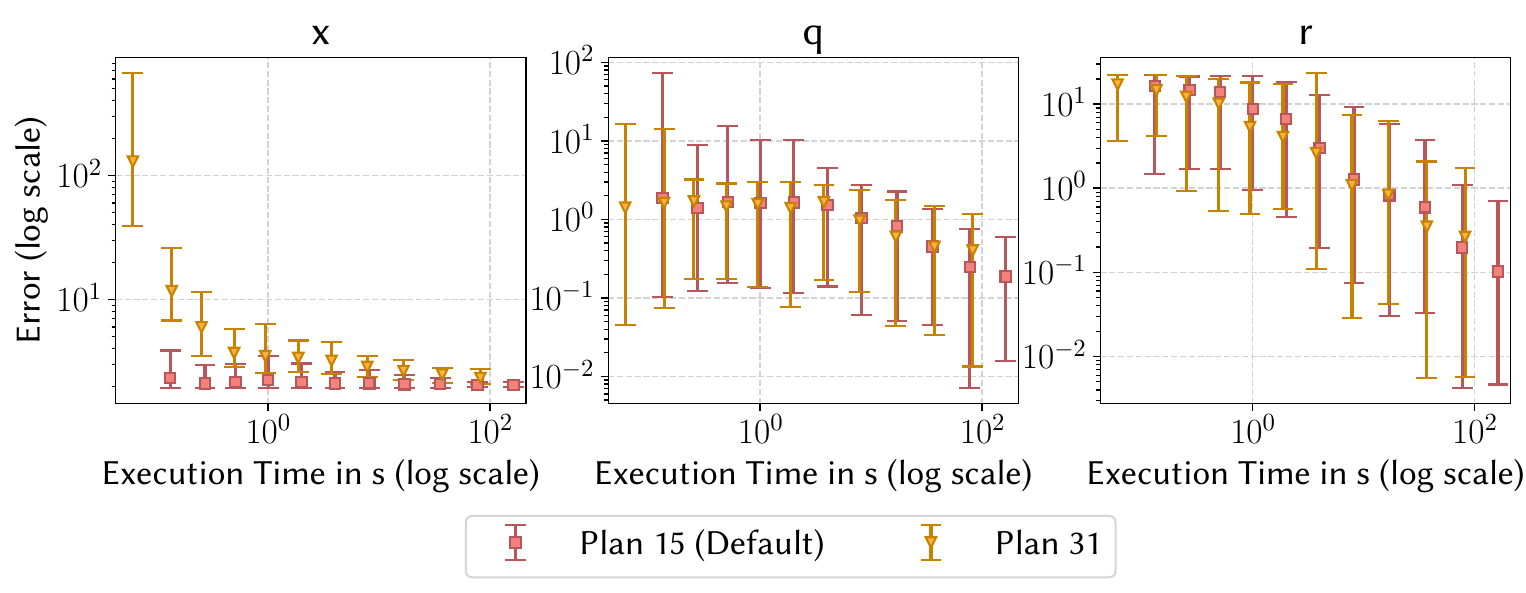}
    \caption{SMC w/ BP.}
  \end{subfigure}%
  \caption{\bRadar{}}
  \label{fig:performance-results-errorbar-radar}
\end{figure}

\begin{figure}[H]
  \centering
  \begin{subfigure}[c]{0.75\textwidth}
    \centering
    \includegraphics[width=1\textwidth]{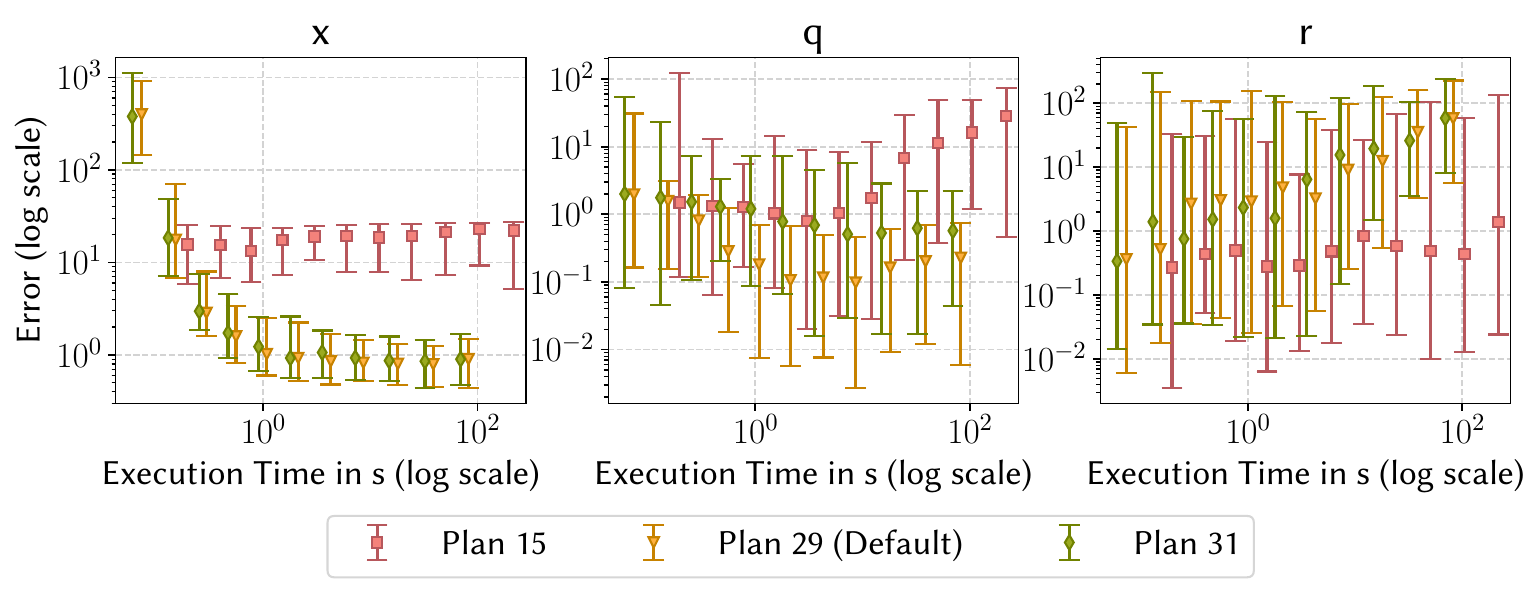}
    \caption{SSI.}
  \end{subfigure}%
  \\
  \begin{subfigure}[c]{0.75\textwidth}
    \centering
    \includegraphics[width=1\textwidth]{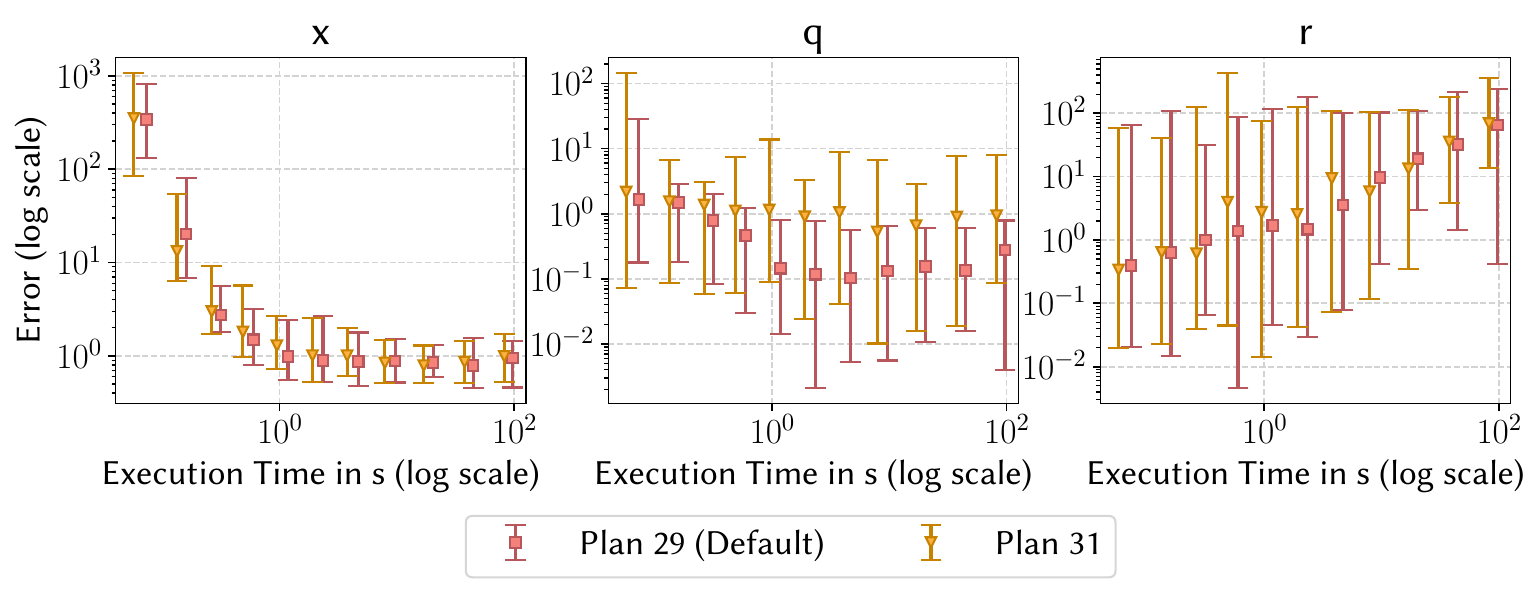}
    \caption{DS.}
  \end{subfigure}%
  \\
  \begin{subfigure}[c]{0.75\textwidth}
    \centering
    \includegraphics[width=1\textwidth]{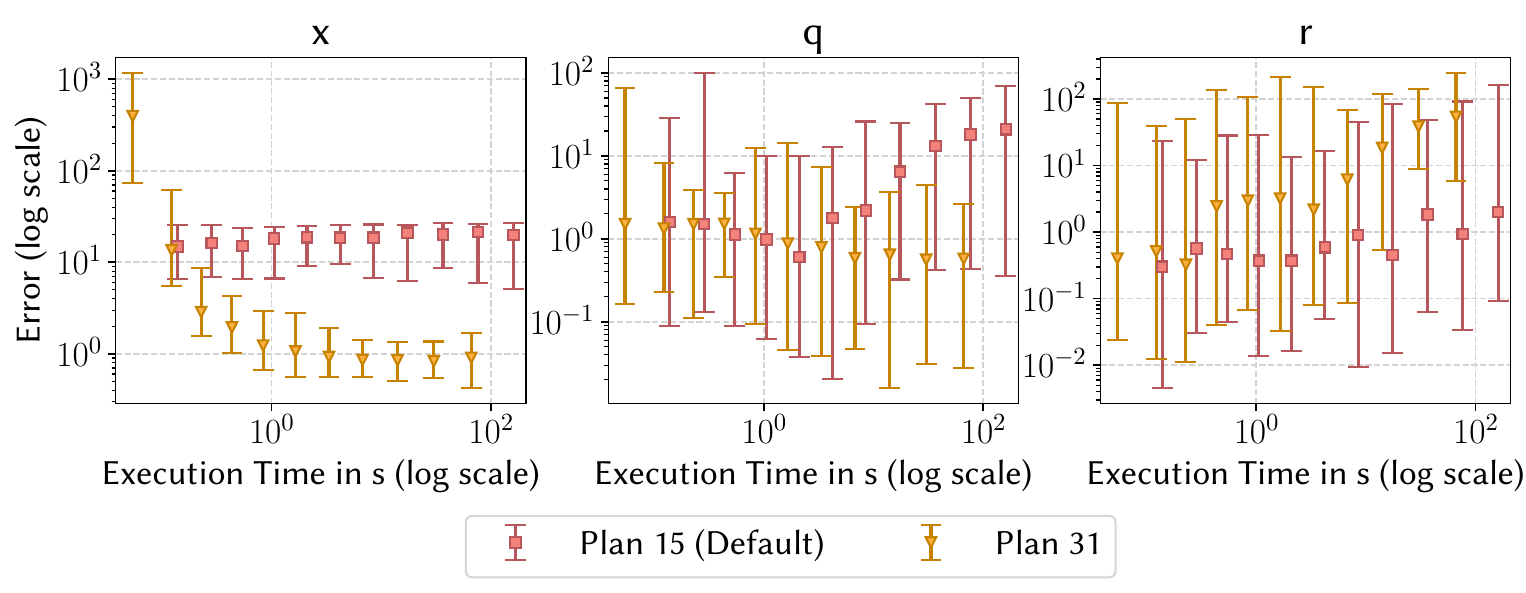}
    \caption{SMC w/ BP.}
  \end{subfigure}%
  \caption{\bEnvnoise{}}
  \label{fig:performance-results-errorbar-envnoise}
\end{figure}

\begin{figure}[H]
  \centering
  \begin{subfigure}[c]{0.5\textwidth}
    \centering
    \includegraphics[width=1\textwidth]{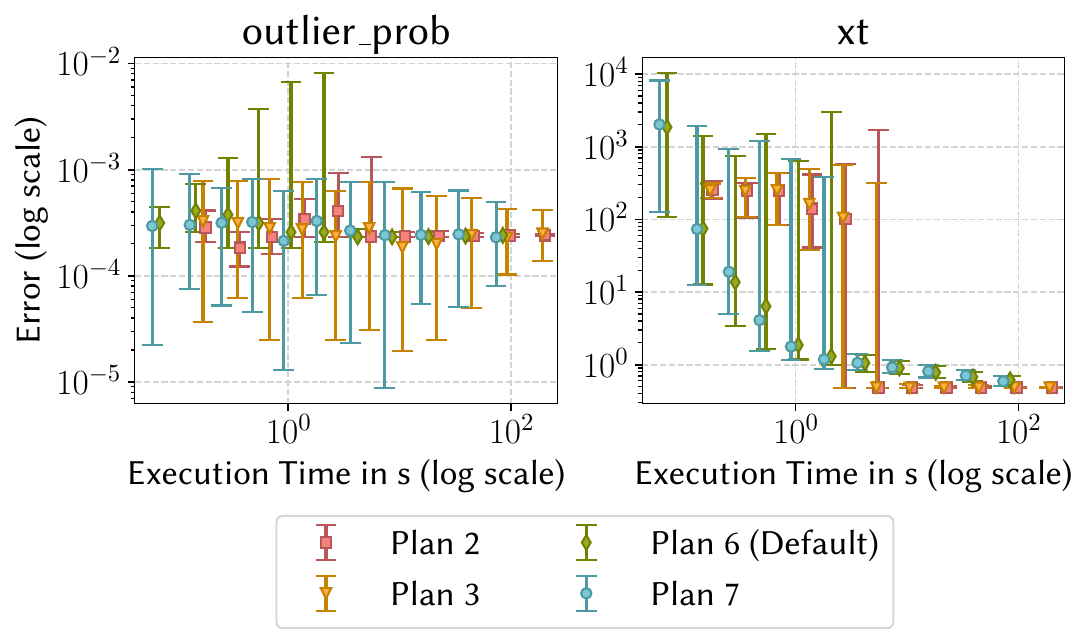}
    \caption{SSI.}
  \end{subfigure}%
  \\
  \begin{subfigure}[c]{0.5\textwidth}
    \centering
    \includegraphics[width=1\textwidth]{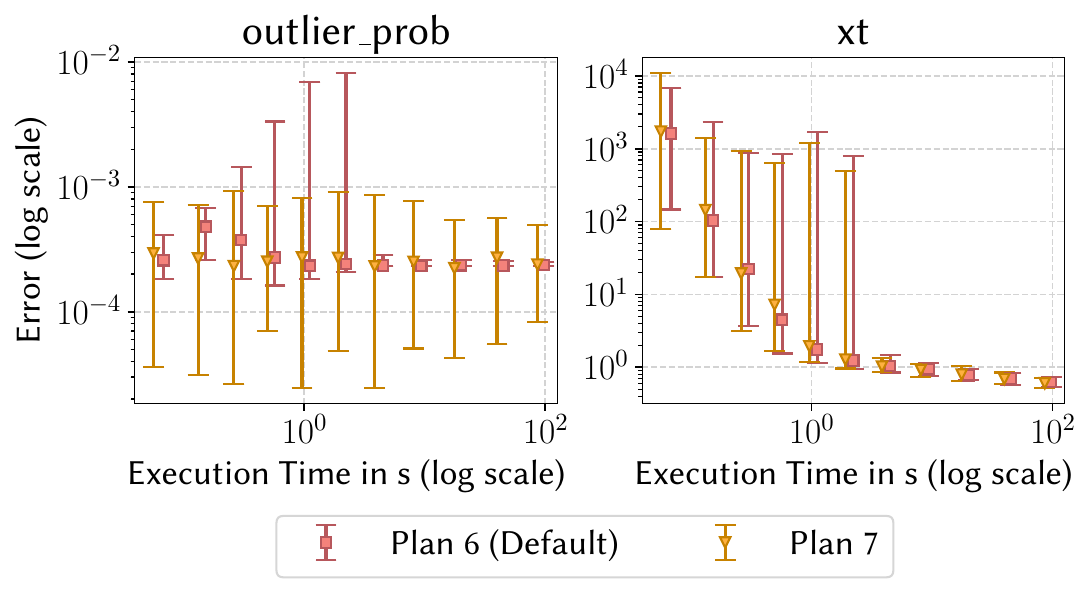}
    \caption{DS.}
  \end{subfigure}%
  \\
  \begin{subfigure}[c]{0.5\textwidth}
    \centering
    \includegraphics[width=0.99\textwidth]{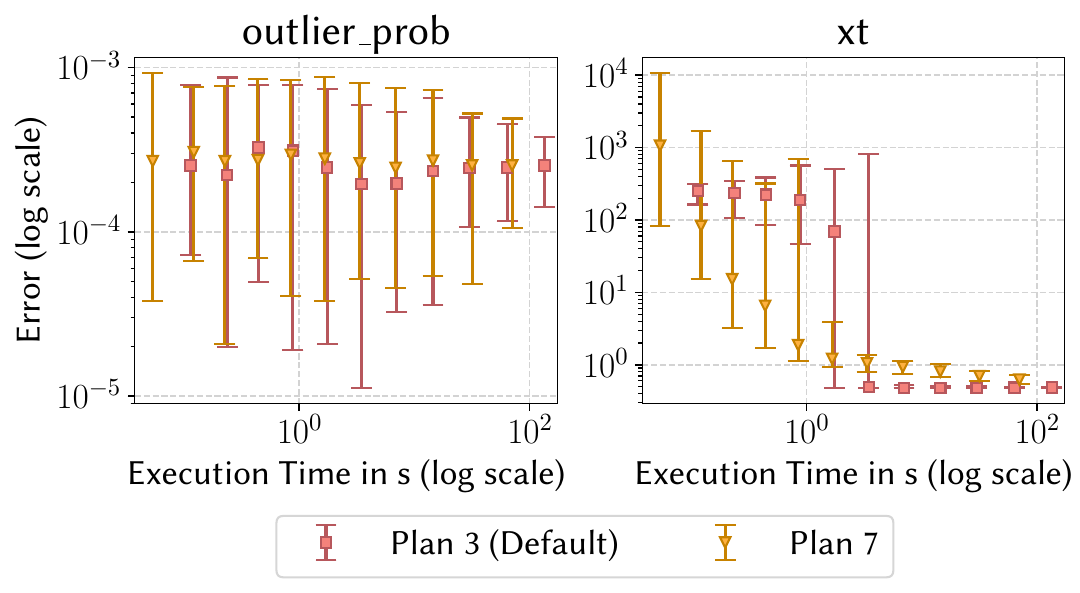}
    \caption{SMC w/ BP.}
  \end{subfigure}%
  \caption{\bOutlier{}}
  \label{fig:performance-results-errorbar-outlier}
\end{figure}

\begin{figure}[H]
  \centering
  \begin{subfigure}[c]{0.5\textwidth}
    \centering
    \includegraphics[width=1\textwidth]{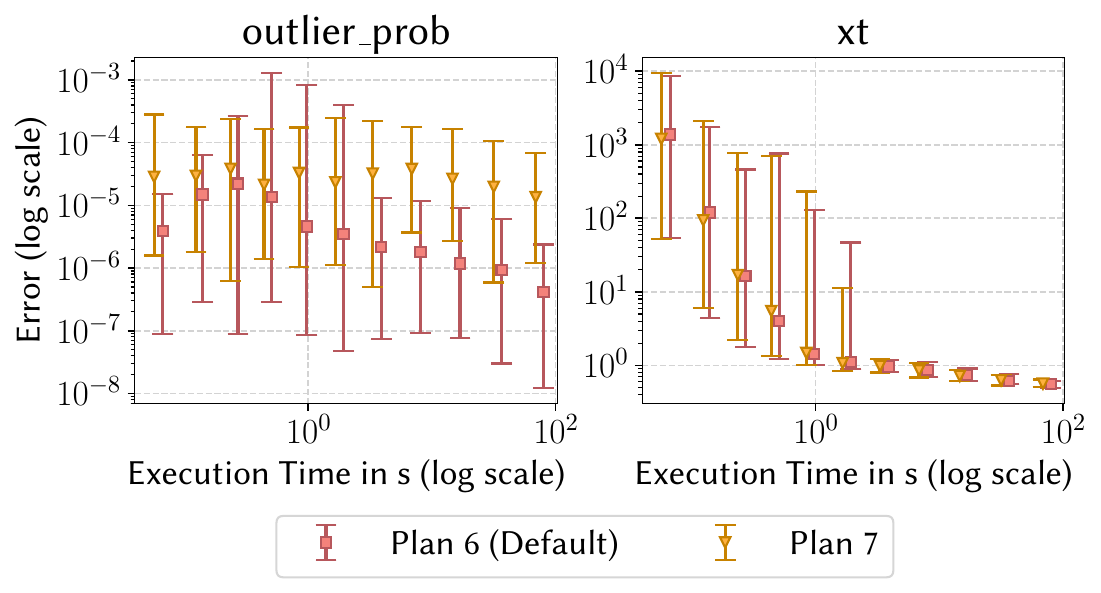}
    \caption{SSI.}
  \end{subfigure}%
  \\
  \begin{subfigure}[c]{0.5\textwidth}
    \centering
    \includegraphics[width=1\textwidth]{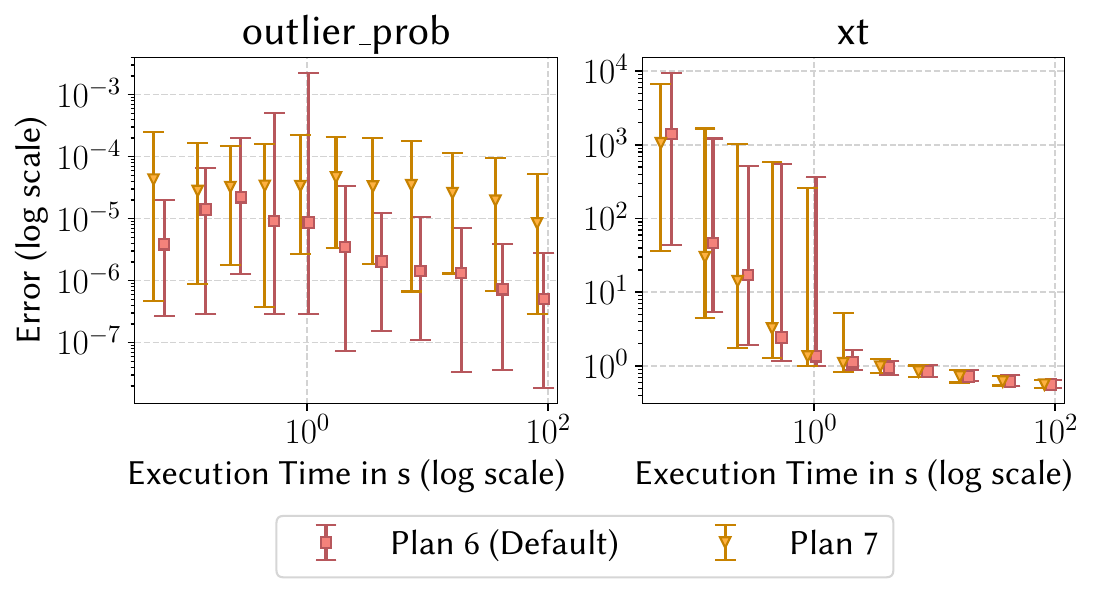}
    \caption{DS.}
  \end{subfigure}%
  \\
  \begin{subfigure}[c]{0.5\textwidth}
    \centering
    \includegraphics[width=1\textwidth]{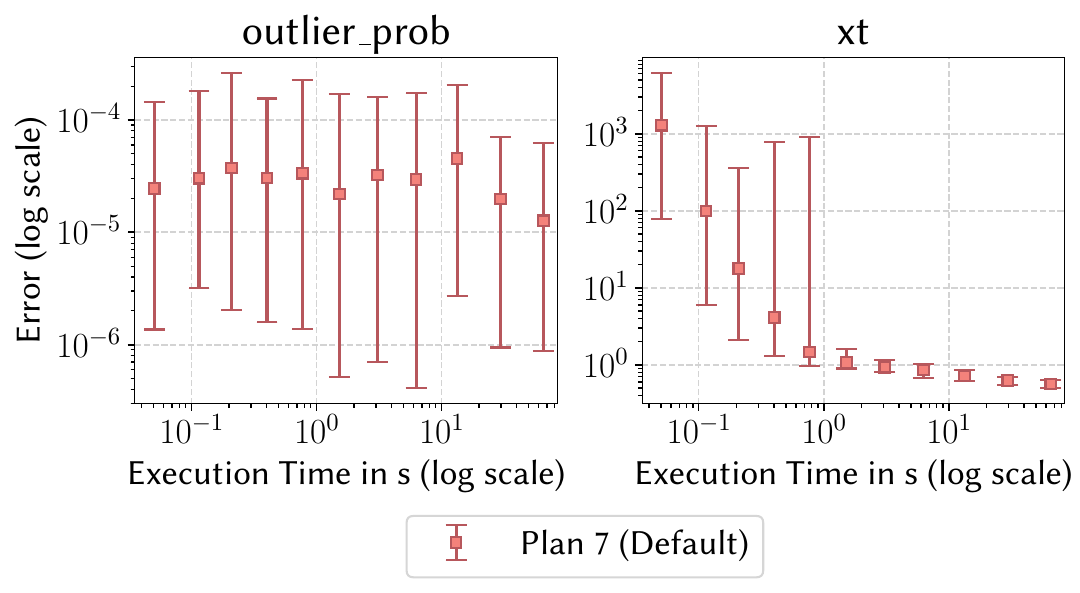}
    \caption{SMC w/ BP.}
  \end{subfigure}%
  \caption{\bOutlierheavy{}.}
  \label{fig:performance-results-errorbar-outlierheavy}
\end{figure}

\begin{figure}[H]
  \centering
  \begin{subfigure}[c]{0.5\textwidth}
    \centering
    \includegraphics[width=1\textwidth]{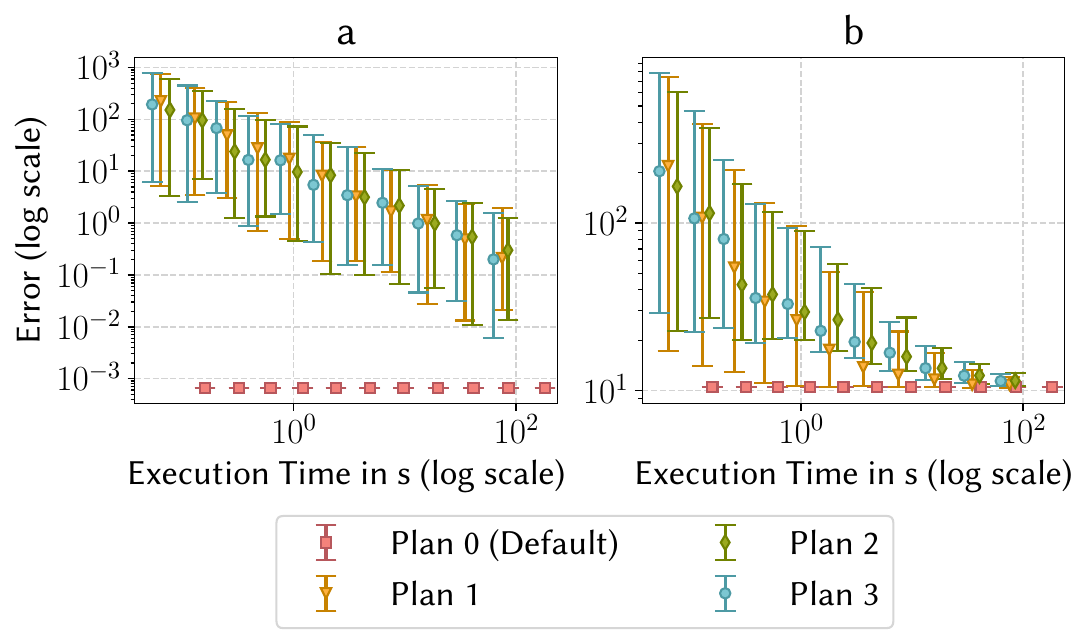}
    \caption{SSI.}
  \end{subfigure}%
  \\
  \begin{subfigure}[c]{0.5\textwidth}
    \centering
    \includegraphics[width=1\textwidth]{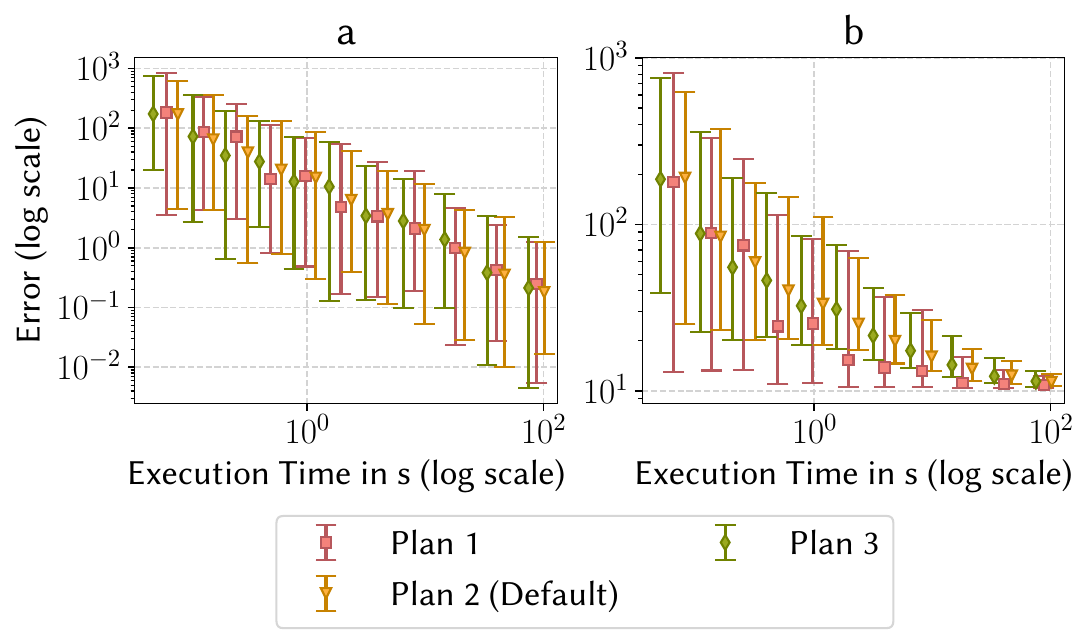}
    \caption{DS.}
  \end{subfigure}%
  \\
  \begin{subfigure}[c]{0.5\textwidth}
    \centering
    \includegraphics[width=1\textwidth]{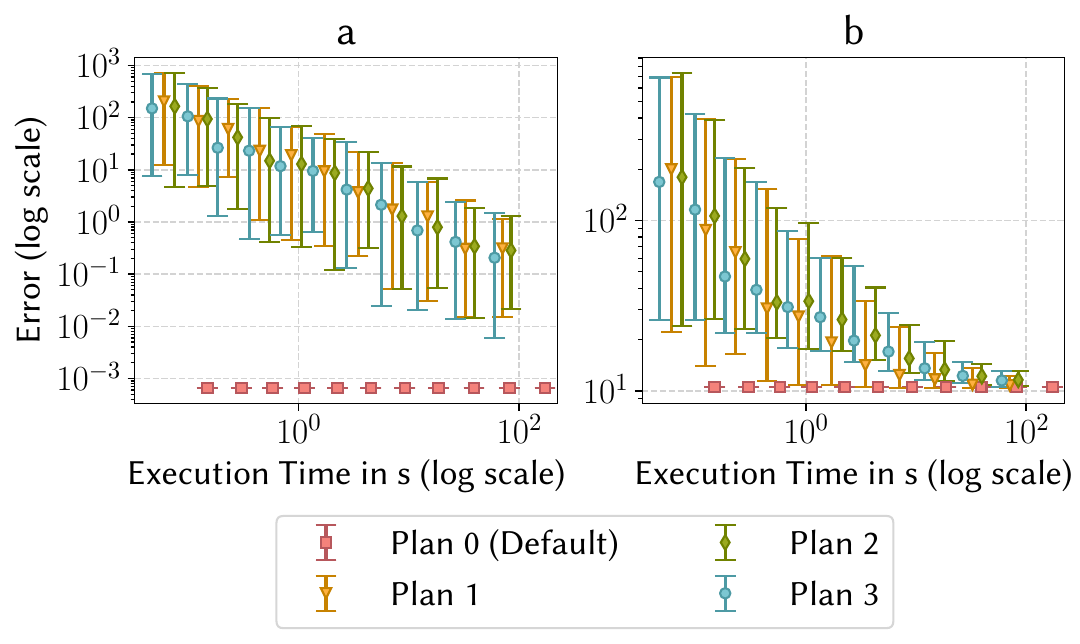}
    \caption{SMC w/ BP.}
  \end{subfigure}%
  \caption{\bGtree{}}
  \label{fig:performance-results-errorbar-tree}
\end{figure}

\begin{figure}[H]
  \centering
  \begin{subfigure}[c]{0.75\textwidth}
    \centering
    \includegraphics[width=1\textwidth]{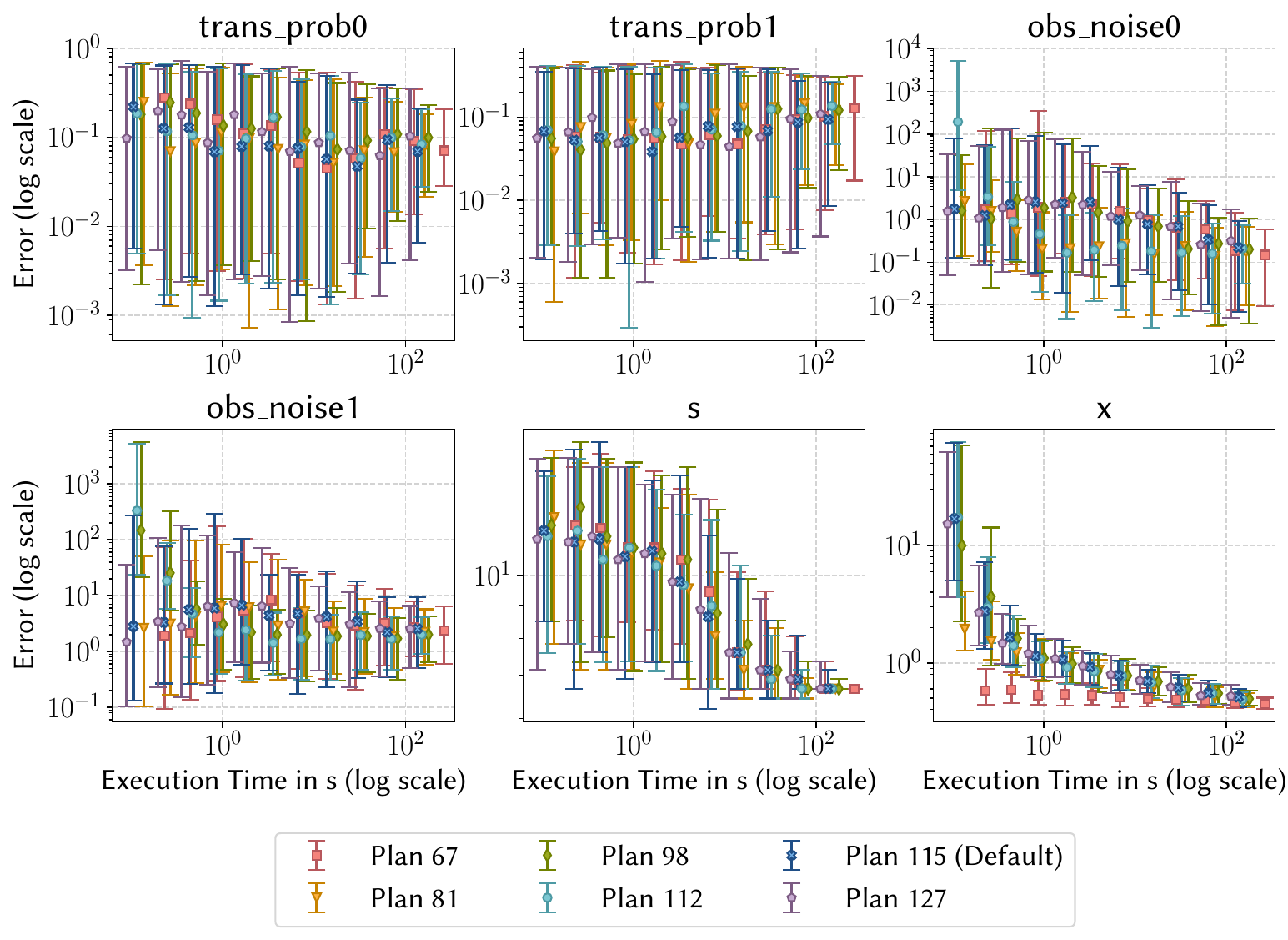}
    \caption{SSI.}
  \end{subfigure}%
  \\
  \begin{subfigure}[c]{0.75\textwidth}
    \centering
    \includegraphics[width=1\textwidth]{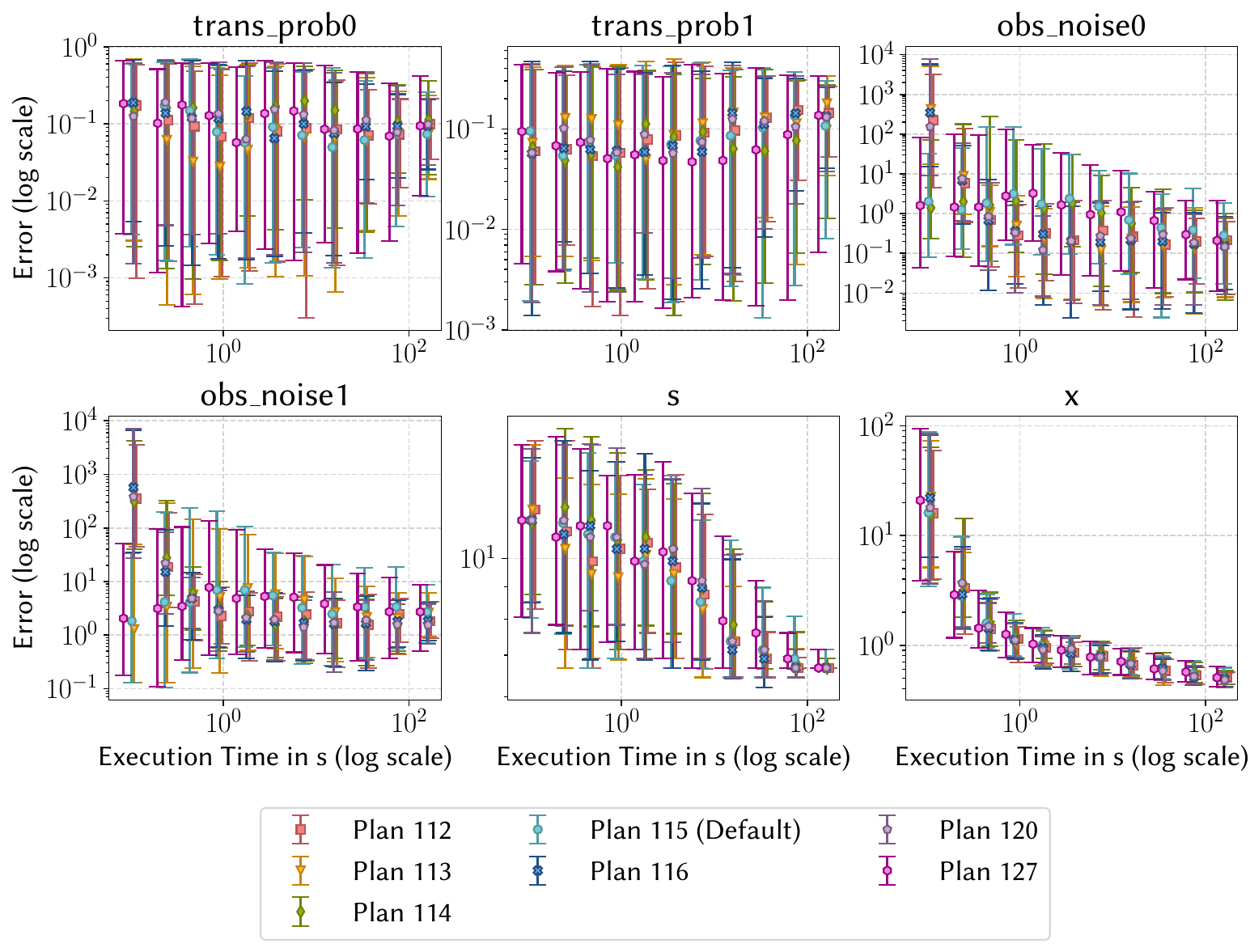}
    \caption{DS.}
  \end{subfigure}%
  \caption{\bSlds{}}
  \label{fig:performance-results-errorbar-slds-1}
\end{figure}

\begin{figure}[H]
  \centering
  \begin{subfigure}[c]{0.75\textwidth}
    \centering
    \includegraphics[width=1\textwidth]{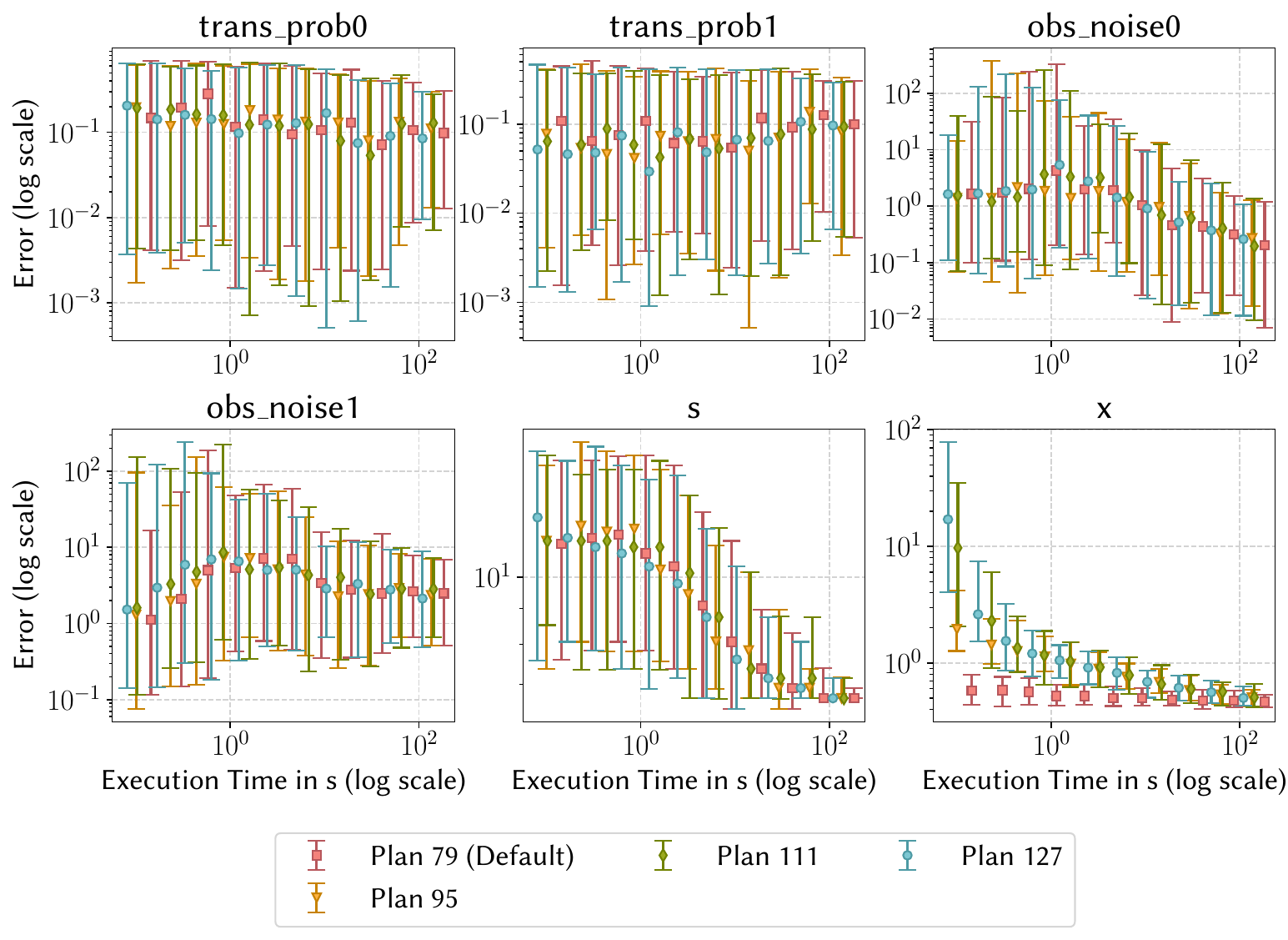}
    \caption{SMC w/ BP.}
  \end{subfigure}%
  \caption{\bSlds{} (continued)}
  \label{fig:performance-results-errorbar-slds-2}
\end{figure}

\begin{figure}[H]
  \centering
  \begin{subfigure}[c]{1\textwidth}
    \centering
    \includegraphics[width=1\textwidth]{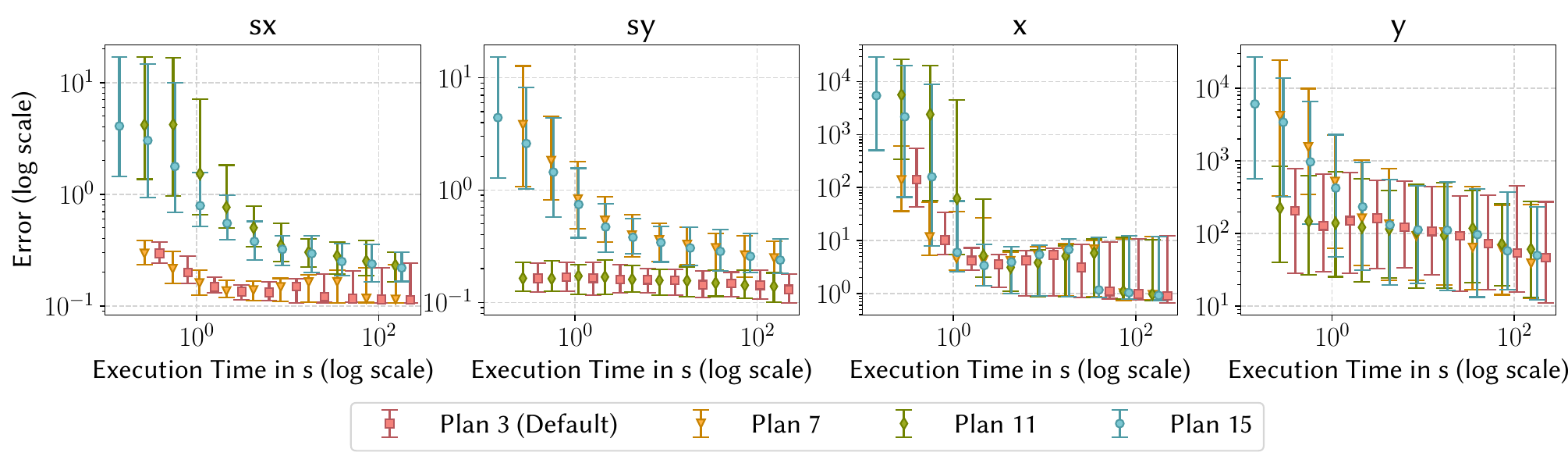}
    \caption{SSI.}
  \end{subfigure}%
  \\
  \begin{subfigure}[c]{1\textwidth}
    \centering
    \includegraphics[width=1\textwidth]{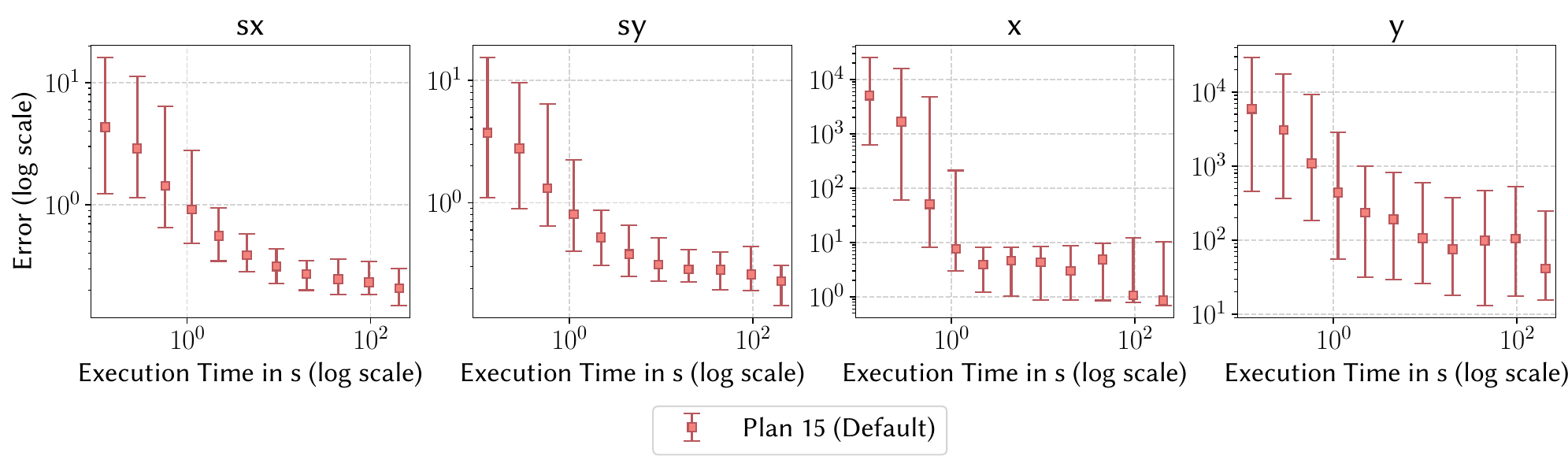}
    \caption{SSI.}
  \end{subfigure}%
  \\
  \begin{subfigure}[c]{1\textwidth}
    \centering
    \includegraphics[width=1\textwidth]{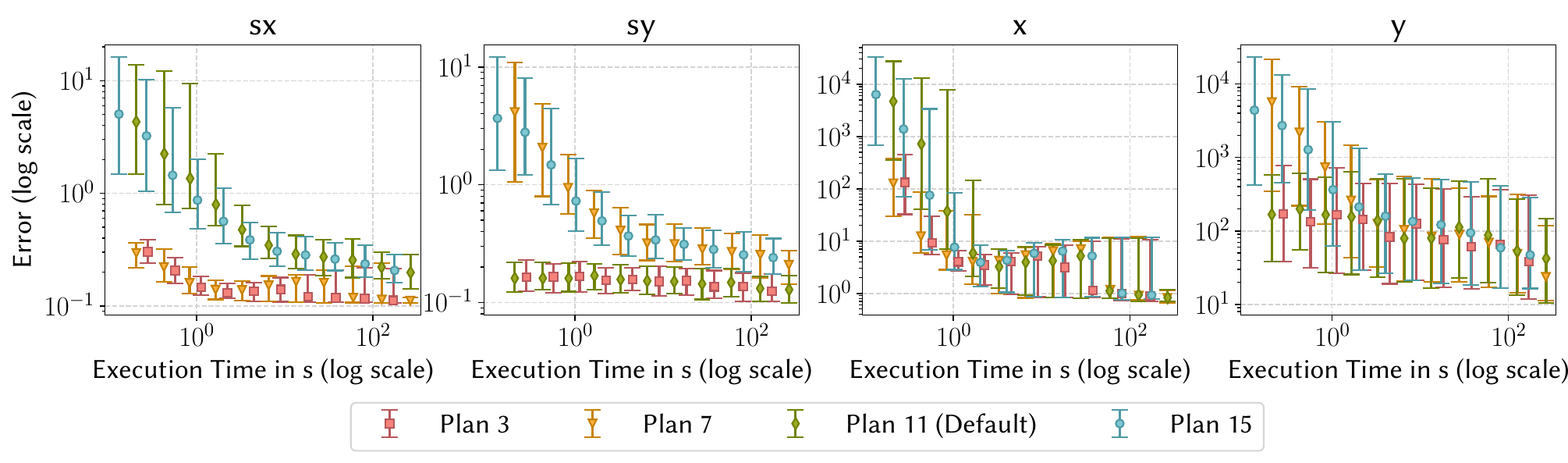}
    \caption{SMC w/ BP.}
  \end{subfigure}%
  \caption{\bRunner{}.}
  \label{fig:performance-results-errorbar-runner}
\end{figure}

\begin{figure}[H]
  \centering
  \begin{subfigure}[c]{1\textwidth}
    \centering
    \includegraphics[width=0.5\textwidth]{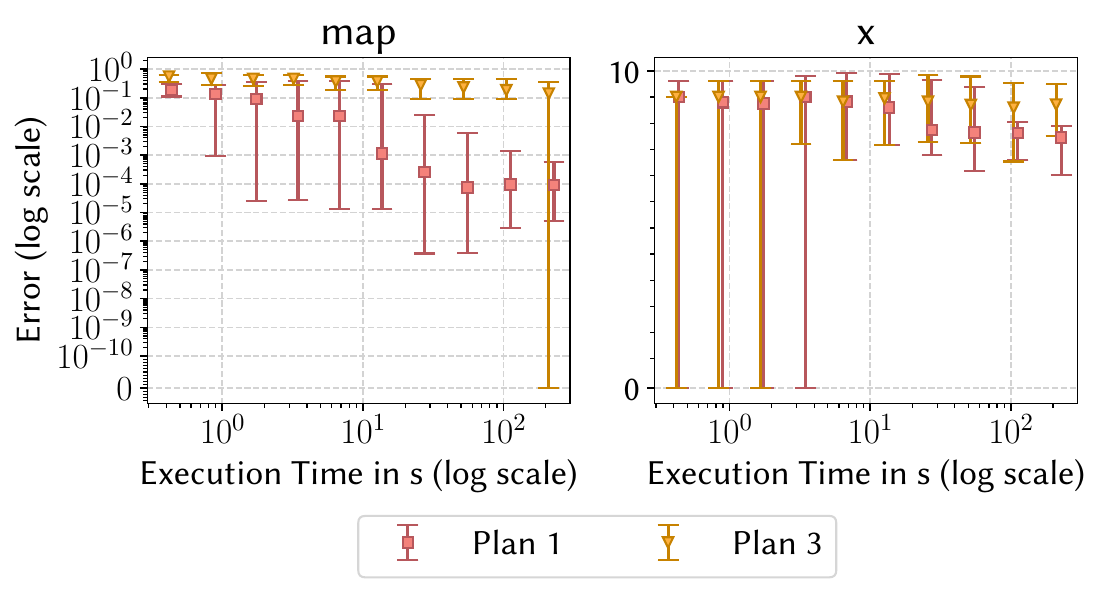}
    \caption{SSI. Plan 0 (Default) and Plan 2 time out for all particles.}
  \end{subfigure}%
  \\
  \begin{subfigure}[c]{1\textwidth}
    \centering
    \includegraphics[width=0.5\textwidth]{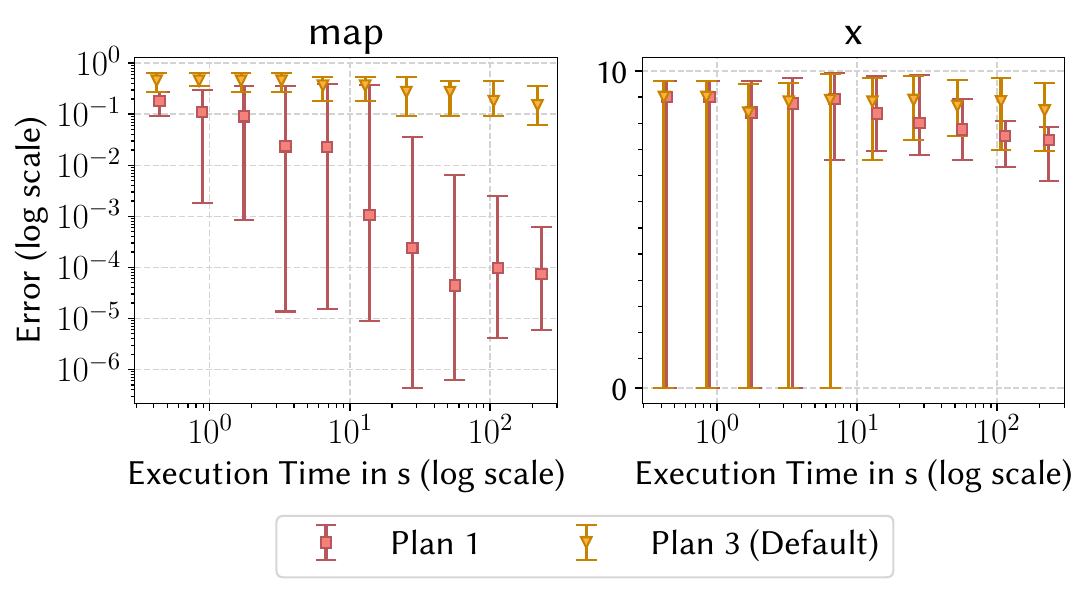}
    \caption{DS.}
  \end{subfigure}%
  \\
  \begin{subfigure}[c]{1\textwidth}
    \centering
    \includegraphics[width=0.5\textwidth]{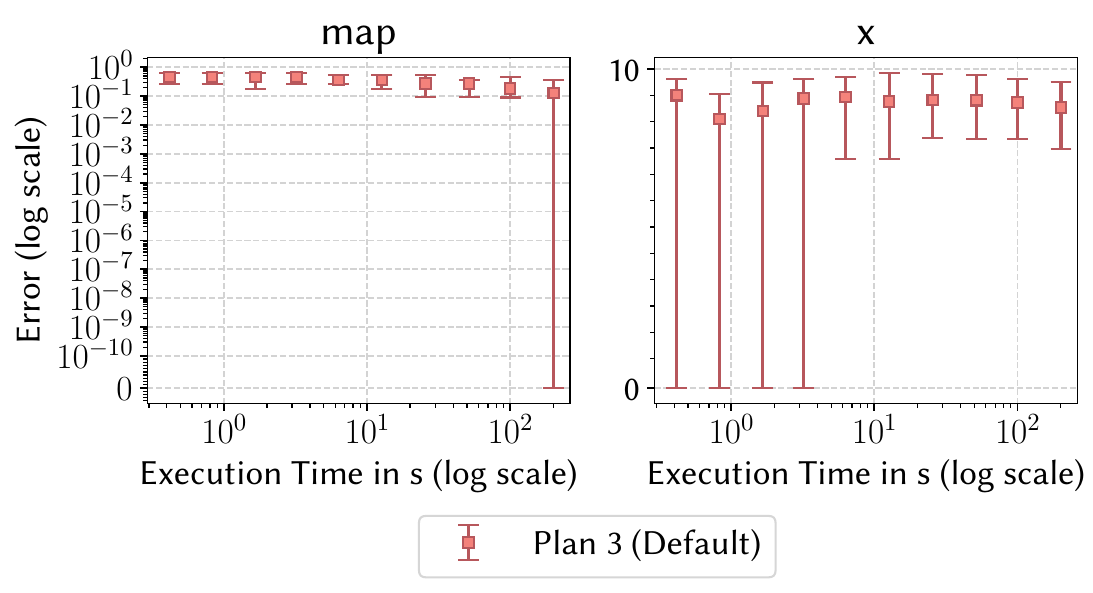}
    \caption{SMC w/ BP.}
  \end{subfigure}%
  \caption{\bSlam{}.}
  \label{fig:performance-results-errorbar-slam}
\end{figure}

\begin{figure}[H]
  \centering
  \begin{subfigure}[c]{1\textwidth}
    \centering
    \includegraphics[width=0.5\textwidth]{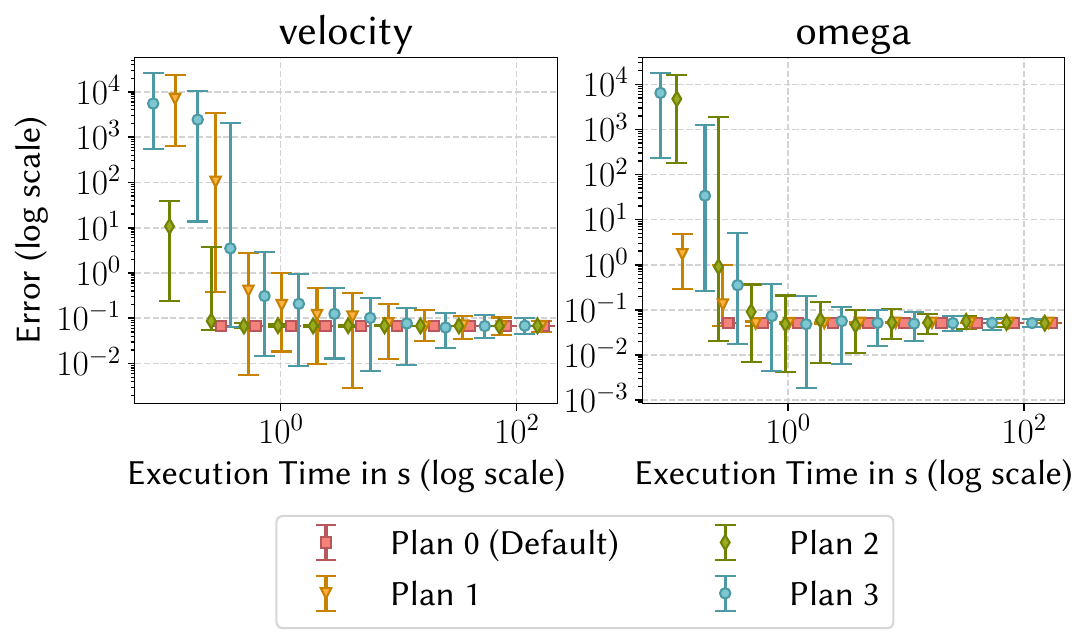}
    \caption{SSI. }
  \end{subfigure}%
  \\
  \begin{subfigure}[c]{1\textwidth}
    \centering
    \includegraphics[width=0.5\textwidth]{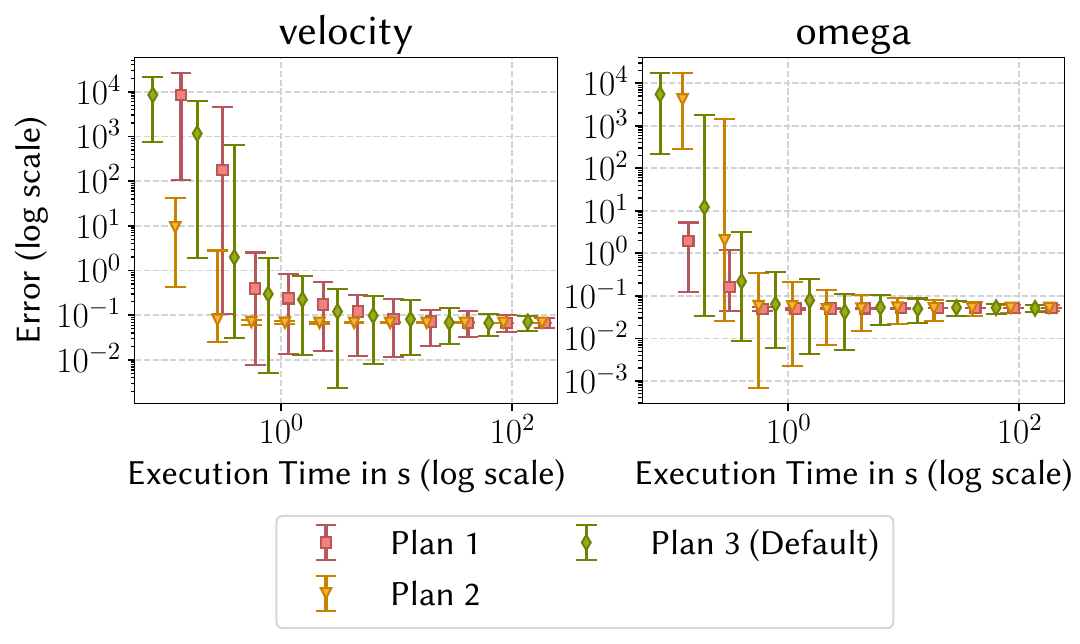}
    \caption{DS.}
  \end{subfigure}%
  \\
  \begin{subfigure}[c]{1\textwidth}
    \centering
    \includegraphics[width=0.5\textwidth]{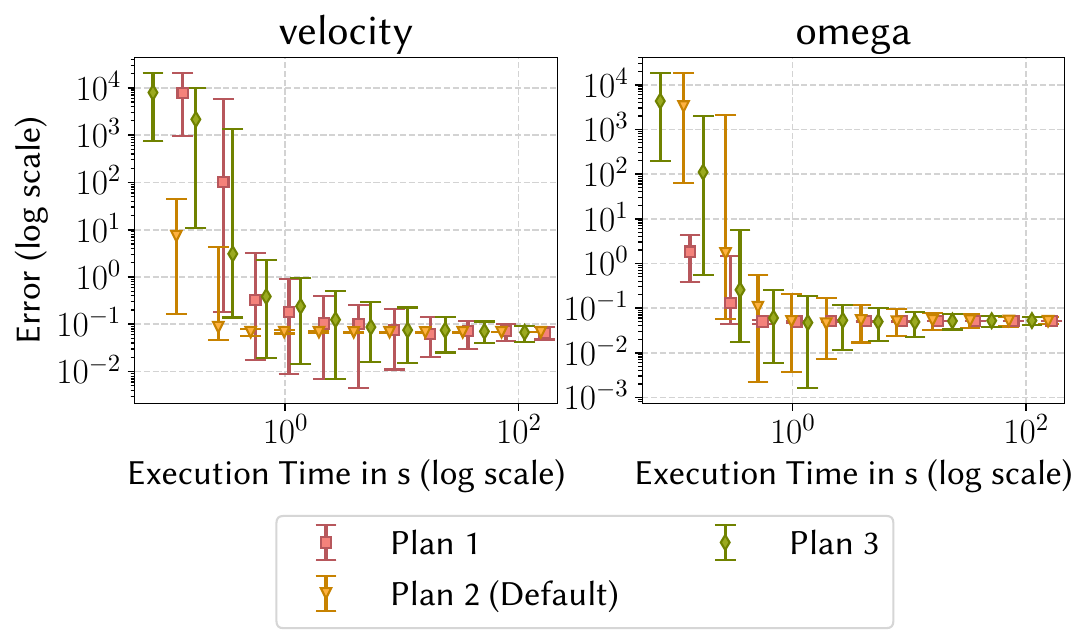}
    \caption{SMC w/ BP.}
  \end{subfigure}%
  \caption{\bWheels{}.}
  \label{fig:performance-results-errorbar-wheels}
\end{figure}

\begin{figure}[H]
  \centering
  \begin{subfigure}[c]{1\textwidth}
    \centering
    \includegraphics[width=1\textwidth]{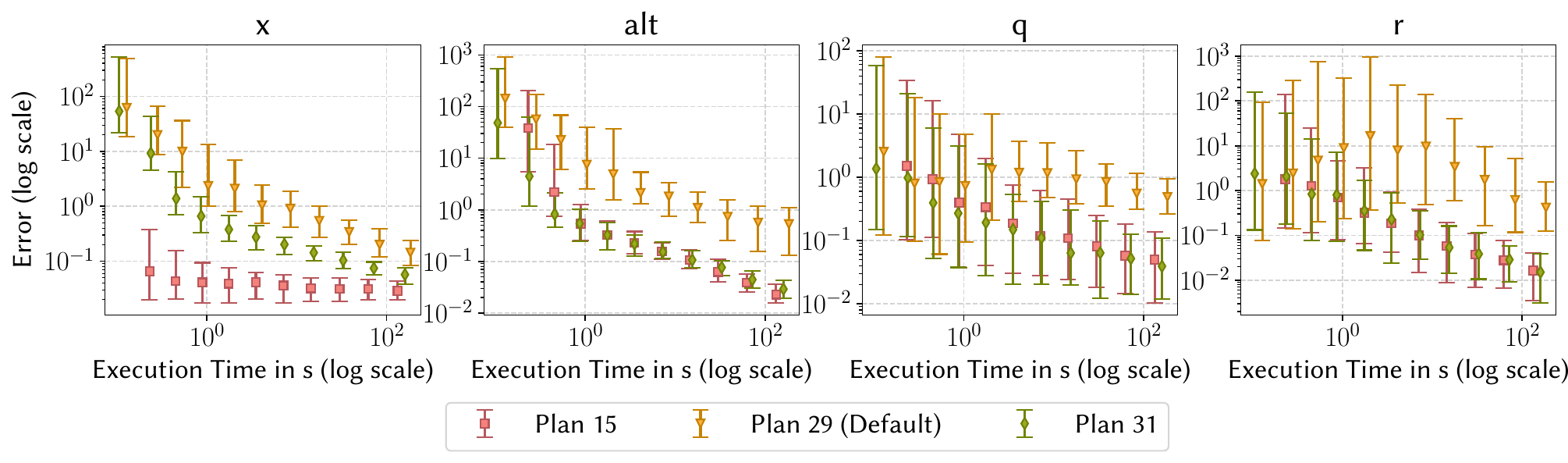}
    \caption{SSI.}
  \end{subfigure}%
  \\
  \begin{subfigure}[c]{1\textwidth}
    \centering
    \includegraphics[width=1\textwidth]{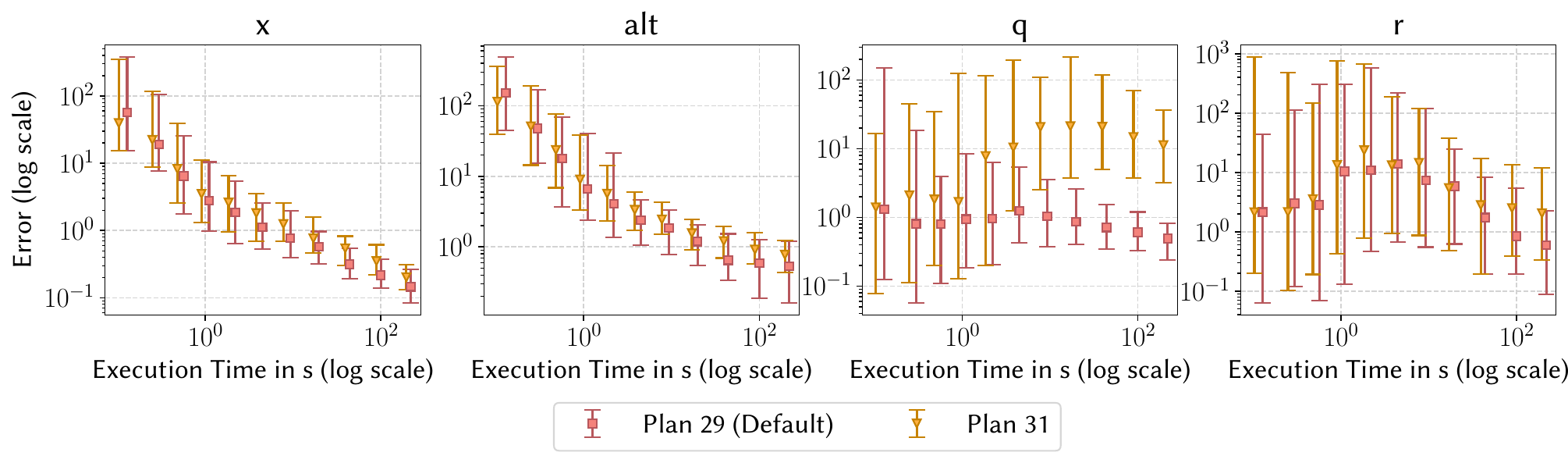}
    \caption{DS.}
  \end{subfigure}%
  \\
  \begin{subfigure}[c]{1\textwidth}
    \centering
    \includegraphics[width=1\textwidth]{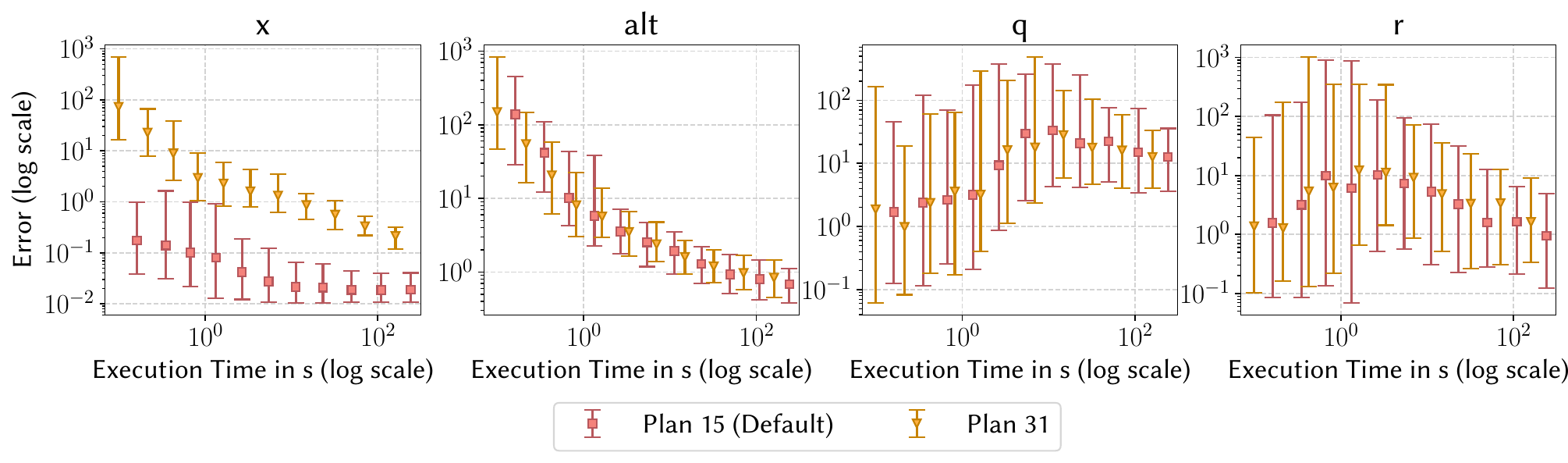}
    \caption{SMC w/ BP.}
  \end{subfigure}%
  \caption{\bAircraft{}. }
  \label{fig:performance-results-errorbar-aircraft}
\end{figure}

\subsection{Execution Time to Reach Target Accuracy}
\begin{table}[H]
  \small
  \centering
  \caption{Best execution time by any satisfiable inference plan and the best execution time by the default inference plan for 90\% of the 100 executions to reach target accuracy, defined as the 90th percentile of error by the default inference plan using greatest number of particles evaluated that did not produce timeouts. Following \citet{atkinson2022semi} and \citet{baudart2020reactive}, reaching target accuracy is defined as $\log(P_{90\%}(\mit{loss})) - \log(\mit{loss}_{\mit{target}}) < 0.5$. -- indicates the default plan timeouts for all particle numbers. \textbf{Bold} marks when an alternative plan performs better than the default plan.}
  \begin{tabular}[h]{llrrrrrr}
    \toprule
    & & \multicolumn{2}{c}{\ssi} & \multicolumn{2}{c}{\ds} & \multicolumn{2}{c}{\bp} \\
     \cmidrule(lr){3-4} \cmidrule(lr){5-6} \cmidrule(lr){7-8}
    Benchmark & Variable & Best (s) & Default (s) & Best (s) & Default (s) & Best (s) & Default (s) \\
    \midrule
\multirow{5}{*}{\bNoise{}}& \mkwm{x} & \textbf{0.19} & 3.62 & \textbf{3.08} & 3.94 & 0.13 & 0.13 \\[-0.4em]
& & \plan{3}& & \plan{7}& & \defaultplan{} \\
& \mkwm{q} & 7.37 & 7.37 & \textbf{9.68} & 16.9 & 69.44 & 69.44 \\[-0.4em]
& & \defaultplan{} & & \plan{4}& & \defaultplan{} \\
& \mkwm{r} & \textbf{11.35} & 15.15 & \textbf{20.2} & 36.95 & 69.44 & 69.44 \\[-0.4em]
& & \plan{3}& & \plan{4}& & \defaultplan{} \\
\midrule
\multirow{5}{*}{\bRadar{}}& \mkwm{x} & \textbf{0.19} & 0.63 & \textbf{0.57} & 0.67 & 0.13 & 0.13 \\[-0.4em]
& & \plan{15}& & \plan{31}& & \defaultplan{} \\
& \mkwm{q} & 9.96 & 9.96 & 2.58 & 2.58 & \textbf{16.66} & 35.71 \\[-0.4em]
& & \defaultplan{} & & \defaultplan{} & & \plan{31}\\
& \mkwm{r} & 46.57 & 46.57 & \textbf{46.73} & 53.95 & \textbf{37.16} & 77.43 \\[-0.4em]
& & \defaultplan{} & & \plan{31}& & \plan{31}\\
\midrule
\multirow{5}{*}{\bEnvnoise{}}& \mkwm{x} & \textbf{0.47} & 0.56 & 0.6 & 0.6 & \textbf{0.12} & 0.14 \\[-0.4em]
& & \plan{31}& & \defaultplan{} & & \plan{31}\\
& \mkwm{q} & 0.3 & 0.3 & 0.32 & 0.32 & \textbf{0.05} & 0.14 \\[-0.4em]
& & \defaultplan{} & & \defaultplan{} & & \plan{31}\\
& \mkwm{r} & \textbf{0.06} & 0.07 & \textbf{0.06} & 0.07 & \textbf{0.05} & 0.14 \\[-0.4em]
& & \plan{31}& & \plan{31}& & \plan{31}\\
\midrule
\multirow{3}{*}{\bOutlier{}}& \mkwm{outlier\_prob} & 0.07 & 0.07 & \textbf{0.06} & 0.07 & \textbf{0.05} & 0.11 \\[-0.4em]
& & \defaultplan{} & & \plan{7}& & \plan{7}\\
& \mkwm{xt} & \textbf{3.6} & 4.19 & \textbf{3.85} & 4.51 & \textbf{3.36} & 7.02 \\[-0.4em]
& & \plan{7}& & \plan{7}& & \plan{7}\\
\midrule
\multirow{3}{*}{\bOutlierheavy{}}& \mkwm{outlier\_prob} & 36.62 & 36.62 & 19.04 & 19.04 & 0.05 & 0.05 \\[-0.4em]
& & \defaultplan{} & & \defaultplan{} & & \defaultplan{} \\
& \mkwm{xt} & \textbf{3.32} & 3.88 & 2.08 & 2.08 & 1.53 & 1.53 \\[-0.4em]
& & \plan{7}& & \defaultplan{} & & \defaultplan{} \\
\midrule
\multirow{3}{*}{\bGtree{}}& \mkwm{a} & 0.16 & 0.16 & \textbf{33.41} & 46.93 & 0.15 & 0.15 \\[-0.4em]
& & \defaultplan{} & & \plan{3}& & \defaultplan{} \\
& \mkwm{b} & 0.16 & 0.16 & \textbf{3.94} & 4.83 & 0.15 & 0.15 \\[-0.4em]
& & \defaultplan{} & & \plan{1}& & \defaultplan{} \\
\midrule
\multirow{10}{*}{\bSlds{}}& \mkwm{trans\_prob0} & \textbf{0.09} & 0.23 & \textbf{0.08} & 0.1 & \textbf{0.08} & 0.14 \\[-0.4em]
& & \plan{127}& & \plan{127}& & \plan{127}\\
& \mkwm{trans\_prob1} & \textbf{0.09} & 0.11 & \textbf{0.08} & 0.1 & \textbf{0.08} & 0.14 \\[-0.4em]
& & \plan{127}& & \plan{127}& & \plan{127}\\
& \mkwm{obs\_noise0} & \textbf{0.51} & 63.32 & \textbf{0.47} & 33.25 & \textbf{22.51} & 40.53 \\[-0.4em]
& & \plan{81}& & \plan{113}& & \plan{127}\\
& \mkwm{obs\_noise1} & \textbf{0.47} & 3.23 & \textbf{0.46} & 14.94 & 0.14 & 0.14 \\[-0.4em]
& & \plan{112}& & \plan{120}& & \defaultplan{} \\
& \mkwm{s} & \textbf{0.09} & 0.11 & \textbf{0.08} & 0.1 & \textbf{0.08} & 0.14 \\[-0.4em]
& & \plan{127}& & \plan{127}& & \plan{127}\\
& \mkwm{x} & \textbf{0.23} & 0.81 & 0.85 & 0.85 & 0.14 & 0.14 \\[-0.4em]
& & \plan{67}& & \defaultplan{} & & \defaultplan{} \\
  \bottomrule
  \end{tabular}
  \label{tab:comparison-results-full-1}
\end{table}

\begin{table}[H]
  \small
  \centering
  \caption{Best execution time by any satisfiable inference plan and the best execution time by the default inference plan for 90\% of the 100 executions to reach target accuracy, defined as the 90th percentile of error by the default inference plan using greatest number of particles evaluated that did not produce timeouts. Following \citet{atkinson2022semi} and \citet{baudart2020reactive}, reaching target accuracy is defined as $\log(P_{90\%}(\mit{loss})) - \log(\mit{loss}_{\mit{target}}) < 0.5$. -- indicates the default plan timeouts for all particle counts. \textbf{Bold} marks when an alternative plan performs better than the default plan.}
  \begin{tabular}[h]{llrrrrrr}
    \toprule
    & & \multicolumn{2}{c}{\ssi} & \multicolumn{2}{c}{\ds} & \multicolumn{2}{c}{\bp} \\
     \cmidrule(lr){3-4} \cmidrule(lr){5-6} \cmidrule(lr){7-8}
    Benchmark & Variable & Best (s) & Default (s) & Best (s) & Default (s) & Best (s) & Default (s) \\
    \midrule
\multirow{7}{*}{\bRunner{}}& \mkwm{sx} & \textbf{0.27} & 0.39 & 2.22 & 2.22 & \textbf{0.21} & 3.29 \\[-0.4em]
& & \plan{7}& & \defaultplan{} & & \plan{7}\\
& \mkwm{sy} & \textbf{0.27} & 0.39 & 2.22 & 2.22 & 0.21 & 0.21 \\[-0.4em]
& & \plan{11}& & \defaultplan{} & & \defaultplan{} \\
& \mkwm{x} & 0.81 & 0.81 & 2.22 & 2.22 & \textbf{265.94} & 267.9 \\[-0.4em]
& & \defaultplan{} & & \defaultplan{} & & \plan{7}\\
& \mkwm{y} & \textbf{0.27} & 0.39 & 9.39 & 9.39 & \textbf{2.25} & 13.35 \\[-0.4em]
& & \plan{11}& & \defaultplan{} & & \plan{3}\\
\midrule
\multirow{3}{*}{\bWheels{}}& \mkwm{velocity} & 0.31 & 0.31 & \textbf{0.55} & 6.28 & 0.51 & 0.51 \\[-0.4em]
& & \defaultplan{} & & \plan{2}& & \defaultplan{} \\
& \mkwm{omega} & 0.31 & 0.31 & \textbf{0.6} & 3.09 & \textbf{0.56} & 1.95 \\[-0.4em]
& & \defaultplan{} & & \plan{1}& & \plan{1}\\
\midrule
\multirow{3}{*}{\bSlam{}}& \mkwm{map} & -- & -- & 0.42 & 0.42 & 0.42 & 0.42 \\[-0.4em]
& & & & \defaultplan{} & & \defaultplan{} \\
& \mkwm{x} & -- & -- & 0.42 & 0.42 & 0.42 & 0.42 \\[-0.4em]
& & & & \defaultplan{} & & \defaultplan{} \\
\midrule
\multirow{5}{*}{\bAircraft{}}& \mkwm{x} & \textbf{0.23} & 38.31 & \textbf{39.34} & 44.06 & 5.43 & 5.43 \\[-0.4em]
& & \plan{15}& & \plan{31}& & \defaultplan{} \\
& \mkwm{alt} & \textbf{0.46} & 8.5 & 9.32 & 9.32 & 11.21 & 11.21 \\[-0.4em]
& & \plan{31}& & \defaultplan{} & & \defaultplan{} \\
& \mkwm{q} & \textbf{1.75} & 17.8 & 19.74 & 19.74 & 0.16 & 0.16 \\[-0.4em]
& & \plan{15}& & \defaultplan{} & & \defaultplan{} \\
& \mkwm{r} & \textbf{0.89} & 183.78 & 99.8 & 99.8 & 49.67 & 49.67 \\[-0.4em]
& & \plan{15}& & \defaultplan{} & & \defaultplan{} \\
  \bottomrule
  \end{tabular}
  \label{tab:comparison-results-full-2}
\end{table}

\subsection{Best Accuracy Given Runtime}

\begin{table}[H]
  \small
  \centering
  \caption{Ratio of the best accuracy achieved by any satisfiable inference plan with median runtime less than or equal to the median runtime of the default inference plan, aggregated across particle counts using geometric mean.  Accuracy is measured as the 90th percentile error. -- indicates the default plan timeouts for all particle counts. }
  \begin{tabular}[h]{llrrrrrr}
    \toprule
    & & \multicolumn{3}{c}{Gmean of Best Accuracy Ratio} \\
     \cmidrule(lr){3-5}
    Benchmark & Variable & \ssi & \ds & \bp \\
\midrule
\multirow{3}{*}{\bNoise{}}& \mkwm{x} & 2.36x & 1.05x & 1.05x \\
& \mkwm{q} & 1.33x & 1.15x & 1.41x \\
& \mkwm{r} & 2.83x & 2.13x & 1.19x \\
\midrule
\multirow{3}{*}{\bRadar{}}& \mkwm{x} & 1.4x & 1.02x & 1.03x \\
& \mkwm{q} & 1.22x & 1.1x & 2.0x \\
& \mkwm{r} & 1.06x & 1.08x & 1.05x \\
\midrule
\multirow{3}{*}{\bEnvnoise{}}& \mkwm{x} & 1.06x & 1.08x & 9.28x \\
& \mkwm{q} & 1.12x & 1.06x & 7.59x \\
& \mkwm{r} & 6.53x & 3.16x & 2.79x \\
\midrule
\multirow{2}{*}{\bOutlier{}}& \mkwm{outlier\_prob} & 2.75x & 2.41x & 1.04x \\
& \mkwm{xt} & 1.89x & 1.25x & 3.02x \\
\midrule
\multirow{2}{*}{\bOutlierheavy{}}& \mkwm{outlier\_prob} & 4.33x & 2.96x & 1.21x \\
& \mkwm{xt} & 1.2x & 1.07x & 1.17x \\
\midrule
\multirow{2}{*}{\bGtree{}}& \mkwm{a} & 1.0x & 1.07x & 1.0x \\
& \mkwm{b} & 1.0x & 1.09x & 1.0x \\
\midrule
\multirow{6}{*}{\bSlds{}}& \mkwm{trans\_prob0} & 1.16x & 1.15x & 1.11x \\
& \mkwm{trans\_prob1} & 1.09x & 1.18x & 1.18x \\
& \mkwm{obs\_noise0} & 7.39x & 6.62x & 2.52x \\
& \mkwm{obs\_noise1} & 4.93x & 5.21x & 2.03x \\
& \mkwm{s} & 1.04x & 1.04x & 1.02x \\
& \mkwm{x} & 1.64x & 1.03x & 1.0x \\
\midrule
\multirow{4}{*}{\bRunner{}}& \mkwm{sx} & 1.16x & 1.0x & 4.72x \\
& \mkwm{sy} & 1.01x & 1.01x & 1.02x \\
& \mkwm{x} & 1.34x & 1.1x & 4.61x \\
& \mkwm{y} & 1.15x & 1.05x & 1.16x \\
\midrule
\multirow{2}{*}{\bWheels{}}& \mkwm{velocity} & 1.0x & 7.12x & 1.0x \\
& \mkwm{omega} & 1.0x & 3.24x & 2.77x \\
\midrule
\multirow{2}{*}{\bSlam{}}& \mkwm{map} & -- & 5.1x & 1.02x \\
& \mkwm{x} & -- & 2.15x & 3.12x \\
\midrule
\multirow{4}{*}{\bAircraft{}}& \mkwm{x} & 27.55x & 1.01x & 1.05x \\
& \mkwm{alt} & 15.46x & 1.07x & 1.05x \\
& \mkwm{q} & 4.19x & 1.41x & 6.15x \\
& \mkwm{r} & 65.0x & 2.47x & 2.71x \\
  \bottomrule
  \end{tabular}
  \label{tab:comparison-results-best-acc-full}
\end{table}

\section{Extending to Other Monte Carlo Methods}
\label{appendix:mh-siren}

We have formalized \siren{} and the inference plan analysis in terms of particle filtering because the existing hybrid inference algorithms that can be implemented with the hybrid inference interface and also support control flow all use particle filtering, to the best of our knowledge. However, the principles and applicability of inference plans can extend to other Monte Carlo-based hybrid inference algorithms that are compatible with the hybrid inference interface. We show this with a basic implementation of SSI, DS, and SMC with BP using the Metropolis Hasting (MH) algorithm as the approximate inference algorithm instead of particle filtering. Devising an efficient implementation of SSI, DS, and SMC with BP using MH is out of the scope of this paper and is left for future work. 

\subsection{MH-\siren{} Operational Semantics}
We present in \Cref{fig:op:sem-full-mh} an interpretation of the \siren{} semantics using a basic implementation of combining MH with the hybrid inference interface. We elide the rules that are the same as in the regular \siren{} semantics. The particle evaluation rules remain the same as in \Cref{sec:operational-semantics} and \Cref{fig:op:sem-full}, except that $\mfResample$ is a no-op and returns a $\false{}$ flag instead of checkpointing execution. 

The particle set evaluation rules now accrue a set of particles\footnote{Typically, MH algorithms call executions \emph{traces}. Here, we use \emph{particle} to refer to an expression and a symbolic state as before to remain consistent.} and weights by iteratively executing the program, using the samples drawn for the previous execution in the next iteration, except for one randomly chosen random variable. The samples and sample scores are stored in the additional information field of the symbolic state, alongside other information used by the instantiated inference algorithm. We denote a lookup in the symbolic state for the sample and sample score using $\symbstateP{}{\randomvar}$. To ensure the samples are assigned to the correct random variable between different particles, each instantiated random variable must be assigned a unique name, as done in \citet{wingate2011lightweight}. The samples are used during $\symvalue$ computation to use instead of drawing fresh samples.

The rules accept new particles according to their weights. The \textsc{Accept} operation returns \true{} with probability $$\alpha = \frac{w'}{w}\cdot\frac{\prod_{\randomvar' \in \dom(\symbstate'), \symbstate'(\randomvar)_\pi = (\val', \weight_\val')} \weight_\val'}{\prod_{\randomvar \in \dom(\symbstate), \symbstateP{}{\randomvar} = (\val, \weight_\val)} \weight_\val} \cdot \frac{|\dom(\symbstate)|}{|\dom(\symbstate')|}$$
An accepted particle is collected into a set of particles. When $N$ particles have been collected, they are collected into an unweighted average of their distributions. The model evaluation initiates the particle set evaluation. The MH-\siren{} particle set and model evaluation rules are shown in \Cref{fig:op:sem-full-mh}.

\begin{figure}[H]
  \begin{small}
  \begin{mathpar}
    \inferrule%
    { }
    {\pclstep{\mfResample}{\symbstate}
             {\mfUnit}{\symbstate}{1}{\false}}

  \inferrule{
    |\traceset| = N \\
    \set{\distribution(\val_i, \symbstate_i) = \distr{}_i \;|\; i \leq N, (\val_i, \symbstate_i, \weight_i) \in \traceset}\\
  }
  {\pclsstepmh{\progexpr, \traceset}{\textstyle{\sum_{i \le N+M}} \dfrac{1}{N} \times \distr{}_i }{N}}
  
  \inferrule{
    |\traceset| < N \\
    (\val, \symbstate, \weight) = \textsc{pop}(\traceset)\\
    \mu = \categorical\left(\set{ (1, \randomvar) \;|\; \randomvar \in \dom(\symbstate)}\right) \\
    \symbstate_1 = \{ \randomvar \mapsto (\emptyset, \emptyset, \symbstateP{}{\randomvar}) \;|\; \randomvar \in \symbstate \}\\
    \symbstate_2 = \symbstate_1[\draw(\mu) \mapsto \emptyset]\\
    \pclstep{\progexpr}{ \symbstate_2}
                  {\val'}{\symbstate'}{\weight'}{\false} \\
    \textsc{Accept}(\symbstate, \weight, \symbstate', \weight')\\
    \pclsstepmh{\progexpr, \traceset \cup \set{(\val', \symbstate', \weight')}}{\distr{}}{N}\\
  }
  {\pclsstepmh{\progexpr, \traceset}{\distr{}}{N}}

  \inferrule{
    |\traceset| < N \\
    (\val, \symbstate, \weight) = \textsc{pop}(\traceset)\\
    \mu = \categorical\left(\set{ (1, \randomvar) \;|\; \randomvar \in \dom(\symbstate)}\right) \\
    \symbstate_1 = \{ \randomvar \mapsto (\emptyset, \emptyset, \symbstateP{}{\randomvar}) \;|\; \randomvar \in \symbstate \}\\
    \symbstate_2 = \symbstate_1[\draw(\mu) \mapsto \emptyset]\\
    \pclstep{\progexpr}{ \symbstate_2 }
                  {\val'}{\symbstate'}{\weight'}{\false} \\
    \neg\textsc{Accept}(\symbstate, \weight, \symbstate', \weight')\\
    \pclsstepmh{\progexpr, \traceset \cup \set{(\val, \symbstate, \weight)}}{\distr{}}{N}\\
  }
  {\pclsstepmh{\progexpr, \traceset}{\distr{}}{N}}

  \inferrule{
      \pclsstepmh{\progexpr, \emptyset}{\distr{}}{N}
  }
  {\pclpstep{N}{\progexpr}{\distr{}}}
  
  \end{mathpar}
  \end{small}
  \caption{Big-step semantics of MH-\siren{}, with a basic implementation of the Metropolis Hasting algorithm. The particle evaluation rules are the same as in \Cref{fig:op:sem-full}, except for the $\mfResample$ rule. We elide the unchanged rules. The particle set evaluation rules and model evaluation rules differ from the particle filtering based \siren{}.}
  \label{fig:op:sem-full-mh}
\end{figure}

The hybrid interface operation $\symvalue$ must use the value stored in $\symbstateP{}{\randomvar}$ rather than drawing new samples if the value is not null. For example, in semi-symbolic inference, after hoisting the input random variable, $\symvalue$ retrieves the cached $\symbstateP{}{\randomvar}$ value. If the value is null, $\symvalue$ draws a sample from the symbolic distribution and intervenes on the random variable with the sample. Otherwise, $\symvalue$ uses the cached $\symbstateP{}{\randomvar}$ value to intervene with. The operation also scores the drawn or pre-stored value on the symbolic distribution of the random variable and stores the sample score in $\symbstateP{}{\randomvar}$ to use during acceptance computation. 
We show here how the implementation of the hybrid inference interface must use the cached samples and sample weights for semi-symbolic inference and delayed sampling. 

\paragraph{Semi-symbolic inference}

\begin{align*}
\symassume(\annotation, \distr{}, \symbstate) &= \; 
    \letin{\symbstate' = \symbstate[\randomvar_\mit{new} \mapsto (\annotation,\distr{}, \symbstateP{}{\randomvar_\mit{new}})]}
    (\randomvar_\mit{new}, \symbstate')\\
\symvalue(\randomvar, \symbstate) &= \;
    \begin{array}[t]{@{}l@{}}
      \mit{let} \; \symbstate' = \textsc{hoist}(\randomvar, \symbstate) \; \mit{in}\; \\
      \mit{let}\; (v, \_) = \symbstateP{}{\randomvar} \;\mit{in}\;\\
      \mit{let}\; v' = \mit{if}\; v = \emptyset \; \mit{then}\; \textsc{draw}(\symbstate'(\randomvar)_d) \; \mit{else}\; v \;\mit{in} \;\\
      \mit{let}\; w = \score(\symbstate'(\randomvar)_d, \val') \; \mit{in}\;\\
      (\val', \textsc{intervene}(\randomvar, \SSIdeltasample{v'}, \symbstate'[\randomvar \mapsto (\symbstateK{}{\randomvar}, \symbstateD{}{\randomvar}, (\val',  w))]))
      \end{array}\\
\symobserve(\randomvar, \val, \symbstate) &= \;
    \begin{array}[t]{@{}l@{}}
    \mit{let} \; \symbstate' = \textsc{hoist}(\randomvar, \symbstate) \; \mit{in}\; 
    \mit{let} \; \scoreval = \score(\symbstate'(\randomvar)_d, \val) \;\mit{in} \;\\
    (\textsc{intervene}(\randomvar, \deltad{\val}, \symbstate'), \scoreval)
    \end{array}
\end{align*}

\paragraph{Delayed Sampling}
\begin{align*}
  \symassume(\annotation, \distr{}, \symbstate) &= \; 
      \begin{array}[t]{@{}l@{}}
        \letin{\distr{\prime}, \symbstate' = \conjdistr(\distr{}, \symbstate)}\\
        \letin{\symbstate'' = \initialize(\randomvar_\mit{new}, \annotation, \distr{\prime}, \symbstate')}
      (\randomvar_\mit{new}, \symbstate'')\\
      \end{array}\\
  \symvalue(\randomvar, \symbstate) &= \;
      \begin{array}[t]{@{}l@{}}
        \mit{let} \; \symbstate' = \textsc{graft}(\randomvar, \symbstate) \; \mit{in}\; \\
        \mit{let}\; (\val, \_) = \symbstateP{}{\randomvar} \;\mit{in}\;\\
        \mit{let}\; \val' = \mit{if}\; \val = \emptyset \; \mit{then}\; \textsc{draw}(\symbstate'(\randomvar)_d) \; \mit{else}\; \val \;\mit{in} \;\\
        \mit{let}\; w = \score(\symbstate'(\randomvar)_d, \val') \; \mit{in}\;\\
        \letin{\symbstate'' = \textsc{realize}(\randomvar, \SSIdeltasample{\val'}, \symbstate')}\\
        (\val', \symbstate'')
        \end{array}\\
  \symobserve(\randomvar, \val, \symbstate) &= \;
      \begin{array}[t]{@{}l@{}}
      \mit{let} \; \symbstate' = \textsc{graft}(\randomvar, \symbstate) \; \mit{in}\;\\
      \mit{let} \; \scoreval = \score(\symbstate'(\randomvar)_d, \val) \mit{in} \;\\
      \letin{\symbstate'' = \textsc{realize}(\randomvar, \deltad{\val}, \symbstate')}\\
      (\symbstate'', \scoreval)
      \end{array}
  \end{align*}

\subsection{Inference Plan Satisfiability Analysis}
The abstract interpretation rules for MH-\siren{} is the same as the particle filtering-based \siren{} shown in \Cref{sec:analysis} and \Cref{fig:abstract-interp-full}. This extensibility falls out from the hybrid inference interface, the abstract domain, and the fact that both MH and particle filtering are Monte Carlo methods. 
The only difference between the particle evaluation rules of \siren{} and MH-\siren{} is whether $\mfResample$ pauses the execution by setting the resample flag to $\true$. The abstract particle evaluation already does not create checkpoints in \siren{}, so the rules are the same for MH-\siren{}. The biggest difference between \siren{} and MH-\siren{} is the particle set evaluation rules. However, at its core, the MH-\siren{} particle set evaluation rule evaluates the program multiple times, similar to particle filtering. The only additional operation is tracking the samples and sample scores. The samples serve only as a cache for the $\symvalue$ function to retrieve a specific constant value instead of drawing a fresh constant value. In either case, the constant values are soundly approximated as $\cunk$. Also, scores are not used in the analysis, as they do not affect the symbolic computation. The abstract interpretation can then ignore the sample sites. This leaves the analysis entirely the same as presented in \Cref{sec:analysis} and \Cref{fig:abstract-interp-full}.

\subsubsection{Collecting Semantics and Soundness}
Because the particle evaluation rules are the same in \siren{} and MH-\siren{} except that $\mfResample$ does not checkpoint in MH-\siren{}, we can reuse the collecting semantics for the particle evaluation rules. 
The collecting semantics of the particle set evaluation still operates on a single set of particles, because the analysis is still interested in the possible particles produced during program execution. As before, if any of the particles return $\fail$, the semantics return a singleton set containing \fail. However, unlike before, all the configurations returned by the rule have $\false$ resample flags. Thus, the collecting particle set evaluation rules return a set of distributions using the \textsc{Distribution} operation. The model evaluation rules are unchanged.

\begin{small}
    \begin{mathpar}
    \inferrule%
    { 
        \pset = \set{(\val',\symbstate') \;|\;  (\progexpr,\symbstate) \in \pset, (\cpclstep{\progexpr, \symbstate}{\evalset'}), (\val',\symbstate',\false) \in \evalset'}\\
        \textstyle{\bigwedge_{(\val, \symbstate) \in \evalset}}\ (\val \neq \fail) \\
        \dset = \set{\distribution(\val,\symbstate) \;|\; (\val, \symbstate) \in \pset}
    }
    {\cpclsstep{\pset}{\dset}}

    \inferrule%
    { 
        \pset = \set{(\val',\symbstate') \;|\;  (\progexpr,\symbstate) \in \pset, (\cpclstep{\progexpr, \symbstate}{\evalset'}), (\val',\symbstate',\false) \in \evalset'}\\
        \textstyle{\bigvee_{(\val, \symbstate) \in \pset}} (\val = \fail)
    }
    {\cpclsstep{\pset}{\set{\fail}}}
    \end{mathpar}
\end{small}

Finally, using the lemmas and theorems in \Cref{appendix:proofs}, we can show that the analysis is sound for the alternative MH-\siren{} semantics.

\begin{theorem}[MH-\siren{} Particle Set Evaluation Soundness]
  \label{thm:particle-set-mh}
  For every particle set $\pset$, and distribution set $\dset$ such that $(\cpclsstep{\pset}{\dset})$, we have that \apclsstep{\set{\progexpr, \abstr(\{\symbstate\}) \;|\; (\progexpr, \symbstate) \in \pset}}{\adistr{}} and $\dset \subseteq \concret(\adistr{})$.
\end{theorem}
\begin{proof}
  Every configuration produced by the collecting particle evaluation rules has the $\false$ resample flag. Then, by \Cref{lem:terminating-particle} and the definition of abstraction and concretization, we have that $\dset \subseteq \concret(\adistr{})$. 

\end{proof}

\begin{corollary}[MH-\siren{} Model Evaluation Soundness]
  If $\cpclpstep{\progexpr}{\{\fail\}}$, then $\apclpstep{\progexpr}{\afail}$.
\end{corollary}

\section{MH-\siren{} Performance Evaluation}
We evaluate and produce the performance profiles of the satisfiable inference plans over the benchmark suite using these algorithms using the same methodology described in \Cref{sec:eval}, except for the number of timesteps. The basic implementation of MH is less efficient at solving these benchmarks than particle filtering and requires longer runtime for the estimation variance to decrease. Instead, we execute each benchmark for 10 timesteps only. Across all benchmarks and the hybrid inference algorithms, the best inference plan achieves a 1.65x speedup compared to the default inference plan, with the maximum speedup being 75x the default plan. Across all benchmarks, variables, and inference algorithms, the best inference plans achieve 2.91x better accuracy with equal or less runtime compared to the default plans, with a maximum of 31475x better accuracy.

\subsection{Performance Profiles with 90th Perecentile Error}

\begin{figure}[H]
  \centering
  \begin{subfigure}[c]{0.75\textwidth}
    \centering
    \includegraphics[width=1\textwidth]{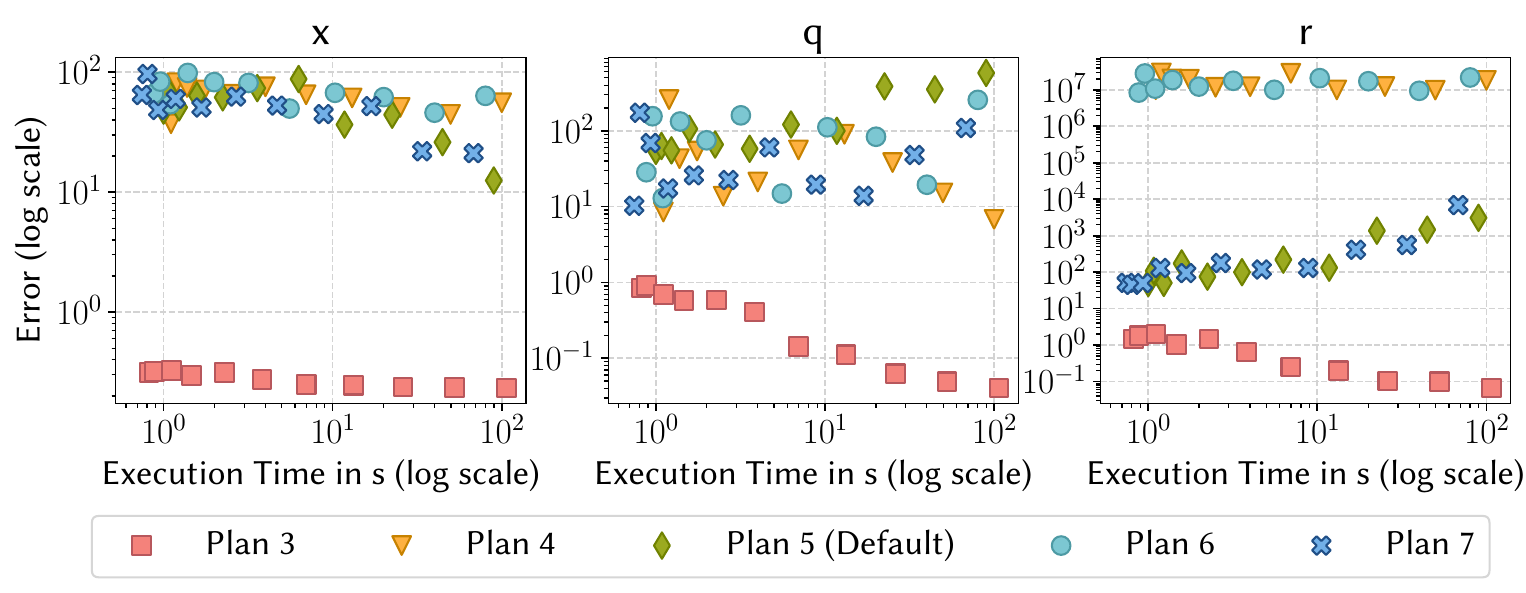}
    \caption{SSI.}
  \end{subfigure}%
  \\
  \begin{subfigure}[c]{0.75\textwidth}
    \centering
    \includegraphics[width=1\textwidth]{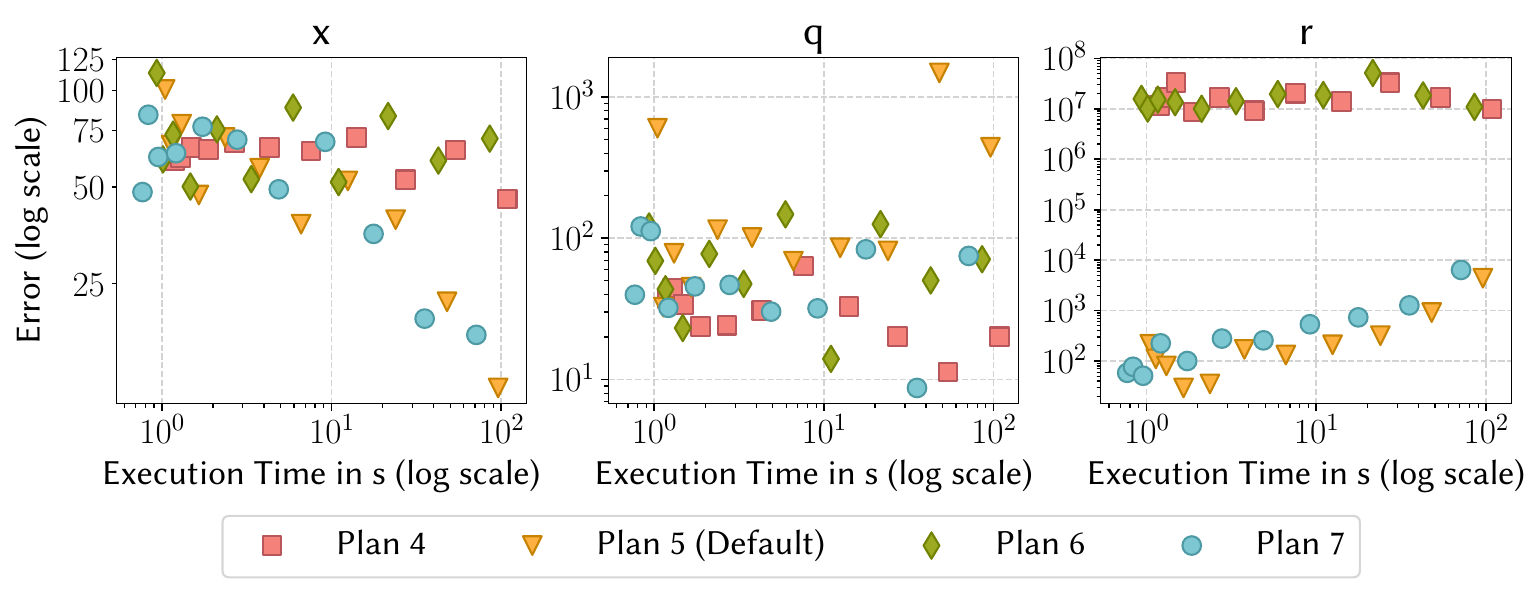}
    \caption{DS.}
  \end{subfigure}%
  \\
  \begin{subfigure}[c]{0.75\textwidth}
    \centering
    \includegraphics[width=1\textwidth]{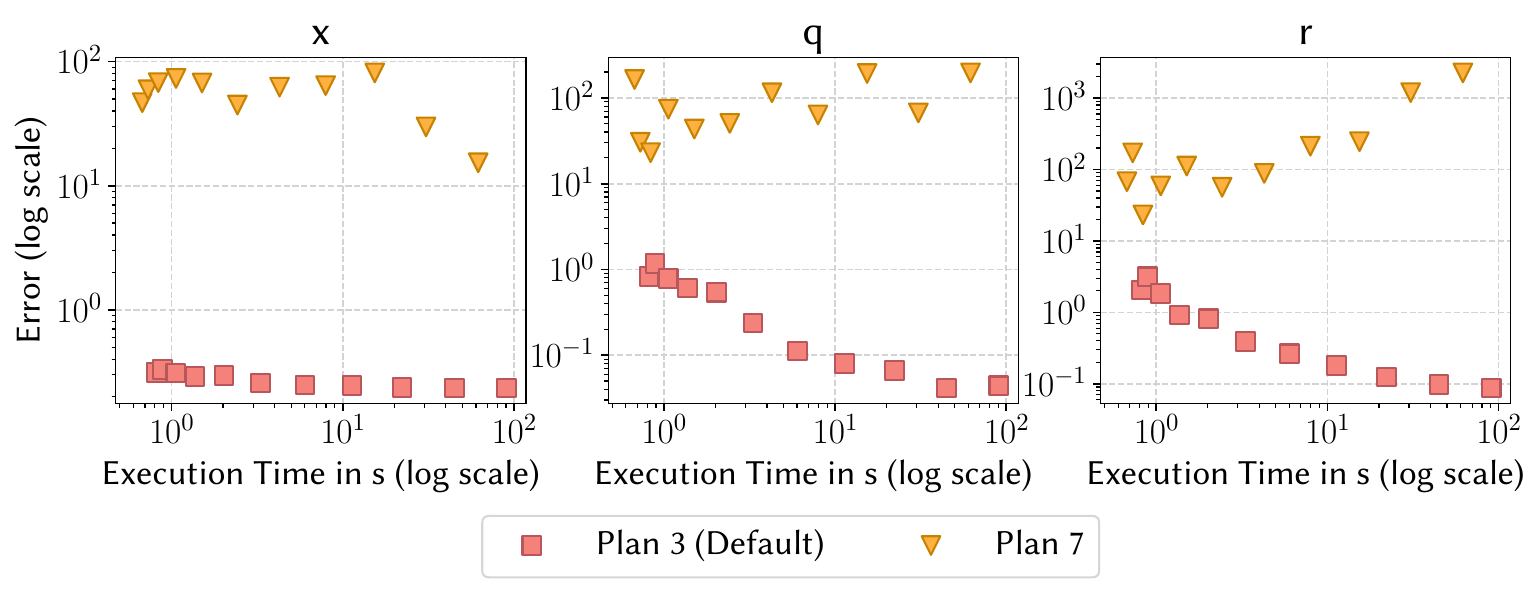}
    \caption{SMC w/ BP.}
  \end{subfigure}%
  \caption{\bNoise{}}
  \label{fig:mh-performance-results-noise}
\end{figure}

\begin{figure}[H]
  \centering
  \begin{subfigure}[c]{0.75\textwidth}
    \centering
    \includegraphics[width=1\textwidth]{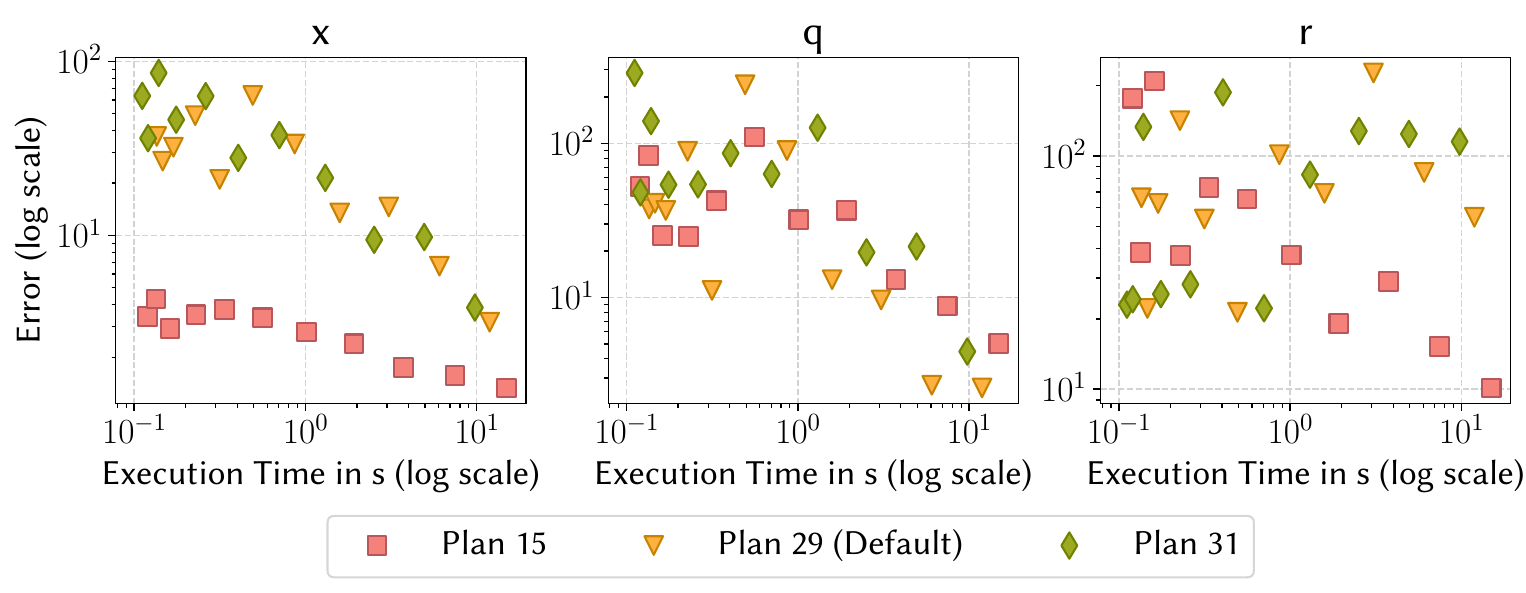}
    \caption{SSI.}
  \end{subfigure}%
  \\
  \begin{subfigure}[c]{0.75\textwidth}
    \centering
    \includegraphics[width=1\textwidth]{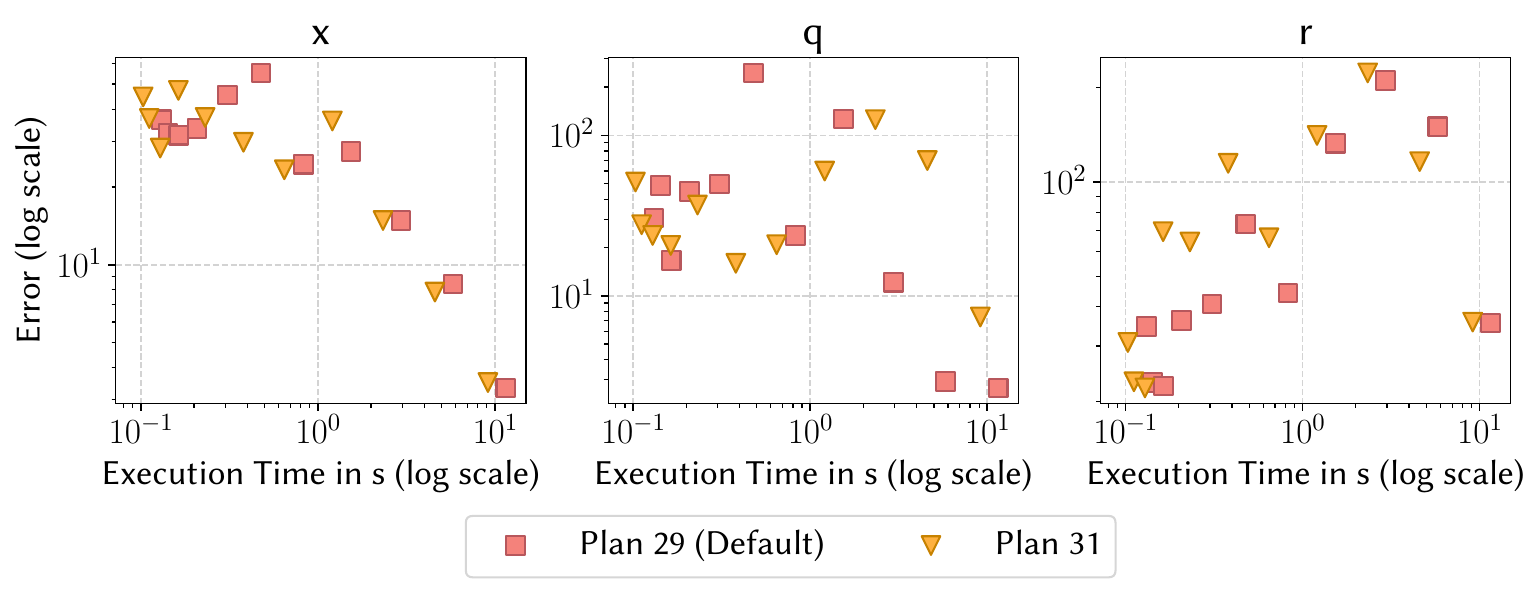}
    \caption{DS.}
  \end{subfigure}%
  \\
  \begin{subfigure}[c]{0.75\textwidth}
    \centering
    \includegraphics[width=1\textwidth]{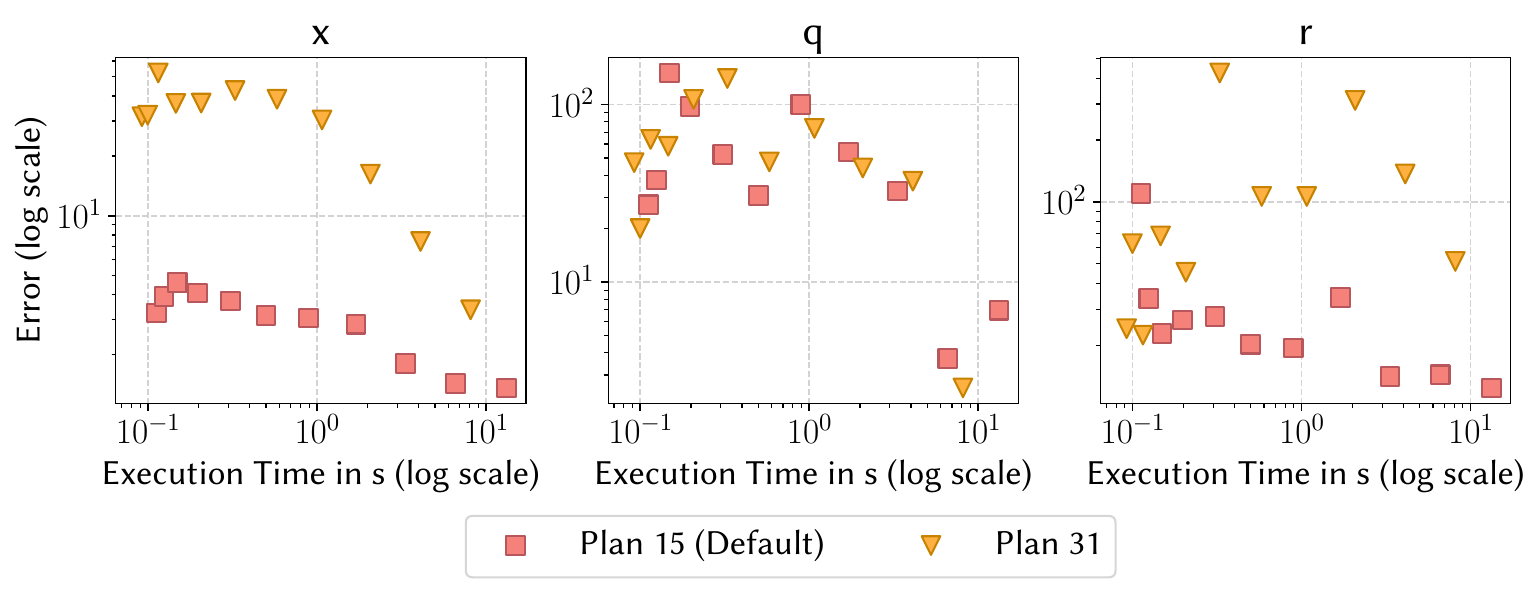}
    \caption{SMC w/ BP.}
  \end{subfigure}%
  \caption{\bRadar{}}
  \label{fig:mh-performance-results-radar}
\end{figure}

\begin{figure}[H]
  \centering
  \begin{subfigure}[c]{0.75\textwidth}
    \centering
    \includegraphics[width=1\textwidth]{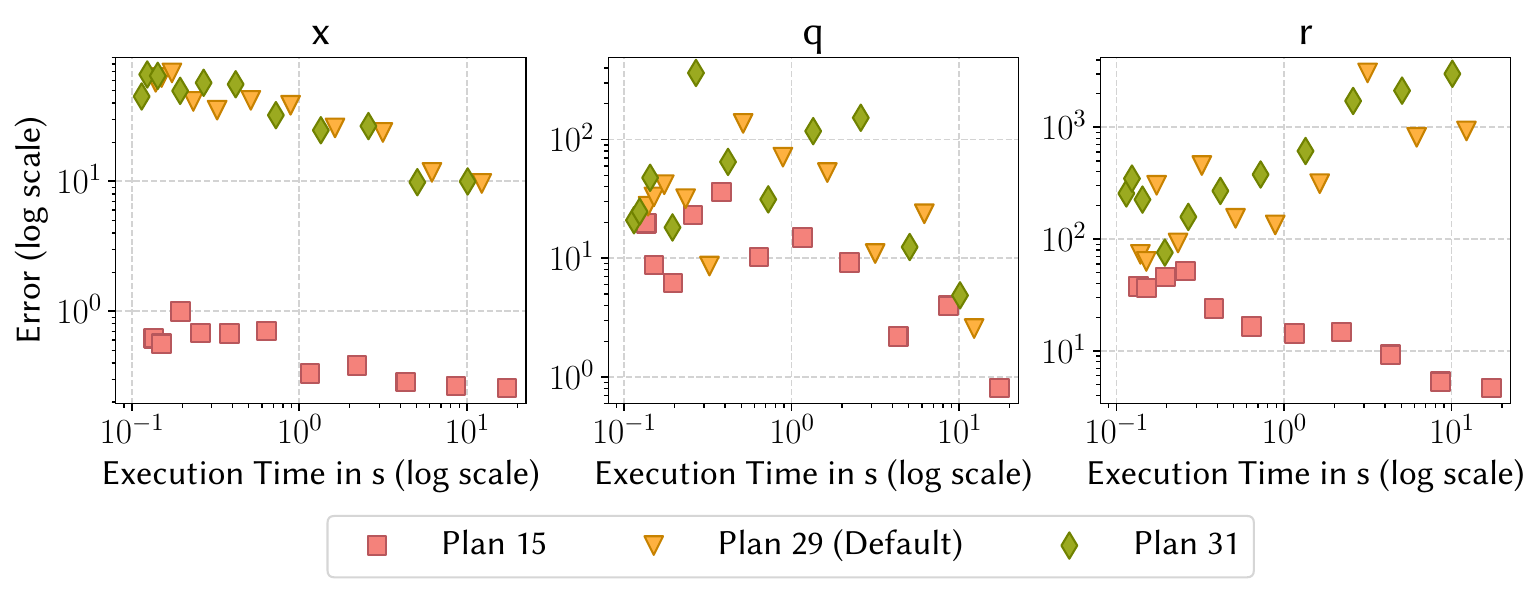}
    \caption{SSI.}
  \end{subfigure}%
  \\
  \begin{subfigure}[c]{0.75\textwidth}
    \centering
    \includegraphics[width=1\textwidth]{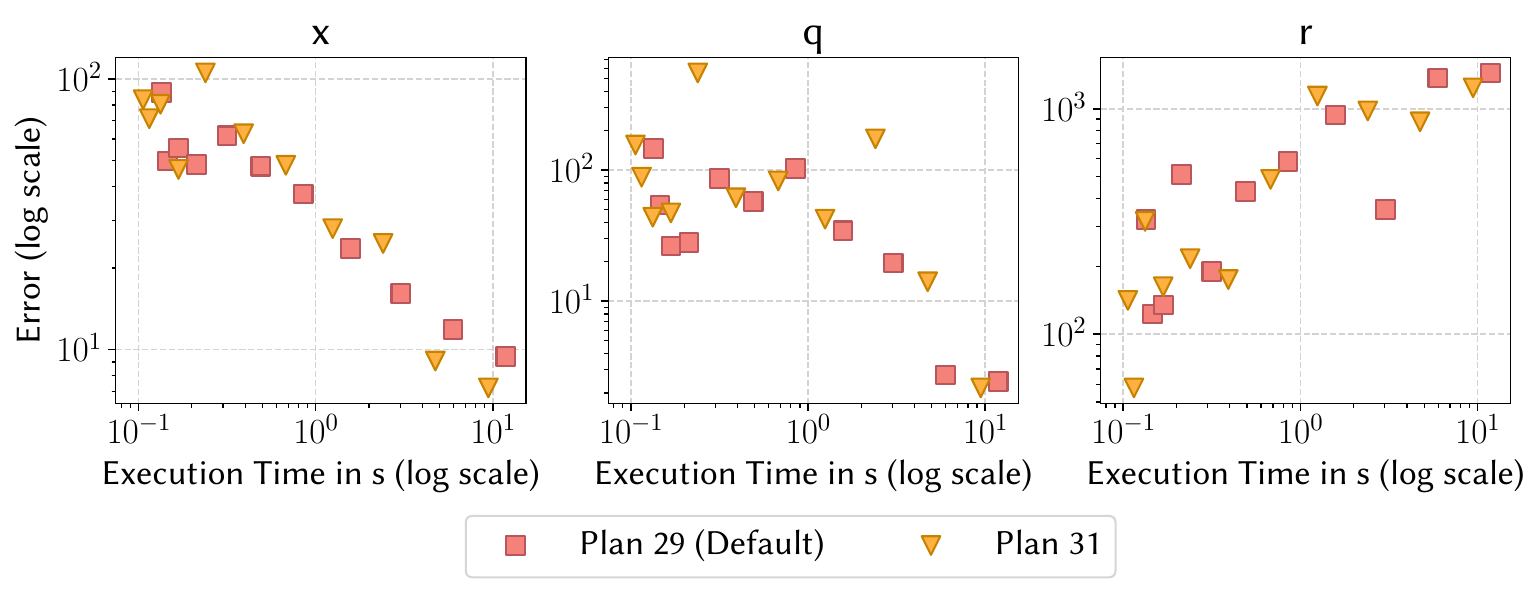}
    \caption{DS.}
  \end{subfigure}%
  \\
  \begin{subfigure}[c]{0.75\textwidth}
    \centering
    \includegraphics[width=1\textwidth]{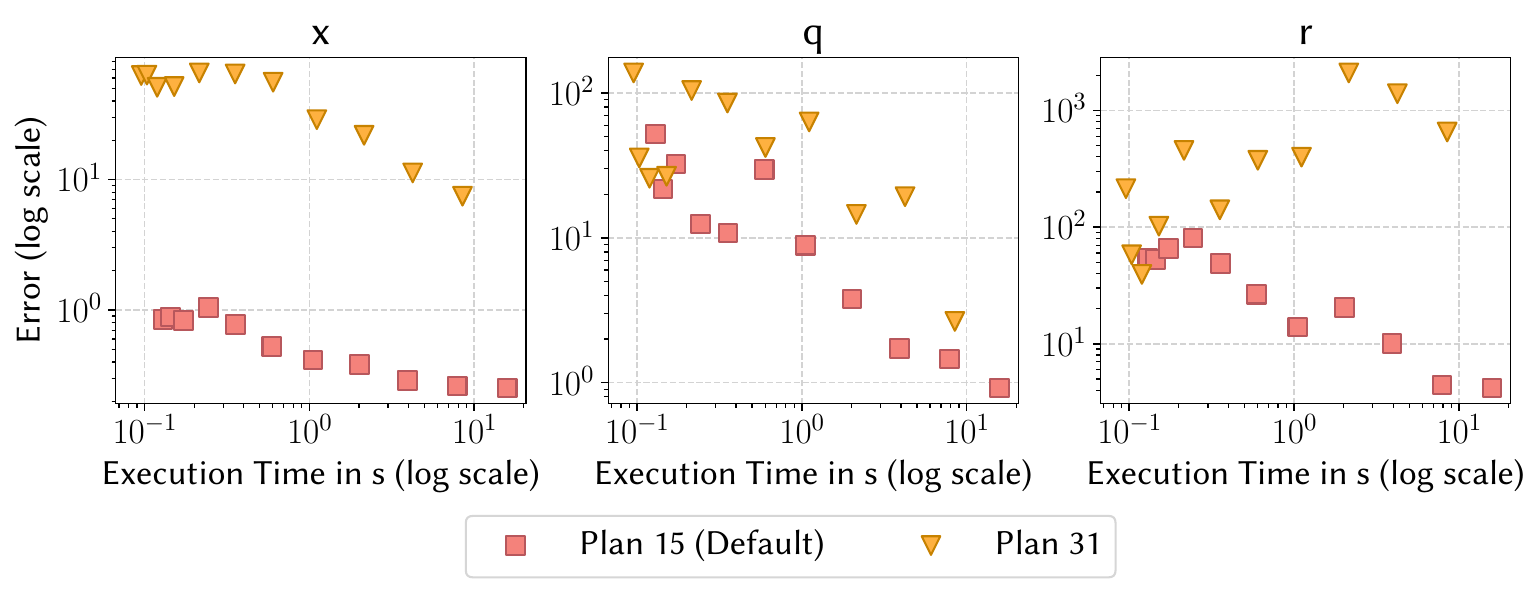}
    \caption{SMC w/ BP.}
  \end{subfigure}%
  \caption{\bEnvnoise{}}
  \label{fig:mh-performance-results-envnoise}
\end{figure}

\begin{figure}[H]
  \centering
  \begin{subfigure}[c]{0.5\textwidth}
    \centering
    \includegraphics[width=1\textwidth]{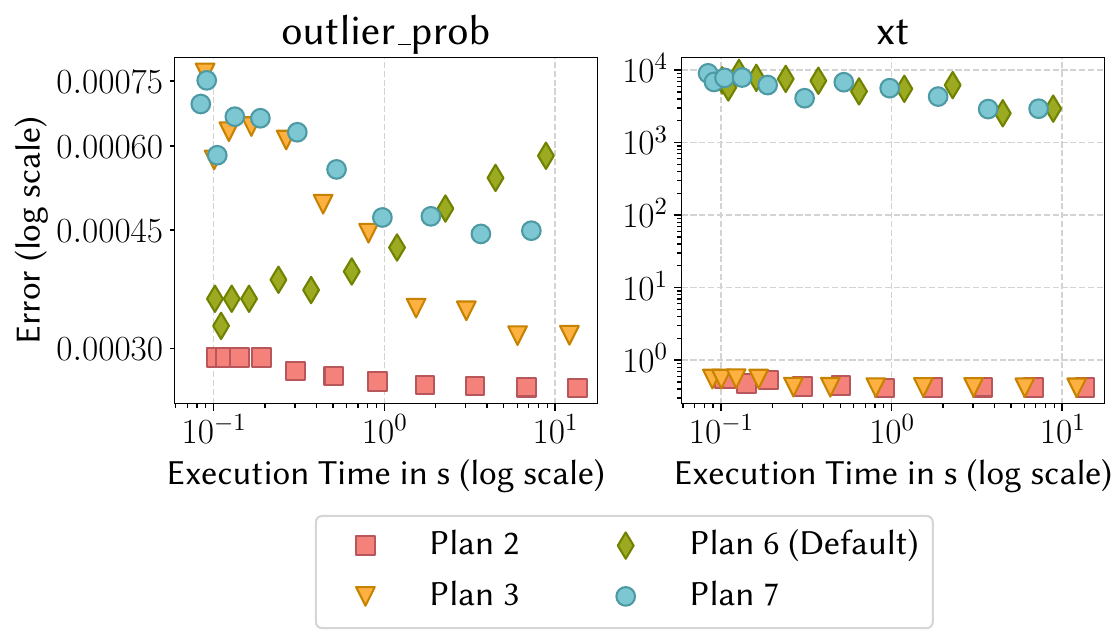}
    \caption{SSI.}
  \end{subfigure}%
  \\
  \begin{subfigure}[c]{0.5\textwidth}
    \centering
    \includegraphics[width=1\textwidth]{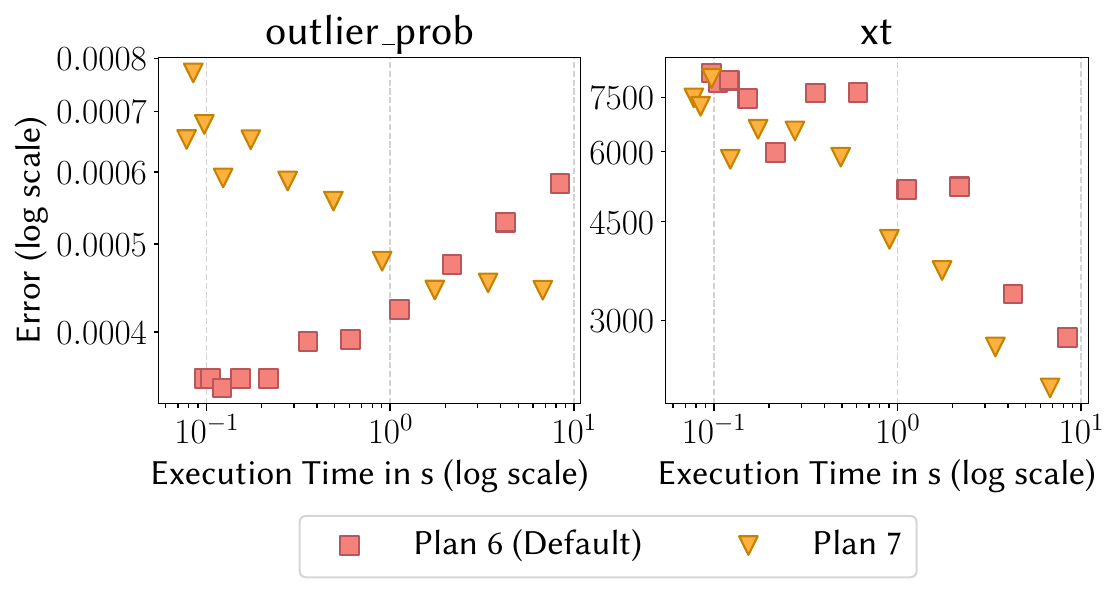}
    \caption{DS.}
  \end{subfigure}%
  \\
  \begin{subfigure}[c]{0.5\textwidth}
    \centering
    \includegraphics[width=0.99\textwidth]{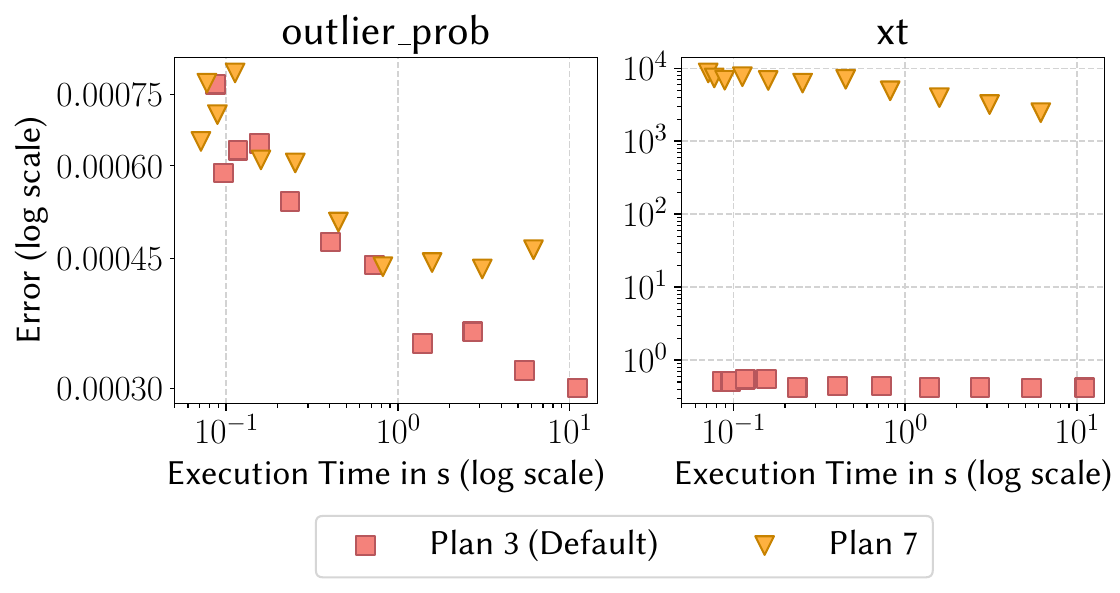}
    \caption{SMC w/ BP.}
  \end{subfigure}%
  \caption{\bOutlier{}}
  \label{fig:mh-performance-results-outlier}
\end{figure}

\begin{figure}[H]
  \centering
  \begin{subfigure}[c]{0.5\textwidth}
    \centering
    \includegraphics[width=1\textwidth]{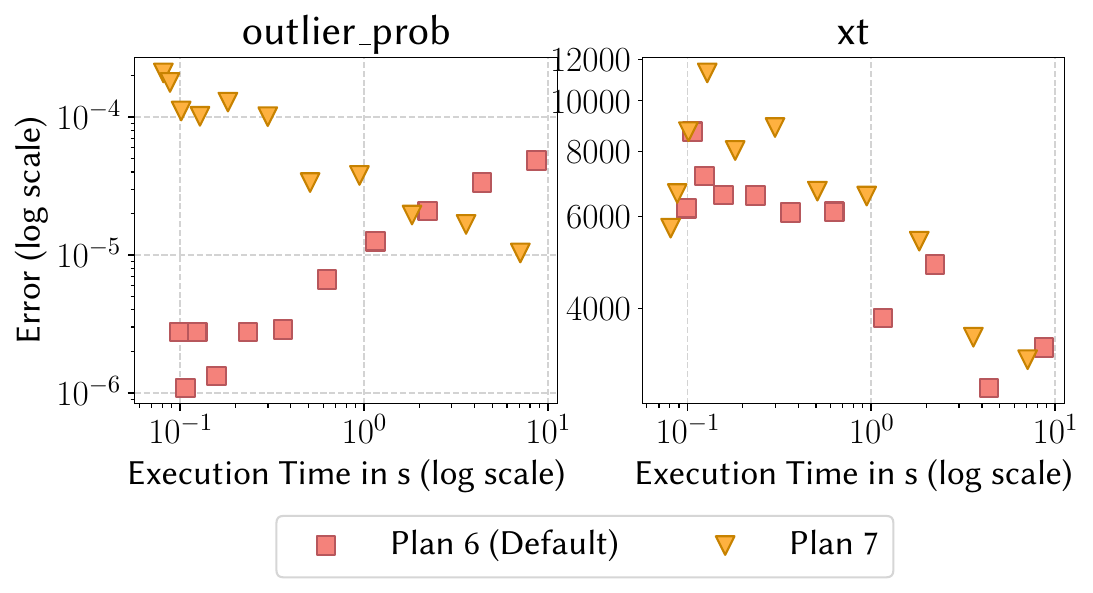}
    \caption{SSI.}
  \end{subfigure}%
  \\
  \begin{subfigure}[c]{0.5\textwidth}
    \centering
    \includegraphics[width=1\textwidth]{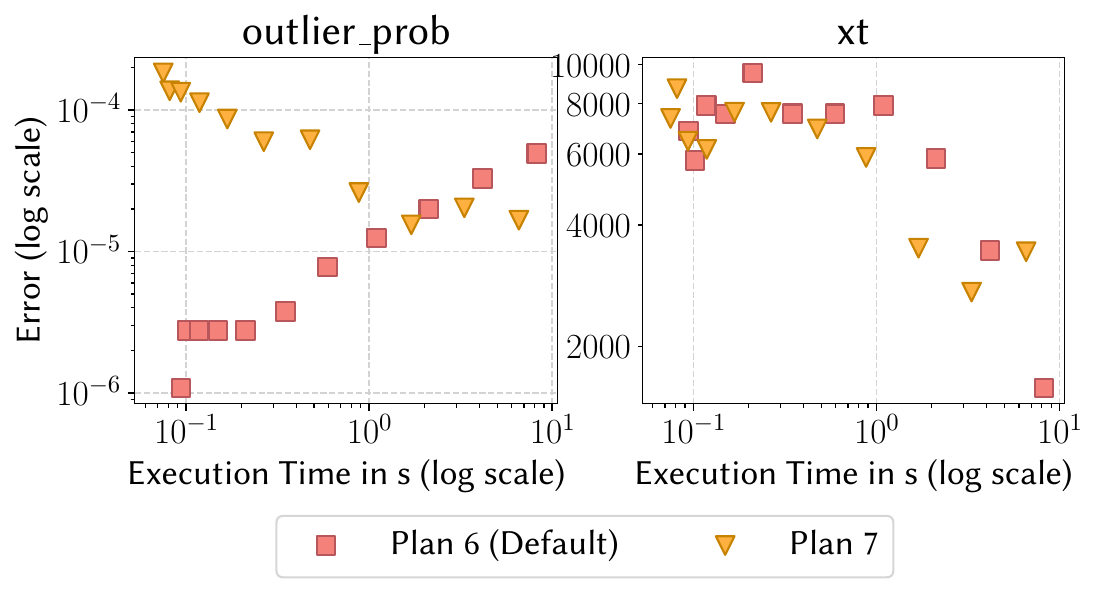}
    \caption{DS.}
  \end{subfigure}%
  \\
  \begin{subfigure}[c]{0.5\textwidth}
    \centering
    \includegraphics[width=1\textwidth]{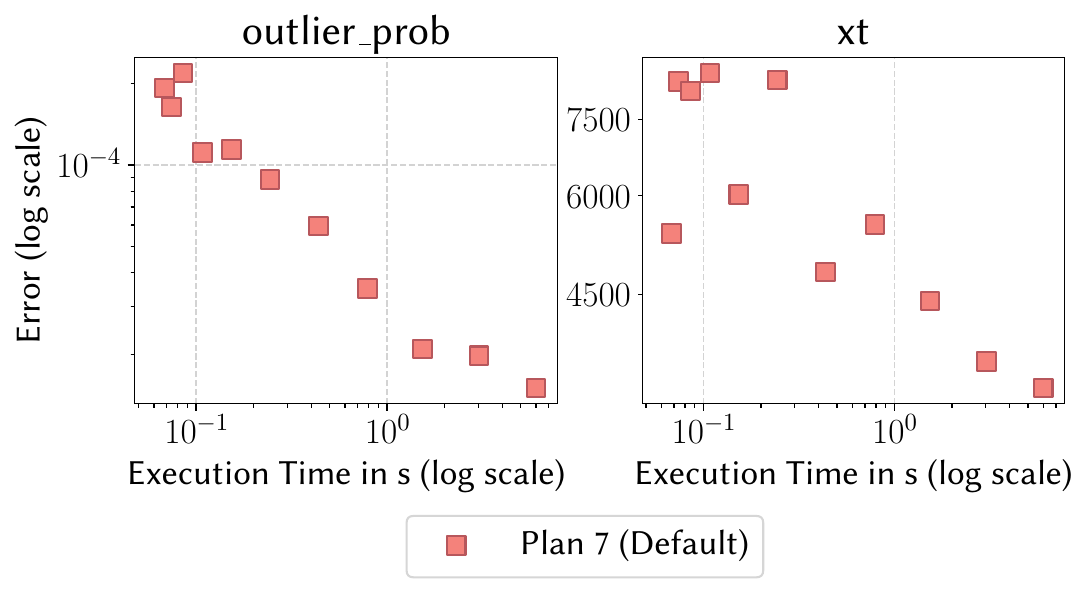}
    \caption{SMC w/ BP.}
  \end{subfigure}%
  \caption{\bOutlierheavy{}.}
  \label{fig:mh-performance-results-outlierheavy}
\end{figure}

\begin{figure}[H]
  \centering
  \begin{subfigure}[c]{0.5\textwidth}
    \centering
    \includegraphics[width=1\textwidth]{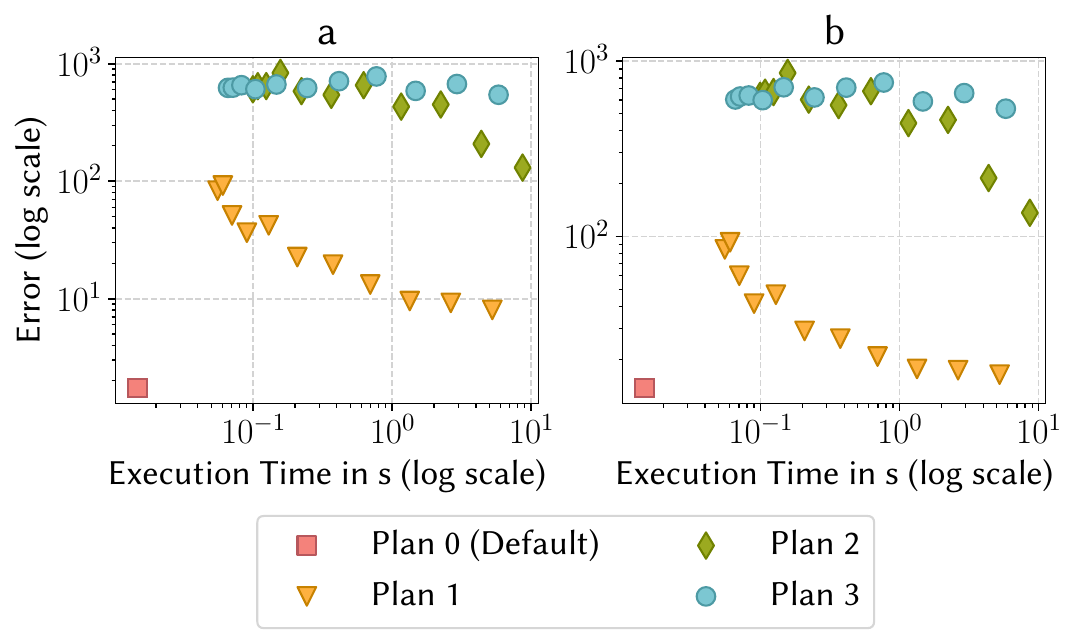}
    \caption{SSI.}
  \end{subfigure}%
  \\
  \begin{subfigure}[c]{0.5\textwidth}
    \centering
    \includegraphics[width=1\textwidth]{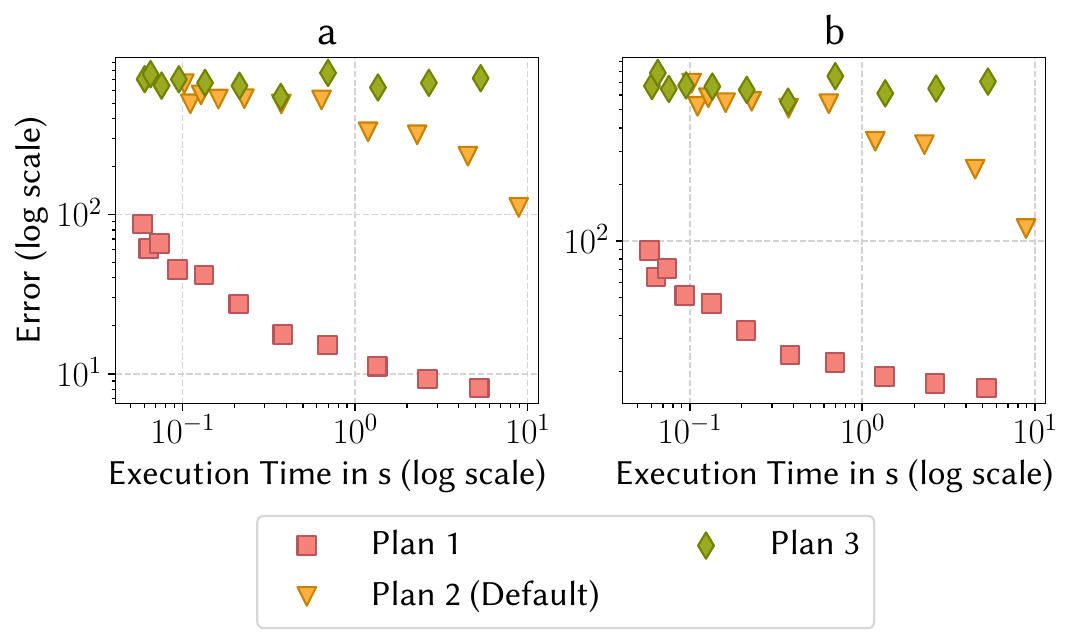}
    \caption{DS.}
  \end{subfigure}%
  \\
  \begin{subfigure}[c]{0.5\textwidth}
    \centering
    \includegraphics[width=1\textwidth]{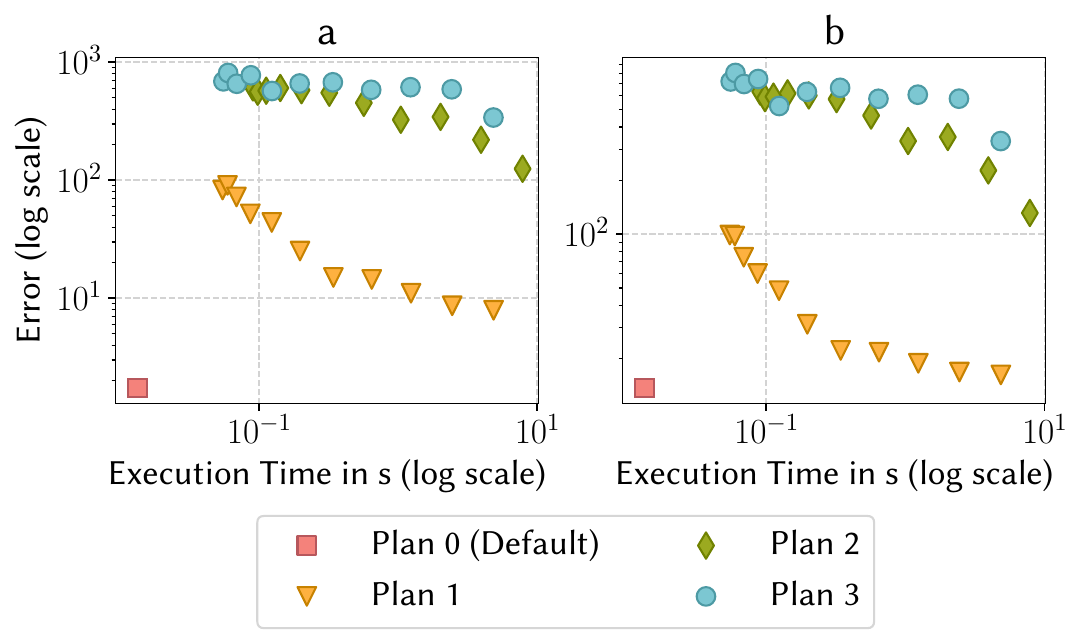}
    \caption{SMC w/ BP.}
  \end{subfigure}%
  \caption{\bGtree{}}
  \label{fig:mh-performance-results-tree}
\end{figure}

\begin{figure}[H]
  \centering
  \begin{subfigure}[c]{0.75\textwidth}
    \centering
    \includegraphics[width=1\textwidth]{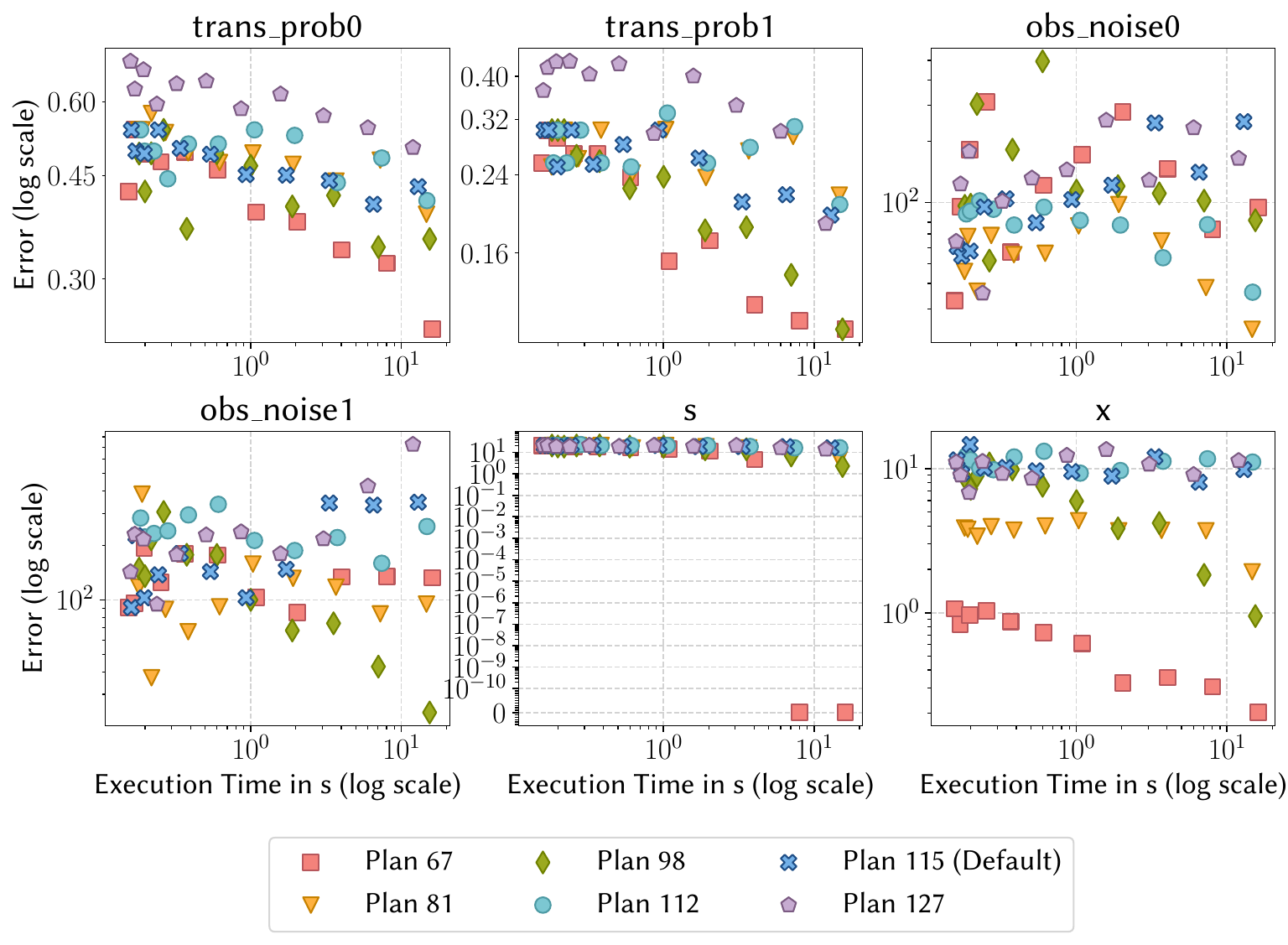}
    \caption{SSI.}
  \end{subfigure}%
  \\
  \begin{subfigure}[c]{0.75\textwidth}
    \centering
    \includegraphics[width=1\textwidth]{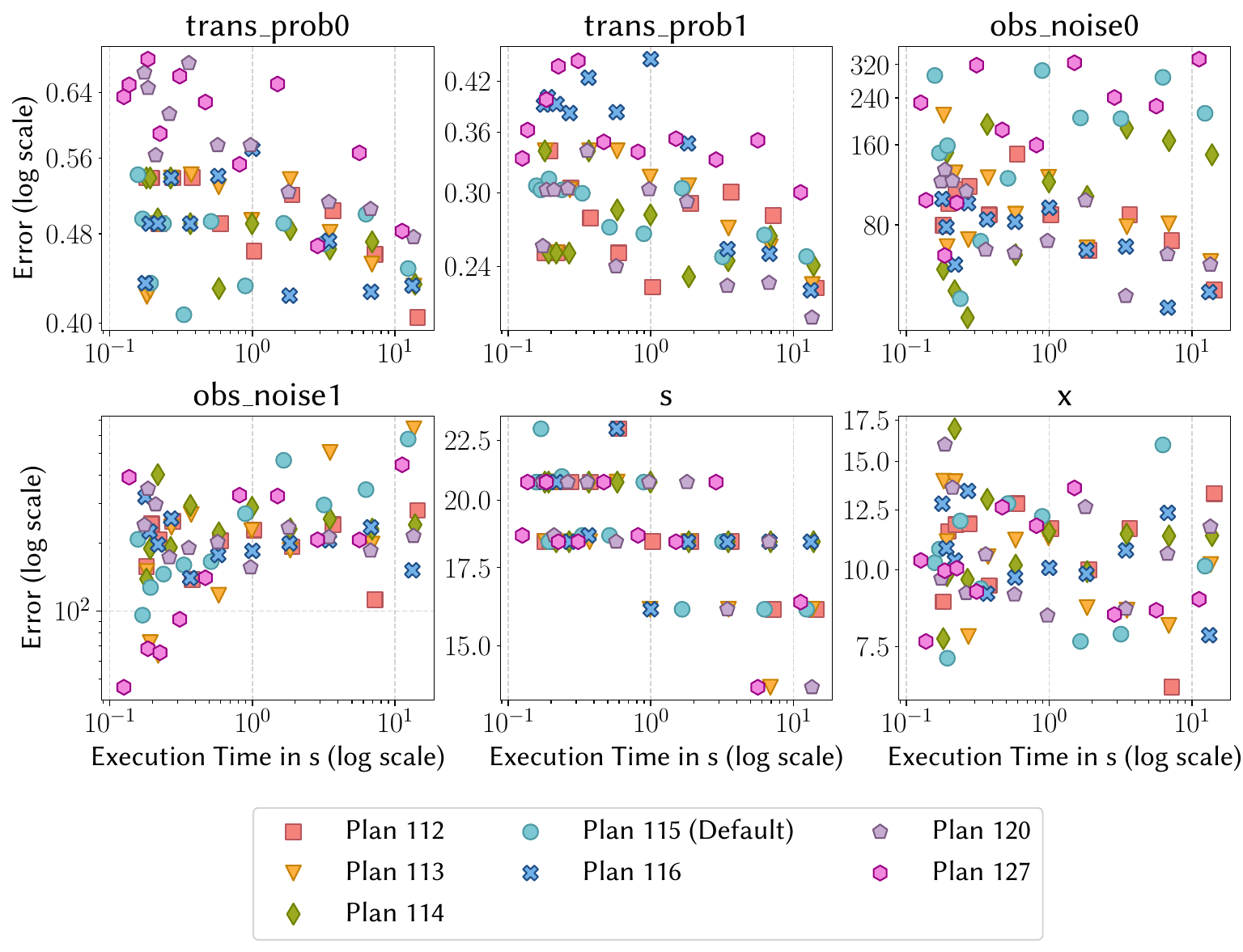}
    \caption{DS.}
  \end{subfigure}%
  \caption{\bSlds{}}
  \label{fig:mh-performance-results-slds-1}
\end{figure}

\begin{figure}[H]
  \centering
  \begin{subfigure}[c]{0.75\textwidth}
    \centering
    \includegraphics[width=1\textwidth]{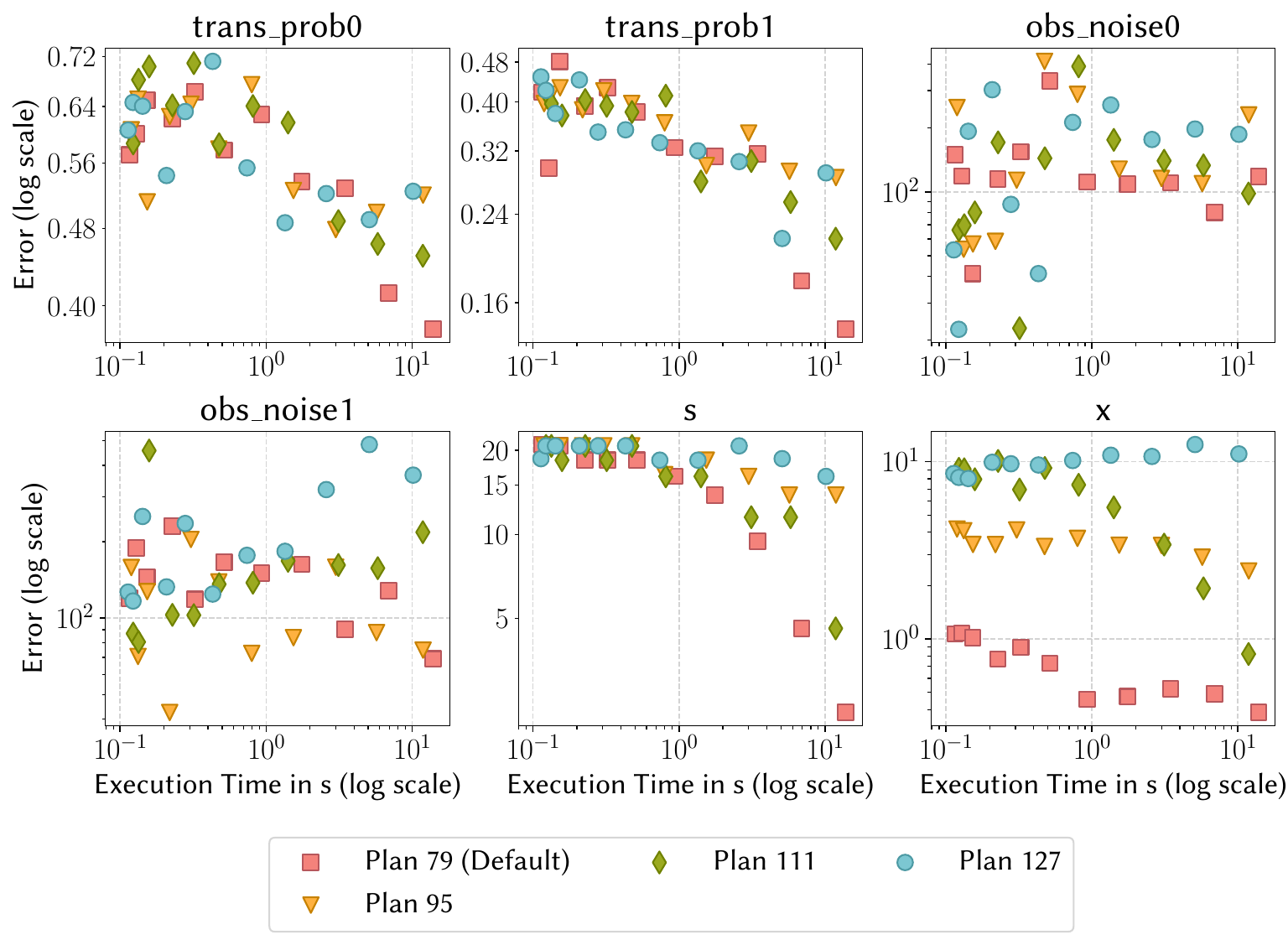}
    \caption{SMC w/ BP.}
  \end{subfigure}%
  \caption{\bSlds{} (continued)}
  \label{fig:mh-performance-results-slds-2}
\end{figure}

\begin{figure}[H]
  \centering
  \begin{subfigure}[c]{1\textwidth}
    \centering
    \includegraphics[width=1\textwidth]{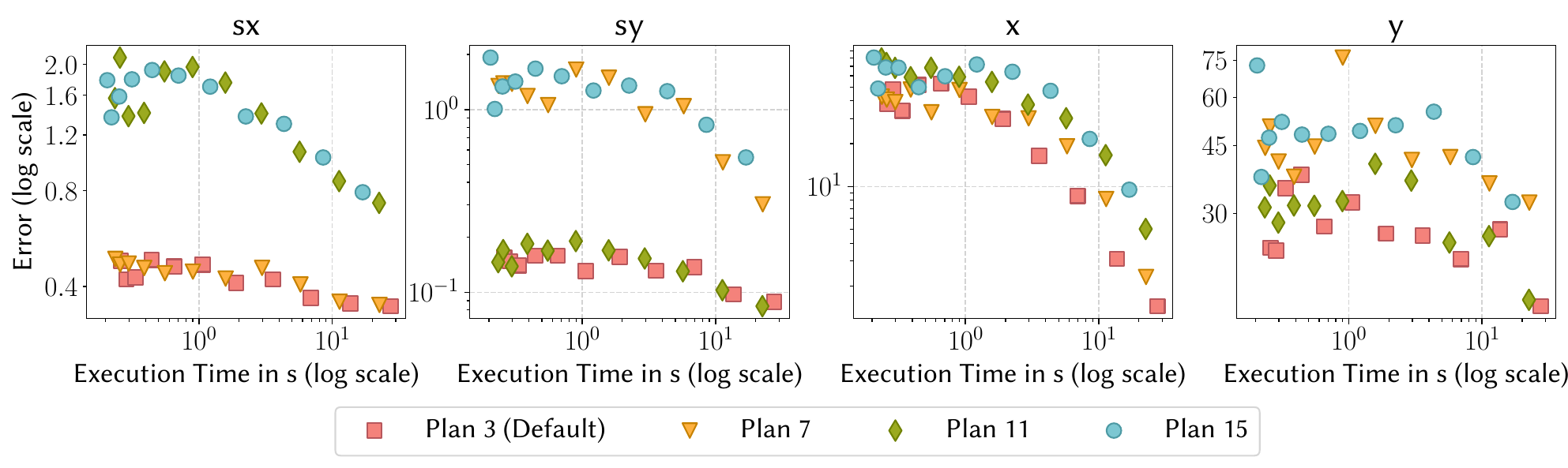}
    \caption{SSI.}
  \end{subfigure}%
  \\
  \begin{subfigure}[c]{1\textwidth}
    \centering
    \includegraphics[width=1\textwidth]{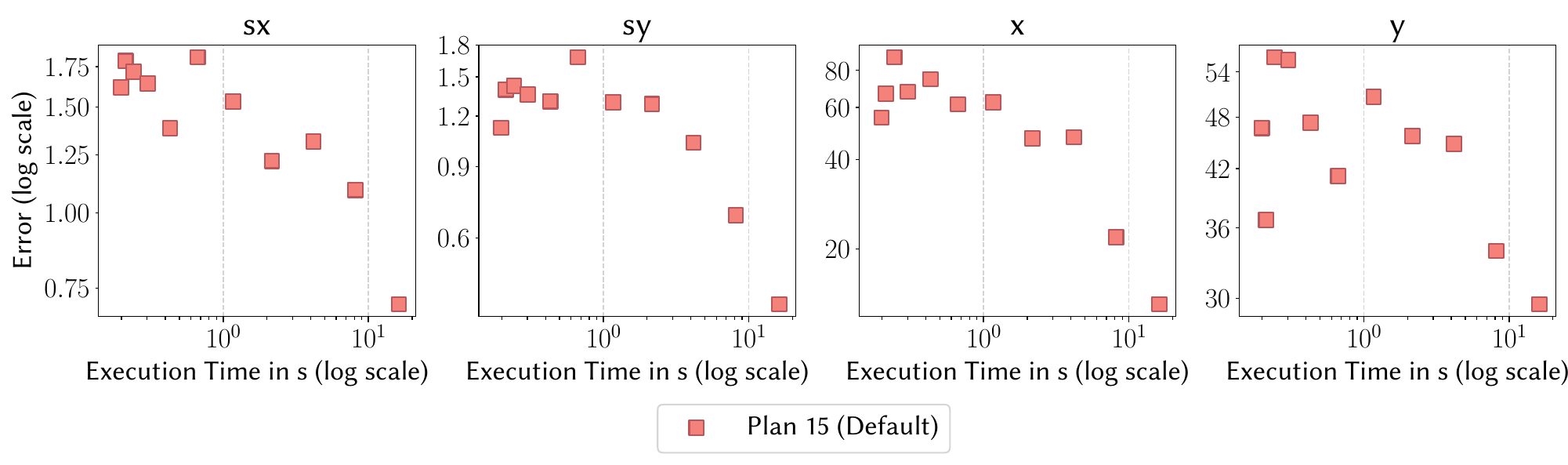}
    \caption{SSI.}
  \end{subfigure}%
  \\
  \begin{subfigure}[c]{1\textwidth}
    \centering
    \includegraphics[width=1\textwidth]{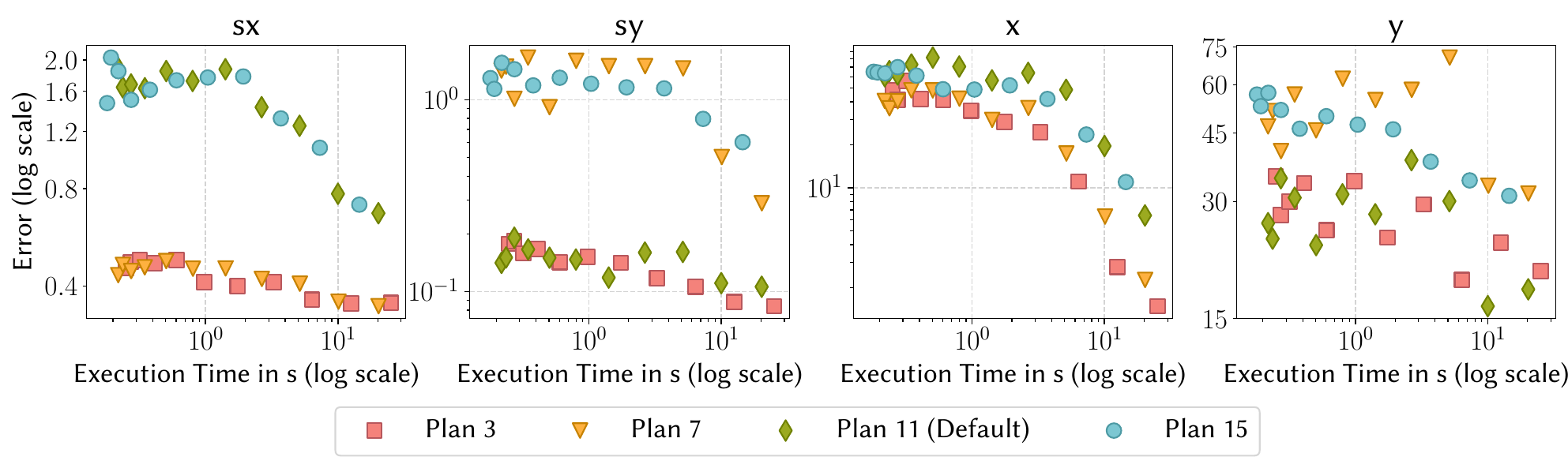}
    \caption{SMC w/ BP.}
  \end{subfigure}%
  \caption{\bRunner{}.}
  \label{fig:mh-performance-results-runner}
\end{figure}

\begin{figure}[H]
  \centering
  \begin{subfigure}[c]{1\textwidth}
    \centering
    \includegraphics[width=0.5\textwidth]{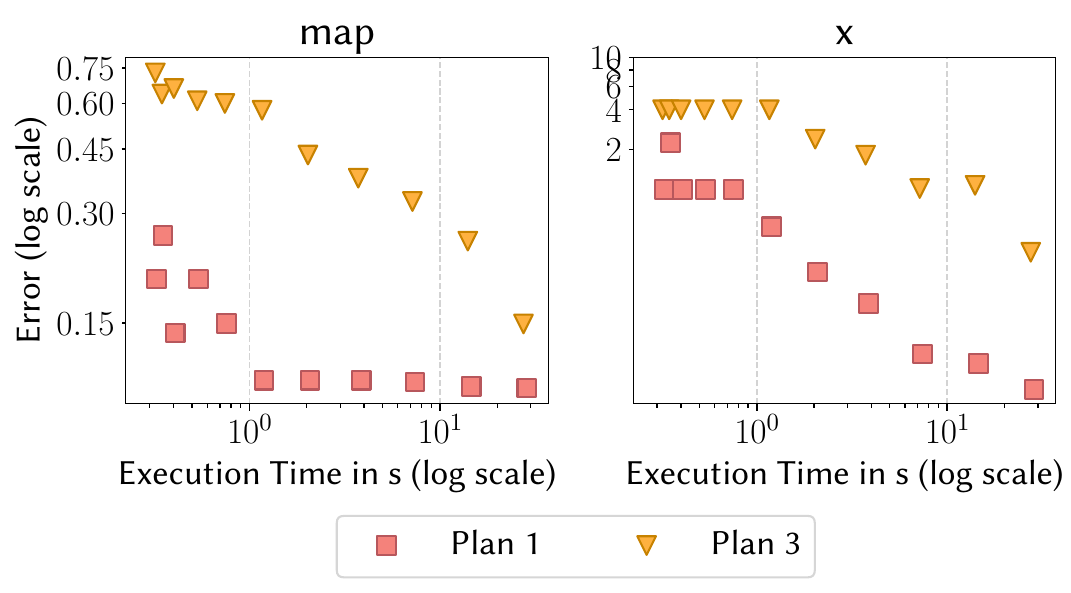}
    \caption{SSI. Plan 0 (Default) and Plan 2 time out for all particles.}
  \end{subfigure}%
  \\
  \begin{subfigure}[c]{1\textwidth}
    \centering
    \includegraphics[width=0.5\textwidth]{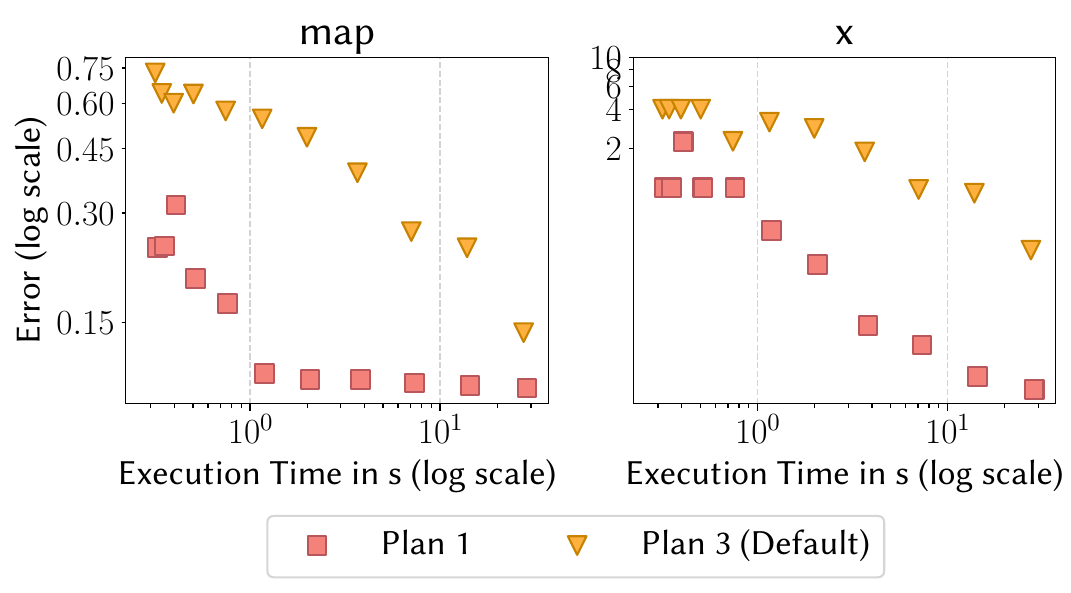}
    \caption{DS.}
  \end{subfigure}%
  \\
  \begin{subfigure}[c]{1\textwidth}
    \centering
    \includegraphics[width=0.5\textwidth]{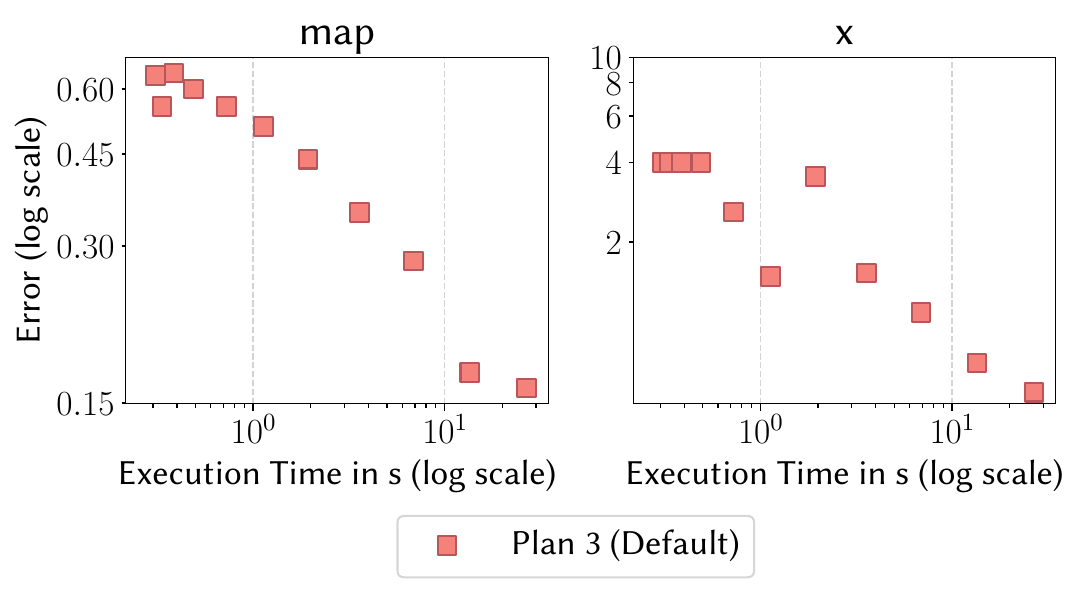}
    \caption{SMC w/ BP.}
  \end{subfigure}%
  \caption{\bSlam{}.}
  \label{fig:mh-performance-results-slam}
\end{figure}

\begin{figure}[H]
  \centering
  \begin{subfigure}[c]{1\textwidth}
    \centering
    \includegraphics[width=0.5\textwidth]{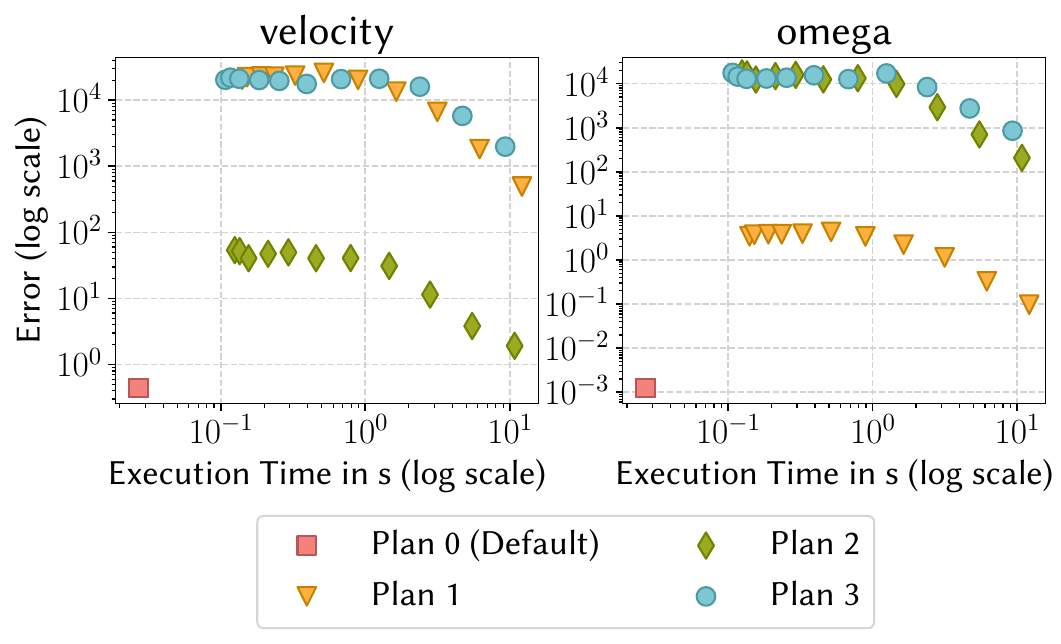}
    \caption{SSI. }
  \end{subfigure}%
  \\
  \begin{subfigure}[c]{1\textwidth}
    \centering
    \includegraphics[width=0.5\textwidth]{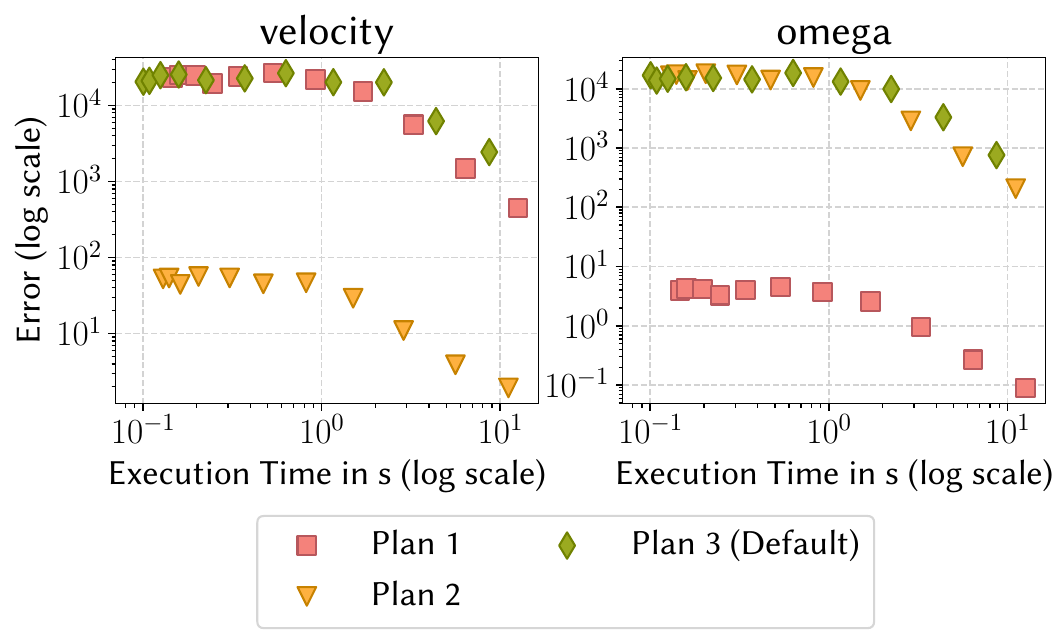}
    \caption{DS.}
  \end{subfigure}%
  \\
  \begin{subfigure}[c]{1\textwidth}
    \centering
    \includegraphics[width=0.5\textwidth]{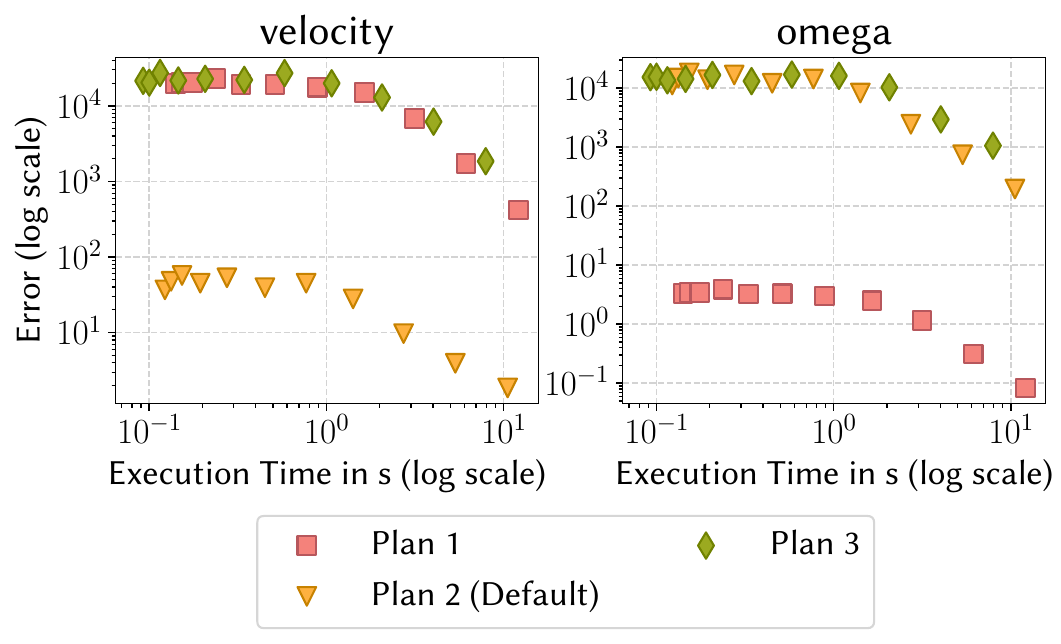}
    \caption{SMC w/ BP.}
  \end{subfigure}%
  \caption{\bWheels{}.}
  \label{fig:mh-performance-results-wheels}
\end{figure}

\begin{figure}[H]
  \centering
  \begin{subfigure}[c]{1\textwidth}
    \centering
    \includegraphics[width=1\textwidth]{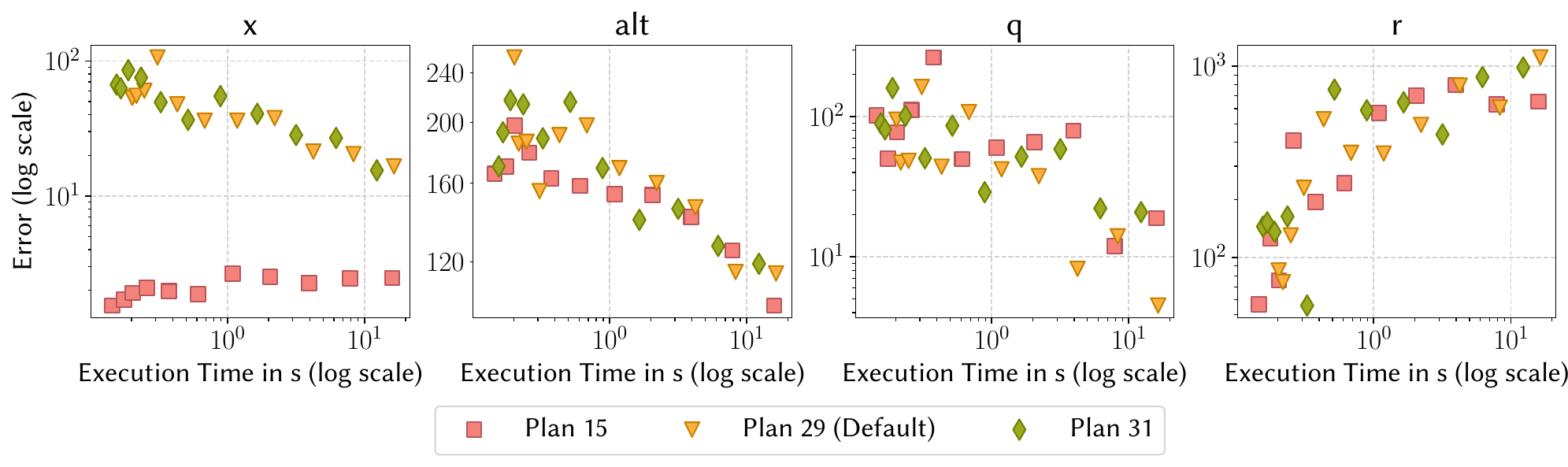}
    \caption{SSI.}
  \end{subfigure}%
  \\
  \begin{subfigure}[c]{1\textwidth}
    \centering
    \includegraphics[width=1\textwidth]{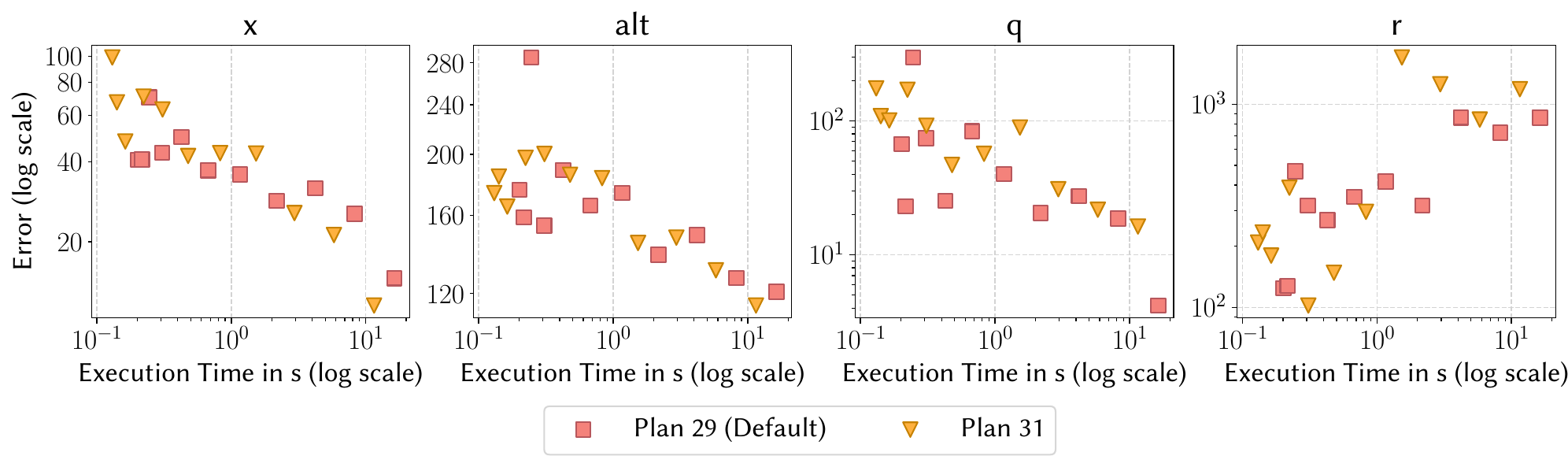}
    \caption{DS.}
  \end{subfigure}%
  \\
  \begin{subfigure}[c]{1\textwidth}
    \centering
    \includegraphics[width=1\textwidth]{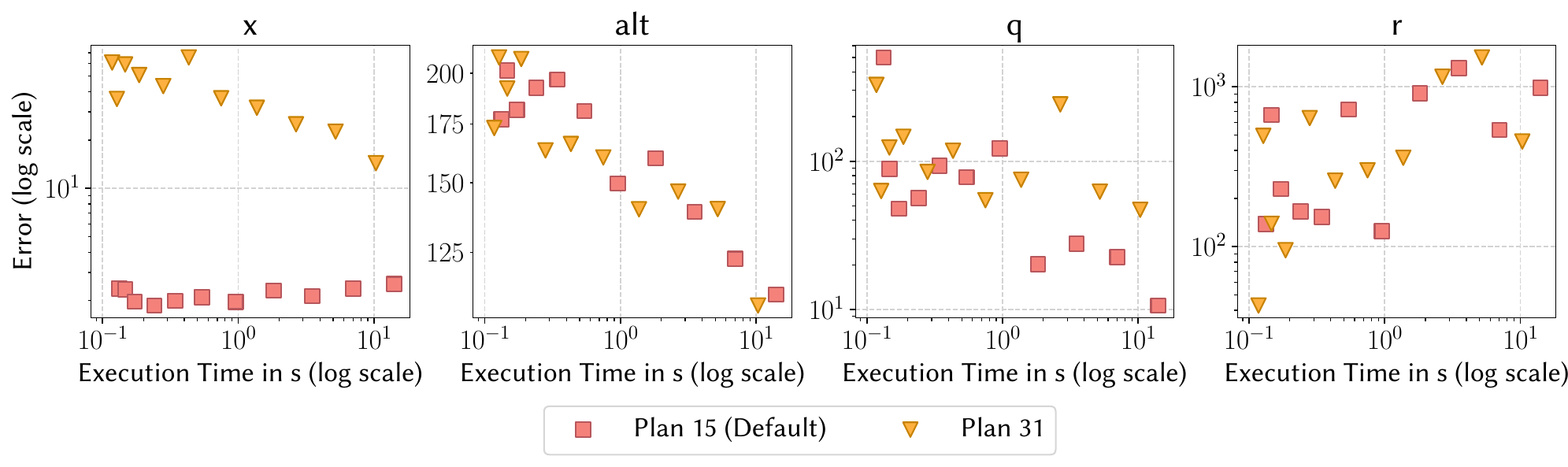}
    \caption{SMC w/ BP.}
  \end{subfigure}%
  \caption{\bAircraft{}. }
  \label{fig:mh-performance-results-aircraft}
\end{figure}

\subsection{Performance Profiles with 10th, 50th, and 90th Perecentile Error}
\begin{figure}[H]
  \centering
  \begin{subfigure}[c]{0.75\textwidth}
    \centering
    \includegraphics[width=1\textwidth]{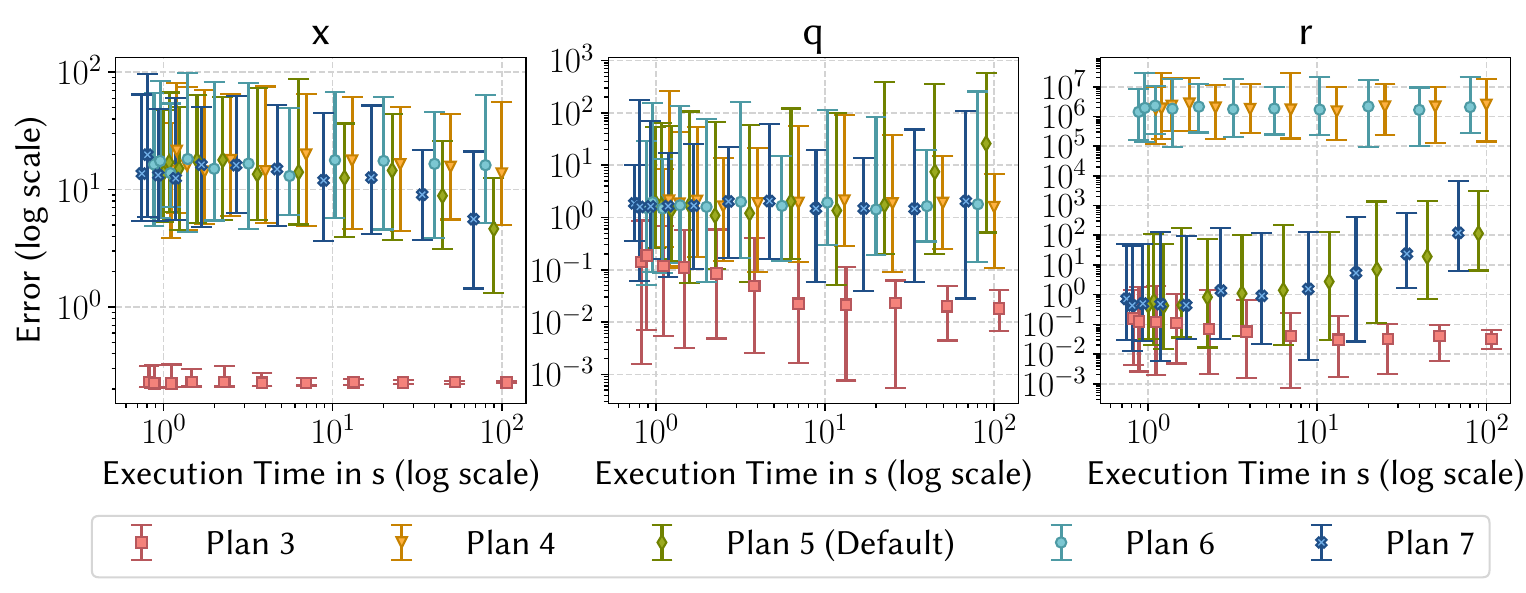}
    \caption{SSI.}
  \end{subfigure}%
  \\
  \begin{subfigure}[c]{0.75\textwidth}
    \centering
    \includegraphics[width=1\textwidth]{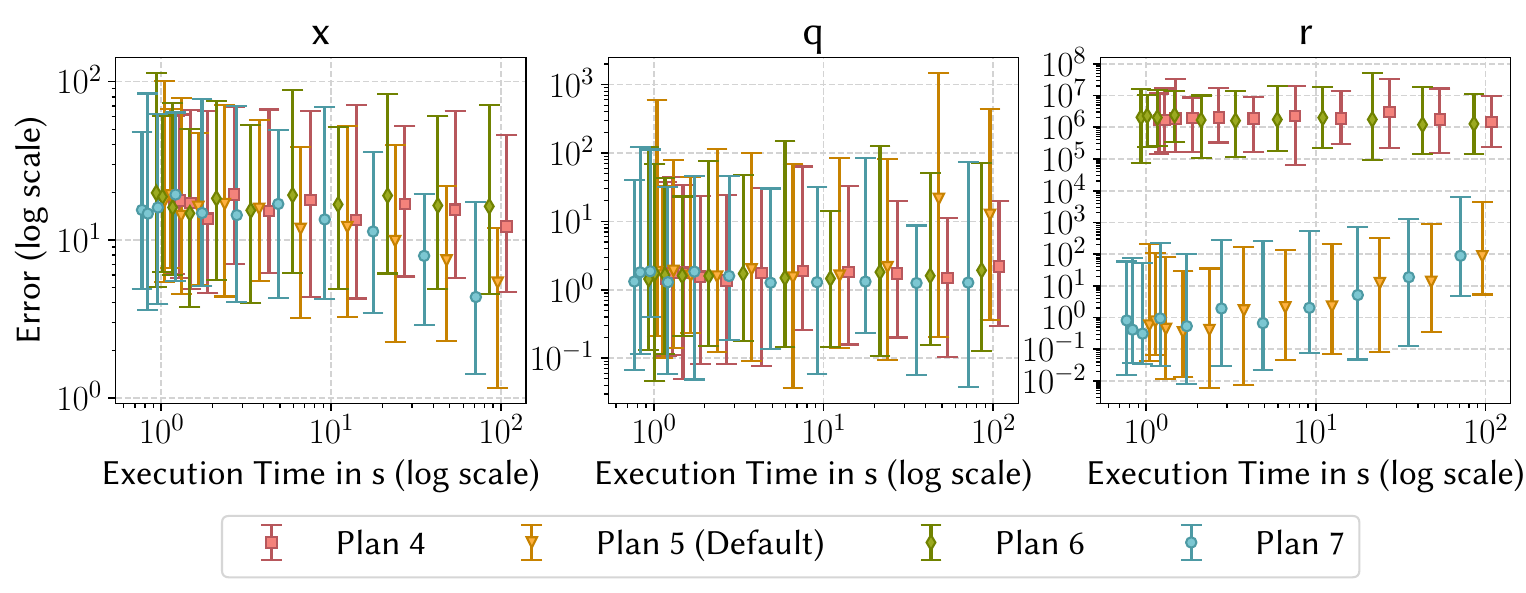}
    \caption{DS.}
  \end{subfigure}%
  \\
  \begin{subfigure}[c]{0.75\textwidth}
    \centering
    \includegraphics[width=1\textwidth]{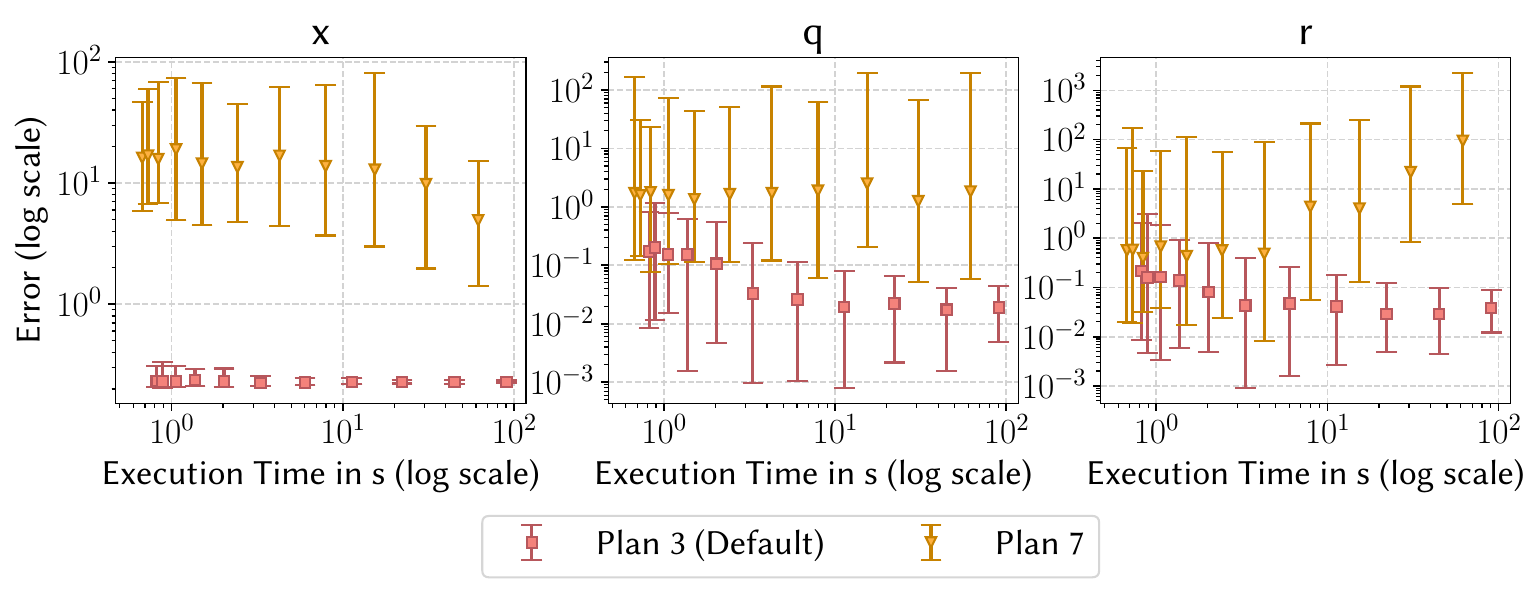}
    \caption{SMC w/ BP.}
  \end{subfigure}%
  \caption{\bNoise{}}
  \label{fig:mh-performance-results-errorbar-noise}
\end{figure}

\begin{figure}[H]
  \centering
  \begin{subfigure}[c]{0.75\textwidth}
    \centering
    \includegraphics[width=1\textwidth]{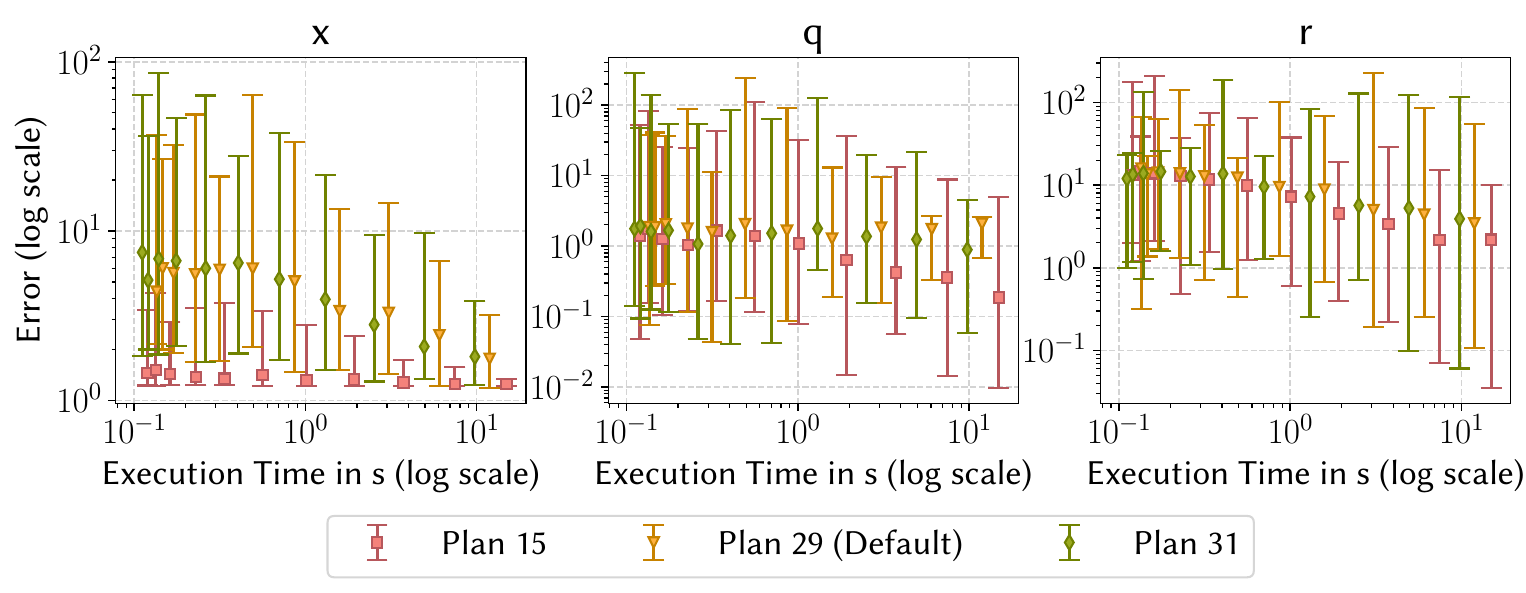}
    \caption{SSI.}
  \end{subfigure}%
  \\
  \begin{subfigure}[c]{0.75\textwidth}
    \centering
    \includegraphics[width=1\textwidth]{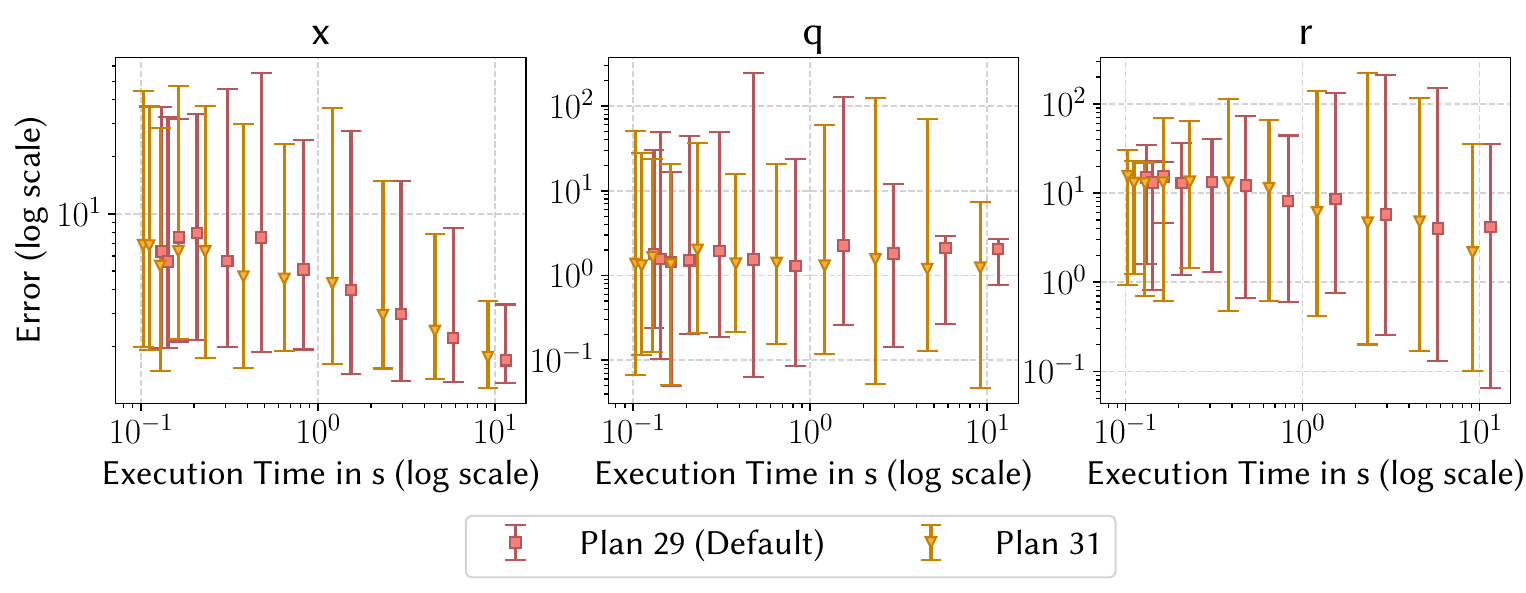}
    \caption{DS.}
  \end{subfigure}%
  \\
  \begin{subfigure}[c]{0.75\textwidth}
    \centering
    \includegraphics[width=1\textwidth]{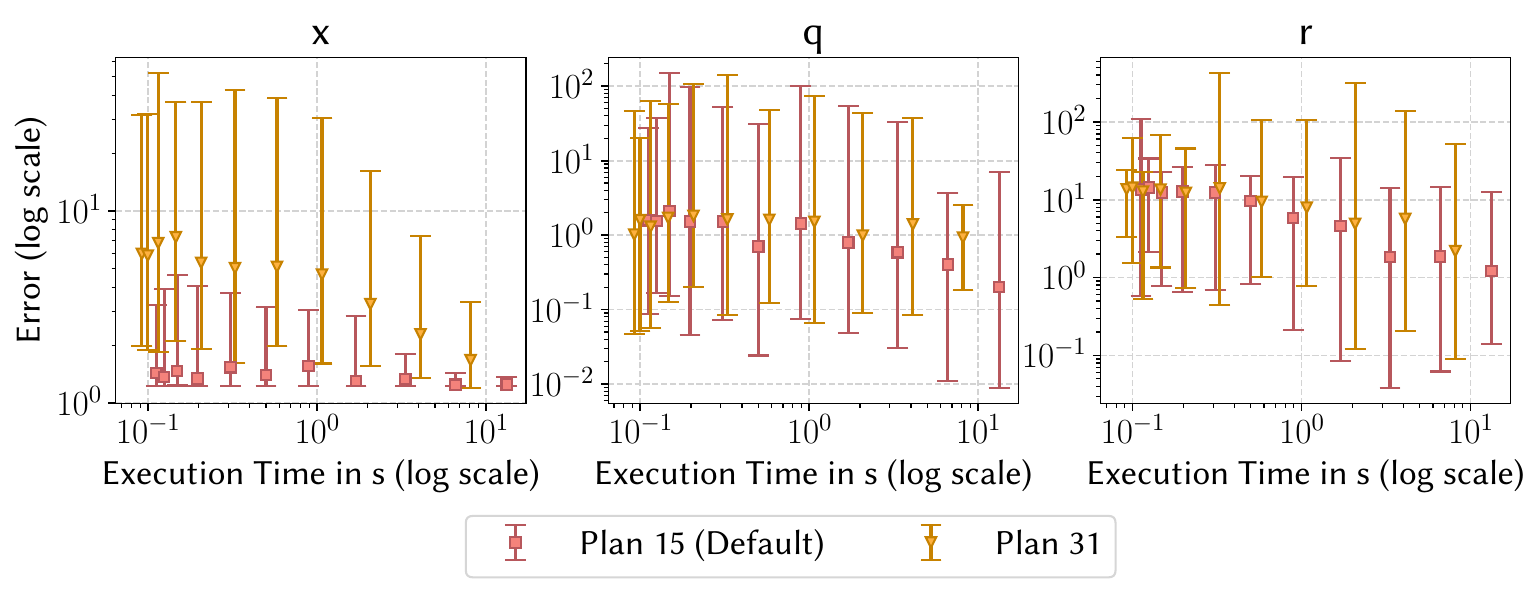}
    \caption{SMC w/ BP.}
  \end{subfigure}%
  \caption{\bRadar{}}
  \label{fig:mh-performance-results-errorbar-radar}
\end{figure}

\begin{figure}[H]
  \centering
  \begin{subfigure}[c]{0.75\textwidth}
    \centering
    \includegraphics[width=1\textwidth]{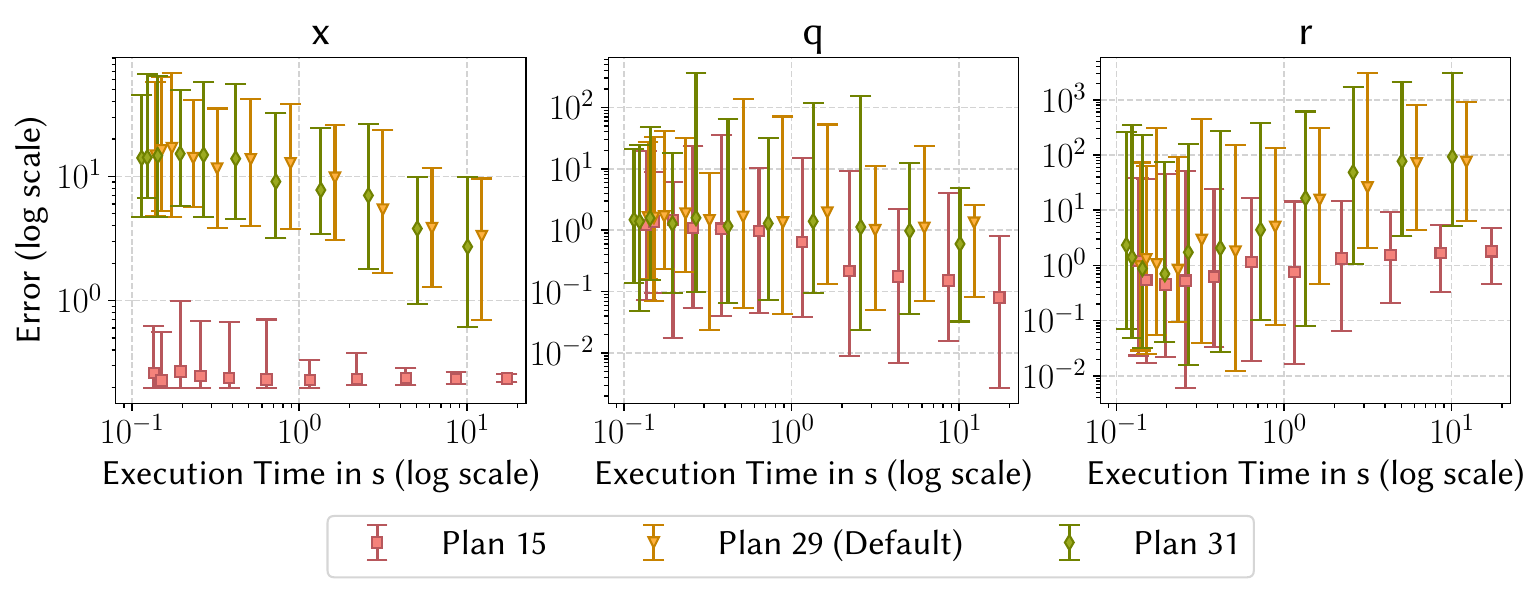}
    \caption{SSI.}
  \end{subfigure}%
  \\
  \begin{subfigure}[c]{0.75\textwidth}
    \centering
    \includegraphics[width=1\textwidth]{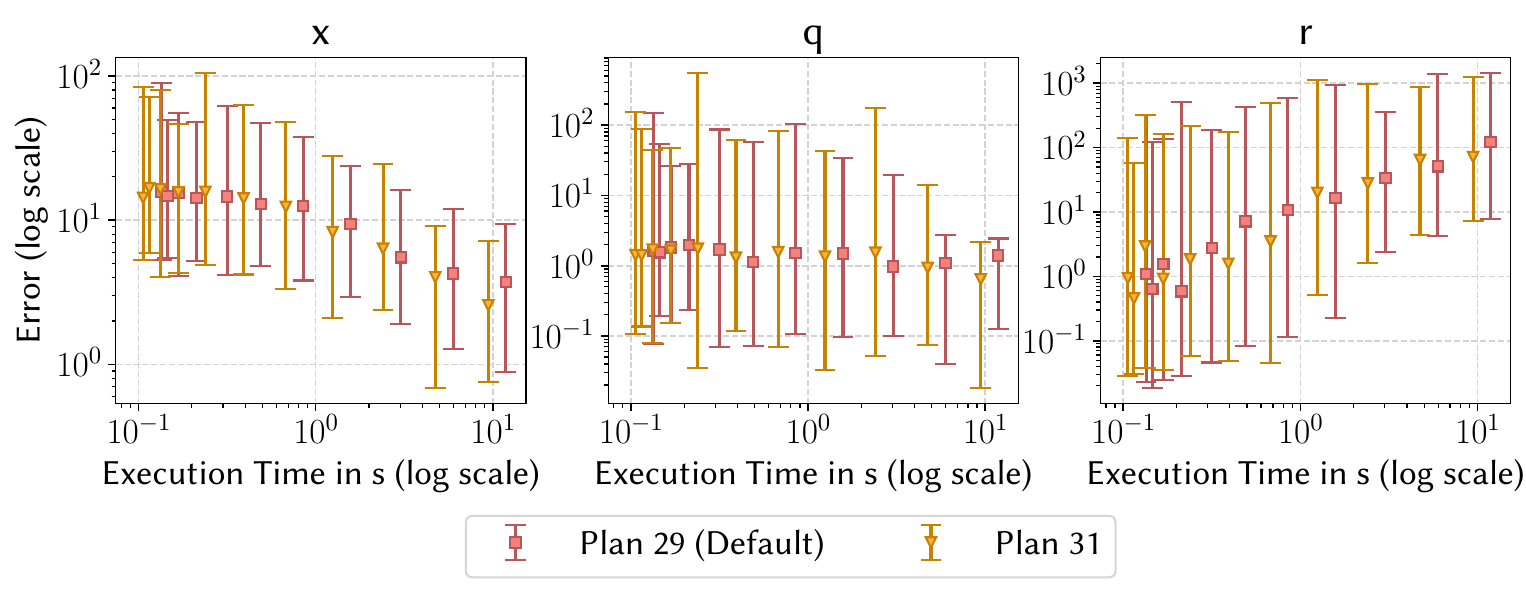}
    \caption{DS.}
  \end{subfigure}%
  \\
  \begin{subfigure}[c]{0.75\textwidth}
    \centering
    \includegraphics[width=1\textwidth]{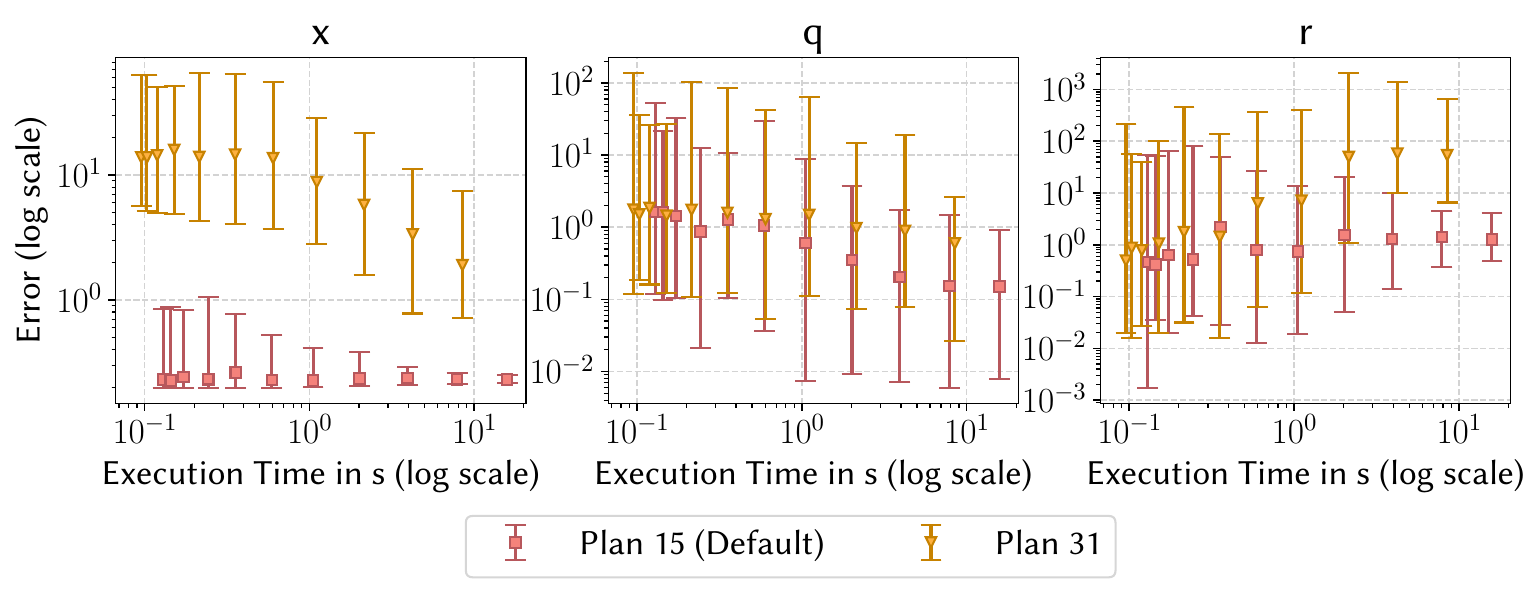}
    \caption{SMC w/ BP.}
  \end{subfigure}%
  \caption{\bEnvnoise{}}
  \label{fig:mh-performance-results-errorbar-envnoise}
\end{figure}

\begin{figure}[H]
  \centering
  \begin{subfigure}[c]{0.5\textwidth}
    \centering
    \includegraphics[width=1\textwidth]{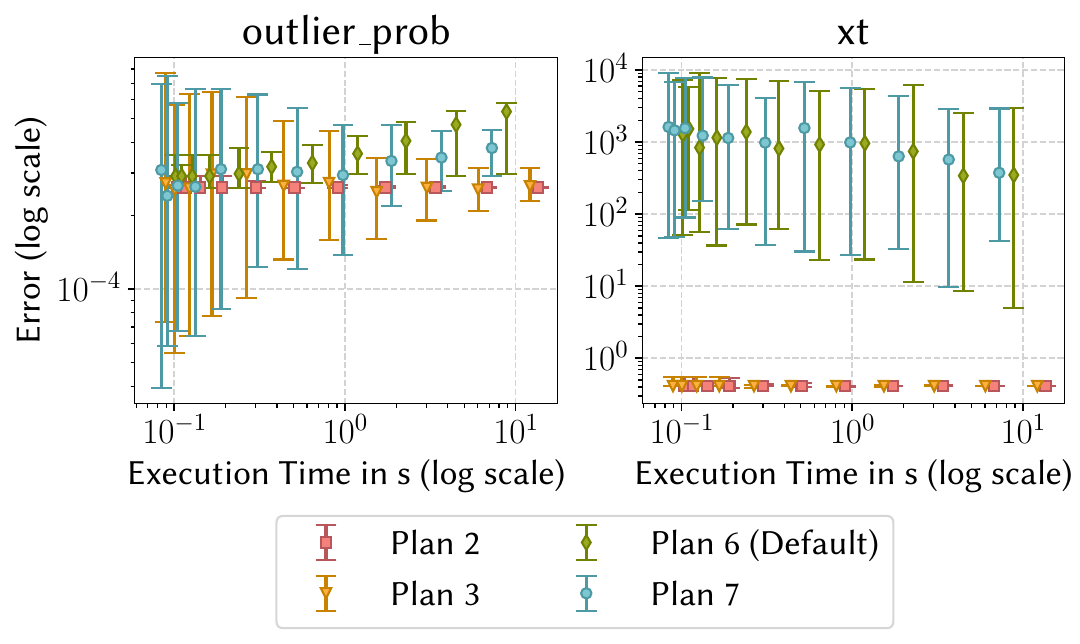}
    \caption{SSI.}
  \end{subfigure}%
  \\
  \begin{subfigure}[c]{0.5\textwidth}
    \centering
    \includegraphics[width=1\textwidth]{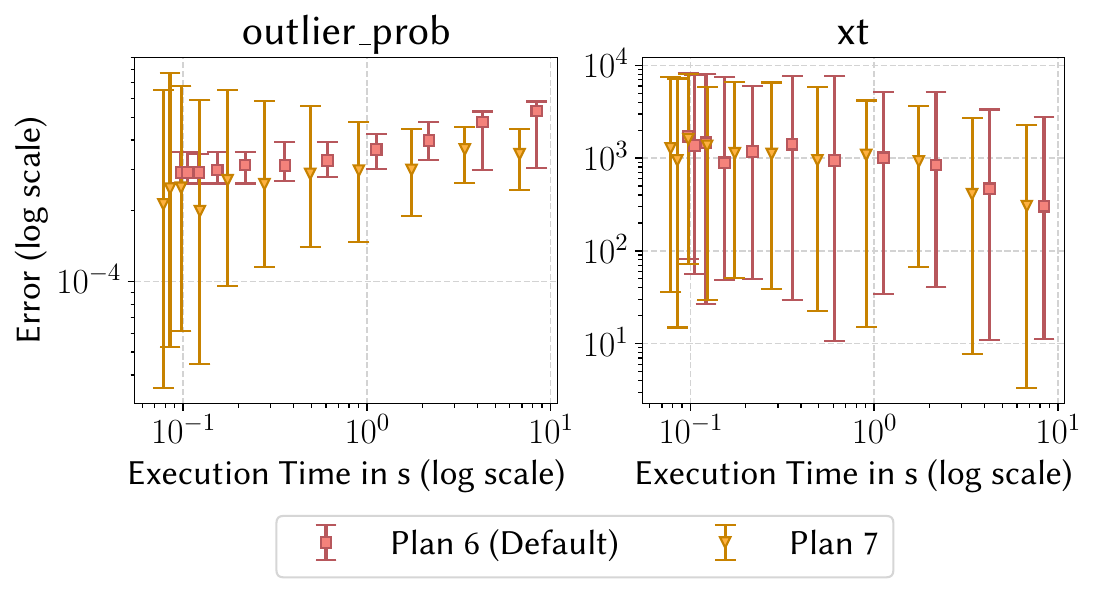}
    \caption{DS.}
  \end{subfigure}%
  \\
  \begin{subfigure}[c]{0.5\textwidth}
    \centering
    \includegraphics[width=0.99\textwidth]{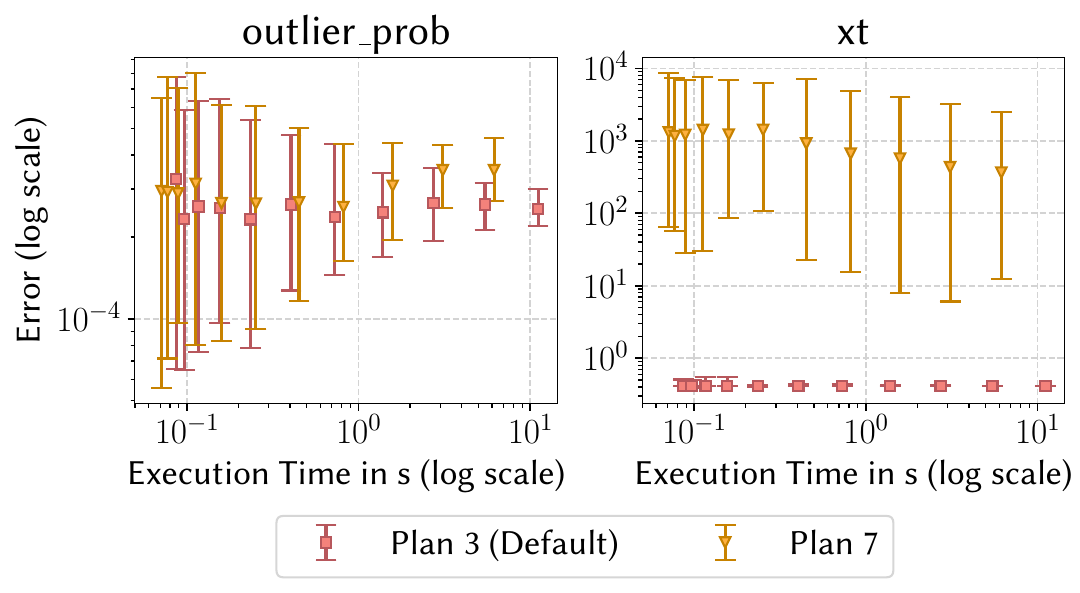}
    \caption{SMC w/ BP.}
  \end{subfigure}%
  \caption{\bOutlier{}}
  \label{fig:mh-performance-results-errorbar-outlier}
\end{figure}

\begin{figure}[H]
  \centering
  \begin{subfigure}[c]{0.5\textwidth}
    \centering
    \includegraphics[width=1\textwidth]{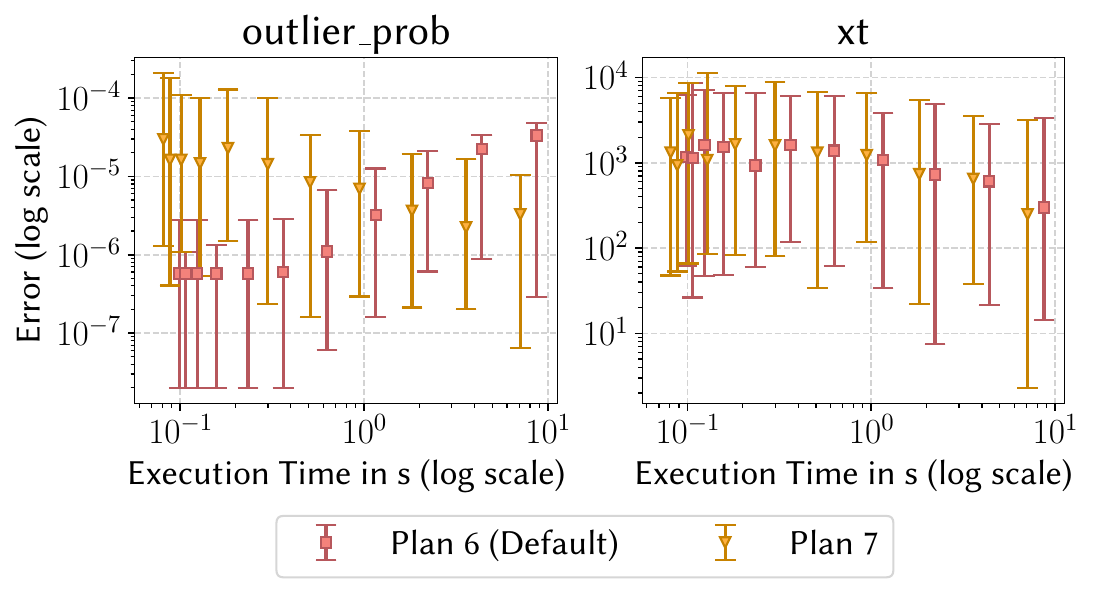}
    \caption{SSI.}
  \end{subfigure}%
  \\
  \begin{subfigure}[c]{0.5\textwidth}
    \centering
    \includegraphics[width=1\textwidth]{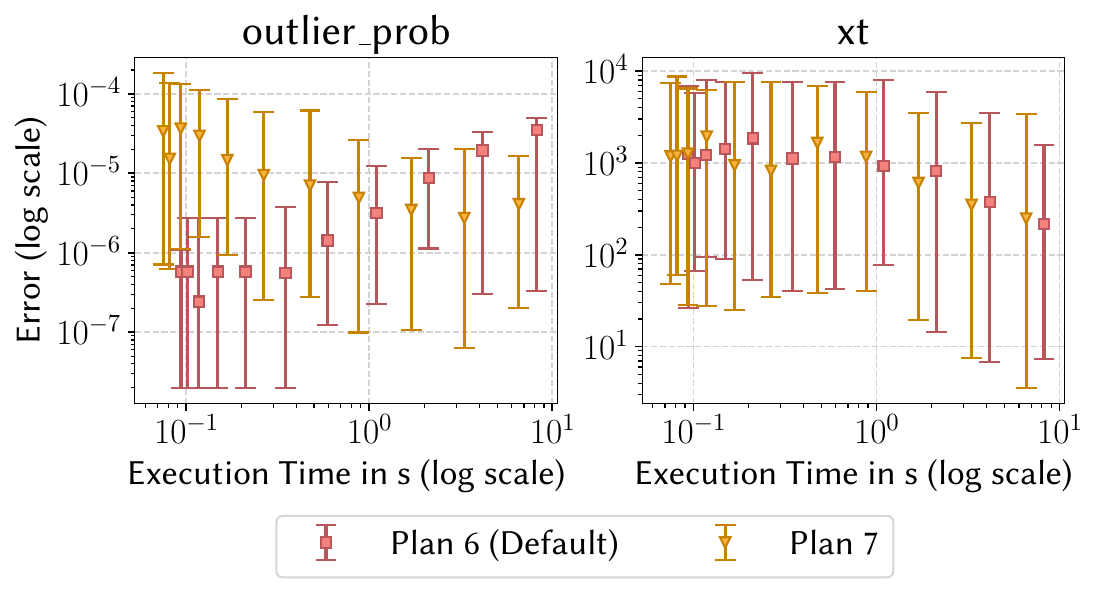}
    \caption{DS.}
  \end{subfigure}%
  \\
  \begin{subfigure}[c]{0.5\textwidth}
    \centering
    \includegraphics[width=1\textwidth]{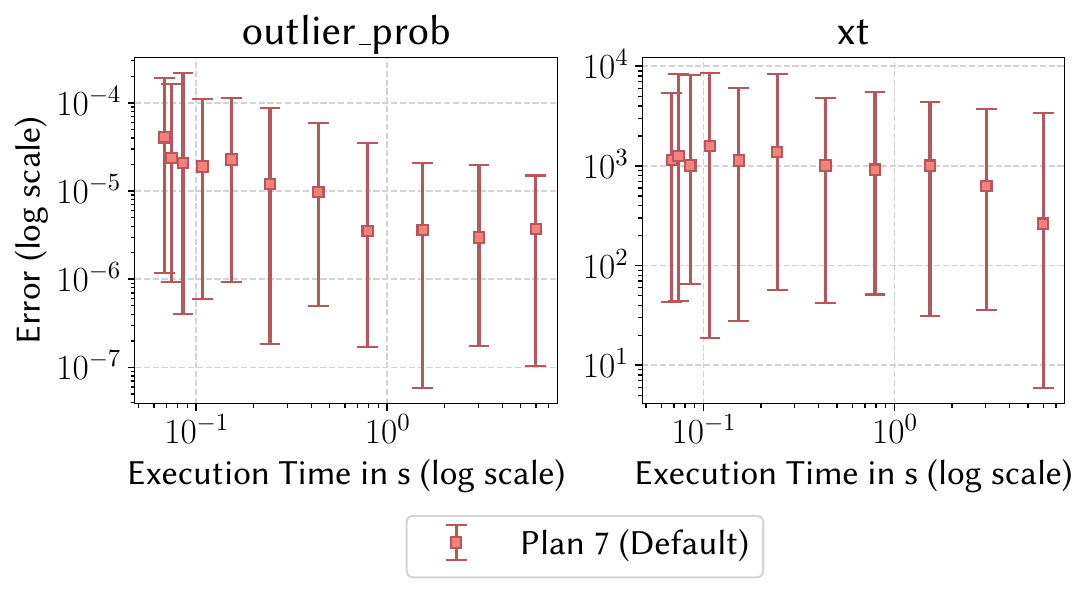}
    \caption{SMC w/ BP.}
  \end{subfigure}%
  \caption{\bOutlierheavy{}.}
  \label{fig:mh-performance-results-errorbar-outlierheavy}
\end{figure}

\begin{figure}[H]
  \centering
  \begin{subfigure}[c]{0.5\textwidth}
    \centering
    \includegraphics[width=1\textwidth]{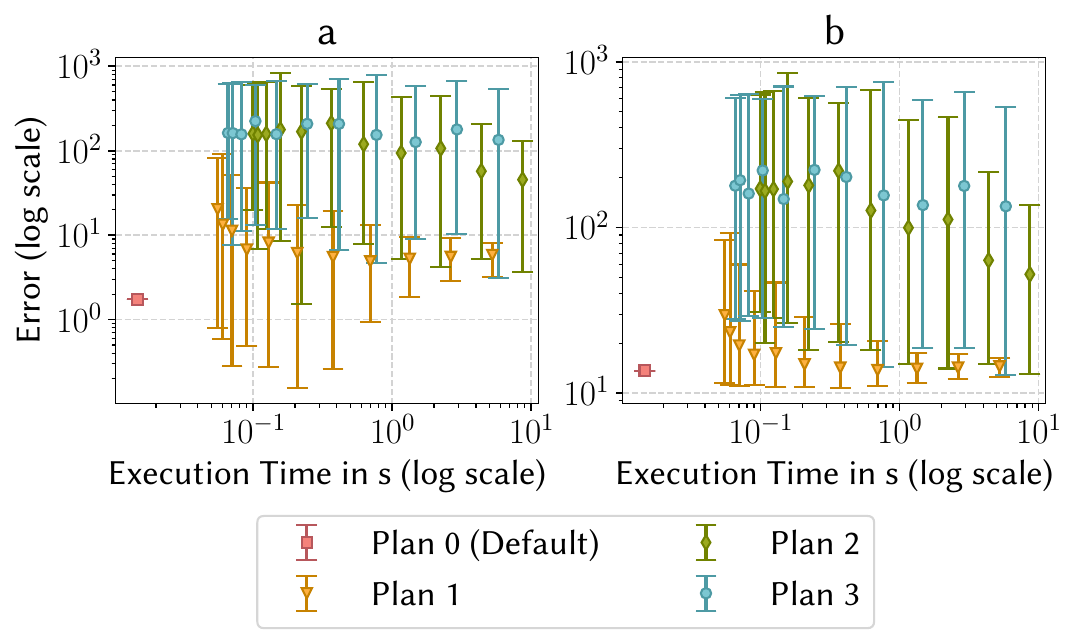}
    \caption{SSI.}
  \end{subfigure}%
  \\
  \begin{subfigure}[c]{0.5\textwidth}
    \centering
    \includegraphics[width=1\textwidth]{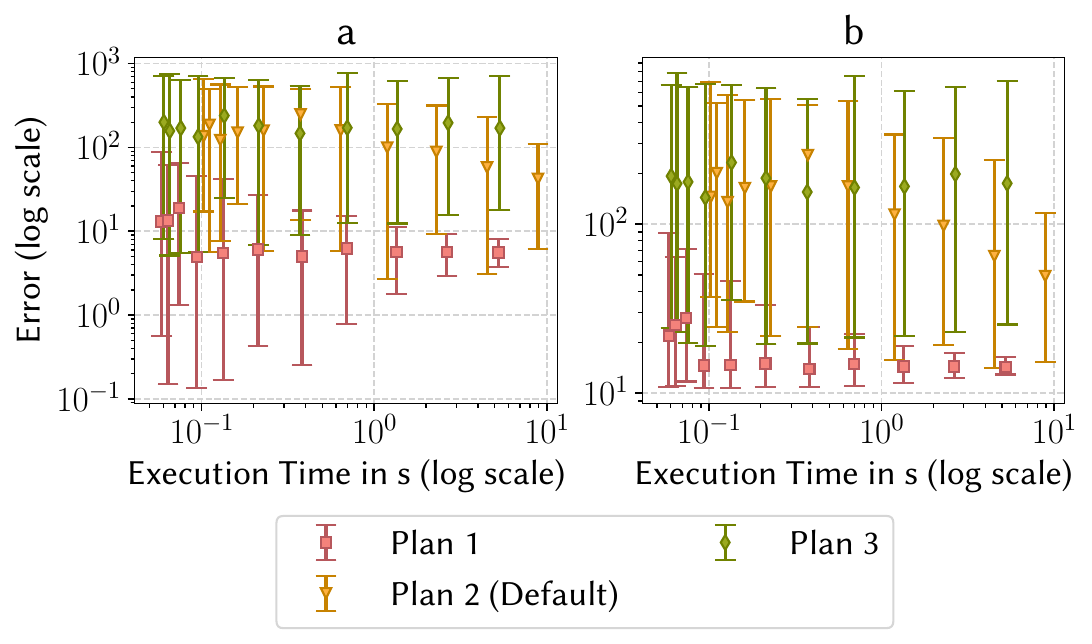}
    \caption{DS.}
  \end{subfigure}%
  \\
  \begin{subfigure}[c]{0.5\textwidth}
    \centering
    \includegraphics[width=1\textwidth]{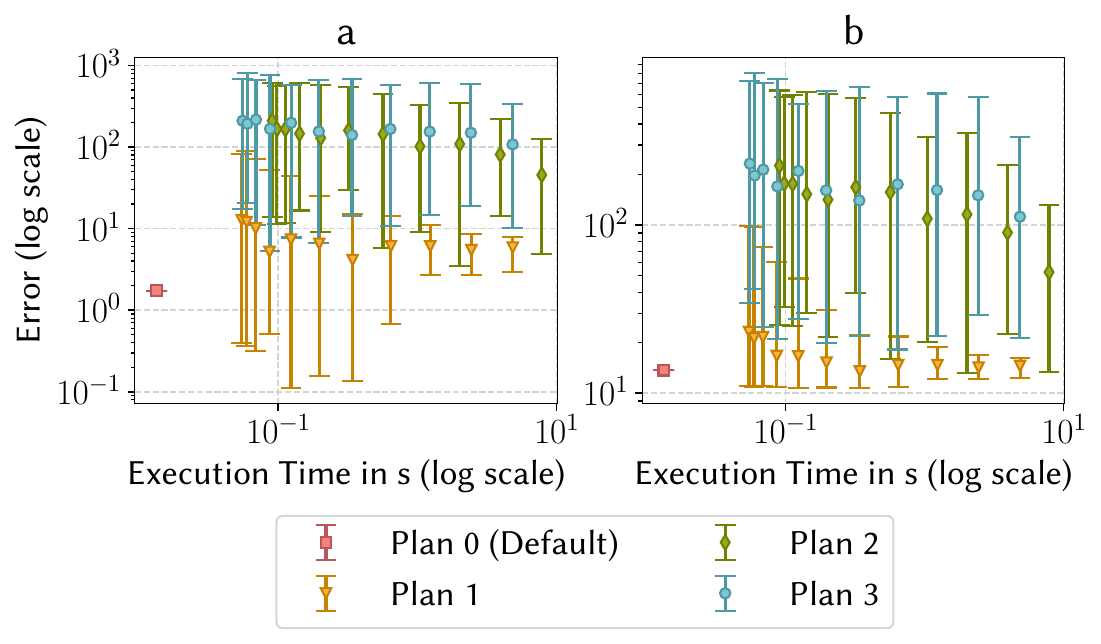}
    \caption{SMC w/ BP.}
  \end{subfigure}%
  \caption{\bGtree{}}
  \label{fig:mh-performance-results-errorbar-tree}
\end{figure}

\begin{figure}[H]
  \centering
  \begin{subfigure}[c]{0.75\textwidth}
    \centering
    \includegraphics[width=1\textwidth]{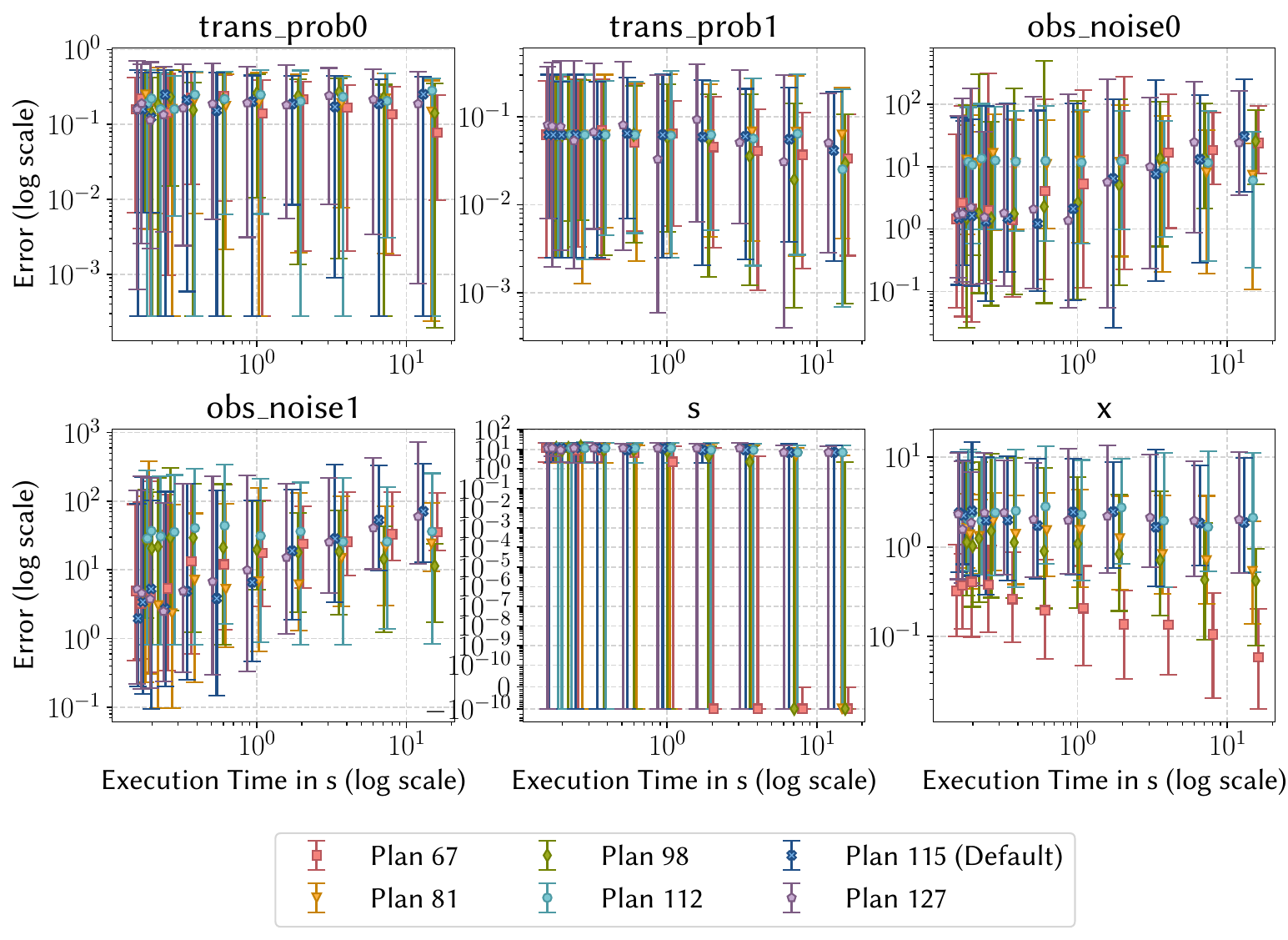}
    \caption{SSI.}
  \end{subfigure}%
  \\
  \begin{subfigure}[c]{0.75\textwidth}
    \centering
    \includegraphics[width=1\textwidth]{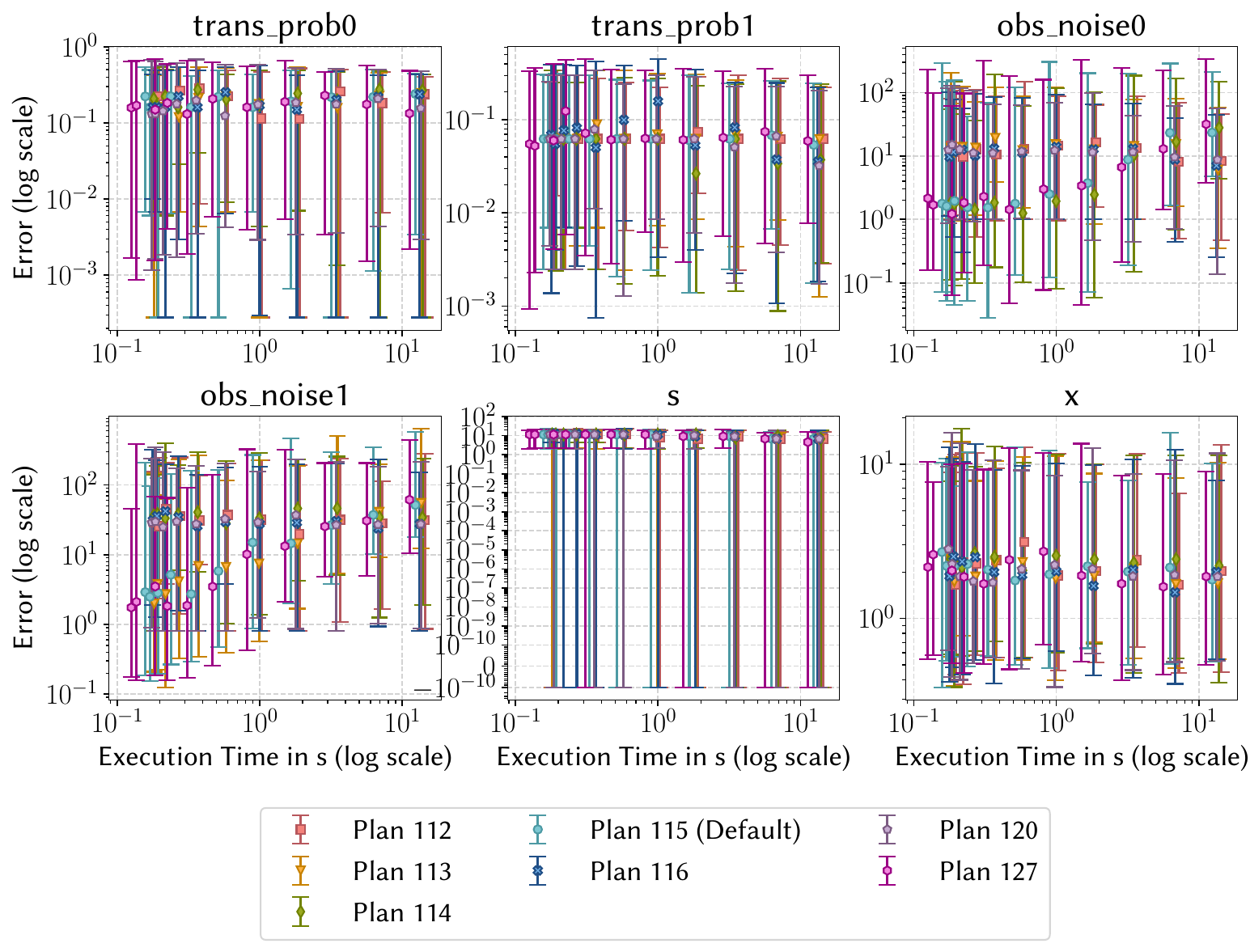}
    \caption{DS.}
  \end{subfigure}%
  \caption{\bSlds{}}
  \label{fig:mh-performance-results-errorbar-slds-1}
\end{figure}

\begin{figure}[H]
  \centering
  \begin{subfigure}[c]{0.75\textwidth}
    \centering
    \includegraphics[width=1\textwidth]{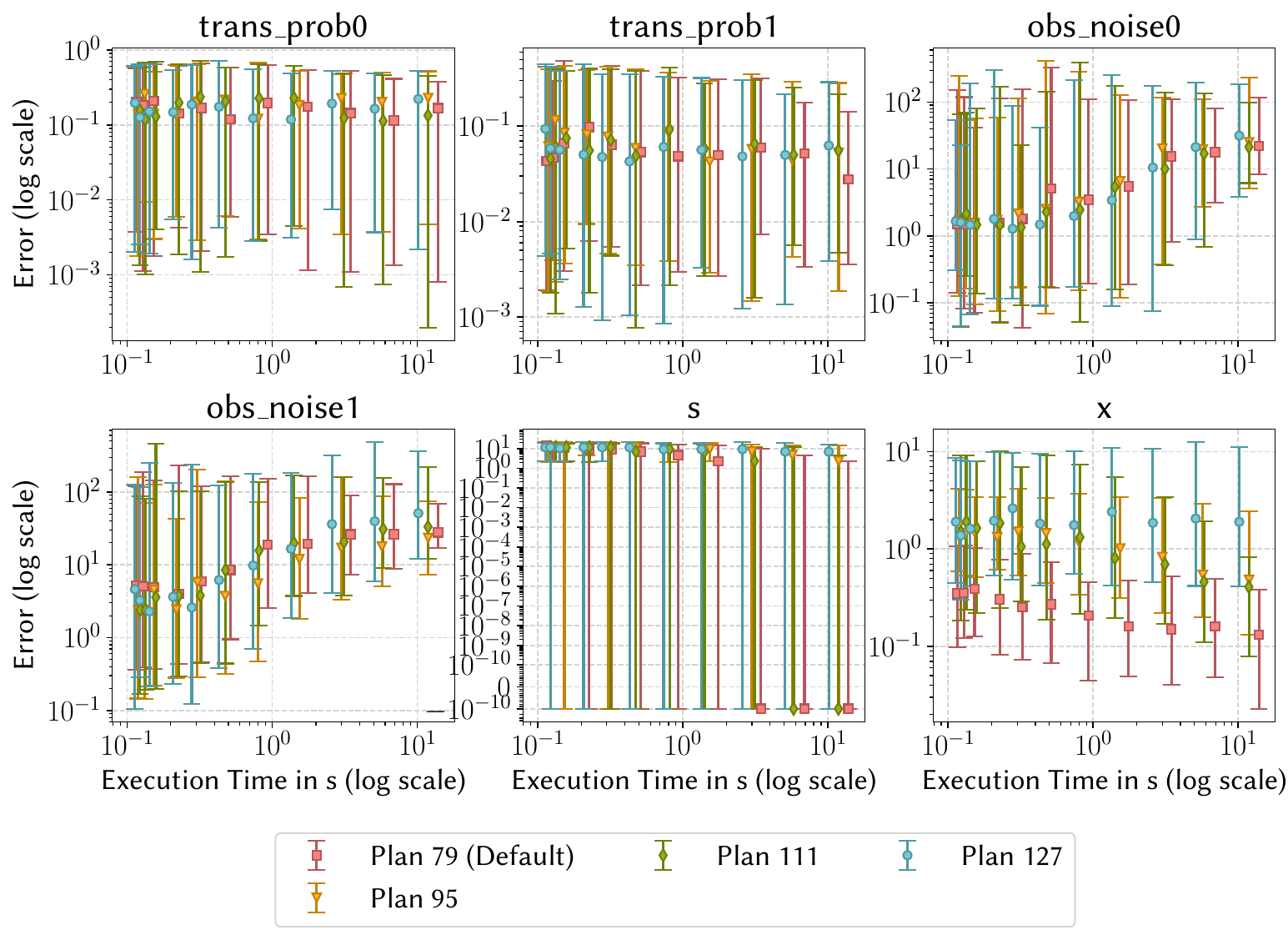}
    \caption{SMC w/ BP.}
  \end{subfigure}%
  \caption{\bSlds{} (continued)}
  \label{fig:mh-performance-results-errorbar-slds-2}
\end{figure}

\begin{figure}[H]
  \centering
  \begin{subfigure}[c]{1\textwidth}
    \centering
    \includegraphics[width=1\textwidth]{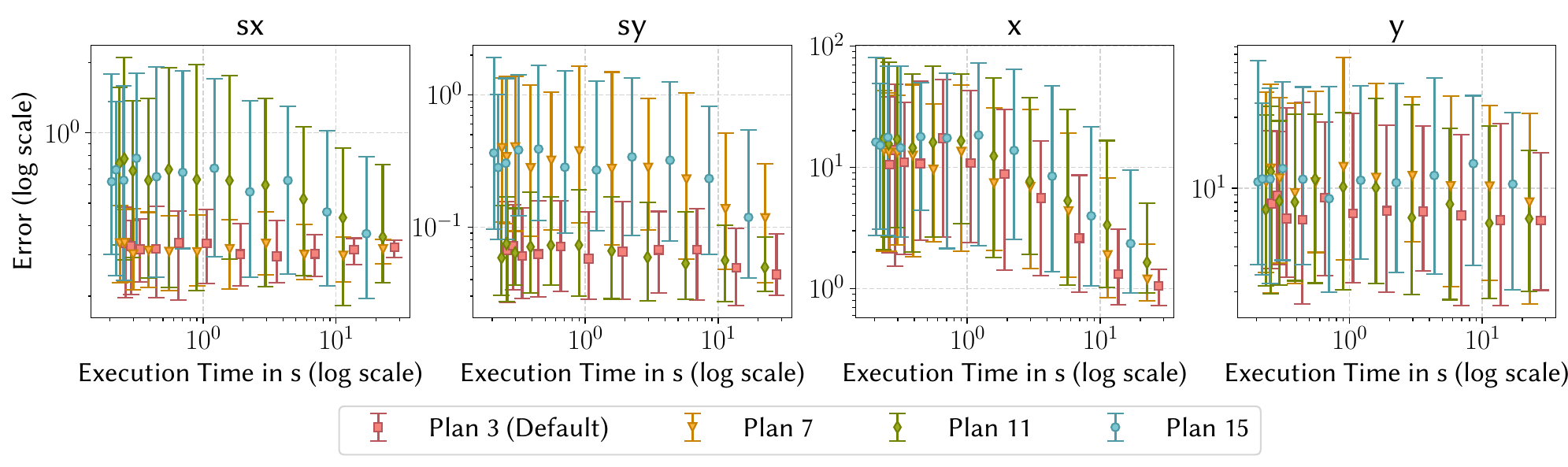}
    \caption{SSI.}
  \end{subfigure}%
  \\
  \begin{subfigure}[c]{1\textwidth}
    \centering
    \includegraphics[width=1\textwidth]{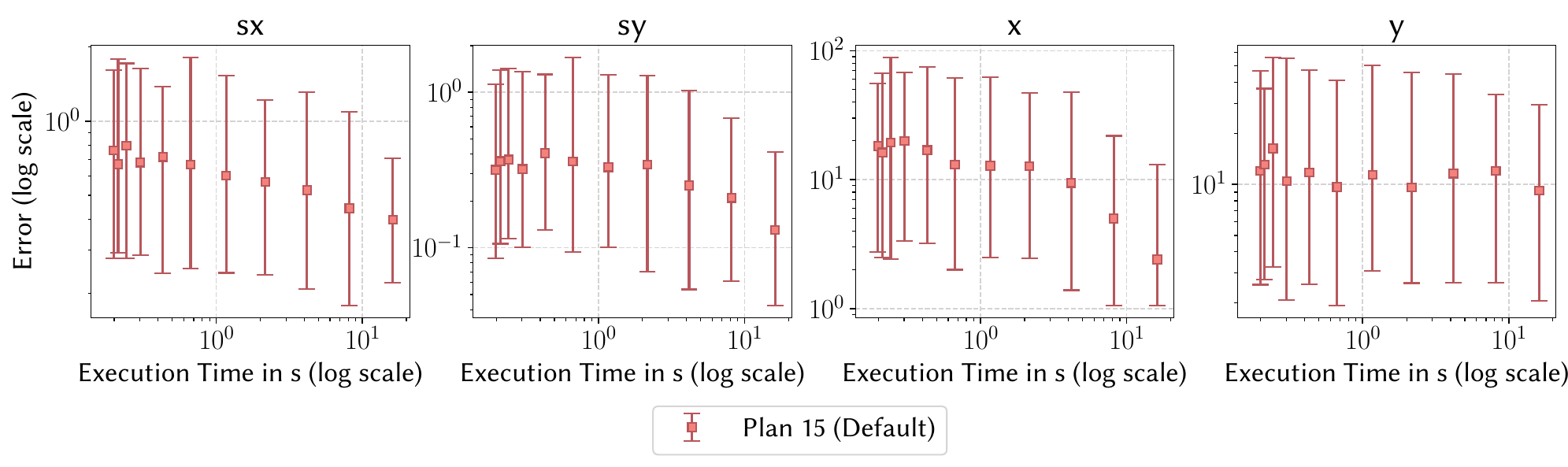}
    \caption{SSI.}
  \end{subfigure}%
  \\
  \begin{subfigure}[c]{1\textwidth}
    \centering
    \includegraphics[width=1\textwidth]{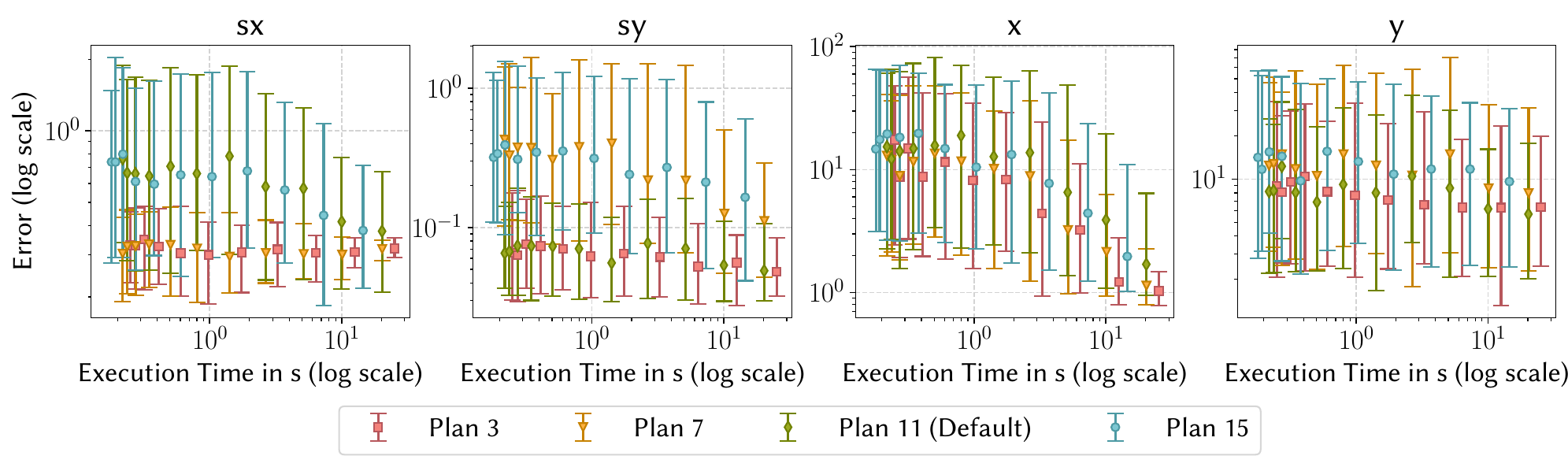}
    \caption{SMC w/ BP.}
  \end{subfigure}%
  \caption{\bRunner{}.}
  \label{fig:mh-performance-results-errorbar-runner}
\end{figure}

\begin{figure}[H]
  \centering
  \begin{subfigure}[c]{1\textwidth}
    \centering
    \includegraphics[width=0.5\textwidth]{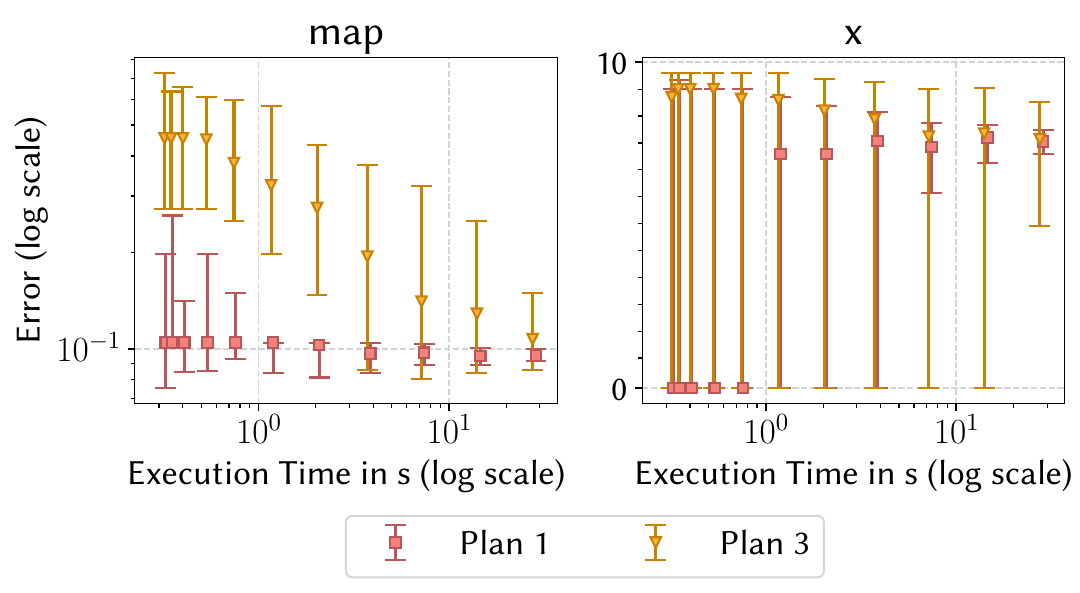}
    \caption{SSI. Plan 0 (Default) and Plan 2 time out for all particles.}
  \end{subfigure}%
  \\
  \begin{subfigure}[c]{1\textwidth}
    \centering
    \includegraphics[width=0.5\textwidth]{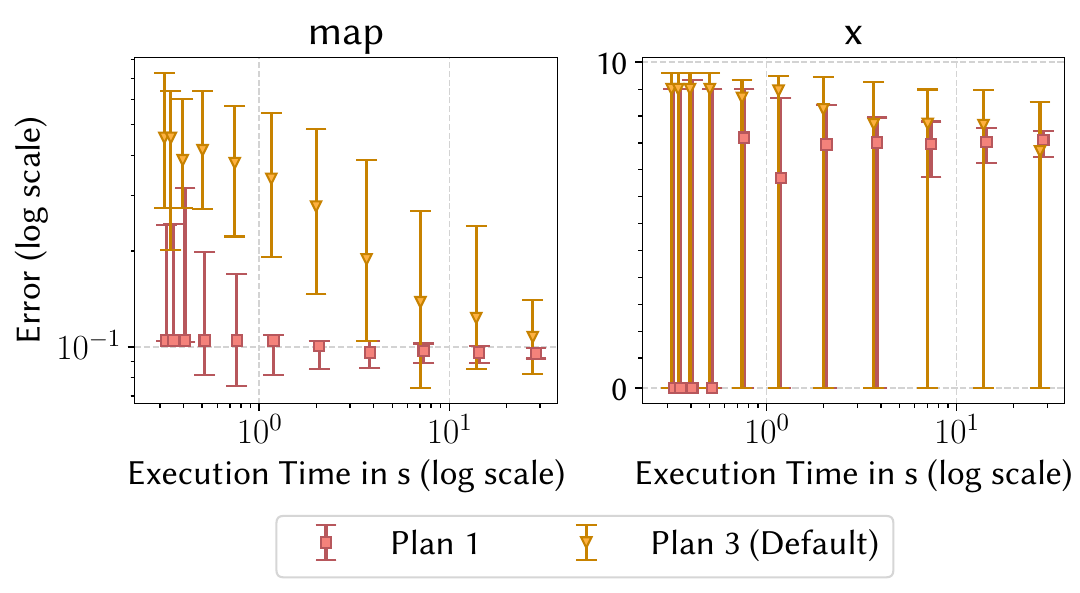}
    \caption{DS.}
  \end{subfigure}%
  \\
  \begin{subfigure}[c]{1\textwidth}
    \centering
    \includegraphics[width=0.5\textwidth]{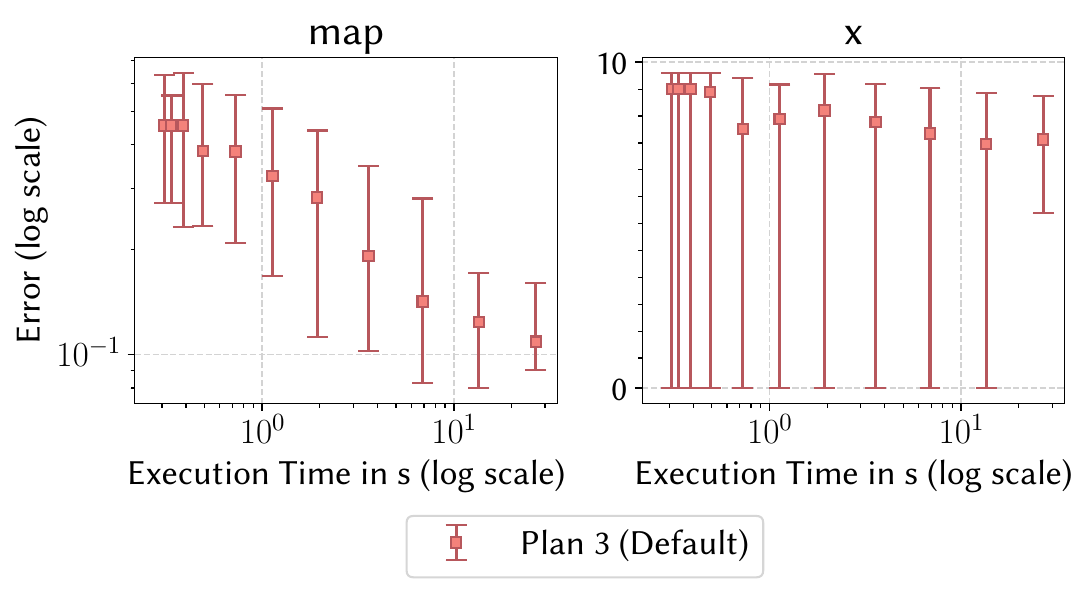}
    \caption{SMC w/ BP.}
  \end{subfigure}%
  \caption{\bSlam{}.}
  \label{fig:mh-performance-results-errorbar-slam}
\end{figure}

\begin{figure}[H]
  \centering
  \begin{subfigure}[c]{1\textwidth}
    \centering
    \includegraphics[width=0.5\textwidth]{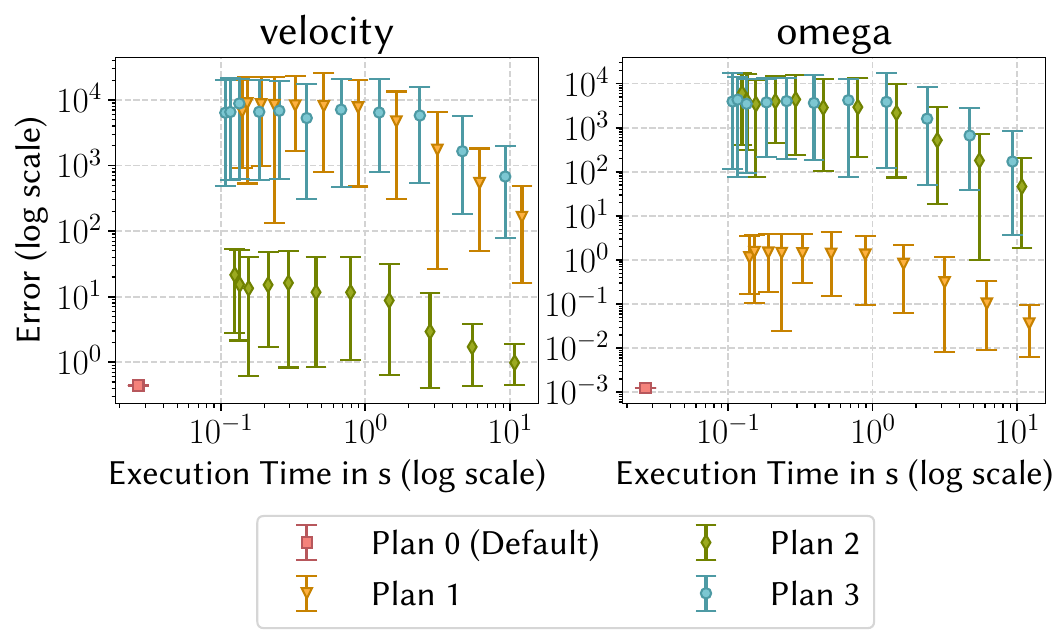}
    \caption{SSI. }
  \end{subfigure}%
  \\
  \begin{subfigure}[c]{1\textwidth}
    \centering
    \includegraphics[width=0.5\textwidth]{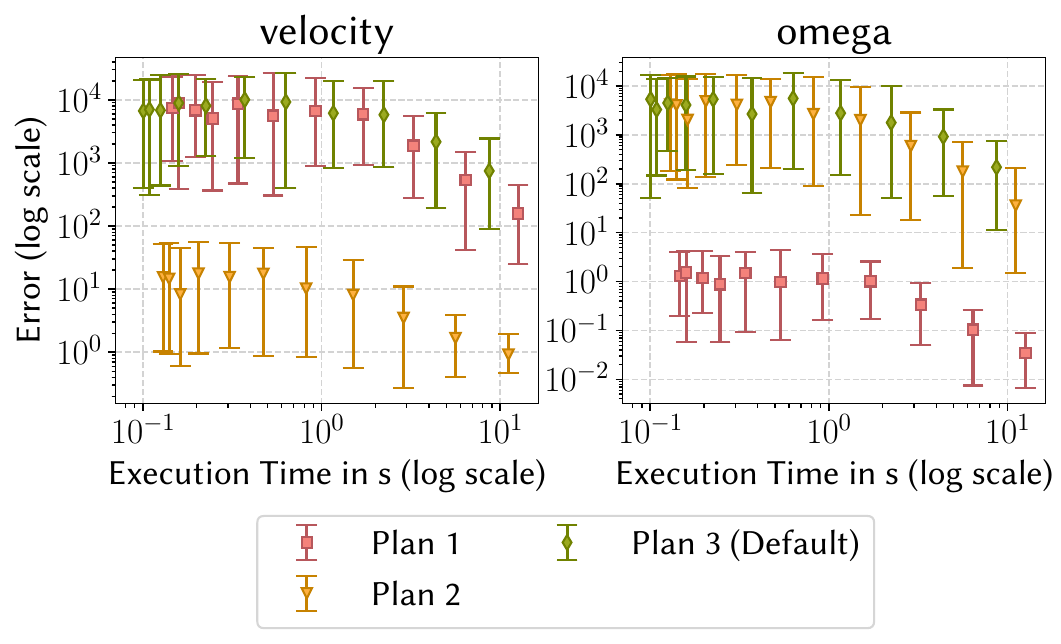}
    \caption{DS.}
  \end{subfigure}%
  \\
  \begin{subfigure}[c]{1\textwidth}
    \centering
    \includegraphics[width=0.5\textwidth]{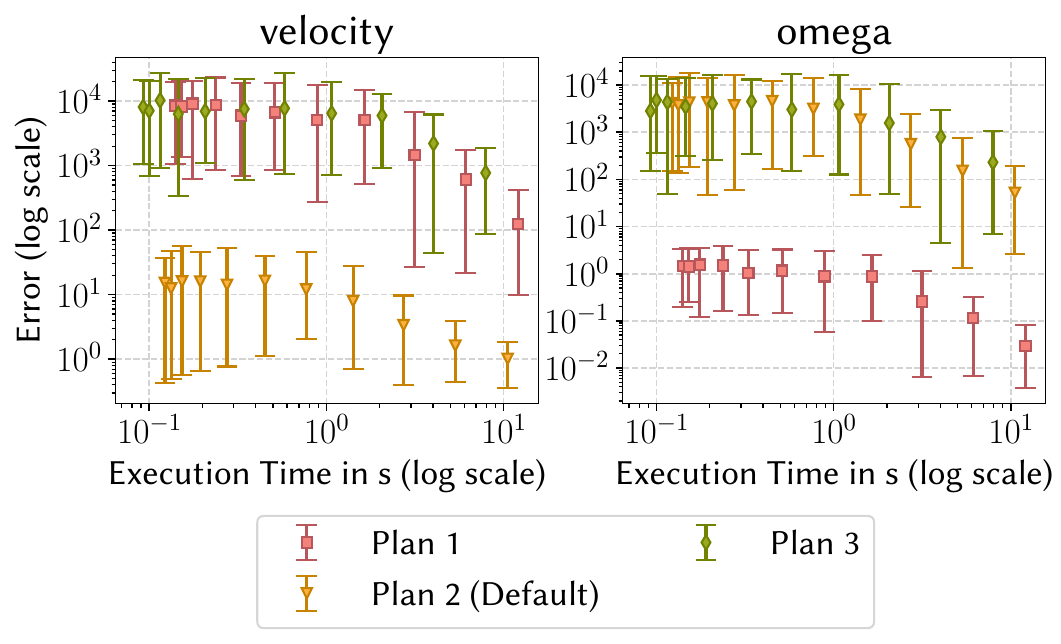}
    \caption{SMC w/ BP.}
  \end{subfigure}%
  \caption{\bWheels{}.}
  \label{fig:mh-performance-results-errorbar-wheels}
\end{figure}

\begin{figure}[H]
  \centering
  \begin{subfigure}[c]{1\textwidth}
    \centering
    \includegraphics[width=1\textwidth]{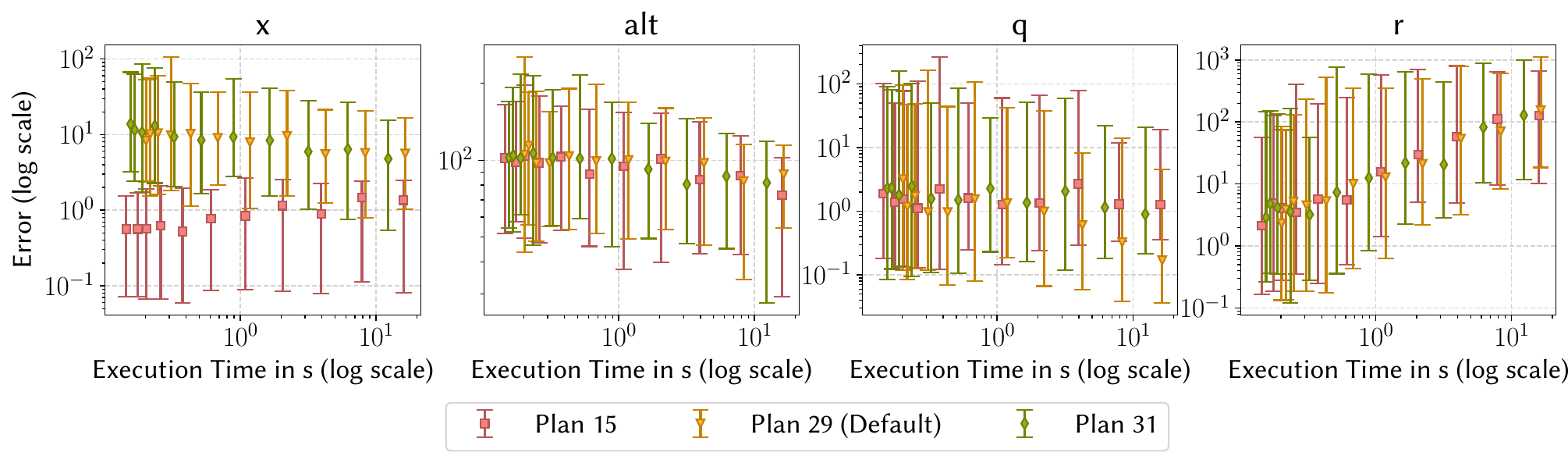}
    \caption{SSI.}
  \end{subfigure}%
  \\
  \begin{subfigure}[c]{1\textwidth}
    \centering
    \includegraphics[width=1\textwidth]{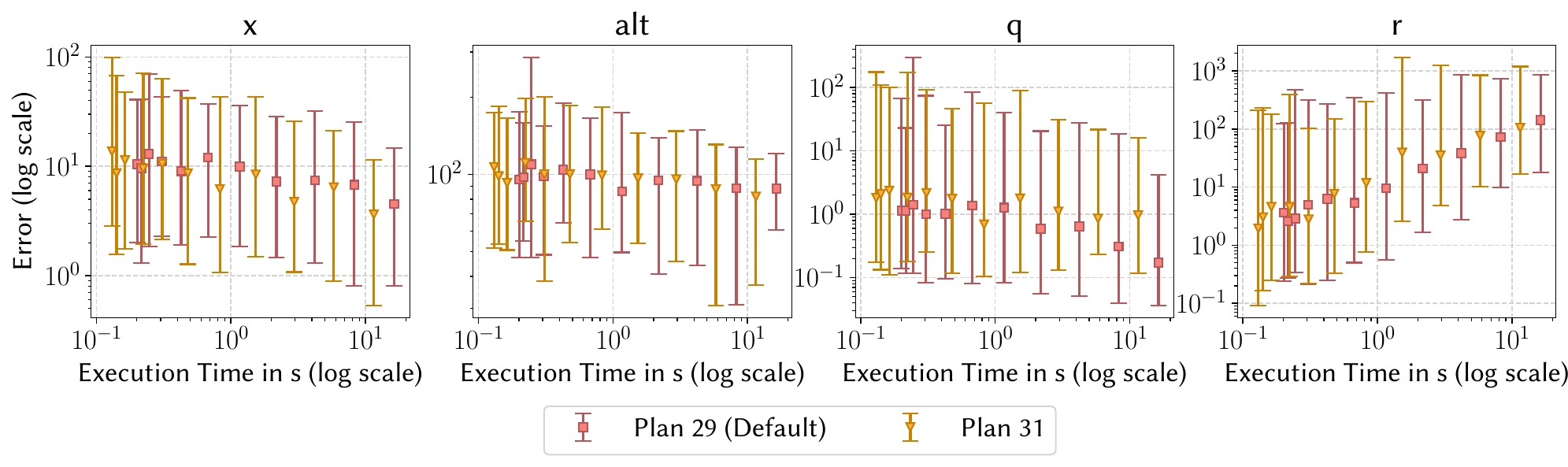}
    \caption{DS.}
  \end{subfigure}%
  \\
  \begin{subfigure}[c]{1\textwidth}
    \centering
    \includegraphics[width=1\textwidth]{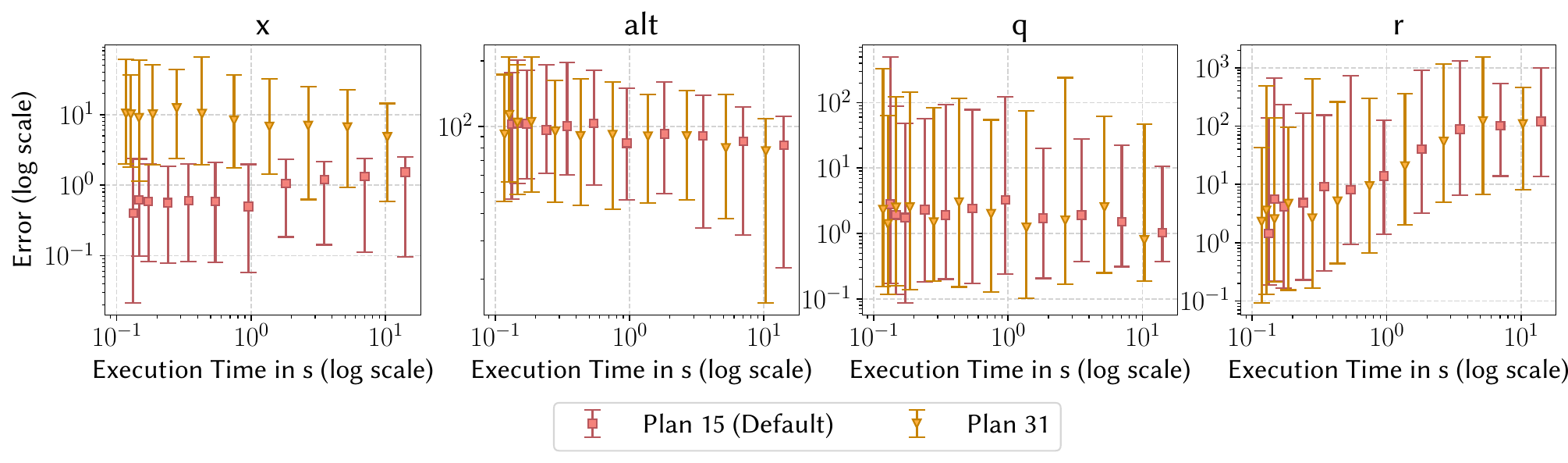}
    \caption{SMC w/ BP.}
  \end{subfigure}%
  \caption{\bAircraft{}. }
  \label{fig:mh-performance-results-errorbar-aircraft}
\end{figure}

\subsection{Execution Time to Reach Target Accuracy}
\begin{table}[H]
  \small
  \centering
  \caption{Best execution time by any satisfiable inference plan and the best execution time by the default inference plan for 90\% of the 100 executions to reach target accuracy, defined as the 90th percentile of error by the default inference plan using greatest number of particles evaluated that did not produce timeouts. Following \citet{atkinson2022semi} and \citet{baudart2020reactive}, reaching target accuracy is defined as $\log(P_{90\%}(\mit{loss})) - \log(\mit{loss}_{\mit{target}}) < 0.5$. -- indicates the default plan timeouts for all particle numbers. \textbf{Bold} marks when an alternative plan performs better than the default plan.}
  \begin{tabular}[h]{llrrrrrr}
    \toprule
    & & \multicolumn{2}{c}{\ssi} & \multicolumn{2}{c}{\ds} & \multicolumn{2}{c}{\bp} \\
     \cmidrule(lr){3-4} \cmidrule(lr){5-6} \cmidrule(lr){7-8}
    Benchmark & Variable & Best (s) & Default (s) & Best (s) & Default (s) & Best (s) & Default (s) \\
\midrule
\multirow{5}{*}{\bNoise{}}& \mkwm{x} & \textbf{0.82} & 11.75 & \textbf{17.68} & 47.89 & 0.82 & 0.82 \\[-0.4em]
& & \plan{3}& & \plan{7}& & \defaultplan{} \\
& \mkwm{q} & \textbf{0.75} & 1.0 & \textbf{0.77} & 1.05 & 6.03 & 6.03 \\[-0.4em]
& & \plan{7}& & \plan{7}& & \defaultplan{} \\
& \mkwm{r} & \textbf{0.75} & 1.0 & \textbf{0.77} & 1.05 & 6.03 & 6.03 \\[-0.4em]
& & \plan{7}& & \plan{7}& & \defaultplan{} \\
\midrule
\multirow{5}{*}{\bRadar{}}& \mkwm{x} & \textbf{0.12} & 6.07 & \textbf{4.59} & 5.81 & 0.11 & 0.11 \\[-0.4em]
& & \plan{15}& & \plan{31}& & \defaultplan{} \\
& \mkwm{q} & 6.07 & 6.07 & 5.81 & 5.81 & \textbf{0.1} & 6.61 \\[-0.4em]
& & \defaultplan{} & & \defaultplan{} & & \plan{31}\\
& \mkwm{r} & \textbf{0.11} & 0.14 & \textbf{0.1} & 0.13 & \textbf{0.09} & 0.12 \\[-0.4em]
& & \plan{31}& & \plan{31}& & \plan{31}\\
\midrule
\multirow{5}{*}{\bEnvnoise{}}& \mkwm{x} & \textbf{0.14} & 1.63 & \textbf{1.25} & 1.58 & 0.36 & 0.36 \\[-0.4em]
& & \plan{15}& & \plan{31}& & \defaultplan{} \\
& \mkwm{q} & \textbf{0.2} & 12.27 & 5.97 & 5.97 & 3.93 & 3.93 \\[-0.4em]
& & \plan{15}& & \defaultplan{} & & \defaultplan{} \\
& \mkwm{r} & \textbf{0.11} & 0.14 & \textbf{0.11} & 0.13 & 3.93 & 3.93 \\[-0.4em]
& & \plan{31}& & \plan{31}& & \defaultplan{} \\
\midrule
\multirow{3}{*}{\bOutlier{}}& \mkwm{outlier\_prob} & \textbf{0.08} & 0.1 & \textbf{0.08} & 0.1 & \textbf{0.07} & 0.09 \\[-0.4em]
& & \plan{7}& & \plan{7}& & \plan{7}\\
& \mkwm{xt} & \textbf{0.08} & 0.1 & \textbf{0.08} & 0.1 & 0.09 & 0.09 \\[-0.4em]
& & \plan{7}& & \plan{7}& & \defaultplan{} \\
\midrule
\multirow{3}{*}{\bOutlierheavy{}}& \mkwm{outlier\_prob} & 0.1 & 0.1 & \textbf{0.08} & 0.09 & 0.79 & 0.79 \\[-0.4em]
& & \defaultplan{} & & \plan{7}& & \defaultplan{} \\
& \mkwm{xt} & \textbf{0.08} & 0.1 & \textbf{1.7} & 4.17 & 0.07 & 0.07 \\[-0.4em]
& & \plan{7}& & \plan{7}& & \defaultplan{} \\
\midrule
\multirow{3}{*}{\bGtree{}}& \mkwm{a} & 0.01 & 0.01 & \textbf{0.06} & 1.19 & 0.01 & 0.01 \\[-0.4em]
& & \defaultplan{} & & \plan{1}& & \defaultplan{} \\
& \mkwm{b} & 0.01 & 0.01 & \textbf{0.06} & 1.19 & 0.01 & 0.01 \\[-0.4em]
& & \defaultplan{} & & \plan{1}& & \defaultplan{} \\
\midrule
\multirow{10}{*}{\bSlds{}}& \mkwm{trans\_prob0} & 0.16 & 0.16 & \textbf{0.13} & 0.16 & \textbf{0.11} & 0.12 \\[-0.4em]
& & \plan{67}& & \plan{127}& & \plan{127}\\
& \mkwm{trans\_prob1} & 0.16 & 0.16 & \textbf{0.13} & 0.16 & 0.12 & 0.12 \\[-0.4em]
& & \plan{67}& & \plan{127}& & \defaultplan{} \\
& \mkwm{obs\_noise0} & 0.16 & 0.16 & \textbf{0.13} & 0.16 & \textbf{0.11} & 0.12 \\[-0.4em]
& & \plan{67}& & \plan{127}& & \plan{127}\\
& \mkwm{obs\_noise1} & 0.16 & 0.16 & \textbf{0.13} & 0.16 & \textbf{0.11} & 0.12 \\[-0.4em]
& & \plan{67}& & \plan{127}& & \plan{127}\\
& \mkwm{s} & 0.16 & 0.16 & \textbf{0.13} & 0.16 & 6.91 & 6.91 \\[-0.4em]
& & \plan{67}& & \plan{127}& & \defaultplan{} \\
& \mkwm{x} & 0.16 & 0.16 & \textbf{0.13} & 0.16 & 0.12 & 0.12 \\[-0.4em]
& & \plan{67}& & \plan{127}& & \defaultplan{} \\
  \bottomrule
  \end{tabular}
  \label{tab:mh-comparison-results-full-1}
\end{table}

\begin{table}[H]
  \small
  \centering
  \caption{Best execution time by any satisfiable inference plan and the best execution time by the default inference plan for 90\% of the 100 executions to reach target accuracy, defined as the 90th percentile of error by the default inference plan using greatest number of particles evaluated that did not produce timeouts. Following \citet{atkinson2022semi} and \citet{baudart2020reactive}, reaching target accuracy is defined as $\log(P_{90\%}(\mit{loss})) - \log(\mit{loss}_{\mit{target}}) < 0.5$. -- indicates the default plan timeouts for all particle numbers. \textbf{Bold} marks when an alternative plan performs better than the default plan.}
  \begin{tabular}[h]{llrrrrrr}
    \toprule
    & & \multicolumn{2}{c}{\ssi} & \multicolumn{2}{c}{\ds} & \multicolumn{2}{c}{\bp} \\
     \cmidrule(lr){3-4} \cmidrule(lr){5-6} \cmidrule(lr){7-8}
    Benchmark & Variable & Best (s) & Default (s) & Best (s) & Default (s) & Best (s) & Default (s) \\
    \midrule
\multirow{7}{*}{\bRunner{}}& \mkwm{sx} & \textbf{0.24} & 0.26 & 0.2 & 0.2 & \textbf{0.18} & 0.22 \\[-0.4em]
& & \plan{7}& & \defaultplan{} & & \plan{15}\\
& \mkwm{sy} & \textbf{0.24} & 0.26 & 0.2 & 0.2 & 0.22 & 0.22 \\[-0.4em]
& & \plan{11}& & \defaultplan{} & & \defaultplan{} \\
& \mkwm{x} & 13.67 & 13.67 & 8.15 & 8.15 & \textbf{5.16} & 10.01 \\[-0.4em]
& & \defaultplan{} & & \defaultplan{} & & \plan{7}\\
& \mkwm{y} & \textbf{0.22} & 0.26 & 0.2 & 0.2 & \textbf{0.19} & 0.22 \\[-0.4em]
& & \plan{15}& & \defaultplan{} & & \plan{15}\\
\midrule
\multirow{3}{*}{\bWheels{}}& \mkwm{velocity} & 0.03 & 0.03 & \textbf{0.13} & 4.38 & 5.35 & 5.35 \\[-0.4em]
& & \defaultplan{} & & \plan{2}& & \defaultplan{} \\
& \mkwm{omega} & 0.03 & 0.03 & \textbf{0.15} & 8.7 & \textbf{0.14} & 10.55 \\[-0.4em]
& & \defaultplan{} & & \plan{1}& & \plan{1}\\
\midrule
\multirow{3}{*}{\bSlam{}}& \mkwm{map} & -- & -- & \textbf{0.32} & 3.67 & 1.94 & 1.94 \\[-0.4em]
& & & & \plan{1}& & \defaultplan{} \\
& \mkwm{x} & -- & -- & \textbf{0.32} & 7.05 & 1.13 & 1.13 \\[-0.4em]
& & & & \plan{1}& & \defaultplan{} \\
\midrule
\multirow{5}{*}{\bAircraft{}}& \mkwm{x} & \textbf{0.14} & 0.43 & 0.2 & 0.2 & 0.13 & 0.13 \\[-0.4em]
& & \plan{15}& & \defaultplan{} & & \defaultplan{} \\
& \mkwm{alt} & \textbf{0.14} & 0.2 & \textbf{0.13} & 0.2 & \textbf{0.12} & 0.13 \\[-0.4em]
& & \plan{15}& & \plan{31}& & \plan{31}\\
& \mkwm{q} & 4.24 & 4.24 & 16.36 & 16.36 & 1.82 & 1.82 \\[-0.4em]
& & \defaultplan{} & & \defaultplan{} & & \defaultplan{} \\
& \mkwm{r} & \textbf{0.14} & 0.2 & \textbf{0.13} & 0.2 & \textbf{0.12} & 0.13 \\[-0.4em]
& & \plan{15}& & \plan{31}& & \plan{31}\\
  \bottomrule
  \end{tabular}
  \label{tab:mh-comparison-results-full-2}
\end{table}

\subsection{Best Accuracy Given Runtime}

\begin{table}[H]
  \small
  \centering
  \caption{Ratio of the best accuracy achieved by any satisfiable inference plan with median runtime less than or equal to the median runtime of the default inference plan, aggregated across particle counts using geometric mean.  Accuracy is measured as the 90th percentile error. -- indicates the default plan timeouts for all particle counts. }
  \begin{tabular}[h]{llrrrrrr}
    \toprule
    & & \multicolumn{3}{c}{Gmean of Best Accuracy Ratio} \\
     \cmidrule(lr){3-5}
    Benchmark & Variable & \ssi & \ds & \bp \\
\midrule
\multirow{3}{*}{\bNoise{}}& \mkwm{x} & 167.82x & 1.28x & 1.01x \\
& \mkwm{q} & 396.85x & 6.79x & 1.04x \\
& \mkwm{r} & 412.56x & 5.57x & 1.04x \\
\midrule
\multirow{3}{*}{\bRadar{}}& \mkwm{x} & 7.79x & 1.21x & 1.09x \\
& \mkwm{q} & 1.9x & 2.12x & 2.49x \\
& \mkwm{r} & 3.2x & 2.47x & 1.3x \\
\midrule
\multirow{3}{*}{\bEnvnoise{}}& \mkwm{x} & 70.65x & 1.13x & 1.03x \\
& \mkwm{q} & 3.97x & 1.6x & 1.21x \\
& \mkwm{r} & 13.67x & 7.39x & 1.25x \\
\midrule
\multirow{2}{*}{\bOutlier{}}& \mkwm{outlier\_prob} & 1.44x & 1.17x & 1.04x \\
& \mkwm{xt} & 12383.78x & 1.21x & 1.02x \\
\midrule
\multirow{2}{*}{\bOutlierheavy{}}& \mkwm{outlier\_prob} & 4.9x & 5.9x & 1.03x \\
& \mkwm{xt} & 1.16x & 1.28x & 1.2x \\
\midrule
\multirow{2}{*}{\bGtree{}}& \mkwm{a} & 1.0x & 17.75x & 1.0x \\
& \mkwm{b} & 1.0x & 13.47x & 1.0x \\
\midrule
\multirow{6}{*}{\bSlds{}}& \mkwm{trans\_prob0} & 1.24x & 1.1x & 1.11x \\
& \mkwm{trans\_prob1} & 1.3x & 1.14x & 1.17x \\
& \mkwm{obs\_noise0} & 3.19x & 3.52x & 4.79x \\
& \mkwm{obs\_noise1} & 3.69x & 4.87x & 2.61x \\
& \mkwm{s} & 1.37x & 1.1x & 1.03x \\
& \mkwm{x} & 16.23x & 1.43x & 1.04x \\
\midrule
\multirow{4}{*}{\bRunner{}}& \mkwm{sx} & 1.04x & 1.06x & 3.46x \\
& \mkwm{sy} & 1.06x & 1.14x & 1.14x \\
& \mkwm{x} & 1.13x & 1.13x & 1.89x \\
& \mkwm{y} & 1.16x & 1.19x & 1.19x \\
\midrule
\multirow{2}{*}{\bWheels{}}& \mkwm{velocity} & 1.0x & 100.12x & 1.15x \\
& \mkwm{omega} & 1.0x & 413.01x & 583.57x \\
\midrule
\multirow{2}{*}{\bSlam{}}& \mkwm{map} & -- & 2.48x & 1.02x \\
& \mkwm{x} & -- & 4.78x & 1.08x \\
\midrule
\multirow{4}{*}{\bAircraft{}}& \mkwm{x} & 25.31x & 1.14x & 1.11x \\
& \mkwm{alt} & 1.15x & 1.11x & 1.07x \\
& \mkwm{q} & 1.44x & 1.72x & 1.59x \\
& \mkwm{r} & 5.65x & 3.33x & 9.08x \\
  \bottomrule
  \end{tabular}
  \label{tab:mh-comparison-results-best-acc-full}
\end{table}

\subsection{Analysis Evaluation}
We also evaluate the precision of the MH-\siren{} analysis using the same methodology as before. We enumerated all satisfiable and unsatisfiable inference plans, and measured if the analysis correctly determines the satisfiability of each plan. We also measure the time taken by the analysis averaged acorss all inference plans for each setting. The MH-\siren{} analysis is as precise as the \siren{} analysis. This is not surprising as the particle evaluation semantics are the same (except for the $\mfResample$ rule). 

\begin{table}[h]
  \small
  \centering
  \caption{Number of satisfiable inference plans identified by the inference plan satisfiability analysis out of the total number of satisfiable inference plans for each benchmark and algorithm. }
  \begin{tabular}[h]{lccccccccccc}
    \toprule
    & \multicolumn{11}{c}{Identified Satisfiable Plans / Satisfiable Plans}\\
    \cmidrule(lr){2-12}
    Algorithm & \bNoise{} & \bRadar{} & \emph{EnvN} & \emph{Out} & \emph{OutH} & \bGtree{} & \bSlds{} & \bRunner{} & \bWheels{} & \bSlam{} & \emph{Air} \\
\midrule
\ssi & 5/5 & 3/3 & 3/3 & 4/4 & 2/2 & 3/4 & 28/36 & 4/4 & 3/4 & 3/4 & 3/3 \\
\ds & 4/4 & 2/2 & 2/2 & 2/2 & 2/2 & 3/3 & 16/16 & 1/1 & 1/3 & 2/2 & 2/2 \\
\bp & 2/2 & 2/2 & 2/2 & 2/2 & 1/1 & 3/4 & 4/4 & 4/4 & 3/3 & 1/1 & 2/2 \\
\midrule
Total Plans & 8 & 32 & 32 & 8 & 8 & 4 & 128 & 16 & 4 & 4 & 32 \\
  \bottomrule
  \end{tabular}
  \label{tab:analysis-results-mh}
\end{table}

\begin{table}[h]
  \small
  \centering
  \caption{Time taken by the analysis averaged over all inference plans for each benchmark and algorithm in seconds. The analysis time for \bSlam{} with SSI is exceptionally long due to the intractable symbolic computation of full exact inference.}
  \begin{tabular}[h]{lccccccccccc}
    \toprule
    Algorithm & \bNoise{} & \bRadar{} & \emph{EnvN} & \emph{Out} & \emph{OutH} & \bGtree{} & \bSlds{} & \bRunner{} & \bWheels{} & \bSlam{} & \emph{Air} \\
\midrule
\ssi & 0.39 & 0.40 & 0.42 & 0.41 & 0.40 & 0.39 & 0.50 & 0.54 & 0.45 & 41.74 & 0.46 \\
\ds & 0.40 & 0.40 & 0.42 & 0.41 & 0.40 & 0.38 & 0.50 & 0.54 & 0.45 & 0.53 & 0.46 \\
\bp & 0.39 & 0.40 & 0.42 & 0.41 & 0.40 & 0.38 & 0.49 & 0.54 & 0.45 & 0.51 & 0.46 \\
  \bottomrule
  \end{tabular}
  \label{tab:analysis-time-mh}
\end{table}

\end{document}